\documentclass[twoside,openright]{report}
\pdfoutput=1 
\usepackage{amsmath,amssymb,array}
\usepackage{amsthm}
\usepackage{graphicx}    
\usepackage{pgf,pgfarrows,pgfautomata,pgfheaps,pgfnodes,pgfshade}
\usepackage{sfheaders}
\usepackage{fancyhdr}
\usepackage{titletoc}
\usepackage{mathrsfs}
\usepackage{pdfpages}
\newtheorem{theorem}{Theorem}[section]

\titlecontents{chapter}
[0pc] 
{}
{\\*\sffamily{ \slshape Chapter\ \thecontentslabel} - \large  \slshape  \sffamily     }
{\large \sffamily \slshape } 
{ \hfill \bfseries \contentspage}
[\addvspace{.5pc}]

\makeatletter
\def\cleardoublepage{\clearpage\if@twoside \ifodd\c@page\else
  \hbox{}
  \vspace*{\fill}
  \begin{center}
    
  \end{center}
  \vspace{\fill}
  \thispagestyle{empty}
  \newpage
  \if@twocolumn\hbox{}\newpage\fi\fi\fi}
\makeatother

\newcommand\dd{\mathrm{d}}
\newcommand\mbb{\mathbb}
\newcommand\mc{\mathcal}
\newcommand\mf{\mathfrak}

\newcommand\nn{\nonumber}

\newcommand\be{\begin{equation}}
\newcommand\ee{\end{equation}}

\newcommand\g{\mathrm{G}}
\newcommand\gp{\mathrm{G}^{+}}
\newcommand\gpp{\mathrm{G}^{++}}
\newcommand\gppp{\mathrm{G}^{+++}}
\newcommand\ag{\mathfrak{g}}
\newcommand\agp{\mathfrak{g}^{+}}
\newcommand\agpp{\mathfrak{g}^{++}}
\newcommand\agppp{\mathfrak{g}^{+++}}

\newcommand\e{\mathZ{E}_{8(8)}}

\newcommand\eppp{\mathZ{E}_{11(11)}}

\newcommand\aep{E_{9(9)}}
\newcommand\aepp{E_{10(10)}}
\newcommand\aeppp{E_{11(11)}}

\newcommand\su{\mathrm{SU}(2,1)}
\newcommand\supp{\mathrm{SU}(2,1)^{++}}
\newcommand\suppp{\mathrm{SU}(2,1)^{+++}}
\newcommand\asu{\mathfrak{su}(2,1)}
\newcommand\asupp{\mathfrak{su}(2,1)^{++}}
\newcommand\asuppp{\mathfrak{su}(2,1)^{+++}}

\newcommand\m{\mu}
\newcommand\n{\nu}
\newcommand\p{\partial}
\newcommand\s{\sigma}

\newcommand\z{\zeta}
\newcommand\eps{\epsilon}
\newcommand\cE{{\cal{E}}}

\newcommand\mmp{{\widehat B}}

\newcommand\cvv{\raise 2pt\hbox{,}}  \newcommand\cvp{\raise 2pt\hbox{.}}

\DeclareFontFamily{OT1}{pzc}{}
\DeclareFontShape{OT1}{pzc}{m}{it}%
             {<-> s * [1.30] pzcmi7t}{}
\DeclareMathAlphabet{\mathZ}{OT1}{pzc}%
                                 {m}{it}
\begin{document}
\begin{titlepage}


\includepdf[page=1]{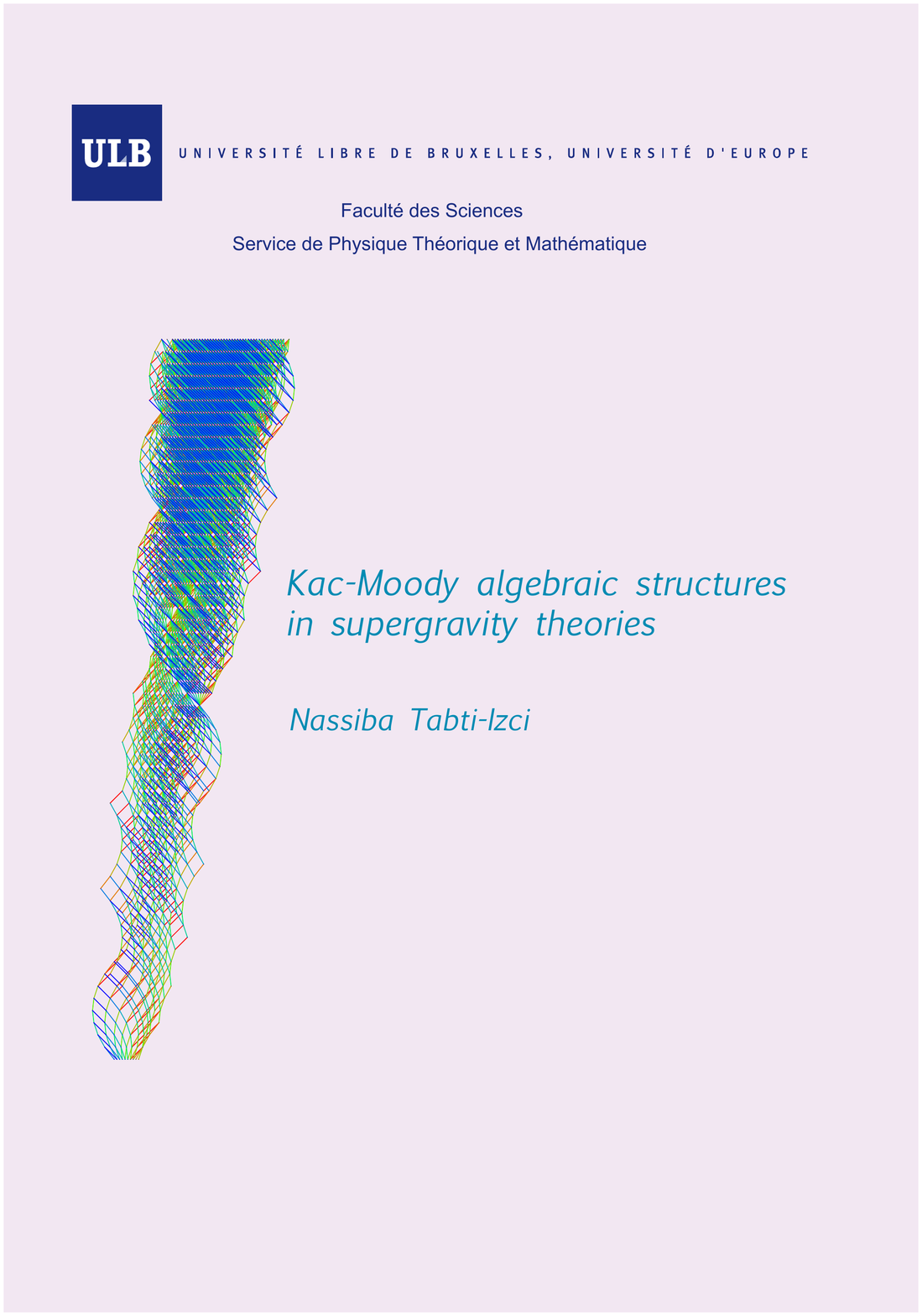}

\newpage

\thispagestyle{empty}
\mbox{}

\newpage

\thispagestyle{empty}
\mbox{}

\newpage

\thispagestyle{empty}
\mbox{} \\

\vspace{7cm}
\begin{flushleft}
\noindent The figure reproduced on the cover is the Hasse diagram \\
of the Kac-Moody algebra $E_{10}$.  This visualization of its\\
 root system is obtained by the program SimpLie created\\
  by Teake Nutma (http://code.google.com/p/simplie/).
\end{flushleft}

\newpage
\thispagestyle{empty}
\mbox{}

\begin{center} \vspace{2cm}
\hrule \vspace{0.5cm}
{\Huge{ \slshape \sffamily Kac-Moody algebraic structures}\\
\vspace{0.4cm}
\slshape \sffamily in supergravity theories}\\
\vspace{0.5cm} \hrule
\vspace{1cm}
\slshape \sffamily  {\LARGE Nassiba Tabti-Izci}\\ \vspace{0.3cm}
\slshape \sffamily  Aspirant F.R.I.A.\\
\vspace{3cm}
\slshape \sffamily   \large Th\`ese pr\'esent\'ee en vue de l'obtention du grade de Docteur en Sciences\\
\slshape \sffamily  Directeur de th\`ese: Laurent Houart\\ \vspace{3cm}  \slshape \sffamily Ann\'ee acad\'emique 2009-2010
\end{center}

\newpage
\thispagestyle{empty}
\mbox{}

\end{titlepage}
\pagenumbering{Roman}
\thispagestyle{plain}
 \begin{center}
\sffamily{ \textsl{ \Large{Kac-Moody algebraic structures in supergravity theories}}}
\vspace{0.5cm}

 \textit{Nassiba Tabti-Izci}\\ 
 \vspace{0.5cm}
 \textit{\large{Abstract}}
 \end{center}
 A lot of developments made during the last years show that Kac-Moody algebras play an important role  in the algebraic structure of some supergravity theories. These algebras would generate infinite-dimensional symmetry groups. The possible existence of such symmetries have motivated the reformulation of these theories as non-linear $\sigma$-models based on the Kac-Moody symmetry groups. Such models are constructed in terms of an infinite number of fields parametrizing the generators of the corresponding algebra. If these conjectured symmetries are indeed actual symmetries of certain  supergravity theories, a meaningful question to elucidate will be the interpretation of this infinite tower of fields. Another substantial problem is to find  the correspondence between the $\sigma$-models, which are explicitly invariant under the conjectured symmetries, and these corresponding space-time theories. The subject of this thesis is to address these questions in certain cases. \\
 
 This dissertation is divided in three parts.
 
 In Part I, we first review the mathematical background on Kac-Moody algebras required to understand the results of this thesis. We then describe the investigations of the underlying symmetry structure of supergravity theories.
 
 In Part II, we focus  on the bosonic sector of eleven-dimensional supergravity which would be invariant under the extended symmetry $\mathZ{E}_{11}$. We study the algebra $E_{10} \subset E_{11}$  and more precisely the real roots of its affine subalgebra $E_9$. For each positive real roots of $E_9$ we obtain a BPS solution of eleven-dimensional supergravity or of its exotic counterparts.  All these solutions are related by U-dualities which are realized via $E_9$ Weyl transformations.
 
 In Part III, we study the symmetries of pure $\mathcal N=2$ supergravity in $D=4$. As is known, the dimensional reduction of this  model  with one Killing vector is characterized by a non-linearly realized symmetry $\su$. We consider the BPS brane solutions of this theory preserving half of the supersymmetry and the action of $\su$ on them. Infinite-dimensional symmetries are also studied and we provide evidence that the theory exhibits an underlying algebraic structure described by the Lorentzian Kac-Mody group $\suppp$. This evidence arises  from the correspondence between the bosonic space-time fields of $\mathcal{N}=2$ supergravity in $D=4$ and a one-parameter sigma-model based on the hyperbolic group $\supp$. It also follows  from the structure of BPS brane solutions which  is neatly encoded in $\suppp$. As a worthy by-product of our analysis, we obtain a regular embedding of $\asuppp$ in $E_{11}$ based on brane physics.
 
 \newpage

\thispagestyle{empty}
\mbox{}
  \newpage

{\slshape Tout a commenc\'e un mois de septembre 2004 alors que je cherchais un sujet pour mon m\'emoire de licence dans les couloirs de l'aile O du sixi\`eme \'etage du b\^atiment  NO. Je n'avais encore aucune id\'ee de ce que je voulais \'etudier mais une chose \'etait s\^ure : je voulais  faire mon m\'emoire dans le service de Physique Th\'eorique et Math\'ematique.  C'est alors que Laurent Houart m'a prise sous son aile et m'a propos\'e de travailler sur les sym\'etries cach\'ees dans les th\'eories de gravitation coupl\'ee \`a la mati\`ere. Il me raconta alors, avec beaucoup de motivation et d'\'energie la fabuleuse histoire de ces sym\'etries \`a travers des diagrammes de Dynkin. Je ne vous cacherai pas que  ses premiers tableaux d\'ebordant de craie (et parfois sans aucune structure) m'ont fait un peu peur au d\'ebut. Mais au final, je me suis retrouv\'ee moi-m\^eme  \`a jouer avec des alg\`ebres de Kac-Moody. Je le remercie de m'avoir embarqu\'ee dans cette grande aventure qui \'etait en quelque sorte \'egalement nouvelle pour lui car j'\'etais sa premi\`ere doctorante.  Sans son soutien, ses pr\'ecieux encouragements et son humour d\'ebordant tout au long de ces cinq ann\'ees, je ne serais pas arriv\'ee \`a ce r\'esultat.  Je lui exprime toute ma gratitude pour son encadrement, sa disponibilit\'e et pour m'avoir fait travailler sur des sujets passionants.

Ce fut un immense honneur pour moi de travailler avec Fran\c{c}ois Englert durant les deux premi\`eres ann\'ees de ma th\`ese. Toute ma gratitude pour sa grande disponibilit\'e et sa gentillesse. Je remercie \'egalement Axel Kleinschmidt  qui a rejoint pour mon plus grand bonheur notre \'equipe de Bruxelles. Je le remercie pour nos diff\'erentes collaborations et pour avoir eu la patience de r\'epondre \`a mes  petites questions. Merci \'egalement aux deux Su\'edois hyper motiv\'es et passionn\'es  qui ont envahi le deuxi\`eme \'etage: Daniel Persson (thank you to have learned me the very useful Swedish word : acid acetylsal \dots suuura) et Josef Lindman H\"ornlund (I am very sorry that Noether is written with ÔoeÕ and not with your swedish  \"o)  pour notre derni\`ere fructueuse collaboration.  J'ai \'egalement eu la chance de collaborer \`a distance avec Hermann Nicolai. Je ne peux oublier  ma `grande soeur de Kac-Moody', Sophie de Buyl. Je la remercie pour ses encouragements et pour son grand soutien.  Malgr\'e son \'eloignement physique, ses corrections ont survol\'e l'oc\'ean Atlantique pour atterrir dans ma boite e-mail. Je la remercie pour sa lecture et sa re-lecture attentive de cette th\`ese.

Je voudrais  exprimer mes sentiments sinc\`eres de reconnaissance \`a la personne qui a du constamment me supporter durant ces quatre ann\'ees (euhh 8 ann\'ees) : Vincent Wens (Fasaa). J'ai eu l'immense chance de partager le m\^eme bureau que lui. Je le remercie pour son  aide si  pr\'ecieuse dans tous les domaines de la physique et des math\'ematiques, pour la clart\'e de ses explications, pour sa contribution durant la r\'edaction de cette th\`ese, pour l'int\'egral des Inconnus et pour son soutien dans les moments un peu plus stressants.

Quelle fut la chance pour moi de travailler dans le service tr\`es dynamique de Physique th\'eorique et math\'ematique. Je remercie les membres de notre service que j'ai pu croiser dans les couloirs du 7\`eme, du 6\`eme et du 2\`eme \'etage : Riccardo Argurio, Glenn Barnich, Frank Ferrari, Marc Henneaux, Axel Kleinschmidt et Christiane Schomblond. Je remercie \'egalement les secr\'etaires Fabienne de Neyn (pour ses encouragements et pour nos petites discussions entre deux couloirs), Isabelle Juif et Marie-France Rogge.\\

Je remercie Fran\c{c}ois Englert, Thomas Hambye, Marc Henneaux, Bernard Julia et Axel Kleinschmidt pour avoir accept\'e de faire partie de mon jury de th\`ese.\\

Je ne peux m'emp\^echer d'\'evoquer ici l'ambiance excellente dans laquelle j'ai pass\'e ces quatre ann\'ees. Je remercie pour cela tous les doctorants et post-doc avec qui jÕai pass\'e de si bons moments: Michael Bebronne, Pierre de Buyl (merci pour m'avoir sauv\'e plusieurs fois en informatique), Cyril Closset,  Fran\c{c}ois Jacques Albert Dehouck (libanais, turc, grec ?  la r\'eponse est turc), Julie Delvax, Jonathan Demaeyer, Cedric Troessaert, ainsi que Jean Sabin Mc Ewen.
Un merci particulier \`a Ella Jamsin, ma `petite soeur de Kac-Moody', pour sa lecture de quelques parties de ma th\`ese et pour avoir partag\'e avec moi la pr\'eparation de travaux pratiques. 

Je remercie \'egalement les ex-doctorants. Je pense d'abord \`a  Docteur Nathan Goldman (gr\^ace \`a lui  il n'y a plus que deux personnes entre moi et le pr\'esident des Etats-Unis). Je le remercie pour les petits breaks qui m'ont fait le plus grand bien, pour son soutien dans mes moments de stress intense et surtout pour m'avoir appris \`a compter les temps en musique.  Je remercie \'egalement les anciens et nouveaux doctorants et post-doctorants qui ont crois\'e ma route: Nazim Bouatta, Sandrine Cnockaert, Geoffrey Comp\`ere, St\'ephane Detournay, Laura Lopez-Honorez (ma marraine de 1\`ere candi), Luca Forte (I need the apple pie), Sergio Montanez Naz, Claire No\"el, Amitabh Virmani (thank you four very interesting discussions about India).\\

Pour m'avoir donn\'e  go\^ut aux Sciences Physique, je remercie mon professeur de physique du secondaire Bernadette Lambert. Je remercie \'egalement toutes mes amies qui m'ont toujours soutenue durant ces derni\`eres ann\'ees. Je ne vous en veux pas d'avoir confondu constamment les mots m\'emoire et th\`ese.  Mais maintenant il faudra m'appeler Docteur. Je remercie du fond du coeur: Noussaiba, Latifus, Fadila, Rabab, Nassiba, Karimus, Sohad, Soumaya, Fatima, Khadigea, Aya, Safiya, Anissa, Sarah, Rukich, Safa et Yasemine.  Je voudrais \'egalement remercier mes `khaltous' bien aim\'ees pour leurs encouragements : les Aichettes, Zazat, Khadra, Houria, Ghogho, Mariam, Salima, Meissa,  Amina, Mimi, Fatima Zohra, Selma, Rahma, Mehdia, Fadila et Khadigea. Un grand merci \'egalement \`a Salimus pour sa contribution non n\'egligeable \`a ce manuscrit.

Pour avoir jou\'e les clowns et pour m'avoir montrer que la vie est coooooool et qu'il ne faut jamais s'en faire, je remercie mes deux petits fr\`eres Zaid et Moshab. Je remercie \'egalement tous ceux qui m'ont apport\'e leur soutien dans ma famille : les OmarÕs, Latifa, Nourya, Abdel Hamid, Soumaya et Houria. Je n'oublie pas non plus ma famille qui se trouve en Alg\'erie. Je voudrais \'egalement remercier Dede pour tous ses encouragements et pour tout l'amour et la gentillesse qu'il d\'everse sur moi. Pour les bons petits plats et les `douas',  je remercie \'egalement ma belle-m\`ere Ane ainsi que toute ma belle famille. \\

Je voudrais avoir une petite pens\'ee pour mes grands-parents que j'ai perdus pendant mon doctorat : mes grand-m\`eres Hana Zolikha et Ebe et mon grand-p\`ere Baba Hbibi. Sans cesse, ils m'ont encourag\'ee \`a poursuivre le plus loin possible mes \'etudes. Vous me manquez \'enorm\'ement.\\

Je terminerai cette liste de remerciements par les trois personnes qui m'ont le plus soutenue et qui ont \'et\'e le plus directement touch\'ees par mes \'emotions parfois un peu d\'ebordantes et par mes diff\'erents caprices : mes parents et mon \'epoux. Sans vous, ce travail n'aurait pas pu voir le jour. Cette th\`ese vous est d\'edi\'ee.\\

Oumi, je te remercie pour ta douceur, pour avoir pris soin de moi de la meilleure des mani\`eres, pour tes bons petits plats, pour tes corrections orthographiques, pour tes encouragements et pour avoir \'et\'e la meilleure des mamans. 
 
Abi, je te remercie pour tes `douas',  pour ton attention et pour ta confiance en mon potentiel.
 
Je remercie enfin, mon Emir, mon askim pour s'\^etre occup\'e de tout depuis que je r\'edige cette th\`ese. Je te remercie pour ton soutien sans faille et pour tes bras toujours grand ouverts. }

\newpage

\thispagestyle{empty}
\mbox{}
\newpage
\thispagestyle{empty}
\mbox{}\\
 \vspace{4cm}

\begin{flushright}
\sffamily {\slshape \large A mes parents Oumi et Abi,\\ \vspace{0.3cm}
A mon Emir,}
\end{flushright}

\tableofcontents

\chapter*{Introduction and discussion \markboth{Introduction and discussion}{Introduction and discussion}}
\addcontentsline{toc}{chapter}{Introduction and discussion}

\hrule
\vspace{2cm}
 \pagenumbering{arabic}
 
\pagestyle{fancy}
\renewcommand{\chaptermark}[1]{\markboth{\chaptername\   \thechapter \ -    #1}{}}
\renewcommand{\sectionmark}[1]{\markright{\thesection\ #1}}
\fancyhf{}
\fancyhead[LE,RO]{\slshape  \thepage}
\fancyhead[LO]{\slshape \sffamily \rightmark}
\fancyhead[RE]{\slshape \sffamily  \leftmark}
\renewcommand{\headrulewidth}{0.5pt}
\renewcommand{\footrulewidth}{0pt}
\addtolength{\headheight}{0.5pt}
\fancypagestyle{plain} { 
\fancyhead{} 
\renewcommand{\headrulewidth}{0pt}
}
\begin{flushright}
\emph{It is only slightly overstating the case to say\\ that physics is the study of symmetry.}\vspace{0.5cm}

Philip Warren Anderson \\
\footnotesize{Nobel Prize in Physics}
\end{flushright}

Major advances in the understanding of the physical world have been achieved by considering symmetries. Let us focus first on the word \emph{symmetry}. It comes from the Greek: $\sigma \upsilon \mu \mu \epsilon \tau \rho \iota \alpha $  which means `with measure'. It has two meanings. The first one is not really  precise and defines the   correct and pleasing proportion of the parts of a thing such that it reflects beauty or perfection. The second one is much less subjective: it is the concept of balance which can be demonstrated with the rigor of mathematics or physics. More precisely  in the mathematical context, symmetries are usually associated  with transformations that leave some properties of an object
 invariant. The set of such transformations forms a group. In physics, the meaning of symmetry is generalized to designate the invariance of physical observables under some transformations. The concept that  the laws of nature originate in symmetries plays a crucial role in theoretical physics and it constitutes  one of the most powerful tools for the  construction and study of a physical theory. 
 
 The exploitation of symmetry groups has accompanied the most successful theories of modern physics. On the one hand,  the Standard model of particle physics is a quantum field theory invariant under the gauge symmetry group $\mathrm{SU}(3) \times \mathrm{SU}(2) \times \mathrm{U}(1)$.
 On the other hand,   Einstein's general relativity lies on the principle of  general covariance which stipulates the independence of the physical laws with respect to the coordinate transformations.  
 However  gravity stands apart from the other three fundamental interactions namely the electromagnetism, the strong and the weak forces which are described by the Standard model. Indeed  the attempts to quantize Einstein's theory following the standard rules of quantum mechanics is problematic because of its non-renormalizability. But our understanding of the fundamental laws of  nature is surely incomplete until general relativity and quantum mechanics are successfully reconciled and unified. Moreover,  some extreme situations in physics such as the first moments of the Universe  
 require  gravity in its quantum regime. New ways must therefore be explored to find the theory which will reconcile general relativity and quantum mechanics. 
 
  
By replacing point-like particles by one-dimensional extended objects, known as strings,  whose specific oscillations modes describe the elementary particles, String theory overcomes the problem of perturbative non-renormalizability. This theory, which arose at the beginning   in an attempt to understand the strong nuclear force, is now devoted to a much more ambitious purpose: the construction of  a quantum theory that unifies the four fundamental forces of nature. 
 The inclusion of fermions in this theory introduces a further symmetry, called supersymmetry, that relates bosons to fermions. 
 There exist  five supersymmetric string theories, each requiring  10 dimensions: type I, type IIA, type IIB,  heterotic $\mathrm{SO}(32)$ and heterotic  $E_8 \times E_8$. The realization that there are five different superstring  theories was somewhat puzzling in the quest of only one theory of unification. The surprising discovering of new kind of transformations, known as dualities,  led to the following statement: what we viewed previously as five theories are in fact five different perturbative limits of a single underlying theory. Two kinds of dualities called T and  S-duality have been identified. 
 
 In the late 1980s it was realized that there is a property known as T-duality that relates the two type II and the two heterotic theories, so  they should not really be regarded as distinct theories. T-duality can be illustrated by considering a compact dimension consisting of a circle of radius $R$.  In this case there are two kinds of excitations to consider. The first one is the Kaluza-Klein momentum excitation on the circle which contributes by  $(n/R)^2$ to the energy squared, where $n$ is an integer. The second one is  the winding-mode excitations due to a closed string winding $m$ times around the circle. T-duality exchanges these two kinds of excitations by the mapping $m\leftrightarrow n$ and $R \leftrightarrow \ell_s^2/R$, where $\ell_s$ is the string length.   
 
Another kind of duality called S-duality was discovered in the mid 1990s. S-duality relates the string couplings constant $g_s$ to $1/g_s$. The two basic examples relate the type I superstring theory to the $\mathrm{SO}(32)$ heterotic string and the type IIB to itself. Therefore with S-duality, a strongly coupled  theory is equivalent to a weakly one. 
The combination of S-duality with T-duality transformations led to  the U-duality group which unifies all the discrete symmetry groups of string theories.

 S-duality explains how three of the five original superstring theories behave  at strong coupling. The understanding of how the other  type IIA and heterotic $E_8\times E_8$  string theories behave when $g_s$ is large came as quite a surprise. In each of these cases there is an eleventh dimension that becomes large at strong coupling and  the elusive eleven-dimensional description of the underlying theory is called M-theory. The relations between M-theory and the two superstring theories, together with T and S dualities, imply a web of dualities connecting the five superstring theories.  We do not have yet a satisfactory description of M-theory but we know that at low energies it is approximated by a classical field theory called eleven-dimensional supergravity.   

Essential elements are lacking in the quest of the unified theory of all fundamental interactions. It is therefore  of interest to inquire into the symmetries which would underlie the M-theory project, using as a guide symmetries rooted in its classical low energy limit, namely eleven-dimensional supergravity.  In this context, the study and the exploitation of hidden symmetries, exhibited by dimensional reduction, could be very useful and would allow a better understanding of the structure of the unified theory. These hopes have been encouraged  by some developments from recent years that certain types of Kac-Moody algebras (which are roughly speaking  infinite-dimensional generalizations of Lie algebras) occur in the algebraic structure of some supergravity theories.  We will now describe how these kind of algebras have made their way in the big project of the quest of  symmetries of eleven-dimensional supergravity and also of other supergravity theories and M-theory.



\section*{ Some history: from a conjecture to some checks}
\addcontentsline{toc}{section}{Some history: from a conjecture to some checks}

Many supergravity theories exhibit continuous global symmetries. These are of importance both for the generation of solutions of the field equations and also for the study of quantization. The symmetries are either inherent in the formulation of the theory, as in the case of type IIB supergravity in $D=10$ which admits an $\mathrm{SL}(2,\mathbb{R})$ symmetry, or act on certain subclasses of solutions admitting Killing vectors. This was first noticed by Ehlers in 1957 in the case of non-supersymmetric $D=4$ gravity with one Killing vector  \cite{Ehlers:1957zz} where it is also  $\mathrm{SL}(2,\mathbb{R})$ that acts on the set of solutions. In 1970, Geroch is the first who discovered infinite-dimensional symmetries appearing in gravity \cite{Geroch:1970nt,Geroch:1972yt}. He demonstrated the existence of an infinite-dimensional symmetry, known as the Geroch group which acts on solutions of Einstein's equations with two commuting Killing vectors (axisymmetric stationary solutions). The group structure was later shown to be  the affine extension of the Ehlers $\mathrm{SL}(2, \mbb R)$ symmetry group \cite{Julia:1980gr,Julia:1982gx,Breitenlohner:1986um}. 

The story received new impetus with the discovery at the end of the seventies by Cremmer and Julia  of hidden symmetries in supergravities. The most famous being ${\cal N}=8$ supergravity in $D=4$ possessing an exceptional $\mathZ E_{7(7)}$ \footnote{The notation $E_{n(n)}$ is used to denote the split real form of the complex exceptional group $E_n$.  In this case, the associated group $\mathZ E_{n(n)}$ is a maximally non-compact group (for more details see Section \ref{sec:realkm}). } symmetry~\cite{Cremmer:1978ds,Cremmer:1979up}. This symmetry can also be viewed as a symmetry of the solutions of ${\cal N}=1$ supergravity in $D=11$ admitting seven commuting space-like Killing vectors, in agreement with the construction of the ${\cal N}=8$ theory in $D=4$ by dimensional reduction of the $D=11$ theory on a seven-torus $T^7$.  Therefore the toroidal compactification of eleven-dimensional supergravity to $11-n$ dimensions reveals a chain of global symmetries described by the Lie groups $\mathZ{E}_{n(n)}$, culminating with the largest exceptional simple Lie group $\mathZ{E}_{8(8)}$ in three dimensions \cite{Marcus:1983hb}. Reduction to dimensions below three is more complicated but it was conjectured by Julia in \cite{Julia:1980gr,Julia:1982gx} that the affine symmetry group $\mathZ{E}_{9(9)}$ (which can be seen as a generalization of the Geroch group) appears in two dimensions, the hyperbolic group $\mathZ{E}_{10(10)}$ in one dimension and the Lorentzian group $\mathZ{E}_{11(11)}$ in zero dimension.  Following the arguments of \cite{Hull:1994ys}, it is expected that the U-duality group of M-theory on the torus $T^n$ is the discrete group $\mathZ{E}_{n(n)}(\mbb Z)$ \cite{Obers:1998rn}.

Subsequently, part of this conjecture was verified through the discovery of a global $\mathZ E_{9(9)}$ symmetry of the space of solutions of $\mathcal N=16$ supergravity in two dimensions \cite{Nicolai:1987kz,Nicolai:1987vy,Nicolai:1988jb}.   These results also provided a direct link to the integrability of these theory in the reduction to two dimensions.

This lead to the question: these conjectured Kac-Moody symmetries are they underlying symmetries of supergravity and M-theory?\\
 
  This idea has received new impetus in a modified form over the last years.
 A new approach was developed by West in 2000 to answer this question.  First, he showed using old results which formulated gravity as non-linear realisation \cite{Ogievetsky:1973ik,Borisov:1974bn} that the entire bosonic sector of eleven-dimensional supergravity and the ten-dimensional II supergravity theory could be formulated as a non-linear realisation \cite{West:2000ga}. In \cite{West:2001as}, he attempted to reformulate eleven-dimensional supergravity as a non-linear realization based on the Lorentzian group $\boldsymbol{\mathZ{E}_{11(11)}}$ (jointly with the conformal group). This led him to conjecture that the full group $\mathZ{E}_{11(11)}$ should be a symmetry of M-theory itself.  This proposal was further elaborated  in \cite{West:2003fc,Kleinschmidt:2003jf,West:2004iz,West:2004wk} according to which it is $\mathZ{E}_{11(11)}$ that should be viewed as the fundamental symmetry. Note that  type IIA  and type IIB supergravities have  also been reformulated as non-linear realisation based on $\mathZ{E}_{11(11)}$ in \cite{Schnakenburg:2001ya,Schnakenburg:2002xx,West:2004st}.\\  

The study of gravitational theories close to a space-like singularity has provided further evidence for the existence of infinite algebraic symmetry structures. It has been found that in this limit, known as BKL (Belinskii-Khalatnikov-Lifshitz) limit,  the dynamics of eleven-dimensional supergravity can be described as a billiard motion in a region of hyperbolic space bounded by hyperplanes which turn out to be the fundamental Weyl chamber of the Kac-Moody algebra $E_{10(10)}$ \cite{Damour:2000hv} (see References \cite{Damour:2002et} and \cite{Henneaux:2007ej} for reviews on this subject).  
 This discovery led Damour, Henneaux and Nicolai to conjecture that the full group $\mathZ{E}_{10(10)}$ should be realized as a symmetry of M-theory \cite{Damour:2002cu}. To explore the possible fundamental significance of these huge
symmetries, they proposed in 2002 a Lagrangian formulation 
  explicitly invariant under $\mathZ{E}_{10(10)}$ \cite{Damour:2002cu}. It was
constructed as a reparametrisation invariant $\sigma$-model of fields
depending on one parameter $t$, identified as a time parameter, living
on the coset space $\mathZ{E}_{10(10)}/\mathrm{K}_{10}^+$ where  $\mathrm{K}_{10}^+$ is the maximal compact subgroup of $\mathZ E_{10(10)}$ invariant  under the Chevalley involution.  This
$\sigma$-model contains an infinite number of fields and is built in a 
recursive way  by the introduction of a gradation called  level decomposition of the adjoint representation of $E_{10(10)}$ with respect to its
subalgebra $A_9$ \cite{Damour:2002cu, Nicolai:2003fw}. 
This level decomposition has the advantage to slice the adjoint representation of $E_{10(10)}$ such that at each level a finite number of representations of the subalgebra $A_9$ appear \footnote{The level of an
irreducible representation of $A_9$ occurring in the decomposition of
the adjoint  representation of $E_{10(10)}$  counts the number of  times
the simple root $\alpha_{11}$ not pertaining to the $A_9$ subalgebra
appears in 
the decomposition (see Figure \ref{fig:kmintro}).}.
The $\mathZ{E}_{10(10)}/\mathrm{K}_{10}^+$ $\sigma$-model, limited to the roots up to
level 3 and height 29, reveals  a perfect  match with the bosonic
equations of motion of eleven-dimensional supergravity in the vicinity of
 spacelike singularity, where fields depend
only on time.  It was therefore possible to establish a  dictionary which relates the truncated fields of the coset  $\mathZ{E}_{10(10)}/\mathrm{K}_{10}^+$ and the degrees of freedom of eleven-dimensional supergravity. This correspondence  was later extended to massive type IIA and type IIB supergravity theories in \cite{Kleinschmidt:2004dy,Kleinschmidt:2004rg,Henneaux:2009ee}. \\


Coset symmetries do not only occur  in  the dimensional reduction of
11-dimensional supergravity \cite{Cremmer:1978km}  but  appeared also in other
theories. They have been the subject of much study, and some classic
examples are given in \cite{Ferrara:1976iq, Cremmer:1977tt, Cremmer:1978ds,Julia:1980gr, Julia:1981wc,Schwarz:1983wa}. More precisely,  the toroidal dimensional reduction of supergravities give rise in three dimensions to theories whose bosonic sectors are described purely in terms of scalar degrees of freedom, which parametrise $\sigma$-model coset spaces. The reduction on purely spacelike torus give rise to a coset space $\mathrm{G}/\mathrm{K}$ where $\mathrm K$ is the maximal compact subgroup of $\mathrm{G}$. For instance, as explain before, the reduction of eleven-dimensional supergravity to three dimensions gives rise to an $\mathZ E_{8(8)}/ \mathrm{SO}(16)$ coset Lagrangian.  We can now try to find all the theories containing gravity suitably coupled to forms and dilatons which may exhibit upon dimensional reduction down to $3$ dimensions a coset structure $\mathrm{G}/\mathrm K$.  This  leads to the procedure of oxidation \cite{Julia:1980gr,Cremmer:1999du} which is roughly speaking the inverse of the dimensional reduction. In this context, we call \emph{maximally oxidised theory}  a Lagrangian theory defined in the highest possible space-time dimension $D$ (namely which is itself not obtained by dimensional reduction) whose dimensional reduction down to three dimensions exhibits a coset $\mathrm{G}/\mathrm K$.
 
  The maximally oxidised theories have been listed for each coset space $\mathrm G/\mathrm K$ in \cite{Cremmer:1999du,Keurentjes:2002vx,Keurentjes:2002xc} when $\g$ is a maximally non-compact simple Lie group (the associated Lie algebra $\mf g$ is a split real form) and the extension to compact groups was discussed (when $\mf g$ is a non-split real form) in \cite{Keurentjes:2002rc}. In particular, we get that  the reduced Lagrangian in $3$ dimensions is invariant under transformations of $\mathrm{SL}
(D-2, \mbb R)/ \mathrm{O}(D-2)$ for  gravity in $D$ dimensions. Note that all maximally oxidised theories do not have necessarily supersymmetry extension such as the low energy effective action of the bosonic string in $26$ dimensions without tachyon  which exhibits the coset $\mathrm{O}(24,24)/(\mathrm O(24)\times \mathrm O(24))$ upon dimensional reduction down to three dimensions.\\

 
 \begin{figure}[h]
\begin{center}
{\scalebox{0.17}
{\includegraphics{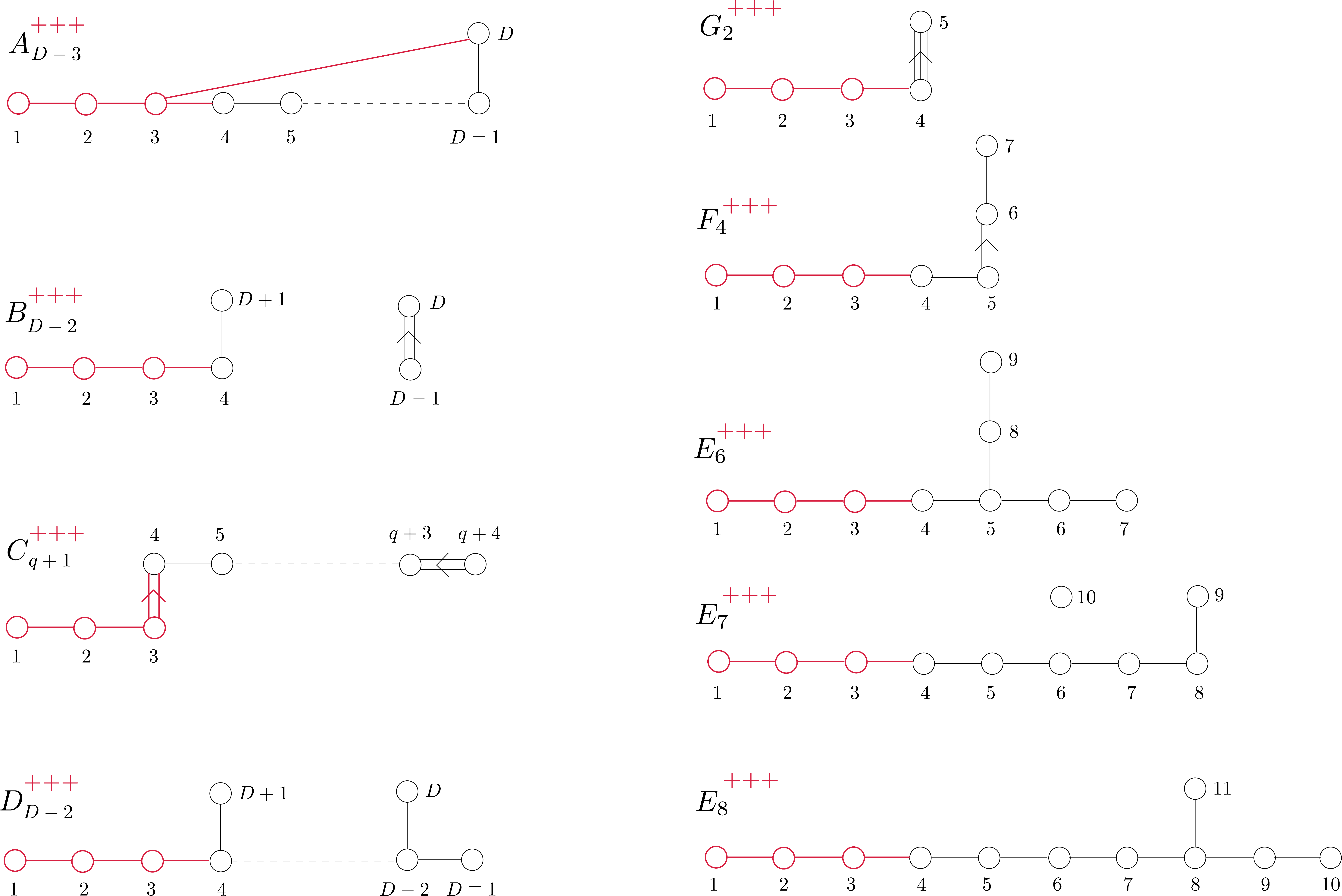}}}
\caption{\small \textsl{Dynkin diagrams of the Kac-Moody algebras $\agppp$ obtained by extension of the finite simple Lie algebras $\ag$ with the nodes labelled by $1,2,3$. The horizontal line starting at $1$ defines the gravity line which is the Dynkin diagram of a $A_{D-1}$ subalgebra. }}
\label{fig:kmintro}
\end{center}
\end{figure}

 The conjecture discussed above on the presence of infinite-dimensional symmetries underlying eleven-dimensional supergravity has been generalized for other gravitational theories and it has been suggested that such actions possess a much larger symmetry than the one
revealed by their dimensional reduction to three space-time dimensions in which all fields, except  $(2+1)$-dimensional gravity itself,  are scalars. The first generalizations were done for the effective action of the bosonic string in \cite{West:2001as} and  for pure D-dimensional gravity where the proposed underlying symmetry algebras is thought to be $A_{D-3}^{+++}= \mathfrak{sl}(D-2)^{+++}$ \cite{Lambert:2001gk}. There is a general construction for extending the symmetry group $\mathrm G$ to an affine group $\mathrm G^+$, a hyperbolic or Lorentzian group $\mathrm G^{++}$ and a Lorentzian group $\mathrm G^{+++}$~\cite{Gaberdiel:2002db} and in the discussion above $\mathZ E_{10(10)}$ and $\mathZ E_{11(11)}$ have to be replaced by $\mathrm G^{++}$ and $\mathrm G^{+++}$ to obtain a set of more general conjectures for a wider class of theories. The corresponding extended algebras $\agp$, $\agpp$ and $\agppp$ are obtained respectively by adding one, two and three nodes  to the Dynkin diagram of the corresponding finite $\ag$ algebra (see Figure \ref{fig:kmintro}).  It has been suggested in \cite{Englert:2003zs} that all maximally oxidised theories (or possibly some unknown
extensions of them) which generate in 3 dimensions a coset $\mathrm G/\mathrm K$ where $\mathrm G$ is in its non-compact form  possess  the very-extended symmetry $\gppp$.  It has also been verified that for all simple $\mathrm G$ the extended symmetry group contains the correct field content to make the conjecture work~\cite{Kleinschmidt:2003mf} but a full dynamical confirmation of the conjectures is still an open problem. 

 The work on cosmological billiards described  above was also extended to the oxidised theories for any  $\mathrm G$  and it was found that the one-dimensional motion takes place in the Weyl chamber of the overextend algebra  $\mathfrak{g}^{++}$ if $\ag$ is a split real form   \cite{Damour:2001sa,Damour:2002fz}. When the algebra $\ag$ is non-split, it is  the Weyl chamber of  its maximal split algebra whose controlled the billiard  \cite{HenneauxJulia}. Note that  in the reference \cite{deBuyl:2003ub}, the billiard approach gives useful informations on the procedure of oxidation of both split and non-split group.                                    


A hybrid approach for uncovering the symmetries of M-theory or other supergravity theories combining the ideas in   \cite{West:2001as} and \cite{Damour:2002cu} has been adopted in \cite{Englert:2003py}.  In this reference, an action explicitly invariant under $\gppp$ is constructed. This action $S_{\gppp}$ is defined in a reparametrisation invariant way on a world-line, a priori
unrelated to space-time, in terms of a infinity of fields $\phi(\xi)$ living in
a coset $\gppp/\mathrm K^{*+++}$ where $\xi$ spans the world-line. The subgroup $\mathrm K^{*+++}$  is invariant under a modified Chevalley involution known as \emph{temporal involution} 
which ensures that the action is $\mathrm{SO}(1,D-1)$ invariant.  This temporal involution allows the identification of the index 1 as a time coordinate.  The action $S_{\gppp}$ is built by the use of a level
decomposition\footnote{ Level decomposition of the very extended $\agppp$ algebras in terms of the
  subalgebra $A_{D-1}$ have been considered in
  \cite{West:2002jj,Nicolai:2003fw,Kleinschmidt:2003mf}.}  of $\agppp$ with respect to the subalgebra $A_{D-1}= \mathfrak{sl}(D,\mbb R)$  where $D$ is identified to the space-time dimension.



To make connection between this new formalism and the covariant space-time theories, it is interesting to analyze several actions invariant under overextended $\gpp$ group. The $\gpp$ content of the $\gppp$ -invariant actions $S_{\gppp}$ has been analysed in reference \cite{Englert:2004ph} where it was
shown that two distinct actions invariant under the overextended
Kac-Moody algebra $\gpp$ exist. The first one $S_{\gpp_C}$ is
constructed from $S_{\gppp}$ by performing a truncation putting
consistently to zero some fields. The corresponding $\gpp$  algebra
is obtained from $\gppp$ by deleting the node labelled 1 from the
Dynkin diagram of $\gppp$ depicted in Figure \ref{fig:kmintro}. This theory carries a
Euclidean signature and is the generalisation to all $\gpp$ of
the $\mathZ{E}_{8(8)}^{++}=\mathZ{E}_{10(10)}$ invariant action of reference
\cite{Damour:2002cu} proposed in the context of M-theory and
cosmological billiards. The parameter $\xi$ is then identified
with the time coordinate and the action restricted to a defined
number of lowest levels is equal to the
corresponding maximally oxidised theory in which the fields depend
only  on this time coordinate. This model is called cosmological model.

A second $\gpp$-invariant action
$S_{\gpp_B}$ is obtained from $S_{\gppp}$ by performing the same
consistent truncation but  after  a  Weyl
reflection of the gravity line  of $\gppp$. The
non-commutativity of the temporal involution with the Weyl
reflection \cite{Keurentjes:2004bv,Keurentjes:2004xx} implies that this second action is
inequivalent to the first one. Contrary to $S_{\gpp_C}$, the action $S_{\gpp_B}$ carries various Lorentzian signatures related by Weyl reflexions \cite{Englert:2004ph}. In $S_{\gpp_B}$, $\xi$ is identified
with a space-like direction and this theory admits for lowest level fields exact solutions
which are identical to those of the corresponding maximally
oxidised theory describing intersecting extremal brane
configurations smeared in all directions but one
\cite{Englert:2004ph,Englert:2003py,Englert:2004it}.  This model is then called brane model. Furthermore the intersection rules  are neatly encoded in the $\gpp$ algebra through orthogonality conditions between the real positive roots corresponding to the
branes  configuration \cite{Argurio:1997gt,Englert:2004it}. 

The precise analysis of the different possible signatures has been  first performed for $\gpp_B=\mathZ{E}_{8(8)}^{++}=\mathZ{E}_{10(10)}$. 
The signatures found in the analysis of references \cite{Keurentjes:2004bv,Keurentjes:2004xx} and in the context of $S_{\gpp_B}$ in \cite{Englert:2004ph} match perfectly with the signature changing dualities and the exotic
phases of M-theories discussed in \cite{Hull:1998vg,Hull:1998fh,Hull:1998ym}.  This analysis p
of signatures was extended in \cite{deBuyl:2005it} and we find for all the
$\gpp_B$-theories all the possible signature $(t,s)$, where $t$
(resp. $s$) is the number of time-like (resp. space-like)
directions, related by Weyl reflections of $\gpp$ to the signature
$(1,D-1)$ associated to the theory corresponding to the
traditional maximally oxidised theories.  Similar results was obtained independently in \cite{Keurentjes:2005jw}.\\


All the results above are restricted to the bosonic sector of supergravity theories.  Considerations over the last years of the fermionic sector of supergravities  have provided more evidence for the presence of a hidden symmetry $\mathZ{E}_{10(10)}$. 
The first analysis included parts of  fermionic degrees of freedom was treated in the reference \cite{Kleinschmidt:2004dy} in the context of the study of type IIA supergravity \footnote{In this reference, a level decomposition of $E_{10(10)}$ is performed under its $D_9$ subalgebra.}. Then,  the one-dimensional $\sigma$-model $\mathZ{E}_{10(10)}/\mathrm{K}_{10}^+$ of reference \cite{Damour:2002cu} was extended to include fermionic degrees of freedom\footnote{In the reference \cite{Damour:2007dt}, it was shown that the coset dynamic truncated at level $\ell \leqslant 3$ can be consistently restricted by requiring the vanishing of a set of constraints which are in one-to-one correspondence with the canonical constraints of supergravity.} in \cite{Damour:2005zs,deBuyl:2005mt,Damour:2006xu}.  This extension requires that the fermionic degrees of freedom are assigned to spinorial representation of the compact subgroup $\mathrm{K}_{10}^+$ \cite{deBuyl:2005zy}. However, a supersymmetric extension of the $\sigma$-model $\mathZ{E}_{10(10)}/\mathrm{K}_{10}^+$ could not yet be built because only finite-dimensional, i.e. unfaithful, spinorial representation of $\mathrm{K}_{10}^+$ have been considered \footnote{ Note that in the reference \cite{Kleinschmidt:2006tm}, the decomposition of spinor representations under subgroup of $\mathrm{K}_{10}^+$  reproduces the right fermionic fields representation of type IIA and type IIB supergravity.}. The replacement of unfaithful representations by faithful infinite-dimensional one is then crucial.

The fermionic side of cosmological billiard was also recently studied  in \cite{Damour:2009zc} for eleven-dimensional supergravity near space-like singularity. In this reference, the structure of `super-billiard' that combines bosonic and fermionic billiards is also considered.\\


The Kac-Moody algebras $\agpp$ and $\agppp$ contain an infinite number of generators and consequently the field content of the corresponding $\mathrm G^{++}$ and $\mathrm G^{+++}$ non-linear realization contains infinitely many fields in addition to those of the corresponding space-time maximally oxidised theories. Indeed, if we focus on the case of $\mathZ E_{10(10)}$ and $\mathZ E_{11(11)}$  we have seen that only low level fields are related  to degrees of freedom of eleven-dimensional, type IIA or type IIB supergravity. If these Kac-Moody algebras are effectively  symmetries of these supergravity theories,  the interpretation of the infinite number of higher level fields must be found. Many proposal was developed these last year to understand the structure of these infinity of fields: 
\begin{itemize}
\item   It was conjectured in \cite{Damour:2002cu}   that  some higher level fields associated to affine representations   contain spatial derivatives of the lowest level fields. This proposition could provide an explication of the way space-like dimensions are encoded in this infinite tower of fields. This conjecture has not yet been proved and the mechanism relating  affine towers to space-time has not been exhibited.
\item The physical interpretation of the fields corresponding to imaginary roots, which are roots of zero or negative norm, is explored in \cite{Brown:2004jb}. They proposed that these fields can be matched with some brane configurations.
\item  In the reference \cite{Damour:2005zb}, the higher derivative quantum corrections  to the action of M-theory is considered.  They found that certain known corrections admit an interpretation in terms of the Kac-Moody structure and are associated to certain negative imaginary roots of $E_{10(10)}$.  This suggest that the $\mathZ E_{10(10)}/ \mathrm K_{10}^+$ $\sigma$-model may incorporate also quantum corrections   of M-theory. This work was extended to other (super-)gravity model (and then to other extended Kac-Moody algebras) in \cite{Damour:2006ez} (see also \cite{Lambert:2006he}).                      
\item  Another interpretation  of affine real roots of $E_{10(10)}$ is given in  \cite{Englert:2007qb}. We consider fields parametrizing the brane model $\mathZ E_{10(10)}/ \mathrm K_{10}^-$  and we obtain for each positive real roots of the affine algebra $E_{9(9)}$ a BPS solutions of eleven-dimensional supergravity, or of its exotic counterparts, depending on two non-compact variables. This analysis is developed in Part II.               
\item Recent studies were developed  to extend the correspondence that can exist between the fields of a non-linear realization and the physical degrees of freedom of some supergravity theories  to non-propagating supergravity fields.  Indeed many supergravity theories allow introduction of $(D-1)$-form potentials called \emph{de-form} and $D$-form potentials called \emph{top-form}. The de-forms correspond to  deformation of supergravity theory with a masslike parameter while the top-forms are gauge fields which couple to space-filling branes such as the $D9$-brane in $10$ dimensions.  Evidences was presented in \cite{Riccioni:2007au} and \cite{Bergshoeff:2007qi} that de-forms and top-forms (identified using previous observations of \cite{Riccioni:2006az}) occur in different decompositions of $E_{11(11)}$. Moreover, the reference   \cite{Bergshoeff:2007qi} argued that the de-forms are in one-to-one correspondence to the embedding tensors that classify the gaugings of all maximal gauged supergravities which was investigated in \cite{Nicolai:2001sv,deWit:2002vt,deWit:2003hr}. 
The same statement was extended to the possible gaugings and massive deformations of half-maximal supergravities in \cite{Bergshoeff:2007vb}. See also \cite{Riccioni:2007ni,deWit:2008ta,Bergshoeff:2008xv,Bergshoeff:2008qd,Riccioni:2009xr} for further references on the subject.
\end{itemize}
\section*{About this thesis}
\addcontentsline{toc}{section}{About this thesis}


\noindent This thesis is structured in 3 Parts:\\

\noindent $\boldsymbol{\star}$ \textbf{Part I}

\noindent In this part, we describe all the tools which will be useful in Parts II and III. As we recalled above, Kac-Moody algebras appears in various aspects of supergravity theories. It is then essential to review some aspects of these algebras. This is done in \textbf{Chapter \ref{chap:math}}. More precisely, we develop in this chapter the preliminary mathematical background materiel required to understand subsequent chapters.  We begin by reviewing the basic definitions and properties which underlie the theory of complex Kac-Moody algebras. Then we focus on particular classes of infinite-dimensional Kac-Moody algebras obtained by adding nodes in the Dynkin diagrams of  finite-dimensional Lie algebras (see Figure \ref{fig:kmintro}).  As  in this dissertation we are going to deal with split real forms in Part II and with non-split real forms in Part III,   we will also present a quick recap of the theory of real forms of complex Kac-Moody algebras. Finally, we will explain in details the process of level decomposition of a Kac-Moody algebra with respect to one of its regular subalgebra. This level decomposition is a very important tool for the construction of non-linear $\sigma$-model over infinite-dimensional coset spaces.\\

The \textbf{Chapter \ref{chap:sugrareformulatedth}} makes use of the content of the previous chapter to describe some recent developments devoted to investigations of the underlying symmetry structures of supergravity theories. We begin by describing how toroidal compactifications of gravity theories reveal hidden  symmetries of the reduced Lagrangian. We also discuss attempts to extend the finite-dimensional symmetry structures to infinite-dimensional ones when the reduction is performed below $3$ dimensions. We then recall the construction of  actions $\mathcal{S}_{\mathrm{G}^{++(+)}}$   explicitly invariant under $\mathrm G^{++(+)}$.  The study of $\gpp$-invariant actions: $\gpp_C$ corresponding to the cosmological model and $\gpp_B$ corresponding to the brane model allows to make connections with the corresponding space-time supergravity theory. We will also review the effect of Weyl reflections on space-time signature of $\mathrm G^{+++}$-invariant theory. This result was published previously in  \cite{deBuyl:2005it}.  \\

\noindent $\boldsymbol{\star}$ \textbf{Part II}

\noindent This part is based on our original work \cite{Englert:2007qb} where an infinite $E_{9(9)}$ multiplet of BPS states is constructed for eleven-dimensional supergravity. 

In \textbf{Chapter \ref{chap:basicbps}} we review the construction of the brane $\sigma$-model $\mathZ E_{10(10)}/K_{10}^-$ and its relation to the basic BPS solutions of eleven-dimensional supergravity and of its exotic counterparts namely: the KK-wave, the M2 brane, the M5 brane and the KK6-monopole.

In \textbf{Chapter \ref{chap:infiniteudualgroup}}, we classify {\em all} $E_{9(9)}$ generators in $A_1^+$
 subgroups with central charge in $E_{10(10)}$. We select  particular
 $A_1^+$ subgroups containing two infinite `brane' towers of
 generators, or one infinite `gravity' tower. The tower generators are
 recurrences of the generators defining the basic M2 and M5,  or the
 KK-wave and the KK6-monopole.  The other $A_1^+$ subgroups needed to
 span all $E_{9(9)}$ generators  are obtained from the chosen ones by  Weyl
 transformations in the $A_8\subset E_{9(9)}$ gravity line.               
 The fields characterising the basic BPS solutions smeared to two
 space dimensions are encoded as parameters in Borel
 representatives of $\mathZ E_{10(10)}/\mathrm{K}_{10}^-$: each basic solution is fully
 determined by a specific positive generator associated to a specific
 positive real root.                                        All
 $E_{9(9)}\subset E_{10(10)}$ real roots are related by Weyl
 transformations. We use sequences of Weyl reflexions to reach any
 positive real root from roots corresponding to basic BPS solutions.
 We then express through dualities and compensations the fields
 defined by a given root  in terms of the eleven-dimensional metric and
 the 3-form potential. We verify 
 that these  fields  yield  a new solution of eleven-dimensional  supergravity or of
 its exotic counterparts.  In this way we generate an infinite
 multiplet of $E_{9(9)}$  BPS solutions  depending on two space
 variables. 
 In the string theory context this constitutes an infinite
 sequence of U-dualities realised  as Weyl transformations of
 $E_{9(9)}$. 
 It is shown that the full BPS multiplet of states is
 characterised by group transformations preserving the analyticity of
 the Ernst potential originally introduced in the context of the
 Geroch symmetry of 4-dimensional gravity with one time and one space
 Killing vectors. 
 
\textbf{Chapter \ref{chap:dualfor}} discusses the nature of the different BPS states. One
introduces the dual formalism which proves a convenient tool to
analyse the charges and masses content of the $E_{9(9)}$ BPS states. The
masses are defined and computed in the string theory context. We show
that the $E_{9(9)}$ multiplet can be split into three different classes
according to the $A_9$ level $\ell$. For $0\le \ell \le 3$ one gets the
basic BPS states smeared to two non-compact space dimensions. For
levels 4, 5 and 6  the BPS states depending on two non-compact space
variables can not be `unsmeared' in higher space dimensions. We
qualify the eight remaining space dimensions, and the time dimension,
as longitudinal ones.  For $\ell >6$ all BPS states admit nine
longitudinal dimensions, including time, and we argue that they are
all compact. 

In this chapter, we will also show that $E_{10(10)}$ fields associated to real roots which
are not in $E_{9(9)}$ are BPS solutions of  $\mathcal{S}^{brane}$ and admit thus a
space-time description with one non-compact transverse space
dimension. They may not admit a direct description but the dual
description is still well defined. These facts are exemplified by a
level 4 field which yields the M9 brane namely the `uplifting' of the $D8$ brane of massive IIA supergravity. 

Appendices complements arguments of this part. \\

\noindent $\boldsymbol{\star}$ \textbf{Part III}

\noindent \noindent This part is based on our original work \cite{Houart:2009ed} where we studied the symmetries of pure $\mathcal{N}=2$ supergravity in $D=4$.  As is known, this theory reduced on one Killing vector is characterised by a non-linear realised symmetry $\su$. This symmetry group is not in its split real form (which would be $\mathrm{SL}(3, \mbb C)$) and one of our motivations for this part was to investigate whether the conjectures discussed above in this introduction (and in Chapter \ref{chap:sugrareformulatedth}) apply also in this case (see also \cite{HenneauxJulia,deBuyl:2003ub,Henneaux:2007ej,Riccioni:2008jz} for related work). Moreover since the symmetry $\su$ mixes the two gravitational charges one can study the question of gravitational dualities analogous to electromagnetic duality in this model. \\

In \textbf{Chapter \ref{chap:finitesymsu215}}, we  first review some facts about pure ${\cal  N}=2$ supergravity in $D=4$ and the group $\mathrm{SU}(2,1)$ acting on its solutions with one time-like or 
one space-like Killing vector in Section~\ref{sec:EinsteinMaxwell}. Then, we go on to study the action of the finite-dimensional $\mathrm{SU}(2,1)$ on the BPS solutions in Section~\ref{sec:GroupAction}. There we show that the four charges transform linearly under the non-compact subgroup $\mathrm{SL}(2,\mathbb{R})\times \mathrm{U}(1)$ of $\mathrm{SU}(2,1)$. In particular, we show that the moduli space of half-BPS solutions can be described as a certain coset space, in agreement with recent results in the literature, and discuss the extension to the quantum theory from a string theory perspective. 

In \textbf{Chapter \ref{chap:infinitedimsymsu21}}, we will consider the conjectured infinite-dimensional symmetries of  four-dimensional pure ${\cal N}=2$ supergravity. By analysing the Lie algebra of $\mathrm{SU}(2,1)^{+++}$ we then demonstrate that the field content of the extended symmetry group is correct in Section \ref{sec:su21+++}.  This requires understanding which generators are present in this particular real form of the Kac-Moody algebra. We recall that the study of the real form $\asuppp$ was performed in Section \ref{subsec:su21andextension}. Starting from this observation, one can construct a correspondence between the one-parameter cosmological model based on $\supp$ and ${\cal N}=2$ supergravity in exactly the same way as for $\mathZ E_{10(10)}$ and this is shown in Section~\ref{sec:su21++}. We demonstrate also how the algebraic structure of $\suppp$ captures the half-BPS solutions in Section~\ref{sec:su21++}. This provides a detailed study of the proposed infinite-dimensional symmetries of ${\cal N}=2$. The extremal BPS solutions that occur in ${\cal N}=2$ supergravity can be derived from intersecting brane construction in M-theory and this leads us to an embedding of the non-split $\mathrm{SU}(2,1)^{+++}$ in the split $\mathrm  E_{11(11)}$, which is described in Section \ref{sec:Embedding}, thus nicely unifying our analysis with existing results. Questions not addressed in this part are the supersymmetric deformations of $\mathcal{N}=2$ supergravity (e.g. adding a cosmological constant) and their consistency of the algebraic structure of $\suppp$ via higher rank forms~\cite{Bergshoeff:2007qi,Riccioni:2007au,Gomis:2007gb,Kleinschmidt:2008jj} as well as the coupling of the fermionic sector.\\

Our results can be taken as evidence that the conjectured $\mathrm{G}^{++}$ and $\mathrm{G}^{+++}$ also appear in situations when $\mathrm{G}$ is not in split real form. Their full verification is subject to the same restrictions regarding the correct interpretation of the infinity of their generators as in the case when $\mathrm{G}$ is split. One can establish a correspondence (or dictionary) between the cosmological coset model based on $\mathrm{G}^{++}$ and the supergravity equations at low levels and account for the algebraic structure of  half-BPS solutions in $\mathrm{G}^{+++}$. The finite $\mathrm{G}$ part of the symmetry acts as a solution generating group in $D=3$. In particular, there are non-linear transformations acting as gravitational dualities on BPS solutions. Furthermore, the construction of ${\cal N}=2$ supergravity as a truncation of the maximal ${\cal N}=8$ theory has an algebraic counterpart since $\mathfrak{su}(2,1)^{+++}$ is contained in $E_{11(11)}$ as a subalgebra. \\


\part{Kac-Moody algebras: an approach to supergravity theories}
  \chapter{Complex and real Kac-Moody algebras} \label{chap:math}

In this chapter, we will review some important aspects of Kac-Moody algebras. Roughly speaking Kac-Moody algebras are infinite-dimensional generalizations of finite-dimensional simple Lie algebras and  their intriguing apparition  in various aspects of   
supergravity theories provides a good motivation to study them. The aim of this chapter is not to re-derive all the theory of Kac-Moody algebras which is still not well understood but we will give all the tools necessary to follow the discussion of this dissertation.

We will first recall in Section \ref{sec:def} the basic definitions which underly the theory of complex Kac-Moody algebras. Although there is no classification of these algebras, we will focus in Section \ref{sec:classkm} on a particular class which is relevant for this work. The following section is devoted to the real forms of Kac-Moody algebras and we will end this chapter by studying a way to approach the infinite number of generators appearing in Kac-Moody algebras.

Section \ref{sec:def} and \ref{sec:classkm} are inspired by \cite{Modave,Danielth}. The Section \ref{sec:realkm} follows the concise presentation of the subject of real forms given in \cite{Henneaux:2007ej,deBuyl:2006gp,Real}. For good references, we recommend also \cite{Helgason:1978}. The last section of this chapter is based on \cite{Axelthesis,Nicolai:2003fw,Kleinschmidt:2003mf}. 
\section{Definitions}\label{sec:def}
\setcounter{equation}{0}
This section presents a survey of the theory of complex Kac-Moody algebras. It includes the definition of these algebras based on their generalized Cartan matrix and covers also the basic properties that we will exploit along this thesis. 

\subsection{Cartan Matrix and the Chevalley-Serre basis}\label{subsec:chevalley}
Let  $A= (A_{ij}) _{i,j=1, \ldots, n}$ be \emph{a generalized Cartan matrix}, i.e.  an $n\times n$ matrix which satisfies the following conditions:
\be \label{eqn:cartan} \begin{split} \begin{aligned}
A_{ii}&=2, \quad & & i=1, \ldots n , \\
A_{ij} &\in \mbb{Z}_{-}, \quad& & \mathrm{for}\  i\neq j,\\
A_{ij}&=0 \Longleftrightarrow A_{ji}=0 & & \mathrm{for}\  i\neq j .
\end{aligned}\end{split} \ee
The associated \emph{Kac-Moody algebra} $\mathfrak{g}(A)$ is a complex Lie algebra on $3n$ \emph{Chevalley generators} $\{h_i,e_i, f_i\}$ ($i=1, \ldots,n$) obeying the following \emph{Chevalley relations}
\be\label{eqn:chevalley} \begin{split} 
[h_i,h_j]&=0 ,\\
[h_i,e_j]&=A_{ij}\, e_j ,\\
[h_i,f_j]&=- A_{ij}\, f_j ,\\
[e_i,f_j]&=\delta_{ij}\,  h_j ,
\end{split} \ee
and subject to the \emph{Serre relations}
\be \label{eqn:serre} \begin{split}
(\mathrm{ad}\, e_i)^{1-A_{ij}} e_j &=0,\\
(\mathrm{ad}\, f_i)^{1-A_{ij}} f_j &=0,
\end{split}\ee
where $\mathrm{ad}\, e_i$ (resp. $\mathrm{ad}\, f_i$) denotes the adjoint action of the generator $e_i$ (resp. $f_i$) on $\mathfrak{g}$. The relations \eqref{eqn:serre} then read explicitly
\be \label{eqn:serreexp} \begin{split}
\underbrace{[e_i,[e_i,[e_i,\ldots,[e_i,e_j]\ldots ]]]}_{1-A_{ij}\ commutators}&=0, \\
\overbrace{[f_i,[f_i,[f_i,\ldots,[f_i,f_j]\ldots ]]]}&=0
\end{split}\ee
and impose restrictions on the generators belonging to $\mf g$. Using the Chevalley relations \eqref{eqn:chevalley} and the Jacobi identity, any multicommutator  can be reduced to a multicommutator involving only the generators $e_i$ or only the generators $f_i$ or is an element of $\{h_i, i=1, \ldots, n\}$ (or is zero). Therefore $\mathfrak{g}$ has the following \emph{triangular decomposition} as a vector space
\be \label{eqn:triangular}
\mathfrak{g}= \mf{n_{-}}\oplus \mf{h}\oplus \mf{n_{+}},
\ee
where the subspace $\mf{n_-}$ and $\mf{n_+}$ are respectively generated by the $f_i$'s and $e_i$'s  modulo the Serre relations \eqref{eqn:serre}. The subalgebra  $\mf{h}$ is the complex vector space spanned by the $h_i$
\be \label{eqn:cartansub} 
\mf{h}= \sum_{i=1}^{n} \mbb C\, h_i,
\ee
which acts diagonally on the Chevalley generators (see \eqref{eqn:chevalley}) and forms an abelian subalgebra of $\mf{g}$ called the \emph{Cartan subalgebra}. Its dimension $n$ is the \emph{rank} of the Kac-Moody algebra $\mf{g}$.

\subsection{Dynkin diagrams and classification of Kac-Moody algebras}

We have seen in the previous section that the structure of the algebra $\mf g(A)$ is completely encoded in its Cartan matrix $A$ whose entries determine the commutation relations between the generators. Therefore the Cartan matrix $A$ characterizes completely the algebra $\mf g$. A very convenient way to encode the Cartan matrix $A$ is by associating to it a \emph{Dynkin diagram} which is constructed as follows:
\begin{itemize}
\item to each triple of Chevalley generators $(h_i, e_i, f_i)$ one associates a white node $\circ$,
\item if $A_{ij} =0$, the nodes $i$ and $j$ are disconnected,
\item if $A_{ij}\neq 0$, the nodes $i$ and $j$ are connected by max$\{|A_{ij}|,|A_{ji}|\}$ lines and one draws an arrow from $j$ to $i$ if $|A_{ij}|>|A_{ji}|$.
\end{itemize}
The form of the multicommutators
\be
[e_{i_1},[e_{i_2},[e_{i_3},\ldots,[e_{i_{p-1}},e_{i_p}]\ldots ]]] \quad \mathrm{and} \quad [f_{i_1},[f_{i_2},[f_{i_3},\ldots,[f_{i_{p-1}},f_{i_p}]\ldots ]]],
\ee
is restricted by the Serre relations \eqref{eqn:serre}. These restrictions might render the algebra $\mf{g}(A)$ finite or infinite-dimensional depending on the properties of the Cartan matrix $A$. In fact, the algebras $\mf g(A)$ have been classified according to the properties of the eigenvalues of their respective Cartan matrix:
\begin{itemize}
\item  $\mf g$ is a \emph{finite-dimensional Lie algebra} if $A$ is positive definite. This kind of algebra are well-know and have been classified by Cartan and Killing. The classification of simple\footnote{A simple Lie algebra is a Lie algebra that contains no proper ideal meaning no proper subset of generators $\in \mf i$ such that $[\mf i, \mf g] \subset \mf i$.} Lie algebras boils down to a classification of Dynkin diagram displayed in Figure \ref{fig:liealgebras} which contains four infinite families: the classical algebras $A_n, B_n, C_n$ and  $D_n$ and  five exceptional cases: $G_2, F_4, E_6, E_7$ and $E_8$.
\item $\mf g$ is an \emph{affine Kac-Moody algebra} if $A$ is positive-semidefinite, i.e. det $A=0$ with one zero eigenvalue. Affine Kac-Moody algebras are infinite-dimensional and  have been also classified \cite{Kac:book}.
\item $\mf g$ is infinite-dimensional and called \emph{indefinite} if $A$ is not part of the two first classes. No general classification exists for this class but in this work we will focus on a particular subclass associated to a Cartan matrix $A$ which has one negative eigenvalue and $n-1$ positive ones. These algebras are called \emph{Lorentzian} by virtue of the signature $(-,+,+,\ldots,+)$ of $A$. A special subclass of the Lorentzian algebras, know as \emph{hyperbolic Kac-Moody algebras} have been classified. Hyperbolic Kac-Moody algebras are such that the deletion of any node from its Dynkin diagram gives a sum of finite or affine algebras. One can show that there exists no hyperbolic algebra with rank higher than $10$.
\end{itemize}

\begin{figure}[ht]
\begin{center}
{\scalebox{0.2}
{\includegraphics{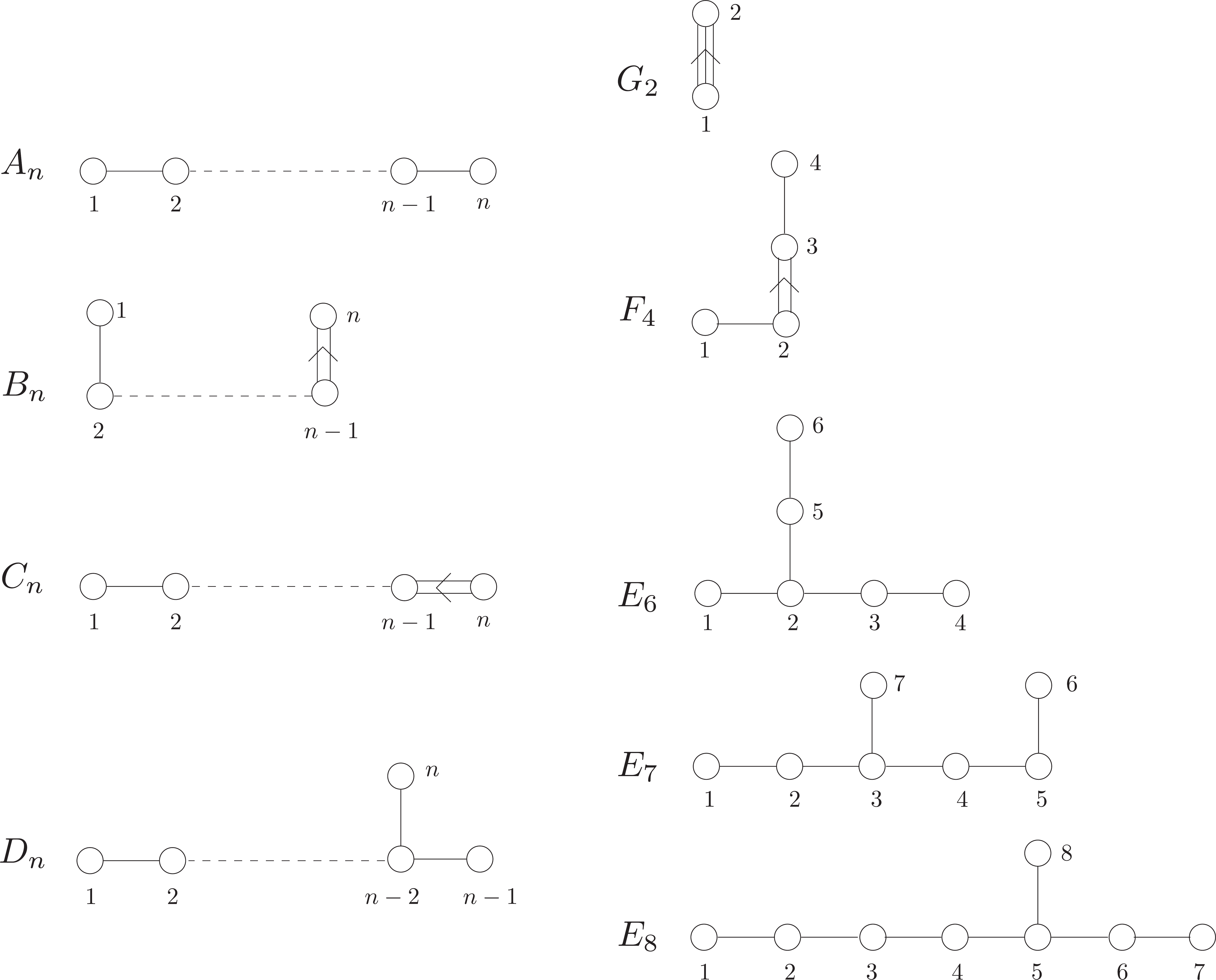}}}
\caption{\small \textsl{Dynkin diagrams of finite-dimensional Kac-Moody algebras. The subscript on the name gives the rank of the algebra. }}
\label{fig:liealgebras}
\end{center}
\end{figure}
\subsection{The root system}\label{subsec:rootsystem1}

A \emph{root} $\alpha$ of the algebra $\mf g$ is a non-zero linear form on the Cartan subalgebra $\mf h$ (i.e. an element of the dual $\mf h^*$) such that the \emph{root space} $\mf g_{\alpha}$ defined by
\be \label{eqn:rootspac}
\mf{g}_{\alpha}= \{ x \in \mf g(A)\, |\, [h,x]= \alpha(h)\, x,\ \forall h \in \mf h\},
\ee
is not empty. The (complex) dimension of the root space is called  \emph{the multiplicity} mult$(\alpha)$
\be\label{multipliroot}
\mathrm{mult}(\alpha)= \mathrm{dim}\, \mf g_{\alpha}.
\ee
The  particular roots denoted by $\alpha_i, \ i=1, \ldots n$ corresponding  to the eigenvalues of the adjoint action of $\mf h$ on the Chevalley generators $e_i$:
\be
[h,e_i] = \alpha_i(h)\, e_i,
\ee
are called \emph{simple roots}. According to the Chevalley relations \eqref{eqn:chevalley}, one gets that
\be \label{componentsiroot}
\alpha_i(h_j)= A_{ji}.
\ee
One also uses the notation $<.,.>$ for the standard pairing between $\mf h$ and its dual $\mf h^*$ 
\be
<\alpha, h>= \alpha(h).
\ee
In this notation the entries of the Cartan matrix can be written as
\be \label{eqn:cartbil}
<\alpha_i, h_j>= A_{ji}.
\ee
The basis of simple roots is denoted by $\Pi=\{\alpha_1, \ldots \alpha_n\}$. Any root $\alpha \in \mf h^*$ can be expressed as an integer linear combination of the simple roots
\be
\alpha= \sum_{i=1}^{n} m_i \alpha_i, \ m_i \in \mbb Z\, .
\ee
where ht$(\alpha)= \sum_{i=1}^{n} m_i $ defines the \emph{height} of the root $\alpha$.
In fact, the adjoint action of $\mf h$ on a multicommutators of $\mf n_{+}$ and $\mf n_{-}$ gives respectively
\be \begin{split}
\bigg[h, \big[e_{i_1},[ e_{i_2}, \ldots [e_{i_{p-1}},e_{i_{p}}] ]\big]\bigg]&=\ \ (\alpha_{i_1}+\alpha_{i_2}+\ldots +\alpha_{i_p} ) (h)\, [e_{i_1},[ e_{i_2}, \ldots \big[e_{i_{p-1}},e_{i_{p}}]]\big],\\
\bigg[h, \big[[[f_{i_p}, f_{i_{p-1}}]\dots,f_{i_2}],f_{i_1}\big]\bigg]
&= - (\alpha_{i_1}+\alpha_{i_2}+\ldots +\alpha_{i_p} ) (h)\, \big[[[f_{i_p}, f_{i_{p-1}}]\dots,f_{i_2}],f_{i_1}\big],
\end{split}\ee
If the generator $e_{\alpha} \equiv \big[e_{i_1},[ e_{i_2}, \ldots [e_{i_{p-1}},e_{i_{p}}] ]\big]$ is non-vanishing, then one says that $\alpha= (\alpha_{i_1}+\alpha_{i_2}+\ldots + \alpha_{i_p} )$ is a \emph{positive root} associated to $e_{\alpha}$ and that $-\alpha$ is a \emph{negative root} associated to $f_{\alpha} \equiv  \big[[[f_{i_p}, f_{i_{p-1}}]\dots,f_{i_2}],f_{i_1}\big]$. We take the convention that the negative generator associated to $-\alpha$ follows from the positive one by writing the multicommutator in the reverse order. From the triangular decomposition \eqref{eqn:triangular} it follows that the roots are either positive (i.e. linear combination of the simple roots $\alpha_i$ with integer non-negative coefficients) or negative (i.e. linear combination of the simple roots $\alpha_i$ with integer non-positive coefficients). The complete set of roots is called \emph{the root system} $\Delta$, and it splits into two disjoint sets
\be
\Delta= \Delta_+ \cup \Delta_-\, ,
\ee
where $\Delta_+$ are the set of positive roots and $\Delta_-$ are the set of negative roots. 

We also have that
\be
\mf h^*= \sum_{i=1}^n \mbb C \,\alpha_i.
\ee
In the Chevalley basis, the simple roots $\alpha_i(h)$ are always integers implying that all possible eigenvalues in the action of $\mf h$ lie on a $n$-dimensional lattice $Q$ spanned by the simple roots,
\be
\Delta \subset Q= \sum_{i=1}^{n} \mbb Z \, \alpha_i \in \mf h^*.
\ee
Using the Jacobi identity, one can see that the root lattice provides a grading of the Kac-Moody algebras $\mf g$. Considering the two root spaces $\mf g_{\alpha}$ and $\mf g_{\beta}$  associated to the roots $\alpha$ and $\beta$ $\in Q$ we have
\be
[\mf g_{\alpha}, \mf g_{\beta}] \subseteq \mf g_{\alpha+\beta}. 
\ee
The whole algebra $\mf g$ decomposes into a direct sum of eigenspaces:
\be \label{rootspacedec}
\mf g = \mf h \oplus \bigoplus _{\alpha \in \Delta} \mf g_{\alpha}.
\ee
This decomposition is called the \emph{root space decomposition}.

\subsection{The Chevalley involution} \label{subsec:chevalleyinvo2}
The  Chevalley-Serre presentation  indicates a symmetry between the positive and negative subalgebras $\mf n_+$ and $\mf n_-$ which leaves the Chevalley-Serre relations \eqref{eqn:chevalley} and \eqref{eqn:serre} invariant. This automorphism on $\mf g$ is known as the \emph{Chevalley involution} $\omega$ defined as follows on the Chevalley generators
\be \label{eqn:involutionchev}
\omega(h_i)= -h_i, \quad \omega(e_i)= -f_i, \quad \omega(f_i)= -e_i.
\ee
The action of $\omega$ can be extended to all generators of $\mf g$ in the standard way. For instance, on a generator  $e_{\alpha} \subset \mf n_+$ obtained by multicommutators of Chevalley generators i.e. $e_{\alpha} \equiv [e_{i_1},[ e_{i_2}, \ldots [e_{i_{p-1}},e_{i_{p}}]\ldots ]]$ one has
\be\begin{split} \label{eqn:ealphagen1}
\omega (e_{\alpha})&= \omega([e_{i_1},[ e_{i_2}, \ldots [e_{i_{p-1}},e_{i_{p}}]\ldots ]]),\\
&= [\omega(e_{i_1}),[\omega( e_{i_2}), \ldots [\omega(e_{i_{p-1}}),\omega(e_{i_{p}})]\ldots ]],\\
&=(-1)^{p}  [f_{i_1},[ f_{i_2}, \ldots [f_{i_{p-1}},f_{i_{p}}]\ldots ]], \\
&= - [[\ldots[f_{i_p}, f_{i_{p-1}}]\dots,f_{i_2}],f_{i_1}],\\
&= - f_{\alpha} \subset \mf n_{-}.
\end{split}\ee
The subset of $\mf g$ which is pointwise fixed under the Chevalley involution $\omega$ defines the \emph{maximal compact subalgebra}
\be
\mf k =\{ x \in \mf g \ | \ \omega(x)= x\}.
\ee
It is generated by the combinations of the Chevalley generators $(e_i-f_i),\, i=1, \ldots, n$. The Chevalley involution induces the following \emph{Cartan decomposition} of $\mf g$ in terms of vector spaces:
\be \label{eqn:cartandecomp124}
\mf g = \mf k \oplus \mf p,
\ee
where the complement $\mf p$ is the subspace of $\mf g$ which is pointwise anti-invariant under $\omega$
\be
\mf p =\{ x \in \mf g \ | \ \omega(x)=- x\}.
\ee
Note that $\mf p$ transforms in some representations of $\mf k$ but is not subalgebra of $\mf g$. The Cartan decomposition yields the following characteristic properties of a \emph{symmetric space}:
\be \label{symmetricspace}
[\mf p, \mf p] \subset \mf k, \quad [\mf k, \mf p] \subset \mf p, \quad [\mf k, \mf k] \subset \mf k.
\ee

For later reference, let us also give another useful decomposition of $\mf g$, known as the \emph{algebraic Iwasawa decomposition} which reads in terms of vector spaces
\be \label{eqn:iwasawa}
\mf g = \mf k \oplus \mf h \oplus \mf n_+\, .
\ee
In the finite-dimensional case this decomposition reduces to the familiar fact that any matrix can be decomposed into an orthogonal part, a diagonal part and an upper triangular part. The subset
\be \label{eqn:borelsubalgebra1}
\mf b_+ = \mf h \oplus \mf n_+\, ,
\ee
is known as the positive \emph{Borel subalgebra}. There is an alternative  Iwasawa decomposition which instead utilizes the negative nilpotent subspace $\mf n_-$
\be
\mf g = \mf k \oplus \mf h \oplus \mf n_-\, ,
\ee
with an associated negative Borel subalgebra
\be
\mf b_{-} = \mf h \oplus \mf n_-\, .
\ee
\subsection{Invariant bilinear form} \label{subsec:bilinear}
The Cartan matrix $A$ can be used to define a complex bilinear form on $\mf h^*$.  To proceed, we assume that the Cartan matrix is invertible and  \emph{symmetrizable} i.e. that there exists an invertible diagonal matrix $D= \mathrm{diag}\, (\epsilon_1, \ldots, \epsilon_n)$ and a symmetric matrix $S= (S_{ij})$, such that
\be \label{cartansymetri}
A= DS.
\ee
Under these assumptions, an \emph{invertible bilinear form }$(.|.)$ is easily defined in the dual $\mf h^*$ of the Cartan subalgebra as follows
\be\label{symbils}
S_{ij}\equiv (\alpha_i | \alpha_j),
\ee
for $\alpha_i$ and $\alpha_j$ $\in \Pi$. It follows from $A_{ii}=2$ that 
\be \label{eqn:eps}
\epsilon_i= \frac{2}{(\alpha_i | \alpha_i)} \cvp
\ee
The norm squared of a simple root $\alpha_i$ is then $\alpha_i^2=(\alpha_i|\alpha_i)= 2/\epsilon_i$. In most cases, the $\epsilon_i$'s can only take two distinct values and the roots corresponding to smaller norm are called \emph{short roots} and the other ones \emph{long roots}.
Using \eqref{symbils} and \eqref{eqn:eps}, one can see that the Cartan matrix (see \eqref{cartansymetri}) can be expressed in terms of bilinear form
\be\label{eqn:cartbil2}
A_{ij}= 2 \, \frac{(\alpha_i | \alpha_j)}{(\alpha_i | \alpha_i)}\,\cvp
\ee
We can now extend  this bilinear form on $\mf h^*$ to the whole algebra $\mf g$. Since the bilinear form is non-degenerate on $\mf h^*$ (since $A$ is), it provides an isomorphism $\mu: \mf h^* \rightarrow \mf h$ defined by
\be \label{eqn:iso1}
<\alpha ,  \mu({\beta})>= (\alpha|\beta),\qquad \alpha, \beta \in \mf h^*,\ \mu(\beta) \in \mf h\, .
\ee
This isomorphism induces a bilinear form on $\mf h$, also denoted by $(.|.)$ considering the inverse isomorphism $\mu^{-1}:\mf h \rightarrow \mf h^* $ such that
\be\label{eqn:iso2}
<\mu^{-1}(h), h'>= (h|h'), \qquad  h, h' \in \mf h,\  \mu^{-1}(h) \in \mf h^*           \, .
\ee
Using equations \eqref{eqn:cartbil}, \eqref{eqn:cartbil2} and \eqref{eqn:eps} and  the isomorphism equations \eqref{eqn:iso1} and \eqref{eqn:iso2}, we get 
\be
\mu^{-1}(h_i)= \epsilon_i\, \alpha_i\  \Leftrightarrow\ h_i= \epsilon_i \, \mu(\alpha_i).
\ee
Using this result, one finds the bilinear form on $\mf h$:
\be
(h_i|h_j)= \epsilon_i\, \epsilon_j\, S_{ij}.
\ee

At this point we have the bilinear form on the Cartan subalgebra $\mf h \subset \mf g$.  To extend this to the full Kac-Moody algebra, one exploits the \emph{invariance} of the bilinear form that it fulfills
\be \label{eqn:invariance}
([x,y] \,|\, z)= (x\,|\, [y,z]), \quad x, y, z \in \mf g.
\ee
Computing $(h_i | [h_k,e_j])= A_{kj} (h_i | e_j)$ and using the invariance property \eqref{eqn:invariance} and the Chevalley relations \eqref{eqn:chevalley}, we find
\be
(h_i|e_j)=0,
\ee
and similarly computing $(h_i | [h_k,f_j])$ one finds
\be
(h_i|f_j)=0.
\ee
In the same way, computing $(h_k | [e_i,f_j])= \delta_{ij} (h_k | h_i)$ one gets
\be
(e_i|f_j)=\epsilon_i\, \delta_{ij}.
\ee
We get all the bilinear forms implying Chevalley generators. They can be extended to the whole algebra using the gradation of $\ag$  \eqref{rootspacedec}, the invariance of the bilinear form \eqref{eqn:invariance} and the orthogonality property:
\be
(x_{\alpha} | x_{\beta})=0 \quad \mathrm{if} \, \alpha+\beta \neq 0,
\ee
where $x_{\alpha} \in \ag_{\alpha}$ and $x_{\beta} \in \ag_{\beta}$.
\subsection{Real and imaginary roots}

For finite-dimensional Lie algebras, the bilinear form $(.|.)$ coincides with the standard Killing form. Therefore, the bilinear form is of Euclidean signature and consequently, the root lattice $Q$ is a Euclidean lattice. In that case, all the roots have positive norm and are called \emph{real root}. The set of  real roots is denoted by $\Delta^{re}$:
\be
\Delta^{re}= \{\alpha \in \Delta \, | \, \alpha^2= (\alpha | \alpha)>0 \}.
\ee
When the algebra $\ag$ becomes infinite-dimensional, the invariant bilinear form does not have Euclidean signature. For the particular case of  Lorentzian algebras that we will study in this dissertation, the bilinear form is a flat metric with signature $(-,+,\ldots,+)$ and consequently $Q$ is Lorentzian lattice. Then besides real roots, this algebra has \emph{imaginary roots} i.e. roots with zero or negative norm. The set of imaginary roots is denoted by $\Delta^{im}$ :
\be
\Delta^{im}= \{\alpha \in \Delta \, | \, \alpha^2= (\alpha | \alpha)\leqslant 0 \}.
\ee
In this way the root system of a Kac-Moody algebra decomposes into two disjoint sets and we have
\be
\Delta = \Delta^{re} \cup \Delta^{im}.
\ee
While the real roots always have a multiplicity equal to one, this is not the case for imaginary roots. The multiplicities of imaginary roots are in general unknown except for affine algebras. For indefinite Kac-Moody algebras the multiplicity of imaginary roots grows exponentially with increasing height. It is a challenge to understand the structure of this type of algebras. 

Let us quickly review some other features of roots. If $\alpha \in \Delta^{re}$ then the only possible multiples of $\alpha$ are $\pm \alpha$. This is not so for imaginary roots, where all multiples of $\alpha$ can be root even if  $\alpha \subset \Delta^{im}$.

There are at most two different root lengths in the finite-dimensional case. In the case of infinite-dimensional Kac-Moody algebras, this is no longer true even for real roots. The algebra is called \emph{simply laced}\footnote{Simply laced algebras correspond to algebras associated to Dynkin diagram not carrying arrows indicating that these algebras have symmetric Cartan matrices. } if all real roots have the same length. In this thesis, we are dealing only with this kind algebra.
\subsection{The Weyl group} \label{subsec:weylgroupthese}
The Weyl group $\mc W$ of the Kac-Moody algebra $\mf g$ is the reflection group generated by the \emph{fundamental Weyl reflections} $s_i\ (i=1, \ldots , n)$ subject to the relations 
\be\label{coxeter1}
(s_i\, s_j)^{m_{ij}}=1, \qquad s_i^2=1,
\ee
where $m=(m_{ij})_{i,j=1, \ldots,n}$ is an $n \times n$ matrix satisfying
\be \label{coxeter2} \begin{split}
m_{ii}&=2,\\
m_{ij}&=2, 3, 4, 6\  \mathrm{or}\  \infty\,  \mathrm{according\ to} \ A_{ij}A_{ji}\ \mathrm{ is}\ 0,1,2,3\ \mathrm{or}\  \geqslant 4, \quad i\neq j,\\
m_{ij}&=m_{ji}.
\end{split} \ee
From the definitions \eqref{coxeter1} and \eqref{coxeter2} one can see that the Weyl group is a particular example of \emph{Coxeter group}. One can realize this group by associating a fundamental reflection $s_i$ to each simple root $\alpha_i, \ i=1, \ldots, n$, such that the action on any $\beta \in \mf h^{*}$ is given by
\be\label{weylreflection}
s_i: \beta \longrightarrow s_i(\beta)= \beta -2\, \frac{(\beta| \alpha_i)}{(\alpha_i | \alpha_i)}\, \alpha_i \,  \cvp
\ee
When acting on the simple roots themselves, the Weyl reflection becomes
\be
s_i(\alpha_j)= \alpha_j - \, A_{ij}\, \alpha_i.
\ee
Geometrically, the fundamental reflection $s_i$ realizes a reflection in the hyperplane orthogonal to the simple root $\alpha_i$. The hyperplane fixed by this reflection is defined by
\be \label{eqn:planinvweyl}
T_i=\{\gamma \in Q \, |\, (\gamma| \alpha_i)=0\}.
\ee
Let us now associate a general Weyl reflection $s_{\alpha}$ to any root $\alpha \in \Delta$ by taking a finite product of  fundamental reflections,
\be
s_{\alpha}= s_{i_1} s_{i_2}\ldots s_{i_k},
\ee
where the minimal number $k$ of fundamental reflections needed to describe $s_{\alpha}$ is called the \emph{lenght} of $s_{\alpha}$. Moreover, we can show that each general reflection $s_{\alpha}$ corresponds to a conjugation of a fundamental reflection $s_i$ by some element $w\in \mc W$:
\be
s_{\alpha}= w s_i w^{-1}.
\ee
The Weyl group has a number of interesting properties:
\begin{itemize}
\item it preserves the root lattice; $\mc W: Q \longrightarrow Q$, and the set of real and imaginary roots are separately invariant under the Weyl group: $\mc W . \Delta^{re}= \Delta^{re}$ and  $\mc W . \Delta^{im}= \Delta^{im}$,
\item it preserves the bilinear form $(.|.)$: $(w(\alpha)|w(\beta))=(\alpha|\beta)$,
\item two roots that are in the same orbit have identical multiplicities: mult$(w(\alpha))=$ mult$(\alpha)$.
\end{itemize}

\section{Classes of Kac-Moody algebras obtained by extensions of finite Lie algebras} \label{sec:classkm}
\setcounter{equation}{0}

Finite simple Lie algebras and affine Kac-Moody algebras are well studied and fully classified. It is not the case for other Kac-Moody algebras whose study constitutes still now a big challenge. However one of the classes of such Kac-Moody algebras that has been studied in some details and that has been classified is the class of  hyperbolic algebras. Their interest was motivated by the study of cosmological solutions in the vicinity of space-like singularity known as cosmological billiards \cite{Damour:2001sa,Damour:2002et,Henneaux:2007ej}. In this work, we consider a larger class of Lorentzian Kac-Moody algebras obtained by extending the Dynkin diagram of finite Lie algebra (displayed in Figure \ref{fig:liealgebras}) by adjoining additional nodes such that the associated Cartan matrix becomes indefinite \cite{Gaberdiel:2002db}. 

In this section, we will first review the case of affine Kac-Moody algebras which constitutes the first crucial step of the extension of finite dimensional Lie algebras. Then we will explain how we get over-extented and very-extended Kac-Moody algebras obtained by adding respectively one and two nodes to the affine extension. Note that in this section, $\mf g$ will denote a finite Lie algebra.

\subsection{Affine Kac-Moody algebras}\label{subsec:affine}

To every finite Lie algebra $\mf g$, we associate an affine extension $\mf g^{+}$ by adding to the Dynkin diagram of $\mf g$ an extra node $\alpha_0$ related to the highest root $\theta$.\footnote{We discuss here only the non-twisted affine extension.} The introduction of this particular simple root has the immediate effect of making the root system of $\agp$ infinite dimensional. We will briefly review how affine Lie algebras are obtained from simple Lie algebra.

We consider the generalization of $\mf g$ in which the elements of the algebra are Laurent polynomials in some variable $t$. This generalization is called the \emph{loop algebra} $\tilde{\mf g}$ and the Laurent mode expansion gives rise to an integer family of generators. Therefore if $\{J^a\}$ denotes a set of generators of the finite Lie algebra $\mf g$ satisfying the following commutations relations
\be
[J^a, J^b]= i\, f^{ab}_{\ \ c}\, J^c\, ,
\ee
then the generators $J^a_n$ of the loop algebra $\tilde{\ag}$ satisfy
\be
[J^a_n, J^b_m]= i\, f^{ab}_{\  \ c}\, J^c_{m+n}\, ,
\ee
where $a,b= 1, \ldots ,\mathrm{dim}\,  \mf g$ and $m, n \in \mbb Z$. This loop algebra becomes an \emph{affine Kac-Moody algebra} $\mf g^+$ by adjoining to $\tilde{\mf g}$ a \emph{central element} $c$
 and a \emph{derivation} $d$. 
 
 In fact an affine Kac-Moody algebra has a degenerate positive semi-definite Cartan matrix $A$ with one zero eigenvalue. This implies that 
the center (i.e. the set of the elements of the algebra which commute with all the algebra) is one-dimensional and is spanned by $c \subset \mf h$. Moreover, as the Cartan matrix $A$ is degenerate, the bilinear form as constructed in Section \ref{subsec:bilinear} is ill-defined. To resolve the difficulty, the algebra itself must be augmented by hand  by the addition of a new generator to the Cartan subalgebra $\mf h$, known as the derivation $d \subset \mf h$ such that $<\alpha_i, d>= \delta_{i0}, \ \ i=0, 1, \ldots n$. Therefore, a non-degenerate bilinear form now exists with the following properties
\be
(c|h_i)=0, \quad (c|c)= 0, \quad (c|d)=1.
\ee
The non-degeneracy of the bilinear form on $\mf h$ follows from the non-vanishing scalar product between the central element $c$ and the derivation $d$.

With the addition of these generators $c$ and $d$, the resulting algebra denoted by $\mf g^+$ can be decomposed as
\be
\mf g^+= \tilde{\mf g} \oplus \mbb C  c \oplus \mbb C d.
\ee
It turns out that the Dynkin diagram for the affine Kac-Moody algebras $\mf g^+$ is simply the one for $\mf g$ with one additional node connected to the Dynkin diagram of $\mf g$. This new node is related to the \emph{affine root} $\alpha_0$ which has the following form
\be\label{affineroot}
\alpha_0= \delta- \theta,
\ee
where $\delta$ is a null root $(\delta | \delta)=0$ and $\theta$ is the highest root of the finite subalgebra $\mf g \subset \mf g^+$. The introduction of the affine root $\alpha_0$ will render the algebra $\ag^+$ infinite-dimensional and will allow the apparition of imaginary roots (which do not exist in the finite case). The root system $\Delta$ of $\agp$ will split into its real and imaginary parts:      
\be
\Delta (\agp)= \Delta^{re}(\agp) \cup \Delta^{im}(\agp)\, 
\ee    
with
\be \begin{split}
\Delta^{re}(\agp) &= \{ \alpha +n\,\delta\ |\ \forall \alpha \in \Delta(\ag); \, n \in \mbb Z \},\\
\Delta^{im}(\agp) &= \{  n \, \delta , \, n \in \mbb Z \}.
\end{split}\ee
We will see in the former section how the affine root $\alpha_0$ connects to the Dynkin diagram of the finite algebra $\ag$.

\subsection{The extension process}

We will now explain the process of extension of finite Lie algebra $\ag $ step by step to obtain first affine, then over-extended and finally very-extended Kac-Moody algebras. 

Let $\Pi= \{\alpha_1, \ldots \alpha_n\}$ be a basis of simple roots for the finite Lie algebra $\mf g(A)$ of rank $n$. The associate root lattice $Q= \sum_{i=1}^n\mbb Z \alpha_i$ is Euclidean. We take $\mbb Z^{1,1}$ the  even two-dimensional unimodular Lorentzian lattice spanned by the vectors $u_1$ and $u_2$. The non-degenerate bilinear form of signature $(-,+)$ is defined on $\mbb Z^{1,1}$ by 
\be \label{scalar156}
(u_1|u_2)=1, \quad (u_1|u_1)=0, \quad (u_2|u_2)=0.
\ee
where these scalar products are induced from the standard Minkowski metric on $\mbb R^{1,1}$, taking $u_1= (1,0)$ and $u_2=(0,-1)$.

\subsubsection{The affine extension}

 We can now extend the Dynkin diagram of $\mf g$ with one node to obtain the Dynkin diagram of the affine algebra $\mf g^+$. This step was already discussed briefly in Section \ref{subsec:affine} and we recall that we obtained the affine Kac-Moody algebra $\mf g^{+}$ by augmenting the set of simple roots with the affine root $\alpha_0$ defined as
\be
\alpha_{0}\equiv u_1- \theta,
\ee
where $u_1 \in \mbb Z^{1,1}$ corresponds to the null root $\delta$ (see \eqref{affineroot}). We have that 
\be
(\alpha_i|u_1)=0, \quad i=1, \ldots,n.
\ee
implying that the scalar product between $\alpha_0$ and any simple root $\alpha_i$ of the finite algebra 
$\mf g$ is
\be \label{affinconec}
(\alpha_0| \alpha_i)= - (\theta| \alpha_i).
\ee
From \eqref{affinconec}, we can understand that the manner  to connect the affine root $\alpha_0$ to the other nodes of the Dynkin diagram of $\mf g$ depends of the highest root $\theta$. In fact, one can see in Figure \ref{fig:dynkinkmalgebras} that the affine root $\alpha_0$ connects differently depending on the finite algebra $\ag$.

 Defining new indices $I,J=(0,i)$, we can write the Cartan matrix of $\mf g^+$ in terms of scalars products between the new simple roots $\Pi^+= \{\alpha_0, \alpha_1, \ldots, \alpha_n\}$ as follows
\be
A^{+}_{IJ}= 2\, \frac{(\alpha_I|\alpha_J)}{(\alpha_I|\alpha_I)}\cvp
\ee


\begin{figure}[t]
\begin{center}
{\scalebox{0.17}
{\includegraphics{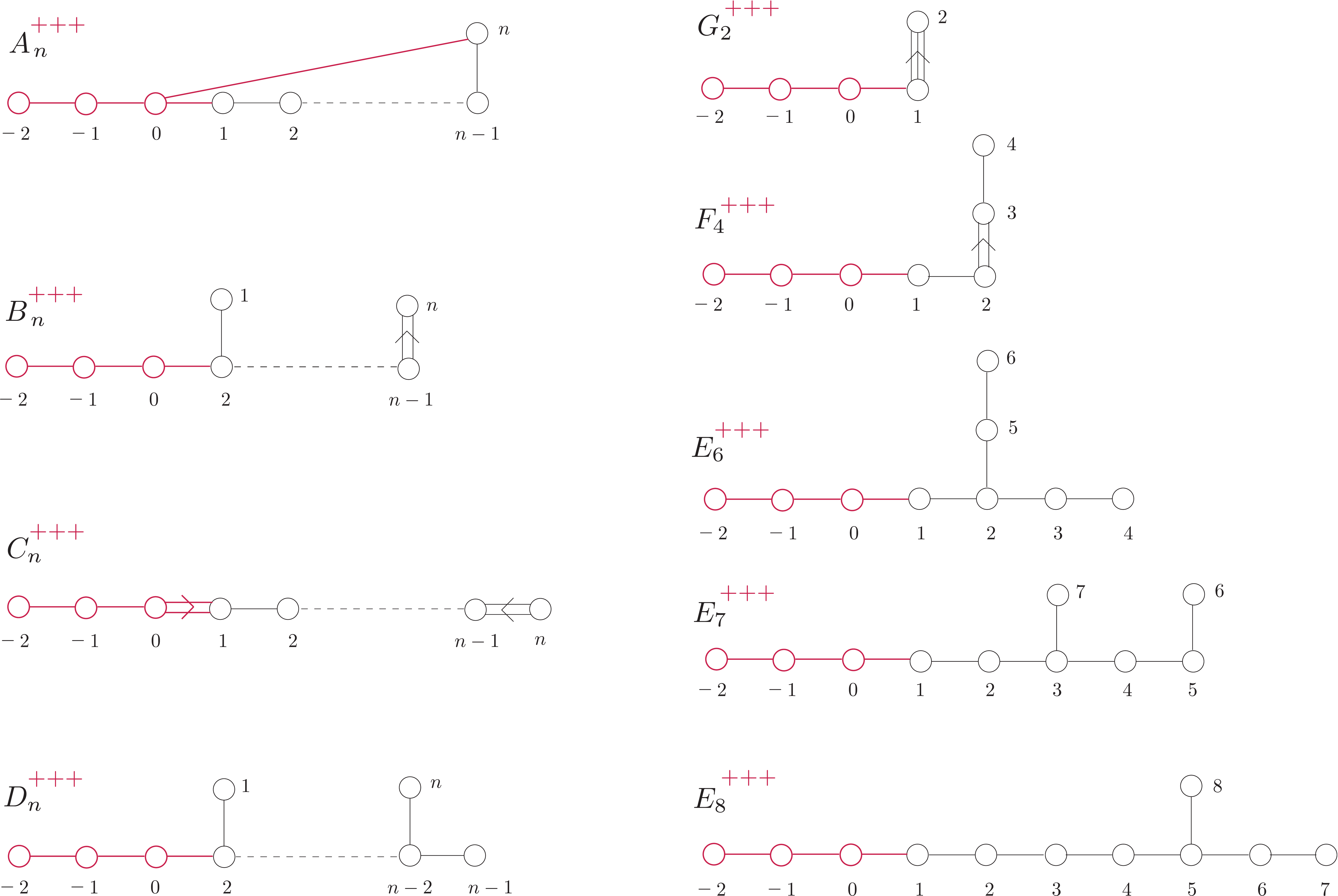}}}
\end{center}
\caption{ \small \textsl{Dynkin diagrams of extensions of finite simple Lie algebras. The nodes labeled by $0$ represent affine roots, nodes labeled by $-1$ represent over-extended roots and by $-2$ very-extended roots.}} \label{fig:dynkinkmalgebras}
\end{figure}

\subsubsection{The over-extension}
In order to obtain a double extension, we will include the second basis vector $u_2 \in \mbb Z^{1,1}$ and adding a new simple root $\alpha_{-1}$ known as the\emph{ over-extended root} defined by
\be
\alpha_{-1} \equiv -u_1-u_2\, ,
\ee
which has non-vanishing scalar product only with $\alpha_0$:
\be
(\alpha_{-1}|\alpha_0)=-1, \quad (\alpha_{-1}|\alpha_i)=0, \quad i=1, \ldots n.
\ee
As we have $(\alpha_{-1}|\alpha_{-1})=(\alpha_{0}|\alpha_{0}) =2$, the node associated to $\alpha_{-1}$ will always attach by a single line to the affine node $\alpha_0$. We can see this statement in the Figure \ref{fig:dynkinkmalgebras}. As before, we define new indices $M,N=(-1,0,i)$ and we write the Cartan matrix of $\agpp$ in terms of scalars products between the new simple roots $\Pi^{++}= \{\alpha_{-1}, \alpha_0, \ldots, \alpha_n\}$ as follows
\be
A^{++}_{MN}= 2\, \frac{(\alpha_M|\alpha_N)}{(\alpha_M|\alpha_M)}\cvp
\ee

\subsubsection{The very-extension}

In order to obtain a triple extension of the finite Lie algebra $\ag$, we introduce another two-dimensional lattice $\tilde{\mbb Z}^{1,1}$ spanned by the basis vectors $v_1$ and $v_2$ which  have the same scalar product as for $\mbb Z^{1,1}$ (see \eqref{scalar156} ). To ensure that the Lorentzian signature of the extended Kac-Moody algebra is preserved we have to introduce a spacelike vector of $\tilde{\mbb Z}^{1,1}$. Therefore, we will include $v_1+v_2$ because $(v_1+v_2| v_1+v_2)= 2$. We can now introduce the new  simple root known as \emph{very-extended root} as
\be
\alpha_{-2}= u_1 -(v_1+v_2).
\ee
Again this root is spacelike $(\alpha_{-2}|\alpha_{-2})=2$ and the associated node in the Dynkin diagram connects always with a single line to the node corresponding to the over-extended root $\alpha_{-1}$ (see Figure \ref{fig:dynkinkmalgebras}) by virtue of the following scalar product
\be
(\alpha_{-2}|\alpha_{-1})=-1, \quad (\alpha_{-2}|\alpha_{I})=0, \quad I=0, 1\ldots n.
\ee
Introducing the new indices $A,B=(-2,-1,0,i)$ we can once again define the Cartan matrix of $\agppp$ in terms of scalar product between the simple roots ${\Pi^{+++}=\{\alpha_{-2},\,  \alpha_{-1}, \, \alpha_{0},} \linebreak[1] {  \,  \alpha_1, \ldots ,  \alpha_n \} }$ as follows
\be
A^{+++}_{AB}= 2\, \frac{(\alpha_A|\alpha_B)}{(\alpha_A|\alpha_A)} \cvp
\ee
\section{Real forms of complex Kac-Moody algebras}\label{sec:realkm}
\setcounter{equation}{0}

In the previous sections, we studied Kac-Moody algebras over the field $\mbb C$ of complex numbers. But later, we are going to deal with Kac-Moody algebras defined on real numbers. One of the real forms $\ag_r$ of the complex algebra $\ag$ is trivially obtained such that  $\ag_r$ possesses the same Chevalley-Serre presentation as the complex algebra $\ag$ but with coefficients restricted to be real numbers. Other real forms of complex Kac-Moody algebras $\ag$ are obtained in a more complicate way.

We will present in this section a quick recap of the theory of real forms of Kac-Moody algebras insisting only on the tools necessary for the Part III. We refer the reader to \cite{Henneaux:2007ej,Helgason:1978} for more details on the theory of real forms of complex semi-simple Lie algebras and to \cite{Rousseau1989,Rousseau1995,BenMessaoud,Tripathy} for informations on real forms associated to infinite-dimensional Kac-Moody algebras.

\subsection{Definitions}

If $\ag$ is a complex Kac-Moody algebra, the subalgebra $\ag_{r}$ of $\ag$ is a \emph{real form} of $\ag$ if $\ag$ is the complexification of $\ag_{r}$. Therefore the complex Kac-Moody algebra $\ag$ can be decomposed when considered as a real algebra as
\be
\ag^{\mbb R}= \ag_{r} \oplus i\, \ag_{r}.
\ee
In other words, a real form of a complex algebra exists if and only if we may choose a basis of the complex algebra such that all the structure constants become real. Such a real form $\ag_r$ determines a mapping $\sigma: \ag \rightarrow \ag$: $x+iy \rightarrow x-iy, \ (x, y \in \ag_r)$ which has the following properties:
\begin{itemize}
\item $\sigma$ is semilinear, i.e. $\sigma(\lambda\,  x+ \rho\,  y )= \bar{\lambda}\,  \sigma(x) + \bar{\rho} \, \sigma(y)$ for $x, y \in \ag$ and $\lambda, \rho \in \mbb C$,
\item $\sigma$ is an involution i.e. $\sigma^2=1$,
\item $\sigma [x,y]= [\sigma(x), \sigma(y)]$ for $x, y \in \ag$.
\end{itemize}
A map $\sigma:\ag \rightarrow \ag $ with these properties is a bijection called \emph{semilinear involution} of $\ag$. Conversely any semilinear involution $\sigma$ of $\ag$ determines a real subalgebra $\ag_r$, known as a real form of $\ag$,  which corresponds to the fixed point subalgebra of $\sigma$:
\be \label{realformdef}
\ag_r= \{x \in \ag \,| \, \sigma\, (x)= x  \}. 
\ee
Therefore, on $\ag$, real forms and  the semilinear involution $\sigma$ are in one-to-one correspondence.

Let us focus first on complex finite Lie algebras. In this case there are two particular real forms associated to these algebras:  the \emph{split real forms} and the \emph{compact real forms}. To define them, we will complete the Chevalley generators defined in Section \ref{subsec:chevalley} into a full basis, the so-called \emph{Cartan-Weyl} basis. Let $\ag$ be the complex Lie algebra, $\mf h$ a Cartan subalgebra of $\ag$ and $\Delta$ the set of roots. It is possible to choose root vectors $x_{\alpha} \in \ag_{\alpha}$ for each $\alpha \in \Delta$ such that the following conditions hold for all $\alpha, \beta \in \Delta$
\be \label{cartanweylbas} \begin{split}
[h, x_{\alpha}]&= \alpha(h)\, x_{\alpha},\\
[x_{\alpha},  x_{\beta}]&= \Bigg{ \{ }\begin{matrix} h_{\alpha}&  \mathrm{if}\  \alpha +\beta=0, \\ N_{\alpha, \beta}\, x_{\alpha+ \beta}&  \mathrm{if}\ \alpha+\beta \in \Delta, \\ 0 & \ \mathrm{if}\ \alpha+\beta \notin \Delta, \end{matrix} 
\end{split}
\ee
and the constants $N_{\alpha, \beta}$ satisfy
\be \begin{split} 
N_{\alpha, \beta}&= - N_{-\alpha,- \beta} = - N_{  \beta, \alpha} ,\\
N_{\alpha, \beta}^{2}&= \frac{1}{2}\, q\, (p+1) (\alpha| \alpha),
\end{split}
\ee
where $p$ and $q$ are non negative integers such that the string of all vectors $\beta +n \alpha$ belongs to  $\Delta$ for $-p\leqslant n\leqslant q$.

We obtain a real form by defining
\be \label{cartanrealf}
\mf h_r= \{h \in \mf h\, |\, \alpha(h) \in \mbb R \ \forall \alpha \in \Delta \,\},
\ee
therefore the structure constants of the commutation relations \eqref{cartanweylbas} may all chosen real and we get a real Lie algebra $\mf s_r$ called split real form which reads explicitly
\be
\mf s_r= \mf h_r \oplus \bigoplus_{\alpha \in \Delta} \mbb R \, x_{\alpha}.
\ee

Another special real form which exists for every complex semi-simple Lie algebra is called \emph{compact real form}. A compact real form is a real form that is compact Lie algebra i.e. if the analytic group Int $\ag$ of inner automorphism is compact. Moreover a theorem states that  if the Killing form on $\mf u_r$ is negative definite, $\mf u_r$ is called  compact real form of $\mf g$. We construct it from the split real form as
\be\label{compactrealf}
\mf u_r= \bigoplus_{\alpha \in \Delta} \mbb R\, i \,h_{\alpha} \oplus \bigoplus_{\alpha \in \Delta} \mbb R \, (x_{\alpha} - x_{-\alpha}) \oplus \bigoplus_{\alpha \in \Delta}\mbb R\,  i \, (x_{\alpha} + x_{-\alpha}).
\ee 

The theory of real forms of complex Kac-Moody algebras generalizes the theory of real semi-simple Lie algebra. But there is  one important difference between them.  The non-conjugacy of the two Borel algebra $\mf b_+$ and $\mf b_-$ yields in the infinite-dimensional case to two different classes of real forms. Indeed, the standard upper triangular and lower triangular Borel subalgebras, $\mf b_+$ and $\mf b_-$, cannot be conjugated into one another \cite{PetersonKac} and depending on whether the semi involution $\sigma$ fixing the real form maps $\mf b_+ \rightarrow \mf b_+$ or $\mf b_- \rightarrow \mf b_-$ the real form is called almost split or almost compact \cite{Rousseau1989,Rousseau1995,BenMessaoud}. Almost split algebras are under better control.

\subsection{Cartan Involution and Cartan decomposition}

An involution $\theta$ of a real semi-simple Lie algebra $\mf g_r$ such that the symmetric bilinear form $B_{\theta}$ defined by $B_{\theta}(x, y) \equiv - B (x, \theta y) \quad \mathrm{for} \ x, y \in \mf g$ is positive definite, is called \emph{Cartan involution}. This involution defines an eigenspace decomposition in an eigenspace $\mf k_r$ to the eigenvalue $+1$ and an eigenspace $\mf p_r$ to the eigenvalue $-1$. The direct sum  decomposition as vector spaces which reads
\be\label{cartandecreal}
\mf g_r= \mf k_r \oplus \mf p_r,
\ee
is called \emph{Cartan decomposition} because the following bracket laws are satisfied (as in the complex case see \eqref{symmetricspace})
\be
 [\mf p_r, \mf p_r] \subset \mf k_r, \quad [\mf k_r, \mf p_r] \subset \mf p_r, \quad [\mf k_r, \mf k_r] \subset \mf k_r, 
\ee
and the bilinear form on $\ag_r$ is negative definite on $\mf k_r$ and positive definite on $\mf p_r$. We will see in Section \ref{subsec:classrealfo} that the involution $\theta$ can be related to the semi-involution $\sigma$.
\subsection{Restricted roots  and Iwasawa decomposition }\label{subsec:restiwa}

Let $\mf a$ be a maximal abelian subspace of $\mf p_r$. The set $\{ \mathrm{ad} h \,|\, h\in \mf a      \}$ is a commuting set of symmetric transformations that can be simultaneously diagonalized on $\mbb R$. Accordingly we may decompose $\mf g_r$ into a direct sum of eigenspaces labelled by elements of the dual space $\mf a^*$:
\be
\mf g_r = \bigoplus_{\lambda} \mf g_{\lambda}, \quad \mf g_{\lambda}= \{x \in \mf g_r\, |\, \mathrm{ad}h(x)= [h,x]= \lambda (h)\, x\}.
\ee
If $\ag_{\lambda} \neq 0$ and $\lambda \neq 0$, the label $\lambda$ is called \emph{restricted root} of $\ag_r$ and $\ag_{\lambda}$ defines \emph{the restricted root space} with its elements called \emph{restricted root vectors}. The set of restricted roots is denoted by $\Sigma$. Restricted roots and root spaces obey to the following properties
\be \begin{split}
&\ \bullet \ag_r = \ag_0 \oplus \bigoplus_{\lambda \in \Sigma}  \ag_{\lambda} \ \mathrm{is\ an\ orthogonal\ direct\ sum}, \\
&\ \bullet [\ag_{\lambda}, \ag_{\mu}] \subseteq \ag_{\lambda + \mu}, \\
& \ \bullet \theta \ag_{\lambda}= \ag_{- \lambda}, \\
& \ \bullet \lambda \in \Sigma \Rightarrow - \lambda \in \Sigma\, \\
&\ \bullet \ag_0 = \mf a\oplus \mf m    \ \mathrm{where} \ \mf m=\{ x \in \mf k_r  |\, [x,y]=0, \ \forall y \in \mf a           \}.
\end{split}
\ee
If $\mf t$ is a maximal abelian subalgebra of $\mf m$, the subalgebra $\mf h_r= \mf a \oplus \mf t$ is a Cartan subalgebra\footnote{$\mf h_r$ is Cartan subalgebra of the real form $\mf g_r$ iff its complexification is  a Cartan subalgebra of the complex algebra $\mf g$.} of the real algebra $\mf g_r$. If $\mf m =0$, then $\ag_r$ is a split real form of $\ag$. 

As for all root systems, we can choose a notion of positivity on $\mf a^*$ and find a way to split the restricted root system $\Sigma$ into positive and negative ones.  Let us consider $\Sigma^+$, the set of positive restricted roots, we can define the nilpotent subalgebra $\mf n_r$ by
\be
\mf n_r = \bigoplus_{\lambda \in \Sigma^+} \ag_{\lambda}.
\ee
Therefore the algebraic Iwasawa decomposition for the real algebra $\ag_r$ becomes
\be\label{Iwasawadecreal}
\mf g_r = \mf k_r \oplus \mf a \oplus \mf n_r.
\ee 

\subsection{Classification of the real forms and Tits-Satake diagrams} \label{subsec:classrealfo}
An interesting aspect in classifying real Kac-Moody algebras is the conjugacy classes of their Cartan subalgebras. In fact, for real Kac-Moody algebras, the Cartan subalgebras are not all conjugate to each other but each Cartan subalgebra $\mf h_r$ is conjugate via Int $\ag_r$ to a $\theta$-\emph{stable Cartan subalgebra} which is a Cartan subalgebra $\mf h_r$ of $\mf g_r$ such that $\theta(\mf h_r) \subset \mf h_r$. Therefore, it is sufficient to focus on $\theta$-stable Cartan subalgebras to study and classify the real Kac-Moody algebras. Let $\mf h_r$ a $\theta$-stable Cartan subalgebra of $\ag_r$, one can decompose $\mf h_r$ into compact and non-compact parts,
\be
\mf h_r= \mf t \oplus \mf a\quad  \mathrm{with} \quad \mf t= \mf h_r \cap \mf k_r, \ \mf a= \mf h_r \cap \mf p_r.
\ee
A $\theta$-stable Cartan subalgebra $\mf h_r$ is called \emph{maximally non compact} (resp. \emph{maximally compact}) if the dimension of its non-compact (compact) part $\mf a$ ($\mf t$) is as large as possible.\footnote{The Cartan subalgebra $\mf h_r$ introduced previously in Section \ref{subsec:restiwa} to define restricted roots, is a maximally non-compact $\theta$-stable subalgebra because its non-compact part $\mf a
$ was defined as the maximal abelian subspace of $\mf p_r$.} The \emph{real rank} of $\ag_r$ is the dimension of its maximally non-compact Cartan subalgebras.

If $\mf h_r$ is a $\theta$-stable Cartan subalgebra, its complexification $\mf h$ is also stable under the involutive automorphism $\theta$. It is possible to extend the action of $\theta$ from $\mf h$ to $\mf h^*$ by duality:
\be \label{thetaroot}
 \theta (\alpha)(h)= \alpha(\theta^{-1} (h))= \alpha(\theta (h)) \quad  \forall h \in \mf h \ \mathrm{and} \ \forall \alpha \in \mf h^*.
\ee
The  Cartan involution $\theta$ allows to define a closed subsystem\footnote{A system $\Delta$ is closed if $\alpha, \beta \in \Delta \Rightarrow -\alpha \ \mathrm{and} \ \alpha+\beta \in \Delta$. } $\Delta_0 \subset \Delta$
\be
\Delta_0= \{ \alpha \in \Delta \ | \ \theta (\alpha)= \alpha \}.
\ee
This condition is equivalent to ask that $\alpha$ vanish on $\mf a$, indeed if $x \in \mf a \Rightarrow \alpha (\theta(x)) = \alpha (-x)= - \alpha(x)$ and using \eqref{thetaroot} we get that $\alpha (x) = -\alpha(x)$ if $\theta (\alpha)= \alpha$. Then $\Delta_0$ project to zero when restricted to the maximally non-compact subalgebra.

 Therefore, from the knowledge of $\theta$, we may obtain the restricted root system by projecting the root space according to
\be \label{restalpha}
\Delta \rightarrow \bar{\Delta}: \alpha \rightarrow \bar{\alpha}= \frac{1}{2} (\alpha - \theta(\alpha)),
\ee
and restricting their action on $\mf a$ since $\alpha$ and $- \theta(\alpha)$ project on the same restricted root.

 Recall that we used Dynkin diagrams to classify some classes of complex Kac-Moody algebras, we will use something similar in the classification of real forms of Kac-Moody  algebras. One can use Vogan diagrams which rest on maximally compact $\theta$-stable subalgebra or in another way, one can use Tits-Satake diagrams that are based on maximally non-compact $\theta$-stable subalgebra.
 As the Iwasawa decomposition \eqref{Iwasawadecreal}, that we will intensely use in  this work, emphasizes the role of the maximally non-compact Cartan subalgebra $\mf a$, we will describe the Tits-Satake diagrams.
 
 A \emph{Tits-Satake diagram} of a real Kac-Moody algebra $\ag_r$  consists in the Dynkin diagram of the complex Kac-Moody algebra $\ag$ with  additional informations:
 \begin{itemize}
 \item Nodes painted in black $\bullet$ are associated to simple roots invariant under the Cartan involution $\theta$: $\theta(\alpha_i) = \alpha_i$. Therefore the Tits-Satake diagrams of compact real forms $\mf u_r$ are simply obtained by painting in black all the roots of the standard Dynkin diagram.
 \item White nodes $\circ$ stand for simple roots not invariant under the involution $\theta$.
 \item  The basis of simple roots $\Pi$ of the complex algebra $\ag$ can always be split into two subsets: $\Pi_0= \{\alpha_{r+1}\ldots \alpha_n\}= \Pi \cap \Delta_0$ whose elements are fixed by $\theta$ (corresponding to black nodes) and $\Pi \backslash \Pi_0= \{ \alpha_1, \ldots \alpha_r\} $ (corresponding to white nodes) such that
 \be \label{thetatits}
 \forall \alpha_k \in \Pi \backslash \Pi_0: \ \theta(\alpha_k)= -\alpha_{\pi(k)} -  \sum_{j=r+1}^{n} a_k^j \, \alpha_j,
 \ee
 where $\pi$ is an involutive permutation of the $r$ indices of the elements $\Pi \backslash \Pi_0$ and $a_k^j$ are non-negative integers which are obtained by solving the linear system given by 
 \be
 (\theta(\alpha_k)+ \alpha_{\pi(k)}\, |\, \alpha_q)= \sum_{j=r+1}^n a_k^j (\alpha_j | \alpha_q).
 \ee
  If $\pi(k)= i$, then a double arrow connects $\alpha_i$ to $\alpha_k$ 
\begin{figure}[ht]  
\begin{center}
\begin{pgfpicture}{0cm}{0cm}{1cm}{1cm}

\pgfnodecircle{Node1}[stroke]{\pgfxy(0,0)}{0.1cm}

\pgfnodecircle{Node2}[stroke]{\pgfxy(1.5,0)}{0.1cm}

\pgfnodebox{Node6}[virtual]{\pgfxy(0,-0.4)}{$\alpha_{i}$}{2pt}{2pt}
\pgfnodebox{Node7}[virtual]{\pgfxy(1.5,-0.4)}{$\alpha_{k}$}{2pt}{2pt}
\pgfline{\pgfxy(0.1,0)}{\pgfxy(0.3,0)} \pgfline{\pgfxy(0.4,0)}{\pgfxy(0.6,0)} \pgfline{\pgfxy(0.9,0)}{\pgfxy(1.1,0)} \pgfline{\pgfxy(1.2,0)}{\pgfxy(1.4,0)} \pgfline{\pgfxy(1.6,0)}{\pgfxy(1.8,0)}  \pgfline{\pgfxy(1.9,0)}{\pgfxy(2.1,0)} \pgfline{\pgfxy(-0.3,0)}{\pgfxy(-0.1,0)} \pgfline{\pgfxy(-0.6,0)}{\pgfxy(-0.4,0)}

\pgfsetstartarrow{\pgfarrowtriangle{4pt}}
\pgfsetendarrow{\pgfarrowtriangle{4pt}}
\pgfnodesetsepend{5pt}
\pgfnodesetsepstart{5pt}
\pgfnodeconncurve{Node2}{Node1}{100}{80}{1cm}{1cm}
\end{pgfpicture} 
\end{center} 
\end{figure}
meaning that
$\alpha_i- \theta(\alpha_i)= \alpha_k- \theta (\alpha_k)$.
 Accordingly to \eqref{restalpha}, this condition is equivalent to ask that $\alpha_i$ and $\alpha_k$ take the same value on $\mf a$ and then project on the same restricted root $\bar{\alpha_i}= \bar{\alpha_k} $.
 \end{itemize} 
 The Tits-Satake diagram determines completely a real form. In fact, from this diagram one can extract the action of $\theta$ on the simple roots. In order to determine the real form $\ag_r$, the action of $\sigma$ on the generators is needed because the real form corresponds to the fixed point subalgebra of $\sigma$ (see \eqref{realformdef}). It can shown that  $\sigma= \theta . \tau$ where $\tau$ is the  involution fixing the maximally compact real form $\mf u_r$ given by \eqref{compactrealf}. Note that  there is a theorem that states that given a real form $\ag_r$ of the complex algebra $\ag$, there   exists always a compact real form $\mf u_r$ associated with it such that the corresponding involution $\tau$ commutes with $\sigma$.
 
 Let $h_{\alpha}$ and $x_{\alpha}$ be respectively the Cartan and  the non zero vector associated with the root $\alpha$, it is possible to extend the action of the $\theta$ involution from $\mf h^*$ to the whole algebra using \eqref{cartanweylbas}  by
 \be \label{thetagenreal}\begin{split}
 \theta(h_{\alpha})&= h_{\theta(\alpha)},\\
 \theta(x_{\alpha})&= \rho_{\alpha}\, x_{\theta (\alpha)},
 \end{split}\ee
 where the constant $\rho_{\alpha}$ satisfies the following relations
 \be
 \rho_{\alpha} \rho_{\theta(\alpha)} =1, \quad \rho_{\alpha} \rho_{-\alpha}=1, \quad \rho_{\alpha} \rho_{\beta} N_{\theta(\alpha), \theta(\beta)}= \rho_{\alpha+ \beta}\, N_{\alpha, \beta}, \quad \rho_{\gamma}= \pm 1 \ (\mathrm{if} \ \gamma \in \Delta_0).
 \ee
 The commutativity of $\sigma$ and $\tau$ implies that
 \be \label{sigmawrttheta}  \begin{split}
 \sigma (\alpha)&= - \theta(\alpha)\\
 \sigma (h_{\alpha})&= h_{\sigma (\alpha)}\\
 \sigma (x_{\alpha})&= - \bar{\rho}_{\alpha}\, x_{\sigma(\alpha)}.
 \end{split} \ee
 
 All the finite real semi-simple Lie algebras have been classified \cite{Araki,Helgason:1978}. The Tits-Satake diagrams of all affine Kac-Moody algebras can be found in \cite{Tripathy} and a complete list of almost split real forms for hyperbolic Kac-Moody algebras is established in \cite{BenMessaoud}.
 
 Let us enclose the theory of real form by defining a notation to design a real form related to a complex algebra. To this end, we define the \emph{character} $\Xi$ of a real form as the number of non-compact generator minus the number of compact generators, namely
 \be
 \Xi = \mathrm{dim}\, \mf{g}_r - 2\, \mathrm{dim}\, \mf{k}_r.
 \ee
 Let $X_{n}$ the complex algebra of rank $n$, then we write $X_{n(n)}$ for  the associated split real form and $X_{n(\Xi)}$ for the other real forms   where the number in bracket is the character of the real form. For instance, $E_{8(8)}$ denotes the split real form associated to the complex Lie algebra $E_8$. 
 
\subsection{$\asu$: a non-split real form of $A_2$ and its extension $\asuppp$}\label{subsec:su21andextension}

Lets us now illustrate the theory of real forms by studying the real algebra $\asu$ which is one of the real forms of the complex algebra $A_2=\mf{sl}(3, \mbb C)$. In the following, this algebra and its associated group $\su$ will play an important role since $\su$ is the global symmetry group of Einstein-Maxwell theory in the presence of a Killing vector \cite{Kinnersley}.  

After reviewing basics features of $\mathfrak{sl}(3, \mbb C)$ and its different real forms, we  will fix the real algebra $\asu$ from its Tits-Satake diagram. Then we will give a complete list of the generators belonging to $\asu$ and to its maximal compact part $\mf k_r= \mf{su}(2) \oplus \mf u(1)$. We will also investigate the restricted root of $\asu$. Finally, we will discuss the very-extended real algebra $\asuppp$.\\
This section is based on \cite{Houart:2009ed}.

 \subsubsection{Generalities on the complex algebra $\boldsymbol{A_2=\mf{sl}(3, \mbb C)}$}
The complex algebra $A_2=\mf{sl}(3, \mbb C)$ is described by the following Cartan matrix:
\be
A [A_2]=(A_{ij})_{i,j =4,5} =  \left(\begin{array}{cc}2 & -1 \\-1 & 2\end{array}\right)\, 
\ee
corresponding to the Dynkin diagram displayed in Figure \ref{figa:su21}a. The rank of this algebra is $2$ and its Chevalley generators $\{h_i, e_i, f_i\}_{ i=4,5}$ satisfy the Chevalley relations \eqref{eqn:chevalley}. Note that the labeling of these generators is chosen to leave room for the extension of $A_2$ to the very-extended Kac-Moody algebra $A_2^{+++}$ done at the end of this section. Using the Serre relations \eqref{eqn:serre}, we get
\be
(\mathrm{ad}\, e_4)^2 e_5=0, \quad (\mathrm{ad}\, f_4)^2 f_5=0.
\ee 
This means that $[e_4, [e_4,e_5]]=0$ and $[f_4, [f_4,f_5]]=0$ and the only multicommutators which survive to the Serre relations are 
\be
e_{4,5}= [e_4,e_5], \quad f_{4,5}= - [f_4,f_5].
\ee
This algebra is then eight-dimensional and according to \eqref{eqn:triangular}, it can be decomposed as
 \be
 \mathfrak{sl} (3,\mathbb{C})= \bigg(\bigoplus_{k=4}^5 \mathbb{C}\, f_k \oplus \mathbb{C} \,f_{4,5}\bigg) \oplus \bigoplus_{k=4}^5 \mathbb{C}\, h_k \oplus  \bigg( \bigoplus_{k=4}^5 \mathbb{C}\, e_k \oplus \mathbb{C} \, e_{4,5} \bigg),
 \ee
 where the eight generators $\{h_i, e_i, f_i\}$ have the following matrix realization in the fundamental representation
 \begin{align} \label{eqn:chevgenA2}
 &h_4 = \left(\begin{array}{ccc}1 & 0 & 0 \\0 & -1 & 0 \\0 & 0 & 0\end{array}\right),  &&h_5= \left(\begin{array}{ccc}0 & 0 & 0 \\0 & 1 & 0 \\0 & 0 & -1\end{array}\right),  \\
 &e_4= \left(\begin{array}{ccc}0 & 1 & 0 \\0 & 0 & 0 \\0 & 0 & 0\end{array}\right), &&e_5=\left(\begin{array}{ccc}0 & 0 & 0 \\0 & 0 & 1 \\0 & 0 & 0\end{array}\right), \quad e_{4,5}= [e_4, e_5]= \left(\begin{array}{ccc}0 & 0 & 1 \\0 & 0 & 0 \\0 & 0 & 0\end{array}\right) \nn,
 \end{align}
 as well as 
 \be
 f_i=(e_i)^{T}. 
\ee 
It easy to check that these generators satisfy the Chevalley-Serre relations \eqref{eqn:chevalley} and \eqref{eqn:serre}.

Let us now focus on the roots of $A_2$. This algebra have two simple roots $\alpha_4$ and $\alpha_5$ and using \eqref{componentsiroot} we get the components of the simple roots vectors $\vec{\alpha}_i=(\alpha_i(h_4), \alpha_i(h_5))$
\be
\vec{\alpha_4}= (2, -1), \quad  \vec{\alpha_5}= (-1, 2).
\ee
In addition to the roots $\pm \alpha_4$ and $\pm \alpha_5$, the algebra $A_2$ have also the roots $\pm (\alpha_4+ \alpha_5)$ as they correspond to the generators $e_{4,5}$ and $f_{4,5}$ which survived to the Serre relations. The whole root system $\Delta$ of $A_2$ is then spanned by $\{\pm \alpha_4, \pm \alpha_5, \pm (\alpha_4 + \alpha_5)\}$ and it is displayed in Figure \ref{fig:restrictedrootdiag}.

Note that the complex Lie algebra $\mf{sl}(3,\mbb C)$ has three real forms:
\begin{itemize}
\item the split real form $\mf{sl}(3,\mbb R)$ obtained by restricting the base field to the real numbers. The corresponding Tits-Satake diagram is just the Dynkin diagram of $A_2$ displayed in Figure \ref{figa:su21}a;
\item the compact real form $\mf{su}(3)$. The associated semi-involution $\tau$ is defined as $\tau(X)=-X^{T}, \ X\in \mf{sl}(3,\mbb C)$ and the combinations of generators which are pointwise fixed by $\tau$ defines $\mf{su}(3)$ as 
\be \nn
\mf{su}(3)= \bigoplus_{k=4}^5 \mbb R (i h_k) \oplus \bigoplus_{k=4}^5 \mbb R(e_k-f_k) \oplus \mbb R(e_{4,5}-f_{4,5}) \oplus \bigoplus_{k=4}^5 \mbb R\, i\, (e_k+f_k) \oplus \mbb R\,  i\,  (e_{4,5}+f_{4,5})
\ee
The corresponding Tits-Satake diagram is the Dynkin diagram of $A_2$ where the two nodes are painted in black (see Figure \ref{figa:su21}b); 
\item the non-split real form $\asu$ whose Tits-Satake diagram is displayed in Figure \ref{figa:su21}c. 
\end{itemize}
\begin{figure}[ht]
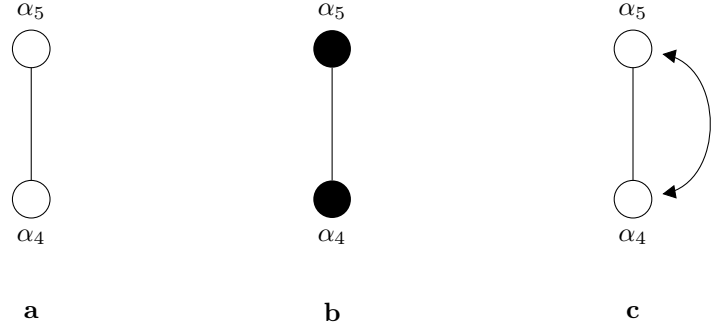

\begin{center}
\begin{pgfpicture}{0cm}{-1cm}{5cm}{4cm}

\pgfnodecircle{Node10}[stroke]{\pgfxy(-2,0.5)}{0.25cm}
\pgfnodecircle{Node20}[stroke]
{\pgfrelative{\pgfxy(0,2)}{\pgfnodecenter{Node10}}}{0.25cm}

\pgfnodebox{Node60}[virtual]{\pgfxy(-2,0)}{$\alpha_{4}$}{2pt}{2pt}
\pgfnodebox{Node70}[virtual]{\pgfxy(-2,3)}{$\alpha_{5}$}{2pt}{2pt}
\pgfnodeconnline{Node10}{Node20} 
\pgfnodebox{Node36}[virtual]{\pgfxy(-2,-1)}{\textbf{a}}{2pt}{2pt}

\pgfnodecircle{Node10}[fill]{\pgfxy(2,0.5)}{0.25cm}
\pgfnodecircle{Node20}[fill]
{\pgfrelative{\pgfxy(0,2)}{\pgfnodecenter{Node10}}}{0.25cm}

\pgfnodebox{Node60}[virtual]{\pgfxy(2,0)}{$\alpha_{4}$}{2pt}{2pt}
\pgfnodebox{Node70}[virtual]{\pgfxy(2,3)}{$\alpha_{5}$}{2pt}{2pt}
\pgfnodeconnline{Node10}{Node20} 
\pgfnodebox{Node36}[virtual]{\pgfxy(2,-1)}{\textbf{b}}{2pt}{2pt}

\pgfnodecircle{Node1}[stroke]{\pgfxy(6,0.5)}{0.25cm}
\pgfnodecircle{Node2}[stroke]
{\pgfrelative{\pgfxy(0,2)}{\pgfnodecenter{Node1}}}{0.25cm}

\pgfnodebox{Node6}[virtual]{\pgfxy(6,0)}{$\alpha_{4}$}{2pt}{2pt}
\pgfnodebox{Node7}[virtual]{\pgfxy(6,3)}{$\alpha_{5}$}{2pt}{2pt}
\pgfnodeconnline{Node1}{Node2} 

\pgfsetstartarrow{\pgfarrowtriangle{4pt}}
\pgfsetendarrow{\pgfarrowtriangle{4pt}}
\pgfnodesetsepend{5pt}
\pgfnodesetsepstart{5pt}
\pgfnodeconncurve{Node2}{Node1}{-10}{10}{1cm}{1cm}

\pgfnodebox{Node63}[virtual]{\pgfxy(6,-1)}{\textbf{c}}{2pt}{2pt}

\end{pgfpicture}
\caption {\label{figa:su21} \sl \small Tits-Satake diagrams of the real forms of $A_2$. \textbf{a}: Dynkin diagram of $A_2$ corresponding also to the Tits-Satake diagram of the split real form $\mf{sl}(3, \mbb R)$. \textbf{b}:  Tits-Satake diagram of the compact real form $\mf{su}(3)$. \textbf{c}:  Tits-Satake diagram of the non-split real from $\asu$. The labeling of nodes is chosen to leave room for the extension of  $A_2$ to the very-extended Kac-Moody algebra $A_2^{+++}$. } 
\end{center}
\end{figure}
\subsubsection{$\boldsymbol{\asu}$: definitions}
The real form $\asu$ is the Lie algebra of $3 \times 3$ complex traceless matrices X, subject to the constraint
\be
\eta \, X+ X^{\dag} \, \eta=0 \, ,
\ee
with
\be
\eta= \left(\begin{array}{ccc}0 & 0 & -1 \\0 & 1 & 0 \\-1 & 0 & 0\end{array}\right)\, .
\ee
The Tits-Satake diagram of $\asu$  displayed in Figure \ref{figa:su21}c determines completely this real form.
In fact, from this Tits-Satake diagram, we can extract the action of the involution $\theta$ on the simple roots of $A_2$. Following \eqref{thetatits}, the presence of the double arrow means that
\be
\alpha_4 -\theta(\alpha_4)= \alpha_5 -\theta(\alpha_5).
\ee
Since there are no black nodes, this implies
\be \label{thetaroots}
 \theta(\alpha_4) = - \alpha_5, \   \theta (\alpha_5) = - \alpha_4\, .
  \ee
This infers, using \eqref{sigmawrttheta}, that the action of the semi-involution $\sigma$ on the simple roots is
 \be \label{eqn:sigmaroots}
 \sigma (\alpha_4) = \alpha_5, \   \sigma (\alpha_5) = \alpha_4\, ,
  \ee
which implies that on the generators of $\mf{sl}(3, \mbb C)$ we have
\be \begin{split}
  \begin{aligned}\label{eqn:sigmagens}
\s (h_4)&=h_5, & \s (h_5)&= h_4,  \\
 \s (e_4)&=e_5, &  \s (e_5)&= e_4, & \s (e_{4,5})&=-  e_{4,5}\,, \\
 \s (f_4)&=f_5, & \s (f_5)&= f_4, & \s (f_{4,5})&=-  f_{4,5}\,. 
\end{aligned} \end{split} \ee
We choose according to \eqref{sigmawrttheta} that $\rho_{\alpha_4}= \rho_{\alpha_5}=-1$.
\noindent The generators of $\asu$ correspond to the ones which are fixed by $\sigma$ and they can be written in terms of the generators of $\mathfrak{sl}(3, \mathbb{C})$ (\ref{eqn:chevgenA2}) as

\be \label{eqn:su21gens}\begin{split} 
\begin{aligned}
\mathbf{h_4}&=h_4+ h_5 &&=  \left(\begin{array}{ccc}1 & 0 & 0 \\0 & 0 & 0 \\0 & 0 & -1\end{array}\right), & \mathbf{h_5}&= i (h_4-h_5)&&= \left(\begin{array}{ccc}i & 0 & 0 \\0 & -2 i & 0 \\0 & 0 & i\end{array}\right)\, , \\
\mathbf{e_4}&= e_4 +e_5&&= \left(\begin{array}{ccc}0 & 1 & 0 \\0 & 0 & 1 \\0 & 0 & 0\end{array}\right), & \mathbf{f_4}&= f_4 +f_5&&=\left(\begin{array}{ccc}0 & 0 & 0\\1 & 0 & 0 \\0 & 1 & 0\end{array}\right),  \\
\mathbf{e_5}&= i (e_4-e_5)&&= \left(\begin{array}{ccc}0 & i & 0 \\0 & 0 & -i \\0 & 0 & 0\end{array}\right) ,& \mathbf{f_5}&= i (f_4-f_5)&&= \left(\begin{array}{ccc}0 & 0 & 0 \\ i & 0 & 0 \\0 & -i & 0\end{array}\right) , \\
\mathbf{e_{4,5} }&=i e_{4,5}&&= \left(\begin{array}{ccc}0 & 0 & i \\0 & 0 & 0 \\0 & 0 & 0\end{array}\right)\, , & \mathbf{f_{4,5}}&=i f_{4,5}  &&=\left(\begin{array}{ccc}0 & 0 & 0 \\0 & 0 & 0 \\i & 0 & 0\end{array}\right), 
\end{aligned}
\end{split}
\ee
where $\mathbf{h_4}$, $\mathbf{h_5}$ are the generators of the Cartan subalgebra $\mathfrak{h}_r$,  $\mathbf{e_4},\, \mathbf{e_5}$ and $ \mathbf{e_{4,5} }$ are positive generators while $\mathbf{f_4},\,\mathbf{f_5}$ and $\mathbf{f_{4,5}}$ are negative ones. Note that we write here the generators of the real form in bold to distinguish them from the complex ones.

\subsubsection{The maximal compact subalgebra of $\boldsymbol{\asu}  : \ \boldsymbol{\mathfrak{k}_r= \mathfrak{su}(2)\oplus \mathfrak{u}(1)}$}
\label{app:k}
The  maximal compact subalgebra $\mf k_r$ of $\asu$ is given as the fixed point set under the Cartan involution $\theta$ 
\be
\mathfrak{k}_r = \big\{ x \in \asu : \, \theta(x)= x \big\}\, .
\ee
Using \eqref{thetagenreal}, we extend the action of $\theta$ involution on the simple roots \eqref{thetaroots} to the whole algebra
\be \begin{split}
  \begin{aligned}\label{eqn:thetagens}
\theta (h_4)&=- h_5, & \theta (h_5)&=-  h_4,  \\
\theta (e_4)&=-f_5, &  \theta (e_5)&= -f_4, & \theta (e_{4,5})&=  f_{4,5}\,, \\
\theta (f_4)&=-e_5, & \theta (f_5)&= -e_4, & \theta(f_{4,5})&=  e_{4,5}\,. 
\end{aligned} \end{split} \ee
On the Borel generators of $\asu$, this corresponds to
\be \begin{split} \begin{aligned} \label{eqn: invsu21}
\theta(\mathbf{h_4})&= -\, \mathbf{h_4},\ &  \theta(\mathbf{h_5})&= \mathbf{h_5}, \\ 
\theta(\mathbf{e_4})&= -\, \mathbf{f_4}, \ & \theta(\mathbf{e_5})&= \, \mathbf{f_5}, \ \theta(\mathbf{e_{4,5}})= \,\mathbf{f_{4,5}} .
\end{aligned} \end{split}
\ee
We find that the subalgebra $\mathfrak{k}_r = \mf{su}(2) \oplus \mf u(1)$ is generated by:
\be \begin{split}
\tilde{u}&= \frac{1}{2} \, (\mathbf{e_{4,5}}+\mathbf{f_{4,5}})+ \frac{1}{6} \mathbf{h_5}, \\
\tilde{t}_1&= \frac{1}{2} \, \big(\mathbf{h_5} - (\mathbf{e_{4,5}}+\mathbf{f_{4,5}})\big), \\
 \tilde{t}_2&= \frac{1}{\sqrt{2} } \, (\mathbf{e_4}-\mathbf{f_4}), \\
   \tilde{t}_3&= \frac{1}{\sqrt{2} }\,  (\mathbf{e_5}+\mathbf{f_5})\, ,
\end{split}
\ee
where $\tilde{u}$ is the $\mathfrak{u}(1)$ generator and the $\tilde{t}_i$ generate a $\mathfrak{su}(2)$ subalgebra.
According to \eqref{cartandecreal},  we get the following Cartan decomposition of $\asu$,
\be
\asu = \mathfrak{k}_r \oplus \mathfrak {p}_r= \mathfrak{su}(2) \oplus \mathfrak{u}(1) \oplus \mathfrak{p}_r\,,  
\ee
where $\mathfrak{p}_r$ is the subset of $\asu$ which is anti-invariant under $\theta$.

\subsubsection{The restricted root system of $\boldsymbol{\asu}$}
\label{app:restricted}

Let $\mathfrak{a}$ be the maximal abelian subalgebra of $\mathfrak{p}_r$ which can be diagonalized over $\mathbb R$. Then
\be
\mathfrak{a}= \mathfrak{p}_r \cap \mathfrak{h} = \mathbb{R}\, \mathbf{h_4} .
\ee

\noindent The eigenvalues under the adjoint action of $\bar{h}_4$ which is the only diagonalizable generator, are the following
\begin{align} \label{eqn:comute}
[\mathbf{h_4}, \mathbf{e_4}]&= \mathbf{e_4},\  && [\mathbf{h_4}, \mathbf{e_5}]= \mathbf{e_5} ,  &&[\mathbf{h_4}, \mathbf{e_{4,5}}]=2\,  \mathbf{e_{4,5}},\\  \label{eqn:comutf}
[\mathbf{h_4}, \mathbf{f_4}]&= - \mathbf{f_4}\ , \  &&[\mathbf{h_4}, \mathbf{f_5}]=-  \mathbf{f_5}\ , \ &&[\mathbf{h_4}, \mathbf{f_{4,5}}]=- 2\,  \mathbf{f_{4,5}}.
\end{align}
The generator $\mathbf{h_5}$ is not diagonalizable over $\mathbb{R}$. Indeed, we have the following commutations relations
\begin{align}
[\mathbf{h_5}, \mathbf{e_4}]&=3 \, \mathbf{e_5}, \ && [\mathbf{h_5}, \mathbf{e_5}]=-3 \, \mathbf{e_4}, \ &&[\mathbf{h_5}, \mathbf{e_{4,5}}]=0, \label{eqn:comuth5e}\\
[\mathbf{h_5}, \mathbf{f_4}]&= -3 \mathbf{f_5}, \  &&[\mathbf{h_5}, \mathbf{f_5}]=3\, \mathbf{f_4}, \ &&[\mathbf{h_5}, \mathbf{f_{4,5}}]=0. \label{eqn:comuth5f}
\end{align}
According to \eqref{eqn:comute} and \eqref{eqn:comutf}, we may decompose $\asu$ into a direct sum of eigenspaces labelled by elements of the dual space $\mathfrak{a}^*$
\be
\asu = \bigoplus_{\lambda} \mathfrak{g}_{\lambda} \ , \ \mathfrak{g}_{\lambda}=\{ x \in \asu :\forall h \in \mathfrak{a}, \mathrm{ad}h (x)= \lambda(h)\, x \}. 
\ee
The non-compact Cartan $\mathbf{h_4}$ then generates a $5$-grading
 of $\asu$ which is given by
 \be
 \asu= \mathfrak{g}_{(-2)} \oplus  \mathfrak{g}_{(-1)} \oplus \mathfrak{h}_r \oplus \mathfrak{g}_{(+1)} \oplus \mathfrak{g}_{(+2)}.
 \ee
\noindent One trivial subspace is $\mathfrak{h}_r$. The other nontrivial subspaces define the restricted root spaces of $\asu$ with respect to $\mathfrak{a}$ and the restricted roots are the $\lambda \in \mathfrak{a}^*$. It is now easy to determine the positive restricted root system $\sigma$ of $\asu$
\be \label{eqn:restricted}
\lambda_1(\mathbf{h_4})= 1 \quad, \quad \lambda_2(\mathbf{h_4})=  2\, \lambda_1= 2.
\ee 
Hence, the restricted root system displayed in Figure \ref{fig:restricted} consists of the restricted root $\lambda_1$ which has multiplicity $2$ and the highest reduced root $2 \lambda_1$, with multiplicity $1$. This can be identified with the non-reduced root system\footnote{The $(BC)_n$ root system is an irreducible non-reduced root system obtained by combining the root system of the algebra $B_n$ with the root system of the algebra $C_n$ such that the long roots of $B_n$ are the short roots of $C_n$. There are in this case three different root lengths. This root system is \emph{non reduced} because if $\alpha \in \Delta$, then it occurs in $\Delta$ an element of type $\pm 1/2 \alpha$ or $\pm 2 \alpha$. For \emph{reduced system} only $\pm \alpha$ can occur as multiple of $\alpha$ \cite{Henneaux:2007ej}.       } $(BC)_1$ \cite{Helgason:1978, HenneauxJulia,Henneaux:2007ej}.

\begin{figure}[ht]
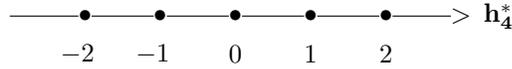

\begin{center}
\begin{pgfpicture}{5cm}{- 0.5cm}{1cm}{1.5cm}
\pgfnodebox{Node1}[virtual]{\pgfxy(1,0.5)}{$\bullet$}{0pt}{0pt}
\pgfnodebox{Node2}[virtual]
{\pgfrelative{\pgfxy(1,0)}{\pgfnodecenter{Node1}}}{$\bullet$}{0pt}{0pt}
\pgfnodebox{Node3}[virtual]
{\pgfrelative{\pgfxy(1,0)}{\pgfnodecenter{Node2}}}{$\bullet$}{0pt}{0pt}
\pgfnodebox{Node4}[virtual]
{\pgfrelative{\pgfxy(1,0)}{\pgfnodecenter{Node3}}}{$\bullet$}{0pt}{0pt}
\pgfnodebox{Node5}[virtual]
{\pgfrelative{\pgfxy(1,0)}{\pgfnodecenter{Node4}}}{$\bullet$}{0pt}{0pt}
\pgfnodebox{Node55}[virtual]
{\pgfrelative{\pgfxy(1,0)}{\pgfnodecenter{Node5}}}{$>$}{0pt}{0pt}
\pgfnodebox{Node0}[virtual]
{\pgfrelative{\pgfxy(-1,0)}{\pgfnodecenter{Node1}}}{}{0pt}{0pt}

\pgfnodebox{Node18}[virtual]{\pgfxy(6.5,0.5)}{$\mathbf{h_4^*}$}{0pt}{0pt}

\pgfnodebox{Node6}[virtual]{\pgfxy(0.9,0)}{$-2$}{2pt}{2pt}
\pgfnodebox{Node7}[virtual]{\pgfxy(1.9,0)}{$-1$}{2pt}{2pt}
\pgfnodebox{Node8}[virtual]{\pgfxy(3,0)}{$0$}{2pt}{2pt}
\pgfnodebox{Node9}[virtual]{\pgfxy(4,0)}{$1$}{2pt}{2pt}
\pgfnodebox{Node10}[virtual]{\pgfxy(5,0)}{$2$}{2pt}{2pt}

\pgfnodeconnline{Node1}{Node2} \pgfnodeconnline{Node2}{Node3}
\pgfnodeconnline{Node3}{Node4} \pgfnodeconnline{Node4}{Node5} 
\pgfnodeconnline{Node5}{Node55} \pgfnodeconnline{Node0}{Node1} 

\end{pgfpicture}
\caption { \sl \small The restricted root system of $\asu$ labeled by the eigenvalues of $\mathbf{h_4}$.} 
\label{fig:restricted}
\end{center}
\end{figure}
We can also read off restricted roots of $\asu$ directly from the root diagram of $A_2$ displayed in Figure \ref{fig:restrictedrootdiag}, by projecting all the roots of $A_2$ onto the horizontal line corresponding to the eigenvalues of the Cartan $\mathbf{h_4}=h_4+h_5$. We see that the positive roots $\alpha_4$ and $\alpha_5$ restrict to the same restricted root $\lambda_1$ and the root $\alpha_4+ \alpha_5$ restrict to the root $2 \lambda_1$. Hence the restricted root system of $\asu$ consists of the restricted roots $\pm \lambda_1$ which have multiplicity $2$ and the restricted root $\pm 2 \lambda_1$ with multiplicity one.


\begin{figure}[ht]
\begin{center}
{\scalebox{0.3}
{\includegraphics{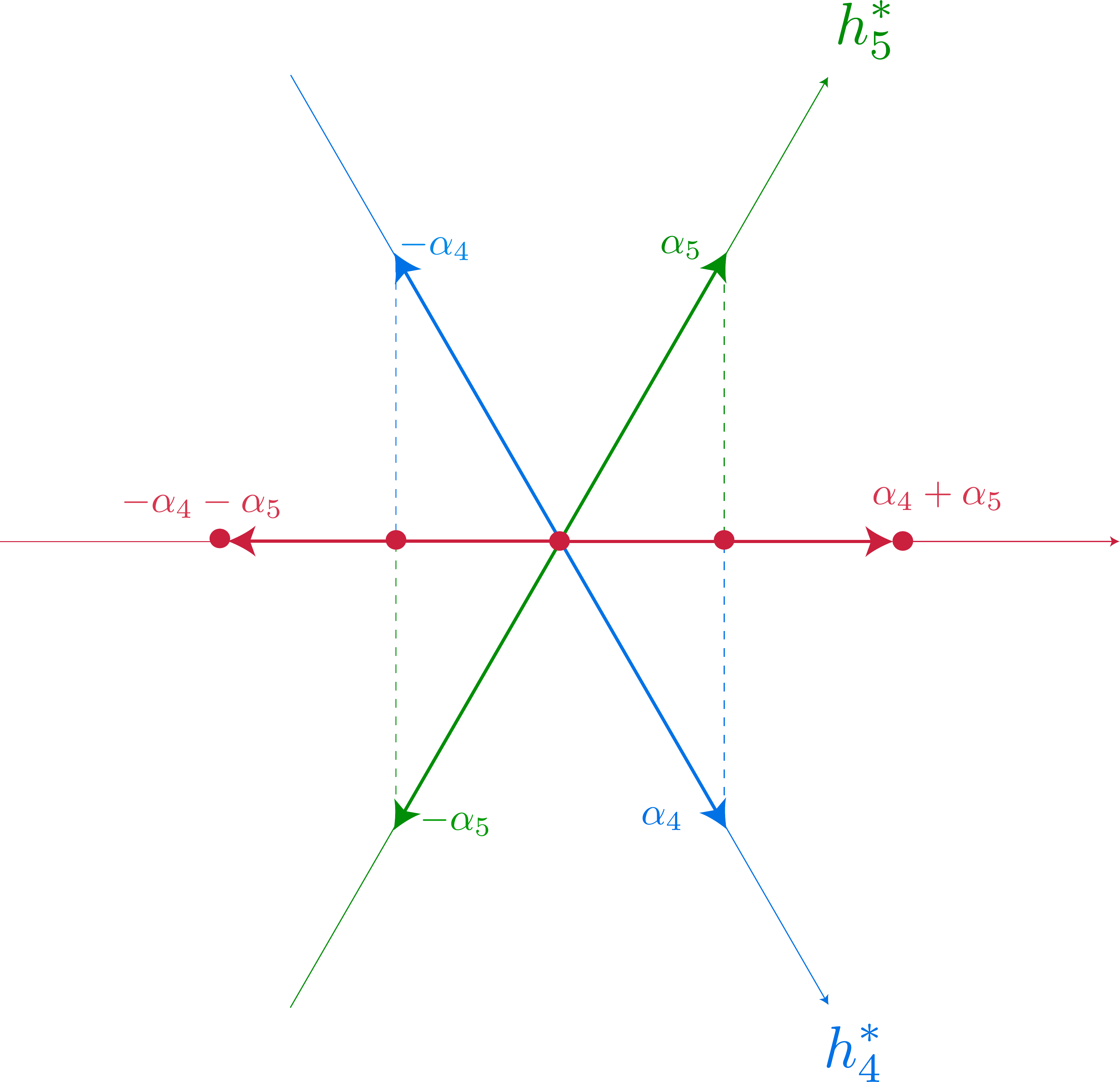}}}
\end{center}
\caption{ \small \textsl{The restricted root of $\asu$ obtained by projecting the root diagram of $A_2$ onto the horizontal red line spanned by $\mathbf{h_4^*}= h_4^* \oplus h_5^*$. The red nodes corresponds to the restricted roots and it consists to the non-degenerate roots $\pm (\alpha_4 + \alpha_5)$ such that $(\alpha_4+ \alpha_5) (\mathbf{h_4}) = 2$ and to the  twice degenerate roots $\pm \tfrac{1}{2} (\alpha_4+\alpha_5)$ such that $\tfrac{1}{2} (\alpha_4+\alpha_5)(\mathbf{h_4}) = \alpha_4(\mathbf{h_4})= \alpha_5(\mathbf{h_4})=1$.  }} \label{fig:restrictedrootdiag}
\end{figure}


Let $\Sigma $ be the subset of nonzero restricted roots and $\Sigma ^{+} $ the set of positive roots, we define a nilpotent subalgebra of $\asu$ as 
\be
\mathfrak{n}_+= \bigoplus _{\lambda \in \Sigma^+} \mathfrak{g}_{\lambda}. 
\ee
Then, the algebraic Iwasawa decomposition of the Lie algebra $\asu$ reads
 \be \label{eqn:Iwasawasu21} \begin{split}
 \begin{aligned}
\asu& =  \ \mathfrak{k}_r \oplus \mathfrak{a}  \oplus \mathfrak{n}_+,\\
 &= \big( \mathfrak{su}(2) \oplus \mathfrak{u}(1)\big) \oplus \mathbb{R} \mathbf{h_4} \oplus \big(\mathbb{R} \mathbf{e_4} \oplus \mathbb{R} \mathbf{e_5} \oplus \mathbb{R} \mathbf{e_{4,5}}\big) .
\end{aligned} \end{split} \ee
It is only $\mathfrak{a}$ that appears in the Iwasawa decomposition of $\asu$ (\ref{eqn:Iwasawasu21}) and not the full Cartan subalgebra $\mathfrak{h}$ since its compact part $\mathbf{h_5}$ belong to $\mathfrak{k}_r$. 

\subsubsection{ The non-compact subalgebra of $\boldsymbol{\asu}$: $\boldsymbol{\mathfrak{k}^*_r= \mathfrak{sl}(2,\mathbb{R})\oplus \mathfrak{u}(1)}$ } \label{app:k*}

 In Chapter \ref{chap:finitesymsu215}  we will consider the reduction on a time-like direction of the four dimensional pure $\mathcal N=2$ supergravity. The reduced Lagrangian will be  identified with a non-linear $\s$-model constructed on the coset  $\mathcal{C}^*= \su/ \mathrm{SL}(2,\mbb R) \times \mathrm{U}(1)$. The generators of the algebra $ \mathfrak{sl}(2,\mathbb{R})\oplus \mathfrak{u}(1)$ associated to the quotient group of this coset are invariant 
 under a $\Omega_4$-involution ( which will be defined later in Section \ref{PERTABsec:tempinv} ): 
\be
\mathfrak{k}^*_r = \big\{ x \in \asu :\, \Omega_4(x)= x \big\}\, ,
\ee
where this involution $\Omega_4$ acts on the generators
of the Borel subalgebra of $\asu$ as:
 \be \label{eqn:omega4} \begin{split} \begin{aligned}
\Omega_4(\mathbf{h_4})&= -\, \mathbf{h_4},\  & \Omega_4(\mathbf{h_5})&= \mathbf{h_5},\\
  \Omega_4(\mathbf{e_4})&= \, \mathbf{f_4}, \ &\Omega_4(\mathbf{e_5})&=- \, \mathbf{f_5}, \quad \Omega_4(\mathbf{e_{4,5}})= \,\mathbf{f_{4,5}}\, .
\end{aligned} \end{split}
\ee
The subalgebra $\mathfrak{k}^*_r$ is generated by:
\be \label{eqn:nocompactsub} \begin{split}
u&= - \frac{1}{2} \, (\mathbf{e_{4,5}} +\mathbf{f_{4,5}})+ \frac{1}{6} \mathbf{h_5},\\
  t_1&= \frac{1}{2} \, \big(\mathbf{h_5} + (\mathbf{e_{4,5}}+\mathbf{f_{4,5}})\big),\\
   t_2&= \frac{1}{\sqrt{2} } \, (\mathbf{e_4}+\mathbf{f_4}), \\
    t_3&= \frac{1}{\sqrt{2} }\,  (\mathbf{e_5} - \mathbf{f_5})\, ,
\end{split}\ee
where $u$ is the $\mathfrak{u}(1)$ generator and the $t_i$ generate a $\mathfrak{sl}(2, \mathbb{R})$ subalgebra.

\subsubsection{The very-extension of $\boldsymbol{\asu}$ = $\boldsymbol{\asuppp}$}
In Chapter \ref{chap:infinitedimsymsu21}, we provide evidence that pure $\mathcal{N}=2$ supergravity in $D=4$ exhibits an underlying algebraic structure described by the Lorentzian Kac-Moody group $\suppp$. The associated Kac-Moody algebra $\asuppp$ can be obtained by adding three nodes $\alpha_1, \alpha_2$ and $\alpha_3$ to the Tits-Satake diagram of $\asu$ displayed in Figure \ref{figa:su21}. The Tits-Satake diagram of this very-extended algebra is given is Figure \ref{fig:su21+++}b. 

In the same way that $\asu$ is a real form of the complex algebra $A_2= \mf{sl}(3, \mbb C)$, the Kac-Moody algebra $\asuppp$ is a real form of the complex algebra $A_2^{+++}$. In  Figure \ref{fig:su21+++}, we can see that the Tits-Satake diagram of $\asuppp$ differs from the $A_2^{+++}$ one by the extra decoration afforded by the double arrow connecting the nodes $\alpha_4$ and $\alpha_5$.   This implies that only the $A_2$-part of the diagram is transformed into the non-split real form $\asu$  \footnote{The fact  that we have an involution given by the double arrow in the Figure \ref{fig:su21+++}b, acting only on a finite-dimensional subalgebra ensures that we are constructing an almost split real form.}  of $A_2$ such that  (see \eqref{thetaroots} )
\be 
\theta(\alpha_4)=-\alpha_5, \quad \theta(\alpha_5)= -\alpha_4.
\ee
 Therefore, to construct $\asuppp$ we have just to  extend the Tits-Satake diagram of $\asu$ with three non-compact simple roots  $\alpha_j$ $(j=1, \ldots, 3)$ such that 
\be \label{eqn:sigmalphanc}
\theta (\alpha_j)= - \alpha_j, \quad \mathrm{or} \quad \s (\alpha_j)= \alpha_j \qquad j=1,2,3\, .
\ee
\begin{figure}[t]
\begin{center}
\begin{pgfpicture}{0cm}{-2cm}{3cm}{2.5cm}

\pgfnodecircle{Node12}[stroke]{\pgfxy(-4.5,0.5)}{0.25cm}
\pgfnodecircle{Node22}[stroke]
{\pgfrelative{\pgfxy(1.5,0)}{\pgfnodecenter{Node12}}}{0.25cm}
\pgfnodecircle{Node32}[stroke]
{\pgfrelative{\pgfxy(1.5,0)}{\pgfnodecenter{Node22}}}{0.25cm}
\pgfnodecircle{Node42}[stroke]
{\pgfrelative{\pgfxy(1.5,1)}{\pgfnodecenter{Node32}}}{0.25cm}
\pgfnodecircle{Node52}[stroke]
{\pgfrelative{\pgfxy(1.5,-1)}{\pgfnodecenter{Node32}}}{0.25cm}

\pgfnodebox{Node62}[virtual]{\pgfxy(-4.5,0)}{$\alpha_{1}$}{2pt}{2pt}
\pgfnodebox{Node72}[virtual]{\pgfxy(-3,0)}{$\alpha_{2}$}{2pt}{2pt}
\pgfnodebox{Node82}[virtual]{\pgfxy(-1.5,0)}{$\alpha_{3}$}{2pt}{2pt}
\pgfnodebox{Node92}[virtual]{\pgfxy(0,-1)}{$\alpha_{4}$}{2pt}{2pt}
\pgfnodebox{Node102}[virtual]{\pgfxy(0,2)}{$\alpha_{5}$}{2pt}{2pt}

\pgfnodeconnline{Node12}{Node22} \pgfnodeconnline{Node22}{Node32}
\pgfnodeconnline{Node32}{Node42} \pgfnodeconnline{Node32}{Node52} 
\pgfnodeconnline{Node42}{Node52}

\pgfnodebox{Node63}[virtual]{\pgfxy(-2.5,-1.5)}{\textbf{a}}{2pt}{2pt}

\pgfnodebox{Node64}[virtual]{\pgfxy(4,-1.5)}{\textbf{b}}{2pt}{2pt}

\pgfnodecircle{Node1}[stroke]{\pgfxy(2,0.5)}{0.25cm}
\pgfnodecircle{Node2}[stroke]
{\pgfrelative{\pgfxy(1.5,0)}{\pgfnodecenter{Node1}}}{0.25cm}
\pgfnodecircle{Node3}[stroke]
{\pgfrelative{\pgfxy(1.5,0)}{\pgfnodecenter{Node2}}}{0.25cm}
\pgfnodecircle{Node4}[stroke]
{\pgfrelative{\pgfxy(1.5,1)}{\pgfnodecenter{Node3}}}{0.25cm}
\pgfnodecircle{Node5}[stroke]
{\pgfrelative{\pgfxy(1.5,-1)}{\pgfnodecenter{Node3}}}{0.25cm}

\pgfnodebox{Node6}[virtual]{\pgfxy(2,0)}{$\alpha_{1}$}{2pt}{2pt}
\pgfnodebox{Node7}[virtual]{\pgfxy(3.5,0)}{$\alpha_{2}$}{2pt}{2pt}
\pgfnodebox{Node8}[virtual]{\pgfxy(5,0)}{$\alpha_{3}$}{2pt}{2pt}
\pgfnodebox{Node9}[virtual]{\pgfxy(6.5,-1)}{$\alpha_{4}$}{2pt}{2pt}
\pgfnodebox{Node10}[virtual]{\pgfxy(6.5,2)}{$\alpha_{5}$}{2pt}{2pt}

\pgfnodeconnline{Node1}{Node2} \pgfnodeconnline{Node2}{Node3}
\pgfnodeconnline{Node3}{Node4} \pgfnodeconnline{Node3}{Node5} 
\pgfnodeconnline{Node4}{Node5}
\pgfsetstartarrow{\pgfarrowtriangle{4pt}}
\pgfsetendarrow{\pgfarrowtriangle{4pt}}
\pgfnodesetsepend{5pt}
\pgfnodesetsepstart{5pt}
\pgfnodeconncurve{Node4}{Node5}{-10}{10}{1cm}{1cm}
\end{pgfpicture}
\caption {  \small \textbf{a}: {\sl Dynkin diagram of $A_2^{+++}$. }\textbf{b}: {\sl Tits-Satake diagram of $\asuppp$.} }
\label{fig:su21+++}
\end{center}
\end{figure}

The action of $\s$  and $\theta$
can be extended from the space of roots to the entire algebra. For $\theta$ this yields to 
\be \begin{split}\begin{aligned}
\label{eqn:thetasu21+++}
\theta (h_j)&=- h_j, \quad &&\theta (e_j)=- f_j, \quad  &&\theta (f_j)=- e_j, \quad & j=1, 2, 3,\\
\theta (h_4)&=- h_5, \quad &&\theta (e_4)=- f_5, \quad  &&\theta (f_4)=- e_5,& \\
\theta (h_5)&=- h_4, \quad &&\theta (e_5)=-f_4, \quad  &&\theta (f_5)=- e_4,&
\end{aligned} \end{split}\ee
while  for $\sigma$ this yields to 
\be \begin{split}\begin{aligned}\label{eqn:sigmasu21+++} 
\s (h_j)&=h_j, \quad &&\s (e_j)=e_j, \quad &&\s (f_j)=f_j, \quad &j=1, 2, 3,  \\
\s (h_4)&=h_5, \quad &&\s (e_4)=e_5, \quad  &&\s (f_4)=f_5,& \\
\s (h_5)&= h_4, \quad &&\s (e_5)=e_4, \quad  &&\s (f_5)=f_4,&
\end{aligned} \end{split}\ee
where $\{h_i, e_i, f_i\}$ are the Chevalley generators of the very-extended algebra $A_2^{+++}$.

The generators of $\asuppp$ then correspond to the subset of $A_2^{+++}$-generators left invariant under $\sigma$. 
They can be written
 in terms of the Chevalley generators of $A_2^{+++}$ as follows
\be \label{eqn:su21+++Chev}
\begin{split}
\begin{aligned}
\mathbf{e_1}&= e_1, \qquad&&\mathbf{f_1}=f_1, \qquad  &&\mathbf{h_1}= h_1\, ,\\
\mathbf{e_2}&= e_2, \qquad&&\mathbf{f_2}=f_2,  \qquad&&\mathbf{h_2}= h_2\, ,\\
\mathbf{e_3}&= e_3, \qquad &&\mathbf{f_3}=f_3, \qquad&&\mathbf{h_3}= h_3\, ,\\
\mathbf{e_4}&= e_4 + e_5, \qquad &&\mathbf{f_4}=f_4 +f_5,\qquad &&\mathbf{h_4}= h_4+ h_5\, ,\\
\mathbf{e_5}&=i\,( e_4 - e_5), \qquad &&\mathbf{f_5}=i\,(f_4 -f_5), \qquad && \mathbf{h_5}= i \,(h_4- h_5)\, .
\end{aligned}
\end{split}
\ee
We stress that for these generators there is no set of standard Chevalley--Serre relations defining the commutators between these elements.

The Cartan subalgebra of $\asuppp$ is given by
\be
\mathfrak{h}_r^{+++}= \bigoplus _{i=1}^{5} \mathbb{R} \ \mathbf{h_i},
\ee
of which $\mathbf{h_1}, \ldots, \mathbf{h_4}$ are non-compact, while $\mathbf{h_5}$ is compact. The generators $\mathbf{h_1}, \ldots, \mathbf{h_4}$ are diagonalizable over the real numbers, and generate the non-compact part $\mathfrak{a}^{+++}$ of the full Cartan subalgebra $\mathfrak{h}_r^{+++}$. 

\subsubsection{The restricted root system of $\boldsymbol{\asuppp}$}

The restricted root system of $\asuppp$ corresponds to the \emph{twisted}\footnote{We refer the reader to  \cite{Henneaux:2007ej} for  the theory of twisted extensions of finite-dimensional Lie algebras.} very-extension of the $(BC)_1$ root system :   $(BC)_1^{+++}= A_2^{(2)++}$. Its  Dynkin diagram is displayed in Figure \ref{fig:twistextension}. The twist appears in the affine extension of the $(BC)_1$ root system ($(BC)_1^+= A_2^{(2)}$) because it is a non-reduced root system and its highest root  $\theta$ is $2 \lambda_1$ instead of $\lambda_1$.  As a consequence the two simple roots of the affine extension $\lambda_1$ and $\alpha_0$ come with different lenghts (see Figure \ref{fig:twistextension}).


\begin{figure}[ht]
\begin{center}
{\scalebox{0.37}
{\includegraphics{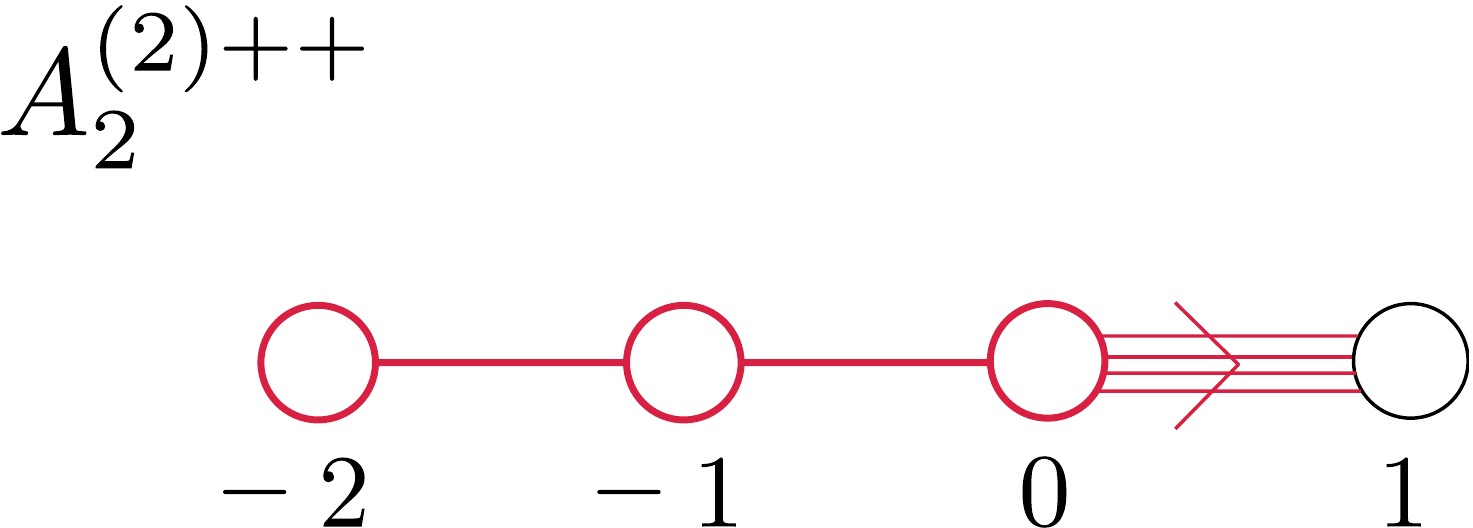}}}
\end{center}
\caption{ \small \textsl{ Dynkin diagram of $A_2^{(2)++}$ which corresponds to the  restricted root system of $\asuppp$.   The node labeled by $1$ represents the simple root $\lambda_1$ of $(BC)_1$. The nodes labeled by $0$, $-1$ and $-2$ are respectively the affine,  overextended and very-extended roots.  }} \label{fig:twistextension}
\end{figure}


\section{Level decomposition}\label{sec:leveldecomposition}
\setcounter{equation}{0}
We will develop in this section a way to approach an infinite-dimensional Kac-Moody algebra $\mf g$  by decomposing its adjoint representation with respect to one of its regular finite subalgebra $\mf{s}$. This decomposition, called \emph{level decomposition}, has the advantage to slice the adjoint representation of $\ag$ such that at each level we get a finite number of representations of the subalgebra $\mf s$. This slicing will be crucial subsequently  in this work because it gives a  recursive way   to construct $\sigma$-models invariant under $\ag$ \cite{Damour:2002cu}. This decomposition yields also to a precise understanding of the lowest levels which will have a physical interpretation. 

This section is organized as follows. We will first explain  the mechanism of the decomposition of a Kac-Moody algebra $\mf g$ with respect to one of its regular finite subalgebras $\mf s$ \cite{Damour:2002cu, Nicolai:2003fw,Kleinschmidt:2003mf}. Then we will recall how  this  level decomposition works  for the split real form $E_{11(11)}$ \cite{Damour:2002cu,Nicolai:2003fw}. This decomposition will be useful in Part II and III.  Finally we will apply this mechanism for a non-split real form $\asuppp$ \cite{Houart:2009ed} that we will use in Part III.

\subsection{The process of level decomposition with respect to a regular  algebra }

 A  \emph{regular} subalgebra $\mf s \subset \ag$ is a Kac-Moody algebra whose generalized Cartan matrix $\bar{A}$ is a proper principal submatrix of the Cartan matrix  $A$ of $\mf g$. At the level of Dynkin diagrams, $\mf s$ is obtained from the Dynkin diagram of $\mf g$ by removing a set of nodes $\mathscr{K}= \{i_1, \ldots ,i_k\}$ and lines connected to them. Let $\mathscr{I}= \{ 1, \ldots,n\}$ the set of indices belonging to the whole algebra $\mf g$ and $\mathscr{Q}= \mathscr{I} \backslash \mathscr{K}$ the set of indices pertaining to $\mf s$. Therefore the rank of $\mf s$ is $n-k$. The adjoint representation of $\mf s$ is viewed as a natural subset of the adjoint $\mf g$.  As the subalgebra $\mf s$ is regular and due to the triangular decomposition of $\mf g$ and to the adjoint action of $\mf s$ upon $\mf g$, it is possible to decompose $\mf g$ into a sum of lowest weight irreducible representations. 

Let us first recall that besides  the root lattice $Q$ there is also the \emph{weight lattice} $P$ which is spanned by the \emph{fundamental weights} $\Lambda_i$ such that $(\alpha_i \, |\, \Lambda_j)= - \tfrac{1}{2}\, \alpha_i^2\, \delta_{ij},\ i,j=1, \ldots, n$. \footnote{Note that we do not choose the standard convention $(\alpha_i \, |\, \Lambda_j)= \tfrac{1}{2} \alpha_i^2 \delta_{ij}$ such  that the fundamental weights point in the positive direction.} As we will always deal with simply laced algebras, this definition becomes $(\alpha_i \, |\, \Lambda_j)= - \delta_{ij}$. Then any root $\alpha \in Q$ can be expressed in the basis of simple roots $\alpha_i \in \Pi$ as
\be
\alpha_i = \sum_{i=1}^n m_i \alpha_i,
\ee
or in the basis of fundamental weights as
\be \label{basefundaweig}
\alpha = \sum_{i=1}^n p_i \Lambda_i,
\ee
where the coefficients $p_i$ are called the \emph{Dynkin labels}. Using the Cartan matrix, we can convert the  coefficients $m_i$ to the $p_i$'s ones by
\be\label{relatbases}
p_i= - \sum_{j=1}^n\,  A_{ij}\, m_j.
\ee
The coefficients $p_i$ belong to $ \mbb Z$ for all $i=1, \ldots ,n$ because the entries of $A$ and the coefficients $m_i$ are integers. We can now write $\alpha$ in terms of its components with respect to the two bases as
\be
\alpha=(m_1, \ldots, m_n)= [p_1, \ldots , p_n],
\ee
using different brackets to distinguish the two bases.

\subsubsection{Lowest weight representations and grading of $\mf g$}

A \emph{lowest weight representation} $\mathcal{R}_{\Lambda}$ of a Kac-Moody algebra $\mf g$ is an $\mf h$-diagonalizable representation with a non-zero vector $v$ in the representation space  $V$ of lowest weight $\Lambda \in \mf h^*$ namely
\be
h (v)= \Lambda(h)\, v, \quad \forall \, h \in \mf h
\ee
such that the action of all negative Chevalley generators $f_i$, $i=1, \ldots ,n $ annihilates it:
\be
f_{i} (v)= 0.
\ee
 All the states in the representation $V$ are obtained by the action of all the generators of $\mf n^+$: $\mf n^+(v)= V$ and the set of eigenvalues of all these states is called the \emph{weight system} of the representation $\mathcal{R}_{\Lambda}$.

We will now introduce a distinction between the fundamental weights of the regular subalgebra $\mf s$ : $\lambda_q, \ q\in \mathscr{Q}$ and the  fundamental weights of the whole algebra $\mf g$: $\Lambda_i, \ i=1, \ldots, n \in \mathscr{I}$. The fundamental weights  $\lambda_q$ satisfy 
\be
(\alpha_p|\lambda_q) =- \delta_{pq}, \quad p,q\in \mathscr{Q} \, ,
\ee
while $\Lambda_i$ must in addition satisfy 
\be
(\alpha_j | \Lambda_i)= - \delta_{ji}, \quad j \in \mathscr{I} \backslash \mathscr{Q}  \ \mathrm{and}\  i \in \mathscr{I} \, .
\ee
As the fundamental weights $\Lambda_{i_a}$ corresponding to nodes $i_a \in \mathscr{K}$ (not pertaining to the Dynkin diagram of the subalgebra $\mf s$) are orthogonal to the simple roots of $\mf s$, the Dynkin labels on these nodes are unchanged under the action of $\mf s$. In the same way  the coefficients $m_{i_a} \,, \ i_a \in \mathscr{K}$ are unchanged under the action of $\mf s$. Therefore,  these coefficients provide a grading of the algebra $\mf g$ with respect to the \emph{level} $\ell=(\ell_1, \ldots ,\ell_k) =(m_{i_1}, \ldots ,m_{i_k})$ such that any root $\alpha$ of $\mf g$ can be decomposed as
\be\label{alphadeclev}
\alpha= \sum_{q \in \mathscr{Q} }m_q \alpha_q+\sum_{i_a \in \mathscr{K}} \ell_{a}\alpha_{i_a}.
\ee
This decomposition of the roots provides a grading of the algebra which reads
\be
\mf g = \bigoplus_{\ell} \mf g_{\ell},
\ee
where $\mf g_{\ell}$ is the subspace spanned by all the vectors $x_{\mu}$ satisfying $[h,x_{\mu}]= \mu(h) x_{\mu}$ such that the level of the root $\mu$ is $\ell$. For given $\ell$, each lowest weight $\Lambda$ generates a representation $\mathcal R_{\Lambda}$ which contains a finite number of  weights. The level decomposition of $\mf g$ with respect to a finite subalgebra $\mf s$ can be then view as a slicing of the forward lightcone in the root lattice of $\mf g$ by spacelike hyperplanes as displayed in Figure \ref{fig:levdec}. Therefore we can see that each section corresponding to a given level $\ell$  will be ellipsoidal and then contains a finite number of roots. By contrast, slicing the lightcone by lightlike or timelike hyperplanes implies a decomposition with respect to an affine or an indefinite subalgebra where each slice will contain an infinite number of roots \cite{Feingold1983,Kacmoodywakimoto,Kleinschmidt:2003pt}. We will consider in this dissertation only level decomposition with respect to finite dimensional subalgebra because it yields a finite amount of information at each level $\ell$ and it gives a concrete realization of the Kac-Moody algebra $\mf g$ in terms of irreducible representations of subgroups that will have a direct physical interpretation.

\begin{figure}[ht]
\begin{center}
{\scalebox{0.35}
{\includegraphics{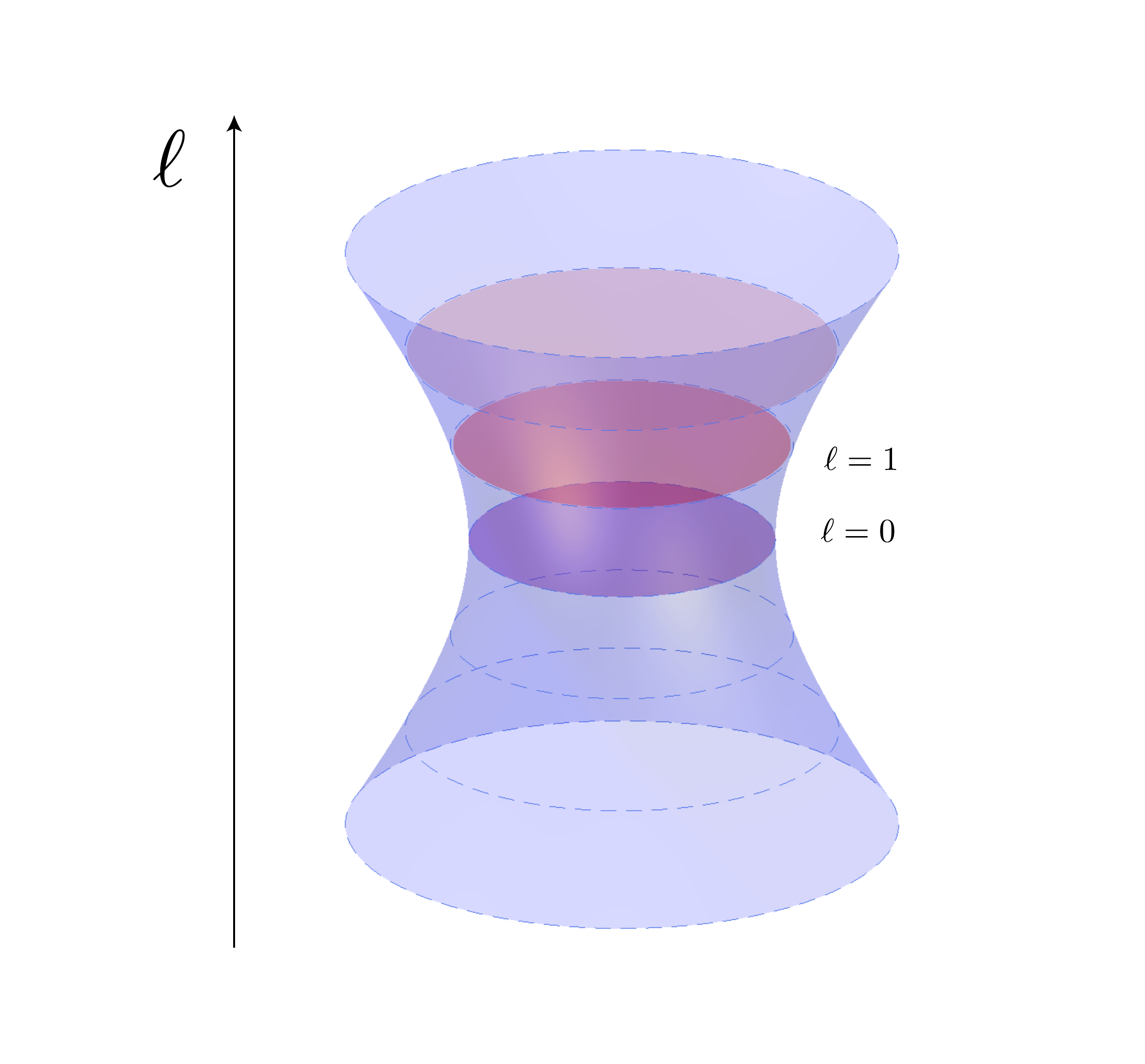}}}
\end{center}
\caption{\small \sl Level decomposition of the adjoint representation of $\mf g$ with respect to the finite regular subalgebra $\mf s$. The lightcone represent the infinite root lattice of $\mf g$ which is sliced  by spacelike sections related to the adjoint action of $\mf s$. For each level $\ell$, the associated section  contains a finite number of weights translating the fact that there are only a finite number of representations. }
\label{fig:levdec}
\end{figure}
The level grading of $\mf g$ implies that 
\be \label{comutelevell}
[\mf g_{\ell}, \ag_{\ell'}] \subseteq \ag_{\ell + \ell'}.
\ee
In particular for $\ell=0=(0, \ldots, 0)$, we have $ [\mf g_0, \ag_{\ell'}] \subseteq \ag_{ \ell'}$ which translates the fact that $\ag_{\ell'}$ is a representation of $\ag_0$ under the adjoint action. Note that the subspace $\ag_0$ corresponding to the level $\ell=0$  contains the whole subalgebra  $\mf s$ enlarged by the Cartan generators $h_{i_a}, \ i_a \in \mathscr{K}$ associated with the removed simple roots $\alpha_{i_a}$. As $\mf s$ is a subalgebra of $\ag_0$, we have that $ [\mf s, \ag_{\ell'}] \subseteq \ag_{ \ell'}$ such that $\ag_{\ell'}$ is left invariant under the adjoint action of the regular subalgebra $\mf s$. 

\subsubsection{Admissible Dynkin labels}

By virtue of \eqref{comutelevell}, we can find all the possible representations that can occur at level $\ell +1$ because it must be contained in the set of representations obtained by taking the product of the level one representation with all the representations appearing at level $\ell$. However, this method becomes more and more fastidious for higher levels.  It is then useful to introduce  a better way to investigate the representations content of $\mf g$ and  to derive all the possible representations for each level $\ell$ by a judicious study of the corresponding lowest weight.

 Let $\Lambda= \alpha$ be the lowest weight of a representation $\mathcal R_{\Lambda}$ and decompose it as \eqref{alphadeclev}. The possible Dynkin labels related to the subalgebra $\mf s$ can be obtained using \eqref{relatbases}  
\be \label{p123}
p_q= -\sum_{i\in \mathscr{I}} A_{qi}\, m_i=  -\sum_{j\in \mathscr{Q}} \bar{A}_{qj} m_j - \sum_{i_a \in \mathscr{K}} A_{q i_a} \ell_a,  \qquad  \forall q \in \mathscr{Q}.
\ee
Inverting this formula, we get
\be\label{m123}
m_q= - \sum_{j \in \mathscr{Q}}  (\bar{A}^{-1})_{qj}\,   \big( \,  p_j + \sum_{i_a \in \mathscr{K}} A_{j i_a} \ell_a \big) ,  \qquad  \forall q \in \mathscr{Q}.
\ee
As the entries of the inverse Cartan matrix $\bar{A}^{-1}$ are always positive for a finite dimensional subalgebra $\mf s$ and because $p_q$ and $m_q$ must be non-negative integers if $\alpha$ is a positive root, the Diophantine equations \eqref{p123} and \eqref{m123} strongly constraints the possible representation  for each level $\ell$.

A further constraint derives from the fact that the lowest weight $\Lambda$ must be a root of $\mf g$ because we are dealing with adjoint representations of $\mf g$. We will exploit a necessary condition for $\Lambda$ to be a root
\be
\Lambda^2= (\Lambda | \Lambda) \geqslant 2\, .
\ee
This additional constraint will eliminate many representations which would still be compatible with \eqref{p123} and \eqref{m123}. 

\subsubsection{Root, weight and outer multiplicities}

We have seen that at level $\ell$, each lowest weight $\Lambda$ generates a representation $\mathcal R_{\Lambda}= \mathcal R(\Lambda)$ with a corresponding finite weight system of the subalgebra $\mf s$.  These weights denoted by $\lambda  \in \mf h^*_{\mf s}$ come with  certain weight multiplicities mult$_{\mathcal R} (\lambda)$. It follows that the root multiplicity of the root  $\beta \in \mf h^*_g$ at a given level $\ell$ denoted by mult$(\beta)$  (defined in \eqref{multipliroot}) is given by the sum of the multiplicities of $\beta$ as a weight in the   various representation $\mathcal{R}^{(\ell)}_i, \ i= 1, \ldots ,n_{\ell}$ of the subalgebra $\mf s$. Thus we can write a relation linking the root multiplicities of $\ag$ with the weight multiplicities of the corresponding $\mf s$-weight as
\be \label{mult157}
\mathrm{mult}(\beta)= \sum_{i=1}^{n_{\ell}}\mu_{\ell} (\mathcal{R}_i^{(\ell)})\,  \mathrm{mult}_{\mathcal{R}_i^{(\ell)}} (\beta),
\ee
where the numbers $\mu_{\ell} (\mathcal{R}_i^{(\ell)})$ are the \emph{outer multiplicities} which counts how many times the representation $\mathcal{R}_i^{(\ell)}$ occurs at level $\ell$. It is clear that if $\mu_{\ell} (\mathcal{R}_i^{(\ell)})= 0$ for some $\ell$ and $i$, then the representation $\mathcal{R}_i^{(\ell)}$ does not appear. 

The only quantities entering in \eqref{mult157} which are known are the weight multiplicities mult$_{\mathcal R}(\lambda)$ for the $\mf s$ representation $\mathcal R_{\Lambda}= \mathcal R(\Lambda)$ with lowest weight $\Lambda$. The computation of these multiplicities follows from the \emph{Freudenthal recursion formula} \cite{Kac:book}
\be
\bigg( (\Lambda|\Lambda)+2\, \mathrm{ht}(\Lambda)- (\lambda|\lambda)-2\, \mathrm{ht} (\lambda)       \bigg)\, \mathrm{mult}_{\mathcal{R}_{\Lambda}}(\lambda)= 2 \, \sum_{\alpha >0} \sum_{k \geqslant 1}  \, (\lambda + k\alpha|\alpha)\, \mathrm{mult}_{\mathcal{R}_{\Lambda}}(\lambda+ k \alpha),
\ee
where the first sum on the r.h.s. ranges over all positive roots of the subalgebra $\mf s$. Starting from the lowest weight $\Lambda$ satisfying mult$_{{\mathcal R}_{\Lambda}}(\Lambda)=1$, all weight multiplicities can be computed by induction on height. 

There is also a recursive way to compute the multiplicity of a root $\alpha \in \mf h^*$ using the \emph{Peterson formula} \cite{Kac:book} which reads 
\be
(\alpha| \alpha +2 \rho)\,  c_{\alpha}= \sum_{ \substack{\gamma, \gamma' \in\,   > 0 \\  \gamma +\gamma'= \alpha }} (\gamma | \gamma')\, c_{\gamma} c_{\gamma'},
\ee  
where $\rho$ is the \emph{Weyl vector} defined as 
\be
\rho = \tfrac{1}{2} \sum_{\beta \in \Delta_+} \beta ,
\ee
and the coefficients $c_{\gamma}$ are defined as
\be
c_{\gamma}= \sum_{k \geqslant 1} \frac{1}{k} \ \mathrm{mult}(\frac{\gamma}{k}).
\ee

For the outer multiplicity $\mu$, there is no closed formula as for the root and the weight multiplicities. Then the computation of the outer multiplicity of each representation at a given level $\ell$ has to be done  via the Freudenthal and Peterson recursion formula and the relation \eqref{mult157}.

\subsubsection{Dynkin labels and Young tableaux}

We will now introduce a useful diagrammatic representation of lowest weights. Let $\lambda$ be a lowest weight of $\mf s$ at a fixed level $\ell$. Its Dynkin labels
\be
\lambda= [p_1, \ldots, p_{n-k}]
\ee
 can equally specified in terms of its \emph{conjugate partition}
\be
\lambda= \{t_1; \ldots ; t_{n-k} \}\quad \mathrm{where} \ t_i= p_1+ p_{2}+ \ldots+ p_{n-k-(i-1)}. 
\ee
To this partition we associate a \emph{conjugate Young tableau}\footnote{The Young tableaux are normally introduced  for highest weight representation $\lambda_{hw}$. Let $\lambda_{hw}=  [\tilde{p}_1, \ldots, \tilde{p}_{n-k}]$, a well defined partition associated to this highest weight is defined by $\lambda_{hw} =  \{ \tilde{t}_1; \ldots ; \tilde{t}_{n-k} \} \quad \mathrm{where} \ \tilde{t}_i= \tilde{p}_1+ \tilde{p}_{2}+ \ldots+ \tilde{p}_{n-k}$. From this partition we associate the Young tableau which is a box array of rows lined up on left such that the lenghts of the $i$-th row is equal to $\tilde{t}_i$. Then the Dynkin labels $\tilde{p}_i$ gives the number of columns of $i$ boxes. The lowest weight can be view as $- \lambda_{hw}$ of the conjugate representation.  This is the reason that we use the conjugate Young tableaux to describe representation of lowest weights. The conjugation amounts to reversing the order of the Dynkin labels.} which is a  array of boxes lined up on the left such that the length of the $i$-th row is equal to $t_i$. For example, the lowest weight $\lambda=[2,0,1,0]$ correspond to the following Young tableau
\be
 \lambda=[2,0,1,0] =\{3,3,2,2\} \quad \longleftrightarrow \quad  {\footnotesize
    \setlength{\tabcolsep}{0.55 em}
    \begin{tabular}{ccc}
      \cline{1-3}
      \multicolumn{1}{|c|}{} &
      \multicolumn{1}{c|}{} &
      \multicolumn{1}{c|}{} \\
      \cline{1-3}
      \multicolumn{1}{|c|}{} &
      \multicolumn{1}{c|}{} &
      \multicolumn{1}{c|}{} \\
      \cline{1-3}
      \multicolumn{1}{|c|}{} &
      \multicolumn{1}{c|}{} \\
      \cline{1-2}
      \multicolumn{1}{|c|}{} &
      \multicolumn{1}{c|}{} \\
      \cline{1-2}
\end{tabular}}\ .
\ee
Dynkin labels provide a dual description of the tableau : $p_i$ gives the number of columns of $(n-k)- (i-1)$ boxes. Therefore in our example, $p_1$ gives the number of columns of $4$ boxes, $p_2$ gives the number of columns of $3$ boxes and so on.

Young tableaux are then in one-to-one correspondence with the irreducible representations of the finite algebra $\mf s$. Any such tableaux always represent a specific symmetrization/antisymmetrization of a tensor of rank $u$, where $u$ is given by the number of boxes in the tableau. The symmetry property of the tensor is such that we first symmetrize all the indices in each row of the tableau, then we antisymmetrize the indices in each column. As a consequence, every column is always completely antisymmetrized but the rows are not necessarily symmetric. Only in special case (i.e. only one row or only one column), the resulting tensor is totally (anti)-symmetric in all indices. 

\subsubsection{The power of computers}

We are now in possession of all the tools necessary  to probe the adjoint representation of the infinite-dimensional Kac-Moody algebra $\ag$ with respect to the irreducible representations of its finite-dimensional Lie algebra $\mf s$. Lowest levels can be computed by hand using all the tools given in this section as in \cite{Damour:2002cu,West:2001as,West:2002jj}. To go further the lowest levels, the calculations can be carried by use of computers \cite{Nicolai:2003fw,Fischbacher:2002fr,Fischbacher:2005fy}. To perform the level decompositions for the algebras which will present a special interest in this work, we will use the nice program SimpLie \cite{SimpLie}. This program created by Teake Nutma  has been specifically developed for level decomposition of infinite-dimensional Lie algebras.
\subsection{Level decomposition of $\aeppp$ with respect to $\mf{sl}(11, \mbb R)$}\label{subsec:leveldece11} 
The Kac-Moody algebra $\aeppp$ has a central role in fundamental physics. The consideration of this algebra is motivated by an attempt to understand the symmetries underlying eleven-dimensional supergravity. We will perform here the level decomposition of $\aeppp$ with respect to the regular subalgebra $\mf s= \mf{sl}(11, \mbb R)$ according to \cite{Nicolai:2003fw,Damour:2002cu}. This level decomposition will be very useful in the continuation of this dissertation.\\

The Kac-Moody algebra $\aeppp$ is entirely characterized by its Cartan matrix  given by 

\be
A [E_{11}]= \left(\begin{array}{ccccccccccc}2 & -1 & 0 & 0 & 0 & 0 & 0 & 0 & 0 & 0 &0 \\-1 & 2 & -1 & 0 & 0 & 0 & 0 & 0 & 0 & 0 &0 \\0 & -1 & 2 & -1 & 0 & 0 & 0 & 0 & 0 & 0 &0 \\0 & 0 & -1 & 2 & -1 & 0 & 0 & 0 & 0 & 0 &0 \\0 & 0 & 0 & -1 & 2 & -1 & 0 & 0 & 0 & 0 &0 \\0 & 0 & 0 & 0 & -1 & 2 & -1 & 0 & 0 & 0 &0\\0 & 0 & 0 & 0 & 0 & -1 & 2 & -1 & 0 & -1 &0 \\0 & 0 & 0 & 0 & 0 & 0 & -1 & 2 & -1 & 0 &-1\\0 & 0 & 0 & 0 & 0 & 0 & 0 & -1 & 2 & -1 &0 \\0 & 0 & 0 & 0 & 0 & 0 & -1 & 0 &-1 & 2&0 \\ 0 & 0 & 0 & 0 & 0 & 0 & 0 & -1 & 0 & 0&2\end{array}\right)\, ,
\ee
which is encoded  in its Dynkin diagram depicted in Figure \ref{fig1:e11}. Erasing the node $1$ defines the regular embedding of its $\aepp$ hyperbolic subalgebra and erasing the nodes $1$ and $2$ yields the regular embedding of the affine $\aep$. As we are dealing with the split real form algebra $E_{11(11)}$, we will reduce the notation and just write $E_{11}$ instead of $E_{11(11)}$. \footnote{From now, we will denote the split real form of the complex algebra $E_n$ by $E_n$ instead of $E_{n(n)}$.} 

We will analyze the level decomposition of the adjoint representation of $E_{11}$ into representations of the $\mf {sl}(11, \mbb R)$ subalgebra defined by the horizontal line in the Dynkin diagram such that the level $\ell$ counts the number of times that the exceptional root $\alpha_{11}$ appears in the decomposition of an arbitrary root $\alpha \in \mf h^*$ in terms of the simple roots $\alpha_i, \ i=1,\ldots,11$ of $E_{11}$:
\be
\alpha= \sum_{i=1}^{10} m_i \alpha_i + \ell\, \alpha_{11}.
\ee
We will concentrate on the lowest levels $\ell \leqslant 4$ given in  Table \ref{tab:levdece11}  which is obtained  from the SimpLie program \cite{SimpLie}. The representation content at each level is represented by irreducible tensors of the $\mathfrak{sl}(11,\mathbb{R})$ subalgebra. Their symmetries properties are fixed by the Young tableaux describing the irreducible representations appearing at a given level $\ell$. We will analyze now this table level by level.

\begin{figure}[h]
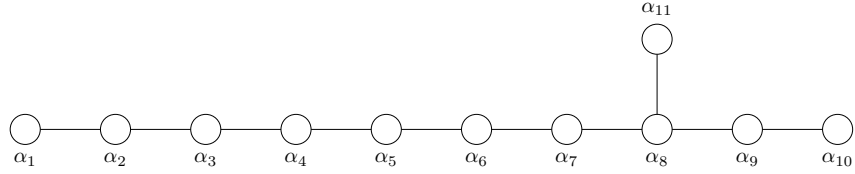
 
\begin{center}
\scalebox{.8}{
\begin{pgfpicture}{15cm}{-0.5cm}{1cm}{2.5cm}
\pgfnodecircle{Node1}[stroke]{\pgfxy(1,0.5)}{0.25cm}
\pgfnodecircle{Node2}[stroke]
{\pgfrelative{\pgfxy(1.5,0)}{\pgfnodecenter{Node1}}}{0.25cm}
\pgfnodecircle{Node3}[stroke]
{\pgfrelative{\pgfxy(1.5,0)}{\pgfnodecenter{Node2}}}{0.25cm}
\pgfnodecircle{Node4}[stroke]
{\pgfrelative{\pgfxy(1.5,0)}{\pgfnodecenter{Node3}}}{0.25cm}
\pgfnodecircle{Node5}[stroke]
{\pgfrelative{\pgfxy(1.5,0)}{\pgfnodecenter{Node4}}}{0.25cm}
\pgfnodecircle{Node6}[stroke]
{\pgfrelative{\pgfxy(1.5,0)}{\pgfnodecenter{Node5}}}{0.25cm}
\pgfnodecircle{Node7}[stroke]
{\pgfrelative{\pgfxy(1.5,0)}{\pgfnodecenter{Node6}}}{0.25cm}
\pgfnodecircle{Node8}[stroke]
{\pgfrelative{\pgfxy(1.5,0)}{\pgfnodecenter{Node7}}}{0.25cm}
\pgfnodecircle{Node9}[stroke]
{\pgfrelative{\pgfxy(1.5,0)}{\pgfnodecenter{Node8}}}{0.25cm}
\pgfnodecircle{Node10}[stroke]
{\pgfrelative{\pgfxy(1.5,0)}{\pgfnodecenter{Node9}}}{0.25cm}
 \pgfnodecircle{Node11}[stroke]
{\pgfrelative{\pgfxy(0,1.5)}{\pgfnodecenter{Node8}}}{0.25cm}

\pgfnodebox{Node12}[virtual]{\pgfxy(1,0)}{$\alpha_{1}$}{2pt}{2pt}
\pgfnodebox{Node13}[virtual]{\pgfxy(2.5,0)}{$\alpha_{2}$}{2pt}{2pt}
\pgfnodebox{Node14}[virtual]{\pgfxy(4,0)}{$\alpha_{3}$}{2pt}{2pt}
\pgfnodebox{Node15}[virtual]{\pgfxy(5.5,0)}{$\alpha_{4}$}{2pt}{2pt}
\pgfnodebox{Node16}[virtual]{\pgfxy(7,0)}{$\alpha_{5}$}{2pt}{2pt}
\pgfnodebox{Node17}[virtual]{\pgfxy(8.5,0)}{$\alpha_{6}$}{2pt}{2pt}
\pgfnodebox{Node18}[virtual]{\pgfxy(10,0)}{$\alpha_{7}$}{2pt}{2pt}
\pgfnodebox{Node19}[virtual]{\pgfxy(11.5,0)}{$\alpha_{8}$}{2pt}{2pt}
\pgfnodebox{Node20}[virtual]{\pgfxy(13,0)}{$\alpha_{9}$}{2pt}{2pt}
\pgfnodebox{Node21}[virtual]{\pgfxy(14.5,0)}{$\alpha_{10}$}{2pt}{2pt}
\pgfnodebox{Node22}[virtual]{\pgfxy(11.5,2.5)}{$\alpha_{11}$}{2pt}{2pt}

\pgfnodeconnline{Node1}{Node2} \pgfnodeconnline{Node2}{Node3}
\pgfnodeconnline{Node3}{Node4}\pgfnodeconnline{Node4}{Node5}
\pgfnodeconnline{Node5}{Node6} \pgfnodeconnline{Node6}{Node7}
\pgfnodeconnline{Node7}{Node8}\pgfnodeconnline{Node8}{Node9}
\pgfnodeconnline{Node9}{Node10} \pgfnodeconnline{Node8}{Node11}

\end{pgfpicture}  }
\caption { \sl \small Dynkin diagram of  $E_{11}= E_{8}^{+++}$. The horizontal nodes correspond to the subalgebra $\mf s= \mf{sl}(11, \mbb R)$ and the node $\alpha_{11}$ corresponds to the exceptional root with respect to which we count the level $\ell$.} 
\label{fig1:e11}
\end{center}
\end{figure} 


\begin{table}[ht]
\begin{center}
\begin{tabular}{|c|c|c|c|c|c|c|}
\hline
$\ell$ &$ \mathfrak{sl}(11,\mathbb{R})$ Dynkin labels& Lowest weights of $E_{11}$ & $\Lambda^2$ & dim & mult & $\mu$\\
 &$  [p_1,\,  p_2,\, \ldots, p_{10}]$& $\Lambda= (m_1,\, m_2,  \ldots, m_{10}, \, \ell)$ &  &$\mathcal{R}_{\Lambda}$ &  $\Lambda$& \\
\hline \hline
$0$ & $[0\, 0 \,  0\,  0\,  0\,  0\,  0\,  0\,  0\,  0]$ &$(0\,  0\,  0\,  0\,  0\,  0\,  0\,  0\,  0\,  0\,  0 )$ &$0$& $1$& $11$& $1$\\
$0$ & $[1\, 0 \,  0\,  0\,  0\,  0\,  0\,  0\,  0\,  1]$ &$(-1\,  -1\,  \ldots   -1\,  0 )$ &$2$& $120$& $1$& $1$\\
$1$ & $[0\, 0 \,  0\,  0\,  0\,  0\,  0\,  1\,  0\,  0]$ &$(0\,  0\,  0\,  0\,  0\,  0\,  0\,  0\,  0\,  0\, 1 )$ &$2$& $165$& $1$& $1$\\
$2$ & $[0\, 0 \,  0\,  0\,  1\,  0\,  0\,  0\,  0\,  0]$ &$(0\,  0\,  0\,  0\,  0\,  1\,  2\,  3\,  2\,  1\,  2 )$ &$2$& $462$& $1$& $1$\\
$3$ & $[0\, 1 \,  0\,  0\,  0\,  0\,  0\,  0\,  0\,  0]$ &$(0\,  0\,  1\,  2\,  3\,  4\,  5\,  6\,  4\,  2\,  3 )$ &$0$& $55$& $8$& $0$\\
$3$ & $[0\, 0 \,  1\,  0\,  0\,  0\,  0\,  0\,  0\,  1]$ &$(0\,  0\,  0\,  1\,  2\,  3\,  4\,  5\,  3\,  1\,  3 )$ &$2$& $1760$& $1$& $1$\\
$4$ & $[0\, 0 \,  0\,  0\,  0\,  0\,  0\,  0\,  0\,  1]$ &$(1\,  2\,  3\,  4\,  5\,  6\,  7\, 8\,  5\,  2\,  4 )$ &$-2$& $11$& $46$& $1$\\
$4$ & $[1\, 0 \,  0\,  0\,  0\,  0\,  0\,  0\,  1\,  0]$ &$(0\,  1\,  2\,  3\,  4\,  5\,  6\,  7\,  4\,  2\,  4)$ &$0$& $594$& $8$& $0$\\
$4$ & $[1\, 0 \,  0\,  0\,  0\,  0\,  0\,  0\,  0\,  2]$ &$(0\,  1\,  2\,  3\,  4\,  5\,  6\,  7\,  4\,  1\,  4)$ &$2$& $715$& $1$& $1$\\
$4$ & $[0\, 1 \,  0\,  0\,  0\,  0\,  0\,  1\,  0\, 0]$ &$(0\,  0\,  1\,  2\,  3\,  4\,  5\,  6\,  4\,  2\,  4)$ &$2$& $8470$& $1$& $1$\\
$\vdots$ & $\vdots$ &$\vdots$ &$\vdots$& $\vdots$& $\vdots$& $\vdots$\\
\hline
\end{tabular}
\caption{\sl \small Level decomposition of $E_{11}$ under $\mathfrak{sl}(11,\mathbb{R})$ up to level $4$ obtained from the SimpLie program \cite{SimpLie}.}
\label{tab:levdece11}
\end{center}
\end{table}

\subsubsection{Level $\boldsymbol{\ell=0}$}
Let us consider first the representations of $\mathfrak{sl}(11,\mathbb{R})$ occurring at level zero. By virtue of the Table \ref{tab:levdece11}, we get $2$ different representations $\mathcal{R}^{(0)}_1$ and $\mathcal{R}^{(0)}_2$ (the notation $\mathcal{R}^{(\ell)}_i$ defined the $i$-th representation at the level $\ell$. The lowest weight of the first one is $\Lambda(\mathcal{R}_1^{(0)})=0$ and of the second one is $\Lambda(\mathcal{R}_2^{(0)})= -\alpha_1 - \alpha_2 \ldots - \alpha_{10}$.

 We will first focus  on the second representation $\mathcal{R}_2^{(0)}$ which corresponds to the adjoint representation of $\mathfrak{sl}(11,\mathbb{R})$. The generator associated to the root  $\Lambda(\mathcal{R}_2^{(0)})$  is the multicommutator $\sim [f_1,[f_2,[ \ldots [f_9,f_{10}]]]]$. Using the Chevalley-Serre relations \eqref{eqn:chevalley} and \eqref{eqn:serre}, we can see that the action of any Chevalley generators $f_i,\ i=1, \ldots, 10$ of $A_{10}$ on this multicommutator will annihilate it 
\be
\big[ f_i, [f_1,[f_2,[ \ldots [f_9,f_{10}]]]] \big]=0     \quad \forall i=1, \ldots ,10\, .
\ee
Therefore it corresponds to the lowest weight  of the representation $\mathcal{R}^{(0)}_{2}$. The dimension of this representation is $120$ and we reach all the weights of this representation by the action of $\mf n_{+} (A_{10})$ on this lowest state. 

The other representation at level zero $\mathcal R^{(0)}_1$  has its  lowest weight equal to zero. This one-dimensional representation consists in the Cartan generator $h_{11}$ not belonging to the subalgebra $\mathfrak{sl}(11,\mathbb{R})$.

Therefore, the $\ell=0$ generators are the $\mathfrak{sl}(11,\mathbb{R})$ generators and the Cartan generator $h_{11}$ not belonging to the $\mathfrak{sl}(11,\mathbb{R})$-subalgebra. This Cartan generator contributes to the enhancement of the  subalgebra $\mf s= \mathfrak{sl}(11,\mathbb{R})$ to $\mathfrak{gl}(11,\mathbb{R})= \mathfrak{sl}(11,\mathbb{R}) \oplus \mbb R$ where the generator $h_{11}$ can be view as the trace part $\mbb R$ of $\mathfrak{gl}(11,\mathbb{R})$. The  generators of the representation $\mathfrak{gl}(11,\mathbb{R})$  are taken to be $K^a_{\ b}$ $(a,b= 1, \ldots ,11)$ with the following commutation relations
\be \label{kabcom1}
[K^{a}_{\ b}, K^{c}_{\ d}]= \delta^c_b \, K^{a}_{\ d}- \delta^a_d\, K^{c}_{\ b},
\ee
and with the following bilinear form
\be \label{bilnzero}
(K^a_{\ b} | K^c_{\ d} ) = \delta^{a}_{d} \delta^{c}_{b} - \delta^{a}_{b} \delta^{c}_{d} .
\ee
We can now express the Chevalley generators corresponding to the horizontal line of $E_{11}$ in terms of this $K^a_{\ b}$-basis
\be \label{gravlinegen} \begin{split}
e_i &= K^i _{\ i+1}, \quad i=1, \ldots 10,\\
f_i&= K^{i+1} _{\ i},\\
h_i&= K^{i}_{\ i}- K^{i+1}_{\ \ i+1}.
\end{split}\ee
This relation between the generators  $K^a_{\ b}$ of $\mathfrak{gl}(11,\mathbb{R})$ and the Chevalley generators of $E_{11}$ follows from the comparison of the commutations relations \eqref{eqn:chevalley} and \eqref{kabcom1}. We get also the identification  for the Cartan generator $h_{11}$
\be
h_{11}= - \frac{1}{3}\, (K^1_{\ 1}+ K^2_{\ 2}+ \ldots + K^8_{\ 8} ) + \frac{2}{3}\,  (K^9_{\ 9}+ K^{10}_{\ 10}+ K^{11}_{\ 11}).
\ee

\subsubsection{Level $\boldsymbol{\ell=1}$}
We will now exhibit the single  representation content at level $\ell=1$.  As indicated in Table \ref{tab:levdece11}, the only  $\mathfrak{sl}(11,\mathbb{R})$ representation at level one is spanned by  an antisymmetric $3$-indices tensor $R^{a_1 a_2 a_3}$. The symmetry properties of this tensor can be encoded by  the following Young tableau 
\be
[0\, 0 \,  0\,  0\,  0\,  0\,  0\,  1\,  0\,  0] \quad \leftrightarrow \quad {\footnotesize
    \setlength{\tabcolsep}{0.55 em}
    \begin{tabular}{|c|}
      \hline
       {} \\
      \hline {} \\
      \hline {} \\
      \hline            \end{tabular}} \quad \leftrightarrow \quad R^{a_1 a_2 a_3}    \ .
\ee 
This level one generator $R^{a_1 a_2 a_3}$ transforms covariantly under the $\mathfrak{gl}(11,\mathbb{R})$ genereators
\be
[K^{c}_{\ d}, R^{ a_1 a_2 a_3}]= \delta^{a_1}_d \, R^{c a_2 a_3 } +  \delta^{a_2}_d \, R^{a_1 c a_3 }+  \delta^{a_3}_d \, R^{a_1 a_2 c }.
\ee
The lowest weight of this representation is $\Lambda (\mathcal{R}^{(1)})= \alpha_{11}$ and the corresponding root vector  $e_{11}$  is identified as
\be
e_{11}= R^{9\, 10\, 11}.
\ee
If at level $\ell=1$ the representation $R^{a_1 \, a_2\, a_3}$ appears, then at level $\ell=-1$ one finds its dual representation $R_{a_1 \, a_2\, a_3}$ which satisfy the following tensor relation
\be
[K^{c}_{\ d}, R_{ a_1 a_2 a_3}]= -\delta^{c}_{a_1} \, R_{d a_2 a_3 } -  \delta^{c}_{a_2} \, R_{a_1 d a_3 }-  \delta^{c}_{a_3} \, R_{a_1 a_2 d }.
\ee
In what follows, positive level generators will always be denoted with upper indices and negative level ones with lower ones. We can now easily get the identification of the negative Chevalley generator $f_{11}$ as
\be
f_{11}= R_{9\, 10\, 11}.
\ee 
By demanding $[e_{11}, f_{11}]= h_{11}$ we find 
\be
[R^{abc}, R_{ d e f}]= 18\, \delta^{[ a b}_{[de}\, K^{c]}_{\ f]} -2 \, \delta^{abc}_{def} \sum_{a=1}^{11}\, K^a_{\ a}\, ,
\ee
where 
\be \begin{split}
\delta^{ab}_{cd} &= \frac{1}{2} (\delta^a_c \delta^b_d- \delta^b_c \delta^a_d)\, ,\\
\delta^{abc}_{def} &= \frac{1}{3!} (\delta^a_d\, \delta^b_e \delta^c_f \pm 5 \mathrm{permutations})\, .
\end{split} \ee
Using the invariance of the bilinear form under the adjoint action (see \eqref{eqn:invariance}), we can extend the bilinear form at level zero \eqref{bilnzero} to the level $\ell=1$. We get at level one
\be
(R^{a_1  a_2 a_3 }| R_{b_1  b_2 b_3 })= 3! \delta^{a_1  a_2 a_3}_{b_1  b_2 b_3},
\ee
such that $(e_{11}| f_{11})= (R^{9\, 10\, 11}| R_{9\, 10\, 11})=1$. This level one representation is of dimension $165$ and we reach all the weight of this representation by the action of $\mf n_{+} (A_{10})$ on this lowest state.
\subsubsection{Level $\boldsymbol{\ell=2}$}
By virtue of the Dynkin labels at this level (see Table \ref{tab:levdece11})
\be
[0\, 0 \,  0\,  0\,  1\,  0\,  0\,  0\,  0\,  0] \quad \leftrightarrow \quad {\footnotesize
    \setlength{\tabcolsep}{0.55 em}
    \begin{tabular}{|c|}
      \hline
       {} \\
      \hline {} \\
      \hline {} \\
      \hline
       {} \\
      \hline {} \\
      \hline {} \\
      \hline          \end{tabular}}    \quad \leftrightarrow \quad  R^{a_1 a_2  a_3  a_4 a_5 a_6} \ .
\ee 
we get that the single representation is spanned by the antisymmetric $6$-indices tensor $R^{a_1 a_2  a_3 a_4 a_5 a_6}$.  Using the graded structure of the level decomposition \eqref{comutelevell}, this level $\ell=2$ generator can be view as the commutator of level $\ell=1$ generators
\be
R^{a_1 a_2  a_3  a_4 a_5 a_6}= [ R^{[a_1  a_2a_3 }, R^{a_4  a_5 a_6]}]\, .
\ee
In the same way as for level $\ell=1$, we get the bilinear form at this level using the invariance of the bilinear form under the adjoint action
\be
(R^{a_1 a_2  a_3  a_4 a_5 a_6}| R_{b_1 b_2  b_3 b_4 b_5 b_6})= 6!\,  \delta^{a_1 a_2  a_3  a_4 a_5 a_6}_{b_1 b_2  b_3 b_4 b_5 b_6} .
\ee
\subsubsection{Level $\boldsymbol{\ell=3}$}
From the Table \ref{tab:levdece11}, we get 2 representations at level $\ell=3$. The first one $\mathcal R^{(3)}_1$ has for Dynkin labels $[0\, 1 \,  0\,  0\,  0\,  0\,  0\,  0\,  0\,  0]$ and corresponds to an $9$-form generator. But the outer multiplicity $\mu$ of this representation is zero meaning that this representation is not allowed. The second representation $\mathcal R^{(3)}_2$ is then the only admissible representation that appears at this level. The corresponding Dynkin labels are associated to a mixed indices tensor $R^{a_1 \, a_2 \ldots a_8 |b}$ where the $a_i$ indices are antisymmetric such that $R^{[a_1 \, a_2 \ldots a_8 |b]}=0$
\be
[0\, 0 \,  1\,  0\,  0\,  0\,  0\,  0\,  0\,  1] \quad \leftrightarrow \quad  {\footnotesize
    \setlength{\tabcolsep}{0.55 em}
    \begin{tabular}{cc}
      \cline{1-2}
      \multicolumn{1}{|c|}{} &
      \multicolumn{1}{c|}{\,    } \\
       \cline{1-2}
       \multicolumn{1}{|c|}{} &
        \multicolumn{1}{c}{} \\
       \cline{1-1}
       \multicolumn{1}{|c|}{\vdots} &
        \multicolumn{1}{c}{} \\
       \cline{1-1}
\multicolumn{1}{|c|}{} &
        \multicolumn{1}{c}{} \\
       \cline{1-1}
\end{tabular}}
\leftrightarrow \quad  R^{a_1 \, a_2 \ldots a_8 |b}\, .
\ee

\subsection{Level decomposition of $ \asuppp$} \label{sec:levdecompo}

In Section  \ref{sec:CosmologicalModel}, we will give the correspondence between the field content of the bosonic part of pure $\mathcal{N}=2$ supergravity and the infinite-dimensional algebra $\asupp$. To this end, we will perform a decomposition of the adjoint representation of $\asuppp$ into representations of an $\mathfrak{sl}(4, \mathbb{R})$ subalgebra defined by the nodes $\alpha_1, \alpha_2$ and $\alpha_3$ in Figure \ref{fig:su21+++}b.  All step operators may then be written as irreducible tensors of the $\mathfrak{sl}(4, \mathbb{R})$ subalgebra. 
\subsubsection{Level decomposition of $A_2^{+++}$}

In order to understand the level decomposition of $\asuppp$, we must first consider the level decomposition of the complex algebra $A_2^{+++}$ under a $A_3\cong \mathfrak{sl}(4,\mathbb{R})$ subalgebra. This level decomposition of $A_2^{+++}$  up to level $\ell= (\ell_1,\ell_2)= (2,2)$ can be obtained  from the SimpLie program \cite{SimpLie} and  it is shown in Table \ref{tab:levdeca2ppp} . The levels $\ell_1$ and $\ell_2$ are respectively associated to the roots $\alpha_4$ and $\alpha_5$ in Figure \ref{fig:su21+++}a. This level decomposition will induce a grading of $A_2^{+++}$ into an infinite set of finite-dimensional subspace $\mathfrak{g}^{+++}_{(\ell_1, \ell_2)}$ such that
\be
A_2^{+++}= \bigoplus_{(\ell_1,\ell_2)} \mathfrak{g}^{+++}_{(\ell_1, \ell_2)}\, ,
\ee 
where the levels $\ell_1$ and $\ell_2$ are either both non-positive or both non-negative.


\begin{table}[t]
\begin{center}
\begin{tabular}{|c|c|c|c|c|c|c|}
\hline
 $(\ell_1, \ell_2)$ &$\mathfrak{sl}(4,\mathbb{R})$ Dynkin labels& Lowest weights of $A_2^{+++}$ & $\Lambda^2$ & dim & mult & $\mu$\\
 &$  [p_1,\,  p_2, \,p_3]$& $(m_1,\, m_2,\, m_3,\, \ell_1, \, \ell_2)$ &  & $\mathcal{R}_{\Lambda}$&$\Lambda$  & \\
\hline \hline
$(0,0)$ & $[0\, 0 \,  0]$ &$(0\,  0\,  0\,  0 \, 0)$ &$0$& $1$& $5$& $2$\\
$(0,0)$ & $[1\, 0 \,  1]$ &$(-1\,  -1\,  -1\,  0\, 0)$ &$2$& $15$& $1$& $1$\\
$(1,0)$ & $[0\, 0 \,  1]$ &$(0\,  0\,  0\, 1\,  0)$ &$2$& $4$& $1$& $1$\\
$(0,1)$ & $[0\, 0 \,  1]$ &$(0\,  0\,  0\, 0\,  1)$ &$2$& $4$& $1$& $1$\\
$(1,1)$ & $[0\, 1 \,  0]$ &$(0\,  0\,  1\, 1\,  1)$ &$0$& $6$& $2$& $1$\\
$(1,1)$ & $[0\, 0 \,  2]$ &$(0\,  0\,  0\, 1\,  1)$ &$2$& $10$& $1$& $1$\\
$(2,1)$ & $[0\, 1 \,  1]$ &$(0\,  0\,  1\, 2\,  1)$ &$2$& $20$& $1$& $1$\\
$(1,2)$ & $[0\, 1 \,  1]$ &$(0\,  0\,  1\, 1\,  2)$ &$2$& $20$& $1$& $1$\\
$(2,2)$ & $[1\, 0 \,  1]$ &$(0\,  1\,  2\, 2\,  2)$ &$-2$& $15$& $5$& $2$\\
$(2,2)$ & $[0\, 2 \,  0]$ &$(0\,  0\,  2\, 2\,  2)$ &$0$& $20$& $2$& $1$\\
$(2,2)$ & $[0\, 1 \,  2]$ &$(0\,  0\,  1\, 2\,  2)$ &$2$& $45$& $1$& $1$\\
$\vdots$ & $\vdots$ &$\vdots$ &$\vdots$& $\vdots$& $\vdots$& $\vdots$\\
\hline
\end{tabular}
\caption{\sl \small Level decomposition of $A_2^{+++}$ under $\mathfrak{sl}(4,\mathbb{R})$ up to level $(2,2)$ obtained from the SimpLie program \cite{SimpLie}.} 
\label{tab:levdeca2ppp} 
\end{center}
\end{table}
\subsubsection{Level $\boldsymbol{(\ell_1, \ell_2)=(0,0)}$}
 At level $\ell=(0,0)$, we get 2 representations: $\mathcal{R}^{(0)}_1$ with lowest weight $\Lambda(\mathcal{R}^{(0)}_1)=0$ and  $\mathcal{R}^{(0)}_2$ with lowest weight $\Lambda(\mathcal{R}^{(0)}_2)= -\alpha_1 - \alpha_2 -\alpha_3$. The representation $\mathcal{R}^{(0)}_2$ is the adjoint representation of $\mathfrak{sl}(4, \mbb R)$. The representation  $\mathcal{R}^{(0)}_1$ has a multiplicity $\mu=2$. This representation consists of the two Cartan generators $h_4$ and $h_5$ not belonging to the subalgebra $\mathfrak{sl}(4, \mbb R)$. Therefore the level zero generators are the $\mathfrak{sl}(4, \mbb R)$ generators and the Cartan generators $h_4+h_5$ and $h_4-h_5$. The Cartan generator $h_4+h_5$ will contribute to the enhancement of the subalgebra $\mathfrak{s}=\mathfrak{sl}(4, \mbb R) $ to $\mathfrak{gl}(4, \mathbb{R})= \mathfrak{sl}(4,\mathbb{R}) \oplus\, \mathbb{R}$ generated by $K^{a}_{\ b}$ $(a, b = 1, \ldots, 4)$.  The extra Cartan generator $T= \tfrac{1}{6}(h_4-h_5)$ enlarges  the  $\mathfrak{gl}(4, \mathbb{R})$ algebra by the addition of a $\mbb R$-factor. The commutation relations at this level are
\be \begin{split}\label{eqn:kab}
\left[K^{a}_{\ b}, K^{c}_{\ d}\right] &= \delta^c_b\,  K^{a}_{\ d} - \delta^a_d\,  K^{c}_{\ b}\, ,\\
\left[T, K^{a}_{\ b}\right]&=0 \, ,
\end{split}
\ee
and the bilinear forms  reads
\be \label{eqn:bilinearzero}
(K^{a}_{\ b} |K^{c}_{\ d}) = \delta^a_d \delta^c_b- \delta^a_b\delta^c_d, \ 
(T|T)= \frac{2}{9}  , \ (T | K^{a}_{\ b})=0\, .
\ee
We can now express the Chevalley generators corresponding to the horizontal line of $A_2^{+++}$ in terms of the $K^a_{\ b}$-basis
\be \label{gravlinegena2ppp} \begin{split}
e_i &= K^i _{\ i+1}, \quad i=1, \ldots 3,\\
f_i&= K^{i+1} _{\ i},\\
h_i&= K^{i}_{\ i}- K^{i+1}_{\ \ i+1}.
\end{split}\ee
We get also the identification for the Cartan generators $h_4$ and $h_5$
\be  \label{eqn:genh4h5level }\begin{split}
h_4 &= - \tfrac{1}{2}\, K + K^4_{\ 4} +3\, T,\\
h_5 &= - \tfrac{1}{2}\, K + K^4_{\ 4} -3\, T, 
\end{split}\ee
where 
\be
K= K^1_{\ 1}+  K^2_{\ 2}+ K^3_{\ 3}+ K^4_{\ 4}\, .
\ee

\subsubsection{Levels $\boldsymbol{(\ell_1, \ell_2)=(0,1) }$ and $\boldsymbol{(1,0) }$}
The positive level generators are obtained through multiple commutators between the generators $R^a$ and $\tilde{R}^a$ on levels $(1,0)$ and $(0,1)$ respectively. They transforms as $\mathfrak{gl}(4, \mathbb{R})$ tensors in the obvious way.

The $T$ commutator relations are
\be
\left[T, R^a \right]= \frac{1}{2} \, R^a, \ \ \left[T, \tilde{R}^a \right]=-  \frac{1}{2} \, \tilde{R}^a .
\ee 
The lowest weight of the representation at level $(1,0)$ is $\Lambda=\alpha_4$. It corresponds to the root vector $e_4$ which is identified as
\be
e_4= R^4.
\ee
In the same way the lowest weight of the representation at level $(0,1)$ is $\Lambda=\alpha_5$. It corresponds to the root vector $e_5$ which is identified as
\be
e_5= \tilde R^4.
\ee
The negative Chevalley generators $f_4$ and $f_5$ are obtained by lowering the indices of the corresponding positive generators. By demanding $[e_4,f_4]=h_4$ and $[e_5,f_5]=h_5$, we find the commutation relations between positive and negative generators
\be \begin{split} \begin{aligned}
\left[ R^a,R_{b}\right] &= \delta^{a}_{b} (- \tfrac{1}{2} K  + 3\, T) + K^{a}_{\ b}\, ,\\
\left[ \tilde{R}^a,\tilde{R}_{b}\right] &= \delta^{a}_{b} (- \tfrac{1}{2} K  - 3\, T) + K^{a}_{\ b}\, .\\
\end{aligned}\end{split}\ee
The bilinear form at these levels are given by 
\be \begin{split}\begin{aligned}
(R^a | R_{b}) &= \delta^{a}_{b},  & (\tilde{R}^a | \tilde{R}_{b})  &= \delta^{a}_{b}.
\end{aligned} \end{split} \ee

Let us now give a complete list of  the identification of the Chevalley generators of $A_2^{+++}$
\be \label{eqn:chevbis}
\begin{split}
\begin{aligned}
&h_1 = K^1_{\ 1} - K^2_{\ 2}, &\qquad&e_1= K^1_{\ 2},    &\qquad&f_1= K^2_{\ 1}                  \, ,\\
&h_2 = K^2_{\ 2} - K^3_{\ 3}, &\qquad&e_2= K^2_{\ 3},    &\qquad&f_2= K^3_{\ 2}                  \, ,\\
&h_3 = K^3_{\ 3} - K^4_{\ 4}, &\qquad&e_3= K^3_{\ 4},    &\qquad&f_3= K^4_{\ 3}                        \, ,\\
&h_4 = - \tfrac{1}{2}\, K + K^4_{\ 4} +3\, T, &\qquad&e_4= R^{4},    &\qquad&f_4= R_{4}           \, ,\\
&h_5 = - \tfrac{1}{2}\, K + K^4_{\ 4} -3\, T, &\qquad&e_5= \tilde{R}^{4},     &\qquad&f_5= \tilde{R}_{4}           \, .
\end{aligned} \end{split}
\ee

\subsubsection{Level $\boldsymbol{(\ell_1, \ell_2)=(1,1)}$}
By virtue of the Dynkin labels at levels $(1,1)$ (see Table \ref{tab:levdeca2ppp})
\be
[0\, 1 \,  0] \ \leftrightarrow \ {\footnotesize
    \setlength{\tabcolsep}{0.55 em}
    \begin{tabular}{|c|}
      \hline
       {} \\
      \hline {} \\
     \hline
         \end{tabular}}    \ \leftrightarrow \  R^{ab} \qquad,  \qquad [0\, 0 \,  2] \ \leftrightarrow \ {\footnotesize
    \setlength{\tabcolsep}{0.55 em}
    \begin{tabular}{|c|c|}
      \hline
       {} &{} \\
      \hline 
         \end{tabular}}    \  \leftrightarrow \  S^{ab} \ ,
         \ee 
we get that the representations at this level are spanned by the antisymmetric-2 indices $R^{ab}$ tensor and by the  symmetric-2 indices $S^{ab}$ tensor. These level $(1,1)$ generators can be obtained through the commutator
\be
\left[ R^a, \tilde{R}^b\right]=  R^{ab} + S^{ab} \,  ,
\ee
where the individual projections are:
\be\label{eqn:rabsab}
S^{ab}= \left[ R^{ ( a}, \tilde{R}^{b ) }\right] , \quad   R^{ab}= \left[ R^{ [ a}, \tilde{R}^{b ] }\right] \  .
\ee

\begin{table}[h]
\begin{center}
\begin{tabular}{|c|c|c|}
\hline
$(\ell_1,\ell_2)$ &$ \mathfrak{sl}(4,\mathbb{R})$ Dynkin labels& Generator of $A_2^{+++}$\\
\hline \hline
$(0,0)$ &$[ 1,0,1] \oplus [ 0,0,0] $ & $K^a_{\ b}$\\
$(0,0)$ &$[ 0,0,0]  $ & $T$\\
$(1,0)$ &$[ 0,0,1]  $ & $R^{\, a}$\\
$(0,1)$ &$[ 0,0,1]  $ & $\tilde{R}^{\, a}$\\
$(1,1)$ &$[ 0,0,2]  $ & $S^{\, s_{1} s_{2}}$\\
$(1,1)$ &$[ 0,1,0]  $ & $R^{\,a_{1} a_{2}}$\\
$(2,1)$ &$[ 0,1,1]  $ & $R^{\, a_{0}|a_{1} a_{2}}$\\
$(1,2)$ &$[ 0,1,1]  $ & $\tilde{R}^{\, a_{0}|a_{1} a_{2}}$\\
$(2,2)$ &$[ 1,0,1]  $ & $R^{\, a_{0}|a_{1} a_{2} a_{3}}$\\
$(2,2)$ &$[ 0,2,0]  $ & $R^{\, a_{1} a_{2} | a_{3} a_{4}}$\\
$(2,2)$ &$[ 0,1,2]  $ & $R^{\,s_{1} s_{2} |a_{3} a_{4} }$\\
$\vdots$& $\vdots$ &$\vdots$\\
\hline
\end{tabular}
\caption{\sl \small Level decomposition of $A_2^{+++}$ under $\mathfrak{sl}(4,\mathbb{R})$ up to level $(2,2)$. The  indices $a_{i}$ are antisymmetric  while the indices $s_{i}$ are symmetric. Note that the generators from the level $(2,1)$ with mixed Young symmetries are subject to constraints.  }
\label{tab:levdeca2}
\end{center}
\end{table}

Negative step operators are defined with lower indices such that the bilinear form evaluated on a positive step operator and its corresponding negative step operator is positive. Then at level $(-1,-1)$, we have
\be 
S_{ab}=  - \left[ R_{ ( a}, \tilde{R}_{b ) }\right] , \quad   R_{ab}= -\left[ R_{ [ a}, \tilde{R}_{b ] }\right] \  ,
\ee
such that the bilinear forms are given by
\be
(R^{ab} | R_{cd})= 3\,   \delta^{ab}_{cd}, \quad   (S^{ab} | S_{cd}) = \bar{ \delta}^{ab}_{cd}, 
\ee
where
\be \begin{split} \begin{aligned}
\delta^{ab}_{cd} &:=  \tfrac{1}{2} (\delta^a_c\, \delta^b_d- \delta^b_c\, \delta^a_d)\, ,\\
\bar{  \delta}^{ab}_{cd} &:=  \tfrac{1}{2} (\delta^a_c\, \delta^b_d+ \delta^b_c\, \delta^a_d)\, .
\end{aligned} \end{split} \ee
The commutations relations between a positive generator and the negative one are given by
\be \begin{split} \begin{aligned}
\left[ R^{ab},R_{cd}\right] &=- 3\,  \delta^{ab}_{cd} K+ 6 \, \delta^{[ a}_{[ c}\, K^{b]}_{\ d]}\, ,\\
\left[ S^{ab},S_{cd}\right] &=-  \bar{\delta}^{ab}_{cd} K+ 2 \, \delta^{(a}_{( c}\, K^{b)}_{\ d)}\,,
\end{aligned} \end{split} \ee
while the generators of different rank commute in the following non-trivial way:
\be \begin{split} \begin{aligned}
\left[ S^{ab}, R_c\right] &= - \delta^{(a}_{c}\,  \tilde{R}^{b)}, && \left[ S^{ab}, \tilde{R}_c\right] =  \delta^{(a}_{c}\, R^{b)}\, ,\\
\left[ R^{ab}, R_c\right] &= - 3 \,  \delta^{[a}_{c}\, \tilde{R}^{b]}, && \left[ R^{ab}, \tilde{R}_c\right] = - 3 \,  \delta^{[a}_{c}\, R^{b]}, \\
\left[ S^{ab}, R_{cd} \right]&=0 \, .
\end{aligned} \end{split} \ee

\subsubsection{Level decomposition of $\asuppp$}
We shall now apply the construction of $\asuppp$ to the level decomposition of $A_2^{+++}$. In this context, we define the level $L$ such that $L=\ell_1 +\ell_2$ and such that the grading of the $\asuppp$ algebra is written as
\be
\asuppp= \bigoplus_L \mathfrak{g}^{+++}_{L}.
\ee
\subsubsection{Level $\boldsymbol{L=0}$}
 At level zero, we have the $\mathfrak{gl}(4, \mathbb{R})$-subalgebra associated to the nodes $\alpha_1, \alpha_2$ and $\alpha_3$. These nodes are non-compact and hence are, as we have seen in (\ref{eqn:sigmalphanc}), invariant under $\sigma$. Thus, the $\mathfrak{gl}(4,\mathbb{R})$ part at $L=0$ is the same as for $A_2^{+++}$. The extra Cartan generators associated to $\alpha_4$ and $\alpha_5$ are however affected by the conjugation $\sigma$. Using (\ref{eqn:su21+++Chev}) and (\ref{eqn:chevbis}), the invariant combinations are 
\begin{eqnarray}\label{eqn:h4}
\mathbf{h_4}&=& h_4+ h_5 = - K + 2 K^4_{\ 4}  ,\\
\mathbf{h_5}&=& i (h_4-h_5) = i 6\, T \label{eqn:h5},
\end{eqnarray}
where $K= \sum_{a=1}^{4} K^a_{\ a}$.
We have already seen that the first one is non-compact, while the second one is compact, i.e.
\be
\theta(\mathbf{h_4})= -\mathbf{h_4}\,, \quad \theta(\mathbf{h_5})= \mathbf{h_5}.
\ee
The effect of the algebraic Iwasawa decomposition (\ref{Iwasawadecreal}) will therefore be to project out the compact Cartan $\mathbf{h_5}$. This was anticipated since the generator $T$ is associated with a dilaton which does not exist in four-dimensional Maxwell-Einstein gravity. 
\subsubsection{Level $\boldsymbol{L=1}$}

We further define the invariant generators at level $L=1$
\be \label{eqn:rarta} \begin{split}
r^a &:= R^a+ \tilde{R}^a,\\
\tilde{r}^a &:= i (R^a - \tilde{R}^a).
\end{split} \ee
The corresponding negative step operators at level $L=-1$ are defined by
\be \label{eqn:rartaneg} \begin{split}
r_a &:= R_a+ \tilde{R}_a,\\
\tilde{r}_a &:= i (R_a - \tilde{R}_a).
\end{split} \ee
More generally, the negative step operators are obtained from the positive ones by lowering the indices as in \eqref{eqn:rartaneg}. The bilinear forms at this level reads
\be \label{eqn:bililevel1}
( r^a | r_b)= 2\, \delta^a_b, \quad ( \tilde{r}^a | \tilde{r}_b) = - 2\, \delta^a_b\, , 
\ee
and the generators of opposite levels $L=\pm 1$ commute as follows
\be \begin{split} \begin{aligned}
\left[ r^a, r_b \right] &= - \delta^a_b\, K + 2\, K^a_{\ b}\, ,\\
\left[ \tilde{r}^a, \tilde{r}_b \right] &= \delta^a_b\, K - 2\, K^a_{\ b}\, ,\\
\left[ r^a, \tilde{r}_b \right] &= 6\, i \,\delta^a_b\, T \, .
\end{aligned} \end{split} \ee
\begin{table}[top]
\begin{center}
\begin{tabular}{|c|c|}
\hline
$L= \ell_1 +\ell_2$ &  Generator of $\asuppp$ \\
\hline \hline
$0$  & $K^a_{\ b}$\\
$0$  & $i \, T$\\
$1$  & $r^{\, a} = R^a + \tilde{R}^a$\\
$1$  & $\tilde{r}^{\, a} = i (R^a - \tilde{R}^a) $\\
$2$  & $s^{\,s_{1} s_{2}} = - 2\, i\, S^{\, s_{1} s_{2}}$\\
$2$  & $r^{\,a_{1} a_{2}}=  2\, R^{\, a_{1} a_{2}}$\\
$3$  & $r^{\, a_{0}|a_{1} a_{2}}$\\
$3$   & $\tilde{r}^{\, a_{0}|a_{1} a_{2}}$\\
$4$  & $r^{\, a_{0}|a_{1} a_{2} a_{3}}$\\
$4$  & $r^{\, a_{1} a_{2} | a_{3} a_{4}}$\\
$4$  & $r^{\,s_{1} s_{2} |a_{3} a_{4} }$\\
$\vdots$ & $\vdots$\\
\hline
\end{tabular}
\caption{\sl \small Level decomposition of $\asuppp$ under $\mathfrak{sl}(4,\mathbb{R})$ up to level $4$. The indices  $a_{i}$ are antisymmetric  while the indices  $s_{i}$ are symmetric. Note that the generators from the level $L=3$ with mixed Young symmetries are subject to constraints.  }
\label{tab:levdecsu}
\end{center}
\end{table}
Using (\ref{eqn:su21+++Chev}) and (\ref{eqn:chevbis}), we get that the invariant combinations of the Chevalley generators at level $L=1$ are
\be \label{eqn:e4ande5}
\mathbf{e_4}= r^4, \quad \mathbf{e_5}= \tilde{r}^4.
\ee
That all other components of $r^a$ and $\tilde{r}^a$ are also invariant follows from the fact that they may be written as commutators between $\mathfrak{gl}(4, \mathbb{R})$ and $r^4$ and $\tilde{r}^4$ which are all invariant. The two Chevalley generators $\mathbf{e_4}$ and $\mathbf{e_5}$ have identical eigenvalues with respect to the four noncompact Cartan
 \be \label{comute4} \begin{split} \begin{aligned}
 \left[\mathbf{h_1}, \mathbf{e_4}\right]&=0,  & [\mathbf{h_2}, \mathbf{e_4}]&=0,& [\mathbf{h_3}, \mathbf{e_4}]&=- \mathbf{e_4}, &[\mathbf{h_4}, \mathbf{e_4}]&=\mathbf{e_4}, \\
 [\mathbf{h_1}, \mathbf{e_5}]&=0, & [\mathbf{h_2}, \mathbf{e_5}]&=0,& [\mathbf{h_3}, \mathbf{e_5}]&=- \mathbf{e_5}, &[\mathbf{h_4}, \mathbf{e_5}]&=\mathbf{e_5}, 
\end{aligned} \end{split} \ee
implying that these generators project into the same root $\vec{\lambda}$ in the restricted root system (see Section \ref{app:restricted} for more details),
\be
\vec{\lambda}= \vec{\alpha}_{\mathbf{e_4}}= \vec{\alpha}_{\mathbf{e_5}}= (0,0,-1,1).
\ee
The generator $\mathbf{h_5}$, being compact, is not diagonalizable over $\mathbb{R}$. Indeed, we have the following commutation relations with $\mathbf{e_4}$ and $\mathbf{e_5}$
\be \label{comutenc45}
[\mathbf{h_5},\mathbf{e_4}] =3\, \mathbf{e_5} \quad, \quad [\mathbf{h_5},\mathbf{e_5}] =- 3\, \mathbf{e_4}.
\ee
\subsubsection{Level $\boldsymbol{L=2}$}
The generators at level $L=2$ are obtained as
\be \label{eqn:sabrab} \begin{split}
s^{ab}&:= [r^a, \tilde{r}^b]\, ,\\
r^{ab}&:=  [r^a,r^b]=  [\tilde{r}^a, \tilde{r}^b].
\end{split} \ee
These generators are separately invariant under $\sigma$. In terms of $A_2^{+++}$ generators, using (\ref{eqn:sabrab}), \eqref{eqn:rarta}, and \eqref{eqn:rabsab} we get
\begin{eqnarray}
s^{ab}&=& - 2i\,  S^{ab},\\
r^{ab}&=& 2\, R^{ab}\, .
\end{eqnarray}
These generators are normalized as
\be
(s^{ab} | s_{cd})= -4\, \bar{\delta}^{ab}_{cd}\,, \quad (r^{ab} | r_{cd})= 12 \, \delta^{ab}_{cd}\, ,
\ee
where  $\delta^{ab}_{cd} :=  \tfrac{1}{2} (\delta^a_c\, \delta^b_d- \delta^b_c\, \delta^a_d)$ and
$\bar{  \delta}^{ab}_{cd} :=  \tfrac{1}{2} (\delta^a_c\, \delta^b_d+ \delta^b_c\, \delta^a_d)$. The commutation relations of generators of opposite levels $L=\pm2$ are
\be \begin{split} \begin{aligned}
\left[ s^{a b}, s_{c d} \right] &= 4\,  \bar{\delta}^{ab}_{cd}\, K  - 8\, \delta^{(a}_{( c}\, K^{b)}_{\ d)}   \, ,\\
\left[ r^{a b}, r_{c d} \right] &= - 12\,  \delta^{ab}_{cd}\, K  + 24\, \delta^{[a}_{[c}\, K^{b]}_{\ d]}   \, ,\\
\left[ r^{a b}, s_{c d} \right] &=  0 \,.
\end{aligned} \end{split} \ee
The generators of different rank commute in the following non-trivial way:
\be \begin{split} \begin{aligned}
\left[ s^{ab}, r_c\right] &= - 2\,  \delta^{(a}_{c}\,  \tilde{r}^{b)}, &&  \left[ s^{ab}, \tilde{r}_c\right] = -2\,  \delta^{(a}_{c}\, r^{b)}\, ,\\
\left[ r^{ab}, r_c\right] &= - 6 \,  \delta^{[a}_{c}\, r^{b]}, &&  \left[ r^{ab}, \tilde{r}_c\right] = 6 \,  \delta^{[a}_{c}\, \tilde{r}^{b]}\, .
\end{aligned} \end{split} \ee
The level decomposition of $\asuppp$ under the $A_3\cong \mathfrak{sl}(4,\mathbb{R})$ subalgebra up to level $L=4$ is shown in Table \ref{tab:levdecsu}. Note that this level decomposition presents the same Young tableaux as in the $A_2^{+++}$ case. We will see in Section \ref{sec:CosmologicalModel} that this representation content up to level $L=2$ where the generator $r^{a_1a_2} $ is projected out,  can be associated with the bosonic field content of pure $\mathcal N=2$ supergravity in $D=4$. \\

   \chapter{Supergravity theories reformulated as non-linear $\sigma$-models}\label{chap:sugrareformulatedth}
 
We will now make use of the contents of the previous chapter to describe some developments devoted to the investigation of the underlying symmetry structures of supergravity theories.\\

 In Section \ref{sec:kmsymcompct14}, we begin by describing how toroidal compactifications of gravity theories reveal hidden global and local symmetries of the reduced Lagrangian. We also discuss attempts at extending these symmetry structures to infinite-dimensional algebras. In Section \ref{app:SigmaModel}, we recall the construction of non-linear $\sigma$-models over coset spaces. We will then apply the tools given in Section \ref{app:SigmaModel} to construct respectively in Sections \ref{sec:gpppinvactionthese1} and \ref{sec:gppinvactionthese2} actions explicitly invariant under very-extended and over-extended Kac-Moody symmetry group .
 
 \setcounter{equation}{0}
 \section{Kac-Moody symmetries through compactifications} \label{sec:kmsymcompct14}
 
In this section, we will explain how hidden Kac-Moody symmetries are exhibited through compactifications on a torus of $D$-dimensional gravitational theories suitably coupled to dilatons $\Phi^u$ and to matter fields associated to $n$-forms $F_{(n)}$ whose Lagrangian is 
\begin{eqnarray}\label{tab:actionmax}
\mathcal{L}= \sqrt{-g}\,\bigg( R -
\frac{1}{2}\sum_{u=1}^{q}\partial_{M}\Phi^{u} 
\partial^{M}\Phi^{u} -\sum_{n}\, \frac{1}{2n!}\, e^{\sum_{u}a_{n}^{u}\Phi^{u}}
F^{2}_{(n)}   \bigg)  +\mathcal{L}_{CS}\, ,
\end{eqnarray}
where $\mathcal{L}_{CS}$ contains the eventual Chern-Simons terms.
We will first recall briefly the Kaluza-Klein reduction of such theories on a circle $S^1$. Then, we will consider in details the reduction of gravity on $T^2$. In this simple example we introduce the essential aspects of enhanced symmetries, most notably that of scalar coset Lagrangians and nonlinear realisations.  We will thereafter consider  the reduction of theories \eqref{tab:actionmax} down to three dimensions. This reduction may exhibits a simple Lie group $\g$ symmetry non-linearly realised. The scalars of the dimensionally reduced theory live in a coset $\g/ \mathrm K$ where $\mathrm K$ is the maximal compact subgroup of $\g$. Finally, we will explain the conjecture that infinite-dimensional Kac-Moody groups should be  symmetries of some supergravity theories.

In the following we shall restrict to the bosonic sectors of supergravity theories and we will consider only reduction on space-like circles.
\subsection{Kaluza-Klein Reduction on $S^1$}
\label{PERTABsec:KKred}

Before starting with the reduction of the bosonic part of supergravity theories on a torus, it will be helpful to briefly recall the Kaluza-Klein reduction on a $S^1$. For that, we will follow  the lecture notes \cite{Pope}.

First we consider the reduction of pure gravity, described by the Einstein-Hilbert Lagrangian in $D$ dimensions
\begin{eqnarray}\label{tab:lag}
\mathcal{L}_{D}= \sqrt{-\widehat{g}}\, \widehat{R}.
\end{eqnarray}
We put hats on the fields to signify that they are constructed from the $D$-dimensional metric $\widehat{g}_{MN} \ (M,N= 0, \ldots, D-1)$. Now suppose that we wish to reduce the theory to $D-1$ dimensions, by compactifying the coordinate $x^{D-1} = z$ on a circle. In general the enhanced symmetries are manifest only in Einstein frame, and therefore we perform the compactification so as to end up with a $(D-1)$-dimensional theory in Einstein frame, with a standard kinetic term for the ``dilatonic''  scalar $\phi$. This requirement fixes the compactification ansatz to be 
\be \begin{split}\label{tab:ans}
\dd\widehat{s}^{2} &= \widehat{g}_{MN}\,  \dd x^M \, \dd x^N,\\
&= e^{2\gamma_{D-1}\phi}\dd s^{2}+e^{2\beta_{D-1}\phi
}(\dd z+A_{\mu}\dd x^{\mu})(\dd z+A_{\nu}\dd x^{\nu})\, ,
\end{split}\ee
where
\begin{eqnarray}\label{tab:alphbet}
 \gamma_{D-1}= \frac {1}{\sqrt{2(D-2)(D-3)}} \  ,\qquad \,  \beta_{D-1}= - (D-3)\gamma_{D-1}  .
\end{eqnarray}
Note that this ansatz implies that the components of the higher-dimensional metric $\widehat{g}_{MN}$ ($M,N= 0, \ldots, D-1$) are given in terms of the $(D-1)$-dimensional fields $g_{\mu \nu}, \, A_{\mu}$ and $ \phi$ ($\mu, \nu= 0, \ldots, D-2$) by
\begin{eqnarray}
\widehat{g}_{\mu\nu}=e^{2 \gamma_{D-1}\phi}g_{\mu\nu}+
e^{2\beta_{D-1}\phi }A_{\mu} A_{\nu} \, ,\qquad
\widehat{g}_{\mu z}=e^{2\beta_{D-1}\phi }A_{\mu} \, ,\qquad
\widehat{g}_{zz}=e^{2\beta_{D-1}\phi } \, .
\end{eqnarray}
All the fields on the right-hand side of (\ref{tab:ans}) are independent of the compact direction $z$ which translates the restriction to the zero modes in the Fourier expansion of the $D$-dimensional metric $\widehat{g}_{MN}$ along the compactified coordinate $z$. \\

The result of the dimensional reduction of \eqref{tab:lag} is
\begin{eqnarray}\label{tab:redl}
\sqrt{-\widehat{g}}\,\widehat{R} \longrightarrow
\sqrt{-g}\,\bigg(R - \frac{1}{2} \partial_{\mu}\phi\,\partial^{\mu}\phi
-\frac {1}{4}e^{-2(D-2)\ \gamma_{D-1}\phi}\, F_{\mu\nu} F^{\mu\nu}\bigg)\, ,
\end{eqnarray}
where $F_{\mu\nu}=2\,   \partial_{[\mu}A_{\nu]}$. The resulting theory \eqref{tab:redl} corresponds to $(D-1)$-dimensional gravity coupled to a scalar field $\phi$ called the \emph{dilaton} and a Maxwell field $A_{\mu}$ also called the \emph{Kaluza-Klein vector}. Note that the original $D$-dimensional general coordinates transformations that preserve the form of the ansatz \eqref{tab:ans} is broken to $(D-1)$- dimensional general coordinates transformations and to $\mathrm{U}(1)$ local gauge transformation of the Kaluza-Klein potential vector $A_{\mu}$. The reduced lagrangian \eqref{tab:redl} has also a constant shift symmetry of the dilaton $\phi$ accompanied by an appropriate rescaling of $A_{\mu}$.\\

Having seen how the Kaluza-Klein reduction of the metric works, we shall now study the reduction of an antisymmetric tensor field strength, from $D$ to $D-1$ dimensions. Suppose we have an $n$-form field strength $\widehat{F}_{(n)} = \dd\widehat{A}_{(n-1)}$ in $D$ dimensions. In terms of indices, it is clear that after reduction on $S^1$ there will two kinds of $D-1$ dimensional potentials, namely one with $(n-1)$ indices lying in the $D-1$ dimensional spacetime, and the other with $(n-2)$ indices lying in the $D-1$ dimensional spacetime, and one index being in the compact direction $z$. This is most easily expressed in terms of differential forms. Thus the ansatz for the reduction of the potential $\widehat{A}_{(n-1)}$ is
\begin{eqnarray}
\widehat{A}_{(n-1)}(x^{\mu},z)=A_{(n-1)}(x^{\mu})+
A_{(n-2)}(x^{\mu}) \wedge \dd z \, .
\end{eqnarray}
After reduction on the circle $S^1$, we thus find
\be\label{tab:redform}
 \sqrt{-\widehat{g}} \, \frac{1}{2n!}
\widehat{F}_{(n)}^{2}
\longrightarrow \sqrt{-g} \, \bigg(\frac{1}{2n!}\,e^{-2(n-1) \gamma_{D-1}
\phi} \, F_{(n)}^{2}+ \frac{1}{2(n-1)!}\,e^{2(D-n-1) \gamma_{D-1} \phi} \,
F_{(n-1)}^{2} \bigg),
\ee
where 
\begin{equation}
F_{(n-1)}= \dd A_{(n-2)}, \qquad  F_{(n)} =
\dd A_{(n-1)}-\dd A_{(n-2)}\wedge A_{(1)}.
\end{equation}
\subsection{Reduction of pure gravity on $T^2$ and $\mathrm{SL}(2, \mathbb{R})$}
\label{PERTABsec:red2dim}

It is clear that having established the procedure for performing a Kaluza-Klein reduction from $D$ to $D-1$ dimensions on a circle, the process can be generalized to reduction on a succession of $k$ circles. This is equivalent to a compactification from $D$ to $D-k$ dimensions on an $k$-torus $T^{k}=S^{1}\times \cdots \times S^{1}$. We shall in this section restrict to the reduction of pure gravity in $D$ dimensions on a $2$-torus $T^2$, because many of the features that appear in this analysis are rather general and apply to any kind of toroidal reduction \cite{Pope}.

\subsubsection{The $\mathrm{SL}(2,\mathbb{R})$-symmetry of the reduced theory}

Let us consider the Einstein-Hilbert action in $D$ dimensions:
\begin{eqnarray}
\mathcal S_{D}= \int \dd ^{D}x \, \sqrt{-g_{D}} \, R_{D}\, ,
\end{eqnarray}
where one puts subscripts on the fields to signify that there are $D$-dimensional objects.
The reduction of this action on $S^{1}$ under the ansatz (\ref{tab:ans}) gives (see (\ref{tab:redl}))
\begin{eqnarray}
\mathcal S_{D-1}= \int \dd ^{D-1}x\, \sqrt{-g_{D-1}} \, \bigg(R_{D-1}-
\frac{1}{2}\,
\partial_{\mu}\phi_{1}
\partial^{\mu} \phi_{1}- \frac {1}{4}\, e^{-2(D-2) \gamma_{D-1} \phi_{1}}\,
F_{(2)\,1}^{2}\bigg)\, ,
\end{eqnarray}
where the indices $\mu$ traverse the non-compact dimensions and the inferior indices $_{()}$ of $F_{(p)}$ indicate that it is a $p$-form associated to a $(p-1)$-form potential $A_{(p-1)}$. The other indices are label indices that specify the torus $T^i$ where the fields appeared.

We perform the same procedure and we compactify again another space-like direction on $S^1$. As a result, we obtained a theory reduced on $T^2= S^1 \times S^1$
{\setlength\arraycolsep{2pt}
\be \begin{split}\label{tab:acred}
\mathcal S_{D-2}=\int & \dd ^{D-2}x\, \sqrt{-g_{D-2}}\\
 {} &\bigg(R_{D-2}- \frac{1}{2}
\partial_{\mu}\phi_{1}
\partial^{\mu}\phi_{1} - \frac{1}{2}\,
\partial_{\mu}\phi_{2}
\partial^{\mu}\phi_{2}  - \frac {1}{4}\,\sum_{i=1}^{2} e^{\vec{c_{i}}\cdot \vec{\phi}}
F_{(2)\,i}^{2}- \frac{1}{2} e^{\sqrt{2}\vec{\alpha}\cdot \vec{\phi}}\,
F^{2}_{(1)12}\bigg)\, ,
\end{split} \ee}
with
\be \begin{split}
\vec{c}_{1}&= \Bigg(-\sqrt{\frac{2(D-2)}{(D-3)}},
-\sqrt{\frac{2}{(D-3)(D-4)}}\Bigg)\, ,\\
\vec{c}_{2}&= \Bigg(0,
-\sqrt{\frac{2(D-3)}{(D-4)}}\Bigg)\, ,\\
\vec{\alpha}&= \Bigg(-\sqrt{\frac{(D-2)}{(D-3)}},
\sqrt{\frac{(D-4)}{(D-3)}}\Bigg)\, .
\end{split}\ee
 The fields resulting from this reduction are 
$g_{\mu\nu}, \, A_{(1)i}, \, A_{(0)12}, \,\textrm{and}\ \vec{\phi}= (\phi_{1},\phi_{2})$. The scalar field  $A_{(0)12}$ called \emph{axion} comes from the dimensional reduction of the first Kaluza-Klein vector $A_{(1)1}$.

Let us now look at the scalars in the reduced theory (\ref{tab:acred}) described  by the scalar Lagrangian 
\be \label{tab:a1}
\mathcal{L}_{scalar} = - \frac{1}{2}\,
\partial_{\mu}\phi_{1}\,
\partial^{\mu}\phi_{1} - \frac{1}{2}\,
\partial_{\mu}\phi_{2}\,
\partial^{\mu}\phi_{2}  - \frac{1}{2}\, e^{\sqrt{2}\vec{\alpha}\cdot \vec{\phi}}\,
\partial_{\mu} \chi \,\partial^{\mu} \chi\, ,
\ee
where $\chi=A_{(0)12} $. Things simplify a lot if we rotate the basis for the two dilatons $ \vec{\phi}=(\phi_1, \phi_2)$. If we make the orthogonal transformations to two new dilaton combinations, which we may call $\phi$ and $\varphi$:
\be \begin{split}
\phi&= \frac{1}{2}\,\Bigg(-\sqrt{\frac{2(D-2)}{(D-3)}}\,\phi_{1}+
\sqrt{\frac{2(D-4)}{(D-3)}} \, \phi_{2} \Bigg)\, ,\\
\varphi&= \frac{1}{2}\,\Bigg(-\sqrt{\frac{2(D-4)}{(D-3)}}\,\phi_{1}-
\sqrt{\frac{2(D-2)}{(D-3)}} \, \phi_{2} \Bigg)\, ,
\end{split}\ee
the Lagrangian (\ref{tab:a1}) becomes
\begin{eqnarray}\label{tab:a2}
\mathcal{L}_{scalar} = - \frac{1}{2}\,
\partial_{\mu}\phi \,
\partial^{\mu}\phi- \frac{1}{2}\,
\partial_{\mu}\varphi \,
\partial^{\mu}\varphi  - \frac{1}{2}\, e^{2\phi}\,
\partial_{\mu} \chi \,\partial^{\mu} \chi \, .
\end{eqnarray}
As $\varphi$ is decoupled from the other scalars, we will consider it independently of $\phi$ and $\chi$
to analyse the global symmetries of this system. On the one hand, $\mathcal{L}_{scalar}$ has a global shift symmetry $\varphi\rightarrow \varphi+ k $. This gives an $\mathbb{R}$ factor in the global symmetry group. On the other hand, the part of $\mathcal{L}_{scalar}$ containing $\phi$ and $\chi$ is invariant under the transformations of  $\mathrm{SL}(2,\mathbb{R})$. Indeed, if we define a complex field $\tau= \chi +i\, e^{-\phi}$ the Lagrangian for $\phi$ and $\chi$ can be written as 
\begin{eqnarray}\label{tab:a3}
\mathcal{L}_{(\phi,\chi)}\equiv -\frac{1}{2}\,
\partial_{\mu}\phi \,
\partial^{\mu}\phi  - \frac{1}{2}\, e^{2\phi}\,
\partial_{\mu} \chi \,\partial^{\mu} \chi= -\frac{\partial_{\mu}\tau \,
\partial^{\mu}\bar{\tau}}{2\, \tau_{2}^{2}}\, ,
\end{eqnarray}
where $\tau_{2}$ is the imaginary part of $\tau= \tau_{1}+
i\tau_{2}$. Now it is not hard to see that this Lagrangian is invariant under the transformation:
\begin{eqnarray}\label{tab:transfot}
\tau \longrightarrow \frac{a\tau + b}{c\tau+d} \qquad \textrm{with} \ ad-bc=1\, ,
\end{eqnarray}
where $a,b,c$ and  $d$ are constants $\in \mathbb{R}$. This transformation can be written under matrix form:
\begin{eqnarray}
g= \left( \begin{array}{cc}
a & b  \\
c & d \end{array} \right)\, ,
\end{eqnarray}
with the condition $\det g= 1$. What we have here is real $2\times 2$ matrices of unit determinant. They form the group $\mathrm{SL}(2,\mathbb{R})$. This symmetry acts nonlinearly on the complex field $\tau$ (see (\ref{tab:transfot})). To conclude, we have seen that the scalar Lagrangian (\ref{tab:a2}) as in total an $\mathrm{SL}(2,\mathbb{R}) \times \mathbb{R}$ global symmetry. This makes a $\mathrm{GL}(2,\mathbb{R})$ symmetry.

We can show that  the non-scalar part of the reduced theory (\ref{tab:acred}) also shares the same symmetry that the scalar part. Note that while the scalars transform nonlinearly under $\mathrm{SL}(2,\mathbb{R})$, the two gauge potentials  $A_{(1)1}$ and $ A_{(1)2}$  transform linearly, as a doublet. Under the $\mathbb R$-factor of $\mathrm{GL}(2,\mbb R)$, the potentials will transform by appropriate constant scaling factors.     Thus the global symmetry of the lower-dimensional Lagrangian is already established by looking just at the scalar fields and their symmetry transformations.

\subsubsection{Scalar Coset Lagrangian $\mathcal{L}_{\mathrm{SL}(2, \mbb R)/ \mathrm{O}(2)}$}\label{PERTABsec:scalcola}
To understand the structure of the global symmetry better, we need to study the nature of the scalar Lagrangian that arise from the dimensional reduction on the torus $T^2$. It leads us into the subject of \emph{nonlinear $\sigma$-models} and \emph{coset spaces}. The example of $\mathrm{SL}(2,\mathbb{R})$ exhibits many of the general features that one encounters in nonlinear $\sigma$-models, while having the merit of being rather simple and easy to calculate explicitly.

The associated Lie algebra of $\mathrm{SL}(2,\mathbb{R})$ is $\mathfrak{sl}(2,\mathbb{R}) \simeq A_1$, the smallest finite  dimensional simple Lie algebra. This algebra is $3$-dimensional and following  Section \ref{sec:def}  we take as a basis the set of the Chevalley generators $\{e, f, h\}$ which take the matrix forms
\be
e= \left(\begin{array}{cc}0 & 1 \\0 & 0\end{array}\right), \ f= \left(\begin{array}{cc}0 & 0 \\ 1 & 0\end{array}\right), \ h = \left(\begin{array}{cc}1 & 0 \\0& -1\end{array}\right)\, .
\ee
These generators satisfy the following commutation relations
\be
[h,e]=2 e, \quad [h,f]=-2f, \quad [e,f]=h.
\ee

Consider now the exponentiation  of the generators $h$ and $e$, and define the  \emph{coset representative}:
\be \begin{split}\label{tab:cos22}
\mathcal{V}(x^{\mu})&= e^{\frac{1}{2} \phi (x^{\mu}) \, h}\, e^{\chi (x^{\mu}) \, e}\, ,\\
&=\left( \begin{array}{cc}
e^{\frac{1}{2}\phi (x^{\mu})} &   \chi(x^{\mu}) \, e^{\frac{1}{2}\phi(x^{\mu})} \\
0 & e^{-\frac{1}{2}\phi(x^{\mu})} \end{array} \right)\, , 
\end{split} \ee
where $\phi$ and $\chi$ are fields depending on the coordinates of a $(D-1)$-dimensional spacetime.
This matrix is in the upper-triangular form and we said that $\mathcal{V}$ is in the upper-triangular gauge or in the \emph{Borel gauge}  because $\mathcal{V}$ is constructed by exponentiation of the Borel subalgebra of $\mathfrak{sl}(2,\mathbb{R})$ defined in  \eqref{eqn:borelsubalgebra1}. If we define
\begin{eqnarray}\label{tab:a6}
\mathcal{M}(x^{\mu})=\mathcal{V}^{T}(x^{\mu})\mathcal{V}(x^{\mu})\, ,
\end{eqnarray}
we can write a Lagrangian as 
\be \begin{split}\label{tab:a4}
\mathcal{L}&= \frac{1}{4}\, tr ( \,
\partial_{\mu}\mathcal{M}^{-1}\,
\partial^{\mu}\mathcal{M}\, )\, , \\
&=- \frac{1}{2} \, \partial_{\mu}\phi \,  \partial^{\mu}\phi
-\frac{1}{2} \, e^{2\phi} \partial_{\mu}\chi \,  \partial^{\mu}\chi\, .
\end{split} \ee
This is exactly the $\mathrm{SL}(2, \mathbb{R})$-invariant scalar Lagrangian $\mathcal{L}_{(\phi,\chi)}$ (\ref{tab:a3}). The coset representative allowed to find a pleasant form for building the scalar Lagrangian using the Lie algebra $\mathfrak{sl}(2,\mathbb{R})$. The advantage now is that we have a very nice  way to see why it is $\mathrm{SL}(2, \mathbb{R})$ invariant. Indeed if we consider a transformation $g \in \mathrm{SL}(2, \mathbb{R})$, then 
\begin{eqnarray}
\mathcal{V} (x^{\mu}) \longrightarrow \mathcal{V}''(x^{\mu})= \mathcal{V}(x^{\mu})g\, ,
\end{eqnarray}thus
\begin{eqnarray}
\mathcal{M}(x^{\mu}) \longrightarrow g^{T}\, \mathcal{M }(x^{\mu})\, g\, ,
\end{eqnarray}
which manifestly leaves the Lagrangian (\ref{tab:a4})  invariant (using the cyclic invariance of the trace). However, we did something improper because the transformation $\mathcal{V} \rightarrow
\mathcal{V}''$ do not leave $\mathcal{V}''$ in the upper-triangular form that the original matrix $\mathcal{V}$ given in (\ref{tab:cos22}). It is thus necessary to make a compensating local transformation $k$ (local in the sense that this matrix depends not only on the constant $\mathrm{SL}(2, \mathbb{R})$ parameters but also on the fields $\phi$ and $\chi$)  that does the job of restoring  $\mathcal{V}''$ to the upper-triangular gauge. The matrix $k(x^{\mu})$ acts on $\mathcal{V}''$ from the left, at the same time as we multiply by $g$ from the right. We define, thus a transformed matrix $\mathcal{V}'$ by 
\begin{eqnarray}\label{tab:a7}
\mathcal{V}'(x^{\mu})= k(x^{\mu})\, \mathcal{V}''(x^{\mu})= k(x^{\mu}) \mathcal{V}(x^{\mu}) g \, ,
\end{eqnarray}
where the matrix $k$ found is orthogonal and belongs to $\mathrm{O}(2)$ the maximal compact subgroup of $\mathrm{SL}(2, \mbb R)$. We can now interpret the action of  $\mathrm{SL}(2, \mathbb{R})$ in terms of transformations on the fields $\phi$ and $\chi$. We can again easily check that the transformation $\mathcal{V}'$ (\ref{tab:a7})  leaves the Lagrangian (\ref{tab:a4}) invariant.

At a given spacetime point, 
we can use the $\mathrm{SL}(2, \mathbb{R})$ transformation to get from any pair of values for $\phi$ and $\chi$, any other pair of values. This means that $\mathrm{SL}(2, \mathbb{R})$  acts transitively on the \emph{scalar manifold} which is the manifold where the fields $\phi$ and $\chi$ take their values. But we must make a compensation transformation $\mathrm{O}(2)$  to make sure that we stay within our original upper-triangular form. Thus we may specify points in the scalar manifold by the \emph{coset } $\mathrm{SL}(2, \mathbb{R})/\mathrm{O}(2)$, consisting of $\mathrm{SL}(2, \mathbb{R})$ motions modulo the appropriate $\mathrm{O}(2)$ compensators. 

\subsection{Compactification on $T^{D-3}$ } \label{subsec:compactificationd3th}
Having established the procedure for performing a Kaluza-Klein reduction from $D$ to $D-1$ dimensions on the circle $S^1$, it is clear that the process can be repeated for a succession of circles. Thus we will first consider a reduction from $D$ to $D-k$ dimensions on a $k$-torus $T^k$ \cite{Lambert:2001gk}. At each $i$'th reduction step, one generates  a Kaluza-Klein vector potential $A_{(1)i}$, and a dilaton $\phi_i$ from the reduction of the metric (see \eqref{tab:redl}). In addition, from $p$-form potential  present in $D-i$ dimensions, one generates a $p$-form  and a $p-1$ form potential (see \eqref{tab:redform}). As a result, one obtains rapidly a proliferating number of fields by compactification on $T^k$. 

In addition, once we reach $n+1$ dimensions, we can dualise the $n$-form $F_{(n)}$ (present in the Lagrangian (\ref{tab:actionmax})) into a scalar. We suppose that the action in $n+1$ dimensions has the form
\begin{eqnarray}\label{tab:act}
\mathcal S= \int \, d^{n+1}x \,\frac{\sqrt{-g}}{2n!}\, e^{\vec{\alpha}\cdot
\vec{\phi}}\, F^{2}_{(n)}\, .
\end{eqnarray}
We introduce in the action (\ref{tab:act}) the field $\varphi$ as a Lagrange multiplier term
\begin{eqnarray}\label{tab:duall}
\mathcal S= \int\,d^{n+1}x \,\frac{\sqrt{-g}}{2n!}\,e^{\vec{\alpha}\cdot
\vec{\phi}}\, F^{2}_{(n)}+ \frac{\sqrt{-g}}{n!}\, \varphi \,
\partial_{\mu}F_{\nu_{1}\ldots \nu_{n}}\, \epsilon^{\,\mu\nu_{1}\ldots \nu_{n}}\, .
\end{eqnarray}
Variation with respect to $\varphi$ simply enforces the Bianchi identity on $F_{(n)}$. However we can also integrate by parts so that there are no derivatives acting on $F_{(n)}$. Eliminating $F_{(n)}$ by its algebraic equation of motion 
\begin{eqnarray}
F_{\nu_{1}\ldots\nu_{n}}=
e^{-\vec{\alpha}\cdot\vec{\phi}}\,\epsilon_{\,\mu\nu_{1}\ldots
\nu_{n}}\, \partial^{\mu}\varphi\, ,
\end{eqnarray}
and substituting back into (\ref{tab:duall}) leads to the equivalent scalar action 
\begin{eqnarray}\label{tab:dualfo}
\mathcal S= \int\,d^{n+1}x \,\frac{\sqrt{-g}}{2}\,e^{-\vec{\alpha}\cdot
\vec{\phi}}\,\partial_{\mu}\varphi \,\partial^{\mu}\varphi \, .
\end{eqnarray}

In conclusion, after the reduction of (\ref{tab:actionmax}) on the torus $T^k$, several scalar fields appears:
\begin{itemize}

\item The dilatons $\Phi^u$, $(u=1, \ldots , q)$ present initially in the $D$ dimensional theory,
\item $k$ scalars $\phi_i$, $(i=1, \ldots, k)$ obtained from the reduction of the metric,
\item scalars from the reduction of the $2$-form $F_{\mu \nu}^k$ (associated to the potential $A_{\mu}^{k}$) generated at each step  by the dimensional reduction of gravity from $D-k+1$ to $D-k$ dimensions,
\item scalars resulting from the potentials associated to the $F_{(n)}$ present initially in the unreduced Lagrangian \eqref{tab:actionmax}. These scalars appear in the compactification on the torus $T^{n-1}$ for $k \geq n-1$,
\item scalars obtained by dualization of all $n$-form $F_{(n)}$ when $k=D-n-1$. In particular, when $D=3$ (compactification on $T^{D-3}$), all the $2$-form $F_{\mu \nu}^k$ ($k=1, \ldots, D-3$), resulting from the reduction of the metric, can be dualize into scalars.

\end{itemize}
The three last kind of scalars will be denoted by $\chi_{\alpha}$.\\

Performing the reduction of the Lagrangian (\ref{tab:actionmax}) to $D=3$, all the physical degrees of freedom which lie in the matter fields can be dualised into scalars. Therefore the reduced  Lagrangian in $D=3$ takes the form: 
\be\label{tab:lagthree}
 \sqrt{-g} \bigg(R - \frac{1}{2}
\partial_{\mu}\vec{\varphi} \cdot \partial^{\mu}\vec{\varphi} -
\frac{1}{2}\sum_{\alpha} e^{\sqrt{2}\,\vec{\alpha}\cdot
\vec{\varphi}}\, (\partial_{\mu}\chi_{\alpha}  + \tfrac{1}{2!} \chi_{\beta} \partial_{\mu}\chi_{\beta '}+...)
(\partial^{\mu}\chi_{\alpha}  + \tfrac{1}{2!} \chi_{\beta} \partial^{\mu}\chi_{\beta '}+...) \bigg),
\ee
where $\vec{\varphi}= \{\Phi^{u},\phi_{i}\}=
(\Phi^{1},\ldots,\Phi^{q};\phi_{1},\ldots,\phi_{D-3})$ and where the vectors $\vec{\alpha}$ are constant vectors with $(q+D-3)$ components\footnote{The vectors $\vec{\beta}$ and $\vec{\beta}'$ in \eqref{tab:lagthree} are such that $\vec{\beta}+ \vec{\beta}' = \vec{\alpha}$. }, characterised the scalars $\chi_{\alpha}$.

One expects that the symmetry of this reduced Lagrangian will be $\mathrm{GL}(D-3, \mathbb{R})$ which is the symmetry group of the moduli space of the $(D-3)$-torus. But for some very specific theories, as eleven-dimensional supergravity, this symmetry is much larger. In fact under certain conditions 
the scalar part of the reduced Lagrangian $\mathcal{L}_3$ \eqref{tab:lagthree} can be identified to a coset Lagrangian $\mathcal{L}_{\g / \mathrm{K}}$ invariant under the transformation $\g/ \mathrm{K}$ where $\g$ is a simple Lie group and $\mathrm{K}$ is the maximal compact subgroup of $\g$. Note that if the time coordinate is compactified, the reduced three-dimensional Lagrangian can be identified to a coset Lagrangian $\mathcal{\mathrm G/ \mathrm K^*}$ where $\mathrm K^*$ is not the maximally compact subgroup of $\mathrm G$. For instance, we will see in Chapter \ref{chap:finitesymsu215} that the reduction of pure $\mathcal N=2$ supergravity on a time-like circle exhibits a coset symmetry $\su/\mathrm{SL}(2, \mbb R) \times \mathrm U(1)$ where the subgroup $\mathrm{SL}(2, \mbb R) \times \mathrm U(1)$ is non-compact .\\

As all the simple Lie group are classified (see Figure \ref{fig:liealgebras}), we can ask an interesting question: What are the theories containing gravity suitably coupled to forms and dilatons which may exhibit upon dimensional reduction down to $3$ dimensions a coset structure $\mathrm{G}/\mathrm K$? This question leads to the procedure of \emph{oxidation} \cite{Julia:1980gr,Cremmer:1999du,Keurentjes:2002xc,Keurentjes:2002rc,Keurentjes:2002vx} which is roughly speaking the inverse of the dimensional reduction. In this context, we call \emph{maximally oxidised theory} such a Lagrangian theory \eqref{tab:actionmax} defined in the highest possible space-time dimension $D$ namely which is itself not obtained by dimensional reduction. In this case, the vectors $\vec{\alpha}$ present in $\mathcal{L}_3$ \eqref{tab:lagthree} are identified to the positive roots \footnote{If the reduction to $3$ dimensions exhibits a non-split real group $\mathrm G$, the vectors $\vec{\alpha}$ correspond to the positive restricted roots of the non-split real algebra $\mathfrak{g}_r$ (see Section \ref{subsec:restiwa}). If the restricted root is degenerate, there will be  more than one scalar field $\chi_{\alpha}$ associated to it.} of the simple Lie algebra $\mf g$ (in its split real form) associated to the Lie group $\mathrm G$. The scalar products of the simple roots allows the construction of the Dynkin diagram of $\ag$. The vertices of the horizontal line of the Dynkin diagram of $\ag$ (see Figure \ref{fig:liealgebras}) define \emph{the gravity line}. It represents simple roots related to scalar fields $\chi_{\alpha}$ coming from the reduction of the metric $g_{MN}$ (except for the algebra $C_n$). The vertices not belonging to the gravity line are called \emph{electric} or \emph{magnetic root} depending if the  simple root is related to a scalar field $\chi_{\alpha}$ resulting from the reduction of an $n$-form $F_{(n)}$ or from its dualization into a scalar.

The maximally oxidised theories have been listed for each coset space $\mathrm G/\mathrm K$ in \cite{Cremmer:1999du} when the associated Lie algebra $\mf g$ is a split real form and in \cite{Keurentjes:2002rc} when $\mf g$ is a non-split real form. For instance, we get that the reduced Lagrangian in $3$ dimensions is invariant under transformations of $\mathrm{SL}
(D-2, \mbb R)/ \mathrm{O}(D-2)$ for the theory containing only gravity and  under transformations  of $\e/\mathrm{SO}
(16)$ for the eleven-dimensional supergravity. Note that all maximally oxidised theories do not have necessarly supersymmetry extension such as the low energy effective action of the bosonic string in $26$ dimensions which exhibits the coset $\mathrm{O}(24,24)/(\mathrm O(24)\times \mathrm O(24))$ upon dimensional reduction down to three dimensions.
\subsubsection{The Coset Lagrangian $\mathcal{L}_{ \g/ \mathrm{K}}$ in $3$ dimensions}
\label{PERTABsec:cosetlagg}

We will now show how the scalar part of the reduced Lagrangian $\mathcal{L}_3$ (\ref{tab:lagthree}) can be identified to a coset Lagrangian $\mathcal{L}_{\g/\mathrm{K}}$ where $\g$ is a simple Lie group and $\mathrm{K}$ is the maximal compact subgroup of $\mathrm{G}$. In the case where the subgroup of $\mathrm G$ is not compact (namely in the case of time-like circle reduction) the construction of the coset Lagrangian is explained in Section \ref{app:SigmaModel}. 

First, let us take a closer look at the subgroup $\mathrm{K}$. In the case of the compactification of gravity on $T^2$, we have found that $\mathrm{K}= \mathrm{O}(2)$. However, it is clear that performing more compactificactions on $T^k$ it would become increasingly complicated to construct the compensator $k$. There is fortunately a general theorem in the theory of Lie algebras that claims that there exists an element of the maximal compact subgroup $\mathrm{K}$ of $\g$ which does the job of compensation. This is the Iwasawa decomposition, which was introduced in Section \ref{subsec:chevalleyinvo2}. Here we introduce the Iwasawa decomposition at the group level. We then have the following statement:  \emph{every element  $g$ in the Lie group $\g$, associated to a Lie algebra $\mathfrak{g}$, can be uniquely expressed as the following product:}
\begin{eqnarray}\label{Iwasawadecomposigroupl}
g= g_{\mathfrak{k}} \, g_{\mathfrak{h}}\, g_{\mathfrak{n}_{
+}}\, ,
\end{eqnarray}
\emph{where $g_{\mathfrak{k}}$ belongs to the maximal compact subgroup $\mathrm{K}$, with Lie algebra $\mathfrak{k}\subset \mathfrak{g}$, $g_{\mathfrak{h}}$ belongs to the Cartan torus $\mathfrak{h}$ of $\mathfrak{g}$ and $g_{\mathfrak{n}^{
+}}$ belongs to the nilpotent part of $\mathrm{G}$.} Our coset representative $\mathcal{V}$ will be constructed by exponentiating the Cartan generators and the full set of positive-root generators (see (\ref{tab:cos22})). Thus our coset representative is written as $\mathcal{V}= g_{\mathfrak{h}}\, g_{\mathfrak{n}_{+}}$. Now, if we act by right-multiplication with a general group element $g$ in $\g$, then $\mathcal{V} g$ is some element of the group $\g$. Now, invoking the Iwasawa decomposition, we must be able to write the group element $\mathcal{V} g$ in the form $g_{\mathfrak{k}} \mathcal{V}'$ where $ \mathcal{V}'$ itself is of the form $g_{\mathfrak{h'}}\, g_{\mathfrak{n}^{,}_{+}}$. This assures that there exists a way of pulling out an element $k$ of the maximal subgroup $\mathrm{K}$ of $\g$ on the left-hand side, such that $\mathcal{V} g= k \mathcal{V}'$. 

To construct the maximal subgroup $\mathrm K$ of $\g$, we use the Chevalley involution $\omega$ 
defined in (\ref{eqn:involutionchev}) which has the effect of reversing the sign of every non-compact generator in the algebra $\mathfrak{g}$, while leaving the sign of every compact generator unchanged.\\

Let us now generalize the construction of the coset Lagrangian $\mathcal{L}_{\mathrm{SL}(2, \mbb R)/ \mathrm O(2)}$ done in Section \ref{PERTABsec:red2dim} for the coset $\mathrm G/\mathrm K$. The coset representative $\mathcal{V}$ is build by exponentiating the Borel subalgebra $\mathfrak{b}$ of $\mathfrak{g}$ (see Section \ref{subsec:chevalleyinvo2})
\begin{eqnarray}
\mathcal{V}(x^{\mu})=e^{\frac{1}{\sqrt{2}}\,\vec{\varphi}(x^{\mu})\cdot \vec{h}}\,
\sum_{{\alpha}>0}\, e^{\chi_{{\alpha}}(x^{\mu})\, 
e_{{\alpha}}}\, ,
\end{eqnarray}
where $\vec{h}=\{R_{u},h_{i}\}$. The generators $R_{u}$, ($u=1,\ldots,q$) are generators associated to the dilatons $\Phi^{u}$ present in the non-reduced theory in $D$ dimensions (see \eqref{tab:actionmax}). The  $h_{i}$ ($i=1,\ldots,D-3$)
and $R_u$ form the Cartan subalgebra $\mathfrak{h}$ of $\mathfrak{g}$ (defined in Section \ref{subsec:chevalley}) and the generator $e_{{\alpha}}$ is the generator associated to the positive root $\alpha$  (see Section \ref{subsec:rootsystem1}). We can identify the scalar part of the reduced Lagrangian (\ref{tab:lagthree}) to a coset Lagrangian
\be \begin{split} \label{eqn:lag3dgsurk25}
\mathcal{L}_{\mathrm{G} / \mathrm{K}} &= \frac{1}{4}\, \big(\,
\partial_{\mu} \mathcal{M}^{-1}\big| \partial^{\mu}
\mathcal{M}\,\big)\, ,\\
&= \frac{1}{4}\, tr ( \,
\partial_{\mu} \mathcal{M}^{-1}\, \partial^{\mu}
\mathcal{M}\,)\, ,
\end{split}\ee
where $\mathcal{M}= \mathcal{V}^{\sharp}\, \mathcal{V}$ and $(\cdot |\cdot) $ is the invariant bilinear form on the Lie algebra (see Section \ref{subsec:bilinear}) that corresponds to the trace in the case of finite-dimensional simple Lie algebra. We define the \emph{generalised transpose}  $x^{\sharp}$ on a generator $x$ of $\mf g$ by
\begin{eqnarray}
x^{\sharp} \equiv - \omega\,(x)\, ,
\end{eqnarray}
where $\omega$ is the Chevalley involution defined in \eqref{eqn:involutionchev}.
In the simple cases, corresponding to orthogonal subgroups $\mathrm{K}$ (like $\mathrm{O}(2)$), $\sharp$ coincides with the transpose. If we normalise in the adjoint representation the Cartan and the positive root generators so that 
\begin{eqnarray}\label{tab:399}
tr ( \, h_{i}\,h_{j} \, ) = \delta_{ij}\, , \qquad tr (
\,e_{{\alpha}}\,e_{{\beta}} \, ) =0 \, ,\qquad tr (
\,e_{{\alpha}}\,f_{{\beta}} \, ) = \delta _{\alpha\beta}\, ,
\end{eqnarray}
then one can show that $\mathcal{L}_{\mathrm{G} / \mathrm{K}}$ is precisely the scalar part of (\ref{tab:lagthree}). Thus it follows that if the vectors $\vec{\alpha}$ obtained from the compactification can be identified with positive roots of an algebra $\mathfrak{g}$, then the action when dimensionally reduced to three dimensions has a symmetry $\g/\mathrm{K}$. Indeed, we can easily check that $\mathcal{L}_{\mathrm{G} / \mathrm{K}}$ is invariant under the global transformation $g \in \g$ and under the local transformations $k \in \mathrm{K}$. 
\begin{eqnarray}
\mathcal{V}(x^{\mu}) \longrightarrow k(x^{\mu})\, \mathcal{V}(x^{\mu})\, g\, .
\end{eqnarray}
\subsection{The conjecture in the infinite-dimensional case}

In the previous section, we have seen the dimensional reduction to $3$ dimensions of  $D$-dimensional gravitational theories  described by the Lagrangian \eqref{tab:actionmax} . That leads to a reduced Lagrangian \eqref{tab:lagthree} containing only scalars coupled to gravity in $3$ dimensions. It is interesting to consider the dimensional reduction beyond $3$ dimensions\footnote{The process using previously to reduce the different theories must be renounced because the parameter $ \gamma_{D-1} = \frac {1}{\sqrt{2(D-2)(D-3)}}$ (see \eqref{tab:alphbet}) appearing through the reduction to $d=D-1$ dimensions in all the components of the root of algebra $\mathfrak{g}$, is badly defined for $d=2$ and $1$. } but it is not any more possible to reduce these theories on the torus $T^k$ for $k > D-3$. Indeed later compactification will not give additional scalar fields in the coset Lagrangian $\mathcal{L}_{\g/\mathrm{K}}$. Nevertheless, the symmetries of the systems must increase because one expects that all the theories posses at least a global symmetry $\mathrm{GL}(k,\mathbb{R})$ when they are reduced on $T^k$. Therefore it was conjectured by Bernard Julia that the symmetry group obtained in $3$ dimensions is extended when we reduce further dimensions \cite{Julia:1980gr,Julia:1982gx}. These additional reductions would give infinite-dimensional Kac-Moody symmetries obtained by adding nodes in the Dynkin diagram of the corresponding finite algebra (see Figure \ref{fig:dynkinkmalgebras}).

It has been shown that the reduced theory in $2$ dimensions are connected to a infinite-dimensional symmetry $\gp$ (affine extension of $\g$) \cite{Nicolai:1987vy} obtained by adding one node to the Dynkin diagram of $\mf g$ where $\mf g$ is a finite Lie algebra (see Section \ref{sec:classkm}). The first example of affine 'hidden' symmetries is the affine symmetry group $\mathrm{SL}(2,\mathbb{R})^{+}$, known as the Geroch group \cite{Geroch:1970nt,Geroch:1972yt,Breitenlohner:1986um}. It comes from 
the reduction of pure four-dimensional gravity on a two-torus $T^2$. The combination of the Ehlers coset $\mathrm{SL}(2, \mbb R)/ \mathrm{SO}(2)$, which appears in the reduction on a circle $S^1$,  and the Matzner-Mizner coset $\mathrm{SL}(2, \mbb R)/ \mathrm{SO}(1,1)$ appearing by further reduction on a  time-like Killing vector,  leads to an infinite-dimensional symmetry known as the Geroch group acting on solutions of Einstein equations with two commuting Killing vectors (axisymmetric stationary solutions).  These results provides  the integrability of these theories in the reduction of two dimensions \cite{Nicolai:1987kz,Nicolai:1988jb}.


Motivated by the dimensional reduction, it has been argued that  the  Kac-Moody algebra $\mathfrak{g}^{++}$ (overextension of $\mathfrak{g} $) defined in Section \ref{sec:classkm}  can play a role in the compactification to $1$ dimension \cite{Julia:1982gx,Mizoguchi:1997si}. Finally, when all the dimensions are compactified, it is obvious to extend the algebra $\mathfrak{g}^{++}$ to $\mathfrak{g}^{+++}$ (triple extension of $\mathfrak{g} $) defined in Section \ref{sec:classkm}   by adding a third vertex. Such $\gppp$ symmetries were first conjectured in the aforementioned cases 
\cite{West:2000ga,West:2001as, Lambert:2001gk} and the extension to all $\gppp$ was proposed in \cite{Englert:2003zs}. So this construction motivates the fact that the eleven-dimensional supergravity could have the symmetry $\mathrm{G}^{+++}=\mathZ{E}_{8(8)}^{+++}=\eppp$ and that the pure gravity in $D$ dimensions could have the symmetry $\mathrm{SL}(D-2, \mbb R)^{+++}$.
 
 More generally all simple maximally non-compact Lie groups $\g$ could be generated from the reduction down to $3$ dimensions of suitably chosen actions \cite{Cremmer:1999du} and it was conjectured that these actions possess the very-extended Kac-Moody symmetries $\gppp$ \cite{Englert:2003zs}. $\mathfrak{g}^{+++}$ algebras are defined by the Dynkin diagrams depicted in Figure \ref{fig:dynkinkmalgebras}, obtained from those of $\mathfrak{g}$ by adding three nodes. 

 \setcounter{equation}{0}
\section{Non-linear $\sigma$-models over cosets spaces}  \label{app:SigmaModel}
In Section \ref{PERTABsec:cosetlagg}, we have constructed a coset Lagrangian $\mathcal{L}_{\mathrm{G}/ \mathrm{K}}$ where $\g$ is a simple Lie group and $\mathrm{K}$ its maximal compact subgroup. In what follows, we will be interested by more general non-linear $\sigma$-models over coset spaces $\bar{\mathrm G}/\bar{\mathrm K} $ where $\bar{\mathrm{G}}$ is not reduced to be a finite-dimensional simple Lie group but it can be an infinite-dimensional Kac-Moody group and the subgroup $\bar{\mathrm{K}}$ is not especially compact. The appearance of such subgroups $\bar{\mathrm{K}}$ will result of the reduction of the time-like coordinate. The associate algebra $\bar{\mf k}$ will be defined  as the subalgebra of $\bar{\mf g}$ invariant under an involution called \emph{temporal involution} that we will describe in details in Section \ref{sec:tempinvolutintro}.

\subsection{The coset lagrangian $\mathcal{L}_{\bar{\mathrm G}/\bar{\mathrm K} }$}

Here we recall how to define a non-linear $\sigma$-model over a coset space $\bar{\mathrm G}/\bar{\mathrm K}$. Let $\bar{ \mathrm G}$ be a connected Kac-Moody group. Consider its real  algebra $\bar{\mathfrak{g}}$ and an involution $\imath: \bar{\mathfrak{g}} \rightarrow \bar{\mathfrak{g}}$. Using the two eigenspaces of $\imath$ we can write the Lie algebra as the direct sum
\be \label{eqn:decopmpolag}
\bar{\mathfrak{g}} = \bar{\mathfrak{k} }\oplus \bar{\mathfrak{p}}, 
\ee  
where $\bar{\mathfrak{k}}$ is the span of the generators fixed under $\imath$ and $\bar{\mf p}$ is the complement of $\bar{\mathfrak{k}}$ which is pointwise anti-invariant under $\imath$.  In the case when the involution $\imath$ corresponds to the Chevalley involution $\omega$, the decomposition \eqref{eqn:decopmpolag} is called the Cartan decomposition (see \eqref{eqn:cartandecomp124}). The vector space $\bar{\mathfrak{k}}$ does  constitute a subalgebra. Let $\bar{\mathrm K}$ be the closed  group corresponding to the subalgebra $\bar{\mathfrak{k}}$. We can now consider $\bar{\mathrm K}$ as a topological subspace of $\bar{ \mathrm G}$ and define the coset space $\bar{ \mathrm G}/\bar{\mathrm K}$ as the set of left cosets $\bar{\mathrm K}g$ where $g \in \bar{\mf g}$. As $\bar{\mathrm K}$ is closed $\bar{ \mathrm G}/\bar{\mathrm K}$ can be endowed with a smooth manifold structure. 

Consider a $k$-dimensional manifold $W$ with metric $h$ and let $x^{\mu}\ (\mu=1, \ldots, k)$ be coordinates on W. 
We can now define a $\sigma$-model for smooth maps $\mathcal{V} : W \rightarrow \bar{ \mathrm G}/\bar{\mathrm K}$, such that $\mathcal{V} : p \mapsto \bar{\mathrm K}\mathcal{V}(p)$. Locally $\mathcal{V}$ can be described, using the exponential map, by a map $v: W \rightarrow \mathfrak{p}$. Let $k : W \rightarrow \bar{\mathrm K}$ be a smooth map, $\mathcal{V}$ and $k\mathcal{V}$ define hence the same map into the coset. We will call such maps $k$ gauge transformations for reasons to be clear below.

 To describe a $\sigma$-model on the coset space $\bar{ \mathrm G}/\bar{\mathrm K}$, one can choose a map $\mathcal{V} : W \rightarrow \bar{ \mathrm G}/\bar{\mathrm K}$ in a fixed gauge, that transforms under global transformations $g \in \bar{\mathrm G}$ as 
\be \label{eqn:tsfonu12}
\mathcal{V}(x^{\mu}) \longrightarrow k(x^{\mu}) \mathcal{V}(x^{\mu}) g,
\ee
where $k(x^{\mu}) \in \bar{\mathrm K}$ is a local compensating transformation requiered to restore the choosen gauge for the coset representative $\mathcal{V}$. A standard gauge-choice is the Borel gauge which relies on the global Iwasawa decomposition \eqref{Iwasawadecomposigroupl}. This gauge choice is always admissible when the subgroup $\bar{\mathrm K} = \mathrm K$ corresponds to the maximal compact subgroup of $\bar{\mathrm G}$ (see the discussion in Section \ref{PERTABsec:cosetlagg} ). But where the subgroup $\bar{\mathrm K}$ is non-compact, we are outside the scope of the global Iwasawa decomposition and the Borel gauge will be not always a good choice \cite{Keurentjes:2005jw}. In this dissertation, we still write maps $\mathcal{V}(x^{\mu})$ in the Borel gauge whatever the form of the subgroup $\bar{\mathrm K}$ of $\bar{\mathrm G}$. But we are to be careful on certain aspects because in the case where $\bar{\mathrm{K}}$ is non-compact, not every element in the coset can be written by exponentiating the Borel generators and in this case the coset representative will describe only a proper subset of $\bar{\mathrm G}/ \bar{\mathrm K}$. Therefore, constructing our coordinate system on the coset $\bar{\mathrm G}/ \bar{\mathrm K}$ via the Borel gauge, we will only consider the subspace of $\bar{\mathrm G}$ where the elements are decomposable in the Iwasawa form.\\

 Our main three criteria on the Lagrangian $\mathcal{L}_{\sigma}$ of the $\sigma$ model is that it should be invariant under $\bar{ \mathrm G}$ acting globally on $\mathcal{V}$ from the right, under gauge transformations and such that its equations of motions should be second order in derivatives of $\mathcal{V}$. It is hence natural to define the action in terms of the $\bar{\mathfrak{g}}$-valued Maurer-Cartan form  on $\bar{ \mathrm G}$ restricted to $\bar{ \mathrm G}/\bar{\mathrm K}$:
\be
\mathrm{d}\mathcal{V} \mathcal{V}^{-1} = \dd x^{\mu}\,  \partial_{\mu} \mathcal V(x) \mathcal V(x)^{-1}.
\ee
According to the decomposition \eqref{eqn:decopmpolag} we can split the Maurer-Cartan form as
\be 
\mathrm{d}\mathcal{V} \mathcal{V}^{-1}= \mathcal{P}+\mathcal{Q}\,,
\ee
where 
$\mathcal{P}$ is the projection of $\mathrm{d}\mathcal{V} \mathcal{V}^{-1}$ into $\bar{\mf p}$ and $\mathcal{Q}$ is the orthogonal projection
\be \begin{split}\label{eqn:pandq} 
\mathcal{P} &=  \frac{1}{2} (\mathrm{d}\mathcal{V}\mathcal{V}^{-1} - \imath (\mathrm{d}\mathcal{V}\mathcal{V}^{-1})) \quad \in \ \bar{\mf p}, \\
\mathcal{Q} &=  \frac{1}{2} (\mathrm{d}\mathcal{V}\mathcal{V}^{-1} + \imath (\mathrm{d}\mathcal{V}\mathcal{V}^{-1}))  \quad  \in \ \bar{\mf k}.
\end{split}\ee
Let us recall the Baker-Hausdorff formulas 
\be\begin{split}
\dd e^{X}\, e^{-X}&= \dd X+ \frac{1}{2!} [X, \dd X]+ \frac{1}{3!}[X,[X,\dd X]] + \ldots ,\\
e^X\, Y\, e^{-X}&= Y + [X, Y]+ \frac{1}{2!}[X,[X,Y]] +\ldots,
\end{split}
\ee
which will be very useful in the computation of the terms of the Maurer-Cartan form \eqref{eqn:pandq}.
We therefore write our Lagrangian as \footnote{The construction used to describe the coset Lagrangian $\mathcal{L}_{\mathrm G/ \mathrm K}$ \eqref{eqn:lag3dgsurk25} is another way to describe the invariant Lagrangian and is of course equivalent to this one.}
\begin{equation}
\label{eqn:SigmaModelAction}
\mathcal{L}_{\sigma} = \sqrt{|h|} h^{\mu \nu} (\mathcal{P}_{\mu} | \mathcal{P}_{\nu}) ,
\end{equation}
where $( \cdot | \cdot )$ is the invariant bilinear form of $\bar{\mathfrak{g}}$ (defined in Section \ref{subsec:bilinear})  and $h = \mathrm{det}\ h_{\mu \nu}$. This Lagrangian is manifestly invariant under the right action of $\bar{\mathrm G}$ on $\mathcal{V}$, and also under gauge transformations acting from the left on $\mathcal{V}$. From \eqref{eqn:tsfonu12}, it follows that under gauge transformation $k$ $\mathcal{P}$ transforms covariantly
\be
\mathcal{P} \longrightarrow k \mathcal{P} k^{-1},
\ee
and the $\bar{\mf k}$-valued field $\mathcal{Q}$ transforms  as a connection
\be\label{eqn:tsfoqconnection}
\mathcal{Q} \longrightarrow k \mathcal{Q} k^{-1}+ \dd k k^{-1}.
\ee

 \subsection{The equations of motion and conserved current}
To derive the equations of motion, consider a variation of $\mathcal{V}$ from the left, i.e. $\mathcal{V}(x) \rightarrow \mathcal{V}'(x)=  e^{\epsilon(x)}\mathcal{V}(x)$, with $\epsilon(x) \in \bar{\mathfrak{p}}$ infinitesimal. As the action is invariant under the action of local $\bar{\mathrm K}$-transformations from the left, this is a non-trivial deformation only for $\bar{\mathfrak{p}}$-valued $\epsilon(x)$. This transformation gives
\begin{equation}
\delta \mathcal{P} =  \mathrm{d}\epsilon + [\mathcal{Q}, \epsilon]\, ,
\end{equation}
 and the equations of motion thus become
\begin{equation}
\label{eqn:SigmaMotion}
\partial_{\mu} (\sqrt{|h|}  h^{\mu \nu} \mathcal{P}_{\nu}) - \sqrt{|h|} h^{\mu \nu}[ \mathcal{Q}_{\mu}, \mathcal{P}_{\nu} ] = 0.
\end{equation}
We have seen in \eqref{eqn:tsfoqconnection} that $\mathcal{Q}$ transforms under gauge transformations as a connection. Its appearance in (\ref{eqn:SigmaMotion}) supports this point of view, so will hereafter refer to $\mathcal{Q}$ as the connection. Note that $\mathcal{Q}$ and $\mathcal{P}$ are not independent, but both derived from the map $\mathcal{V}$. The equations of motion (\ref{eqn:SigmaMotion}) is hence invariant under both gauge transformations and global $\bar{ \mathrm G}$-transformations. We can furthermore use Noethers theorem to derive a gauge invariant Lie algebra valued Noether ($k-1$)-form
\begin{equation}\label{eqn:currapp}
\mathcal J^{\mu} =  \sqrt{|h|} h^{\mu \nu} \mathcal{V}^{-1} \mathcal P_{\nu} \mathcal{V} ,
\end{equation}
which is conserved by virtue of the equations of motion. The Noether-form transform in the adjoint representation of $\mathrm G$, so that when $\mathcal{V} \rightarrow \mathcal{V}'=  \mathcal{V} g^{-1}$,
\begin{equation}
\mathcal J^{\mu}  \rightarrow \mathcal{J}'^{\mu}=  g\mathcal{J}^{\mu} g^{-1}.
\end{equation}
Note also that $\mathcal J$ being conserved implies the equations of motion, so (\ref{eqn:SigmaMotion}) and (\ref{eqn:currapp}) are equivalent, which is the natural consequence of the arbitrariness in defining the action of $\bar{ \mathrm G}$ from the right or from the left when deriving the equations of motion.\\

Let us consider the dynamics of this model. Choose a grading of $\bar{\mathfrak{g}}$ such that
\begin{equation}
\bar{\mathfrak{g}} = \underset{\ell}\bigoplus \ \bar{\mathfrak{g}}_\ell, 
\end{equation}
respected by the involution $\imath$, in the sense that $\imath(\bar{\mathfrak{g}}_\ell) \subset \bar{\mathfrak{g}}_{-\ell}$. The restricted root space decomposition \cite{Helgason:1978} provide for example such a grading, another is given by the level decomposition under a regularly embedded subalgebra (see Section \ref{sec:leveldecomposition}). We can then choose a base  of every level $\ell$ and $-\ell$ in terms of generators $E^{(\ell)}$ and $F^{(\ell)}$ such that $\imath(E^{(\ell)}) = -F^{(\ell)}$. Note that for a finite Lie algebra, the spaces $\bar{\mathfrak{g}}_\ell$ are zero for $|\ell|$ bigger than some given $\mathrm N$, and if $\mathrm{dim}\ \bar{\mathfrak{g}}_\ell > 1$, $E^{(\ell)}$ (and $F^{(\ell)}$) has some additional index, enumerating these generators.  The algebra-valued function $\mathcal{P}$ can now be expanded, with respect to this grading \cite{Damour:2004zy}, as
\begin{equation}
\label{eqn:generalP}
\mathcal{P} = \frac{1}{2}P_{(0)} K^{(0)} + \frac{1}{2} \sum_{\ell \geq 1} P_{(\ell)} (E^{(\ell)} + F^{(\ell)})\, ,
\end{equation}
and we write $K^{(0)}$ for the elements in $\bar{\mathfrak{p}}$ at level zero, i.e. $K^{(0)} \in P_{\imath}(\bar{\mathfrak{g}}_0)$ (with $P_{\imath}: \bar{\mf g} \rightarrow \bar{\mf p}$ the projection onto the coset algebra) defining $J^{(0)}$ to span their complement in $\bar{\mathfrak{g}}_0$. The connection $\mathcal{Q}$ is similarly written
\begin{equation}
\label{eqn:generalQ}
\mathcal{Q} = \frac{1}{2}Q_{(0)} J^{(0)} + \frac{1}{2} \sum_{\ell \geq 1} P_{(\ell)} (E^{(\ell)} - F^{(\ell)})\, .
\end{equation}
Inserting these expressions into (\ref{eqn:SigmaMotion}) we see that, as all the generators are linearly independent, (\ref{eqn:SigmaMotion}) split into one equation for every generator. These equations are to be interpreted as equations of motion, but also as generalized Bianchi identities and constraints, as the parameters $P_{(l)}$ are not all simultaneously physical fields. More explicitly, inserting the expansions \eqref{eqn:generalP} and \eqref{eqn:generalQ} into (\ref{eqn:SigmaMotion}), we get the equation of motion for $P_{(0)}$,

\begin{equation} \begin{split}
\label{eqn:levelzeromotion}
\frac{1}{\sqrt{|h|}} \partial_{\mu} (\sqrt{|h|} h^{\mu \nu} {P_{(0)}}_{\nu}) K^{(0)} + \frac{1}{2} {P_{(0)}}^{\mu}{Q_{(0)}}_{\mu}[J^{(0)}, K^{(0)}] &\\
 + \sum_{\ell \geq 1} {P_{(\ell)}}^{\mu} {P_{(\ell)}}_{\mu} [ E^{(\ell)}, F^{(\ell)}] = 0 ,&
\end{split} \end{equation}
and for the $P_{(\ell)}$'s we get
\be
\label{eqn:higherlevelmotion} \begin{split}
&\frac{1}{\sqrt{|h|}} \partial_{\mu} (\sqrt{|h|} h^{\mu \nu} {P_{(\ell)}}_{\nu})(E^{(\ell)}+F^{(\ell)}) + \frac{1}{2}  {Q_{(0)}}^{\mu} {P_{(\ell)}}_{\mu}  [J^{(0)}, E^{(\ell)}+F^{(\ell)}] \\
&- \frac{1}{2}  {P_{(0)}}^{\mu} {P_{(\ell)}}_{\mu}  [K^{(0)}, E^{(\ell)}-F^{(\ell)}]\\
&+ \frac12 \sum_{\substack{k,m \geq 1\\ k-m = \ell}} {P_{(k)}}^{\mu} {P_{(m)}}_{\mu} [ E^{(k)}-F^{(k)}, E^{(m)}+F^{(m)}] = 0 .
\end{split}\ee

\setcounter{equation}{0}
 \section{The $\gppp$-invariant action} \label{sec:gpppinvactionthese1}
In  Section \ref{sec:kmsymcompct14}, we have seen how hidden Kac-Moody symmetries are exhibited through compactifications of suitably chosen actions \eqref{tab:actionmax}. It was conjectured that theses actions possess the very-exyended Kac-Moody symmetries $\gppp$ \cite{Englert:2003zs}. Now, we would like to make these Kac-Moody symmetries manifest and in this context, an action $\mathcal{S}_{\gppp}$
explicitly invariant under the infinite-dimensional group $\gppp$ was constructed in  \cite{Englert:2003py}. We will review in this section the building of such an action  using the tools given in Section \ref
{app:SigmaModel}. 

The
action $\mathcal{S}_{\gppp}$ is defined in a reparametrisation invariant way  on a
world-line, a priori unrelated to space-time, in terms of an infinity of fields
$\psi_a(\xi)$ where $\xi$ spans the world-line. The fields $\psi_a(\xi)$
live in a coset space
$\gppp/\mathrm{K}^{*+++}$ where the subgroup $\mathrm{K}^{*+++}$ is invariant under a
\emph{temporal involution} $\Omega_1$ (defined later in Section \ref{PERTABsec:tempinv})
 which ensures that the action is $\mathrm{SO}(1,D-1)$ invariant and  which allows identification of index $1$ to a time coordinate.
 
 To avoid the confusion,  we will be interested in this section only by groups $\gppp$ whose the corresponding algebras are split real form. The subtleties introduced by non-split real forms will be highlighted in Part III. The Dynkin diagrams of the $\agppp$ algebras that we will consider in this section are represented in Figure \ref{fig:kmalgebrasgood}. We have  re-labelled the nodes of Figure \ref{fig:kmalgebrasgood} comparing to Figure \ref{fig:dynkinkmalgebras} to simplify the notations of this section.
 
 \begin{figure}[h]
\begin{center}
{\scalebox{0.17}
{\includegraphics{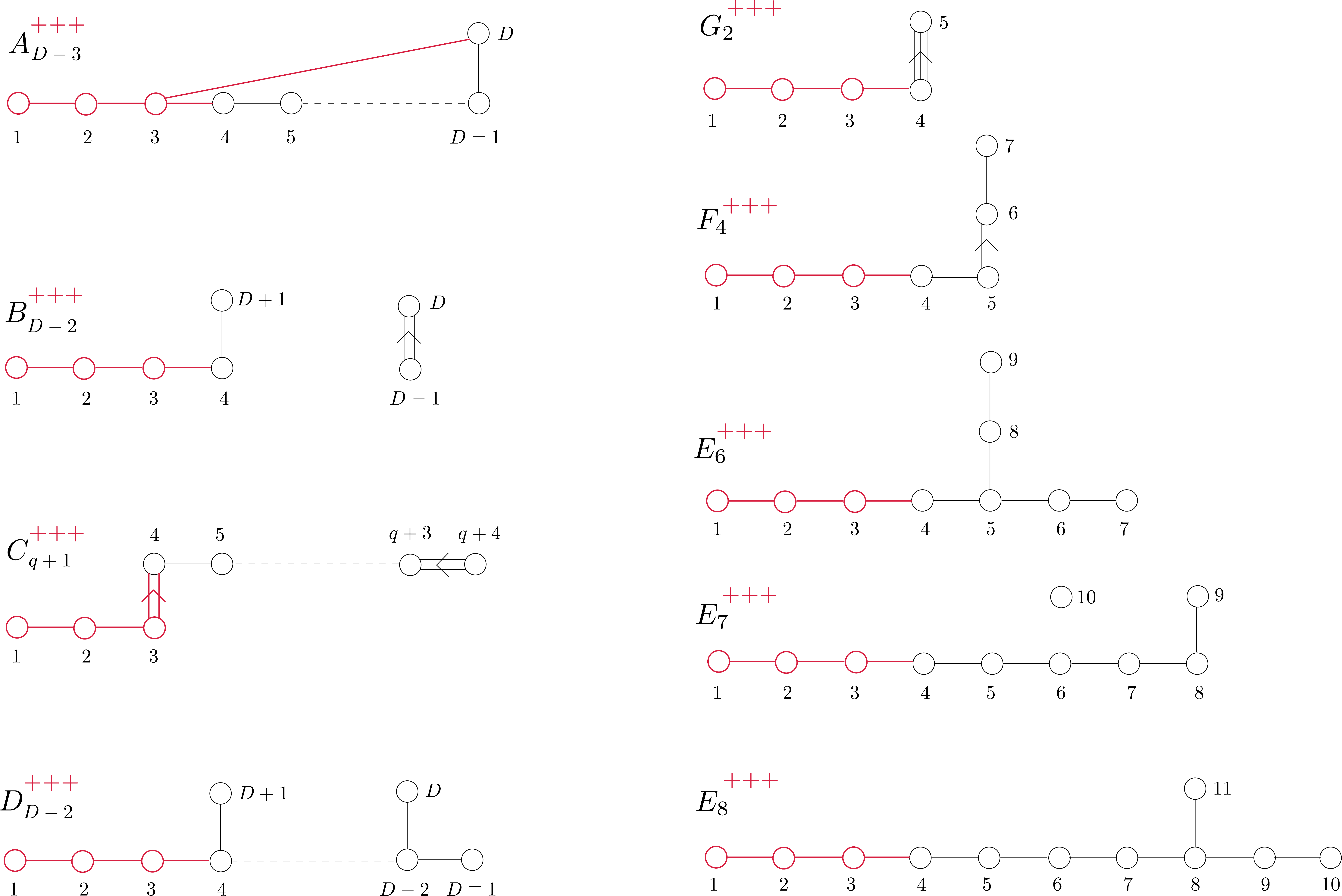}}}
\caption{\small \textsl{Dynkin diagrams of the Kac-Moody algebras $\agppp$ obtained by extension of the finite simple Lie algebras $\ag$ with the nodes labeled by $1,2,3$. The horizontal line starting at $1$ defines the gravity line which is the Dynkin diagram of a $A_{D-1}$ subalgebra. }}
\label{fig:kmalgebrasgood}
\end{center}
\end{figure}

\subsection{On the level decomposition of $\agppp$}
According to the discussion in Section \ref{app:SigmaModel}, we will write an element of the coset $\gppp/\mathrm K^{*+++}$ by exponentiating the Borel subalgebra of $\mathfrak{g}^{+++}$ as:
 \begin{eqnarray}\label{tab:corep}
 \mathcal{V} (\xi)= \exp^{ \psi_a(\xi) \mathcal{B}^{a} }\, .
 \end{eqnarray}
  To each Borel generator $\mathcal{B}^a$, we associate a field $\psi_a(\xi)$. As there is an infinity of generators $\mathcal{B}^a$, there is an infinite number of fields $\psi_a(\xi)$. Therefore we have to organise this summation on the infinity of fields  in  a recursive way. We will use the level decomposition (introduced in Section \ref{sec:leveldecomposition}) of the algebra $\agppp$ with respect to the finite subalgebra $A_{D-1}=\mathfrak{sl}(D, \mbb R)$.  Each $\mathfrak{g}^{+++}$ contains indeed a subalgebra $\mathfrak{gl}(D, \mbb R)$ such that $\mathfrak{sl}(D, \mbb R) \subset \mathfrak{gl}(D, \mbb R) \subset
\mathfrak{g}^{+++}$ where $D$ is interpreted as the space-time dimension.\\

The level zero generators are the $\mathfrak{sl}(D, \mbb R)$ generators and the Cartan generators not belonging to the $\mathfrak{sl}(D, \mbb R)$-subalgebra. One of this Cartan generator will always contribute to the enhancement of the algebra $\mathfrak{sl}(D, \mbb R)$ to $\mathfrak{gl}(D, \mbb R)$.
The generators of the representation $\mathfrak{gl}(D, \mbb R)$  are taken to be
$K^a_{~b}\ (a,b=1,2,\ldots ,D)$   with the following commutation relations
\begin{equation}
\label{dualpap:Kcom} [K^a_{~b},K^c_{~d}]   =\delta^c_b
K^a_{~d}-\delta^a_dK^c_{~b}\,  .
\end{equation}  
When the  maximally oxidised theory \eqref{tab:actionmax} has $q$ dilatons\footnote{All the maximally oxidised  theories have at most one dilaton except
the $C_{q+1}$-series characterised by $q$ dilatons. } the $\mathfrak{gl}(D, \mbb R)$ algebra is enlarged by the addition of $q$ $\mathbb R$-factors through abelian generators $R_u \,( u=1, \dots , q)$.
The $K^a_{~b}$ along with  the $R_u$  generators are the level zero
generators. The step operators of level greater than zero are  tensors
$R^{ a_1\dots a_r }_{\quad b_1 \dots b_s}$ of the $A_{D-1}$ subalgebra. We recall that each  tensor forms an  irreducible representation of $A_{D-1}$ characterised
by some Dynkin labels (see Section \ref{sec:leveldecomposition}). The lowest levels contain antisymmetric tensor step  operators
$R^{a_1a_2 \dots a_r}$ associated to electric and magnetic
roots arising from the dimensional reduction of field strength forms in \eqref{tab:actionmax}. They satisfy the tensor and scaling relations
\be \begin{split}
\label{dualpap:root} [K^a_{~b},R^{a_1\dots a_r}]   &=\delta^{a_1}_b R^{aa_2
\dots a_r} +\dots +
\delta^{a_r}_b R^{a_1 \dots a_{r-1}a}\, ,\\
 [R, R^{a_1\dots a_r}
] &=   -\frac{\varepsilon_A
a_A}{2}\,  R^{a_1\dots a_r}\, ,
\end{split}\ee
 where $a_A$ is the dilaton coupling constant to the
 field strength form   and
$\varepsilon_A$ is $+1\, (-1)$ for an
electric (magnetic) root \cite{Englert:2003zs}. The generators obey the invariant
scalar product relations
\begin{eqnarray}
\label{dualpap:killing}
&&( K_{~a}^a | K_{~b}^b) =G_{ab}\, ,\quad(
K^b_{~a} |K_{~c}^d)=
\delta_c^b\delta_a^d \ a\neq b\, , \quad( R| R)
=\frac{1}{2}
\,,\\
\label{dualpap:step}
&&(
R^{ a_1\dots a_r }_{\quad b_1 \dots b_s} | \bar{R}_{ a'_1\dots a'_r }^{\quad b'_1 \dots b'_s})
=\delta^{a_1}_{a'_1}\dots\delta^{a_r}_{a'_r}\delta^{b'_1}_{b_1}\dots\delta^
{b'_s}_{b_s}\,.
\end{eqnarray}
Here $G= I_D -
\frac{1}{2}\Xi_D$ where $\Xi_D$ is a $D$-dimensional matrix with all
entries   equal to unity and   $\bar{R}_{ a'_1\dots a'_r }^{\quad b'_1 \dots b'_s}$ designates the negative step operator conjugate to $ R^{ a'_1\dots a'_r }_{\quad b'_1 \dots b'_s}$ interchanging upper and lower indices.\\

Using this level decomposition, it is possible to rewrite the coset representative  (\ref{tab:corep}) as
\begin{equation}
\label{dualpap:positive} {\cal V(\xi)}= \exp \big(\sum_{a\ge b}
h_b^{~a}(\xi)K^b_{~a} -\sum_{u=1}^q
\phi^u(\xi) R_u\big) \exp \big(\sum
\frac{1}{r!s!} A_{ a_1\dots a_r |s_1 \dots s_p }(\xi) R^{ a_1\dots a_r |s_1 \dots s_p } \big)
\end{equation}
where we have associated to each Borel generators $\mathcal{B}^a =\{K^b_{\ a}, R_u, R^{ a_1\dots a_r}_{\quad b_1 \dots b_s }, \dots \}$ a field $\psi_a (\xi)=\{h^{\ a}_b(\xi), \phi_u(\xi), A_{ a_1\dots a_r}^{\quad b_1 \dots b_s } (\xi), \dots \} $. Note that the first exponential in \eqref{dualpap:positive} contains only level zero operators and the second one the positive root generators of levels strictly greater than zero.
\subsection{The temporal involution and the coset $\gppp/\mathrm K^{*+++}$}\label{PERTABsec:tempinv}
The metric $g_{\mu \nu}$ at a fixed space-time point parametrises the coset $\mathrm{GL}(D)/\mathrm{SO}(1,D-1)$. To construct a $\gppp$-invariant action containing such a tensor, we shall build a non linear realisation of $\gppp$ in a coset space $\gppp/\mathrm K^{*+++}$ where the subgroup $\mathrm K^{*+++}$ contains the Lorentz group $\mathrm{SO}(1,D-1)$. As $\mathrm{SO}(1,D-1)$ is non-compact, we cannot use the Chevalley involution $\omega$ (defined in Section \ref{subsec:chevalleyinvo2}) to construct $\mathrm K^{*+++}$ that is now non-compact. Rather we will use the \emph{temporal involution} $\Omega_1$ from which the required non-compact generators of $\mathrm K^{*+++}$ can be selected.

The temporal involution $\Omega_i$ generalizes the Chevalley involution to allow identification of the index $i$ to a to a time coordinate. Let us define the action of a general temporal involution $\Omega_i$  on the Chevalley generators $\{h_j, e_j, f_j \}$ of $\agppp$:
\be \begin{split}
&\Omega_i \\
h_j \quad &\rightarrow \quad  -\, h_j \, ,\\
e_j\quad &\rightarrow  \quad - \, \epsilon_i \, f_j\, , \\
f_j
\quad &\rightarrow\quad  - \, \epsilon_i \,e_j\, .
\end{split} \ee
Here $e_j$  is expressed as $A_{D-1}$ tensors and $f_j$  as the tensor with upper and lower indices interchanged. $\epsilon_i$ is defined as $+1(-1)$ if the number of $'i'$ indices (that is the number of time indices) is even (odd). In the same way, all the generators $e_{\alpha}$ and $f_{\alpha}$ (obtained by multiple commutators of $e_j$ and $f_j$ as in \eqref{eqn:ealphagen1} ) are mapped under the temporal involution to $(e_{\alpha}, f_{\alpha}) \longrightarrow \, -\epsilon_{\alpha}\, (f_{\alpha}, e_{\alpha})$ where $\epsilon_{\alpha}$ is equal to $+1(-1)$ depending that the number of indices $'i'$ are even (odd). As all the generators of $\mathfrak{g}^{+++}$ are expressed  in terms of $K^a_{\ b}, R_u$ and tensors $R^{ a_1\dots a_r}_{\quad b_1 \dots b_s }$ (see Section \ref{sec:leveldecomposition}), we can write the effect of the temporal involution on these quantities as
\be \begin{split}
\label{tabeq:inomegk}
&\Omega_i  \\
K^{a}_{\ b} \quad &\rightarrow \quad  -\, \epsilon_{a}\,\epsilon_{b}\,
K^{b}_{\ a}\, ,\\
R_u \quad  &\rightarrow  \quad  -R_u\, , \\
R^{ a_1\dots a_r}_{\quad b_1 \dots b_s }
\quad &\rightarrow \quad -\,\epsilon_{a_{1}}\ldots
\epsilon_{a_{r}}\,\epsilon_{b_{1}}\ldots\, \epsilon_{b_{s}}\,
\bar{R}_{ a_1\dots a_r}^{\quad b_1 \dots b_s }\, , 
\end{split}\ee
with $\epsilon_a = -1$ if $a=i$ and $\epsilon_a = 1$ otherwise. For example we have:   $\Omega_1(R^{9\, 10\, 11}) \, =\,-\,  \epsilon_9\, \epsilon_{10}\, \epsilon_{11} \, R_{9\, 10\, 11}\, =\, - \, R_{9\, 10\, 11}$. The temporal involution leaves invariant the subalgebra $\mf k^{*+++}$ of $\agppp$. Its generators are $(e_{\alpha}- \epsilon_{\alpha}  f_{\alpha})$.\\
 

Following the steps of Section \ref{app:SigmaModel} with the coset representative given by \eqref{dualpap:positive} and the involution $\imath=\Omega_1$, we are now able to construct an  $\mathcal{S}_{\gppp}$ invariant under
global $\gppp$ transformations, defined on  the
coset ${\gppp}/\mathrm{K}^{*+++}$
\begin{equation}
\label{dualpap:actionG} \mathcal{S}_{ \gppp}=\int \dd\xi  \frac{1}{n(\xi)}(\mathcal{P} (\xi) | \mathcal{P}(\xi))\, ,
\end{equation} where
$n(\xi)$ is an arbitrary \emph{lapse function} ensuring reparametrisation
invariance on the world-line. This lapse function ensures the motion to be a null geodesic on the coset through its equation of motion $(\mathcal{P}|\mathcal{P})=0$.  Writing
\begin{equation}
\label{dualpap:full} \mathcal{S}_{\gppp} =\mathcal S_{\gppp}^{(0)}+\sum_A
\mathcal S_{\gppp}^{(A)}\, ,
\end{equation} where $\mathcal S_{\gppp}^{(0)}$ contains all level
zero contributions, one obtains (see \cite{Englert:2003py} for more details)
\begin{subequations} \begin{align}
 \mathcal S_{\gppp}^{(0)} &=\frac{1}{2}\int \dd \xi
\frac{1}{n(\xi)}\left[\frac{1}{2}(g^{\mu\nu}g^{\sigma\tau}-
\frac{1}{2}g^{\mu\sigma}g^{\nu\tau})\frac{\dd g_{\mu\sigma}}{\dd \xi}
\frac{\dd g_{\nu\tau}}{\dd \xi}+\sum_{u=1}^q
\frac{\dd \phi^u}{\dd \xi}\frac{\dd \phi^u}{\dd \xi}\right], \label{dualpap:fullzero}\\
\mathcal S_{\gppp}^{(A)}&=\frac{1}{2 r! s!}\int \dd \xi
\frac{ e^{- 2\lambda
\phi}}{n(\xi)}\left[
\frac{\mathrm{D}A^{\mu_1\dots \mu_r}_{ \quad \nu_1\dots
\nu_s }}{\dd \xi} g_{\mu_1{\mu}^\prime_1}...\,
g_{\mu_r{\mu}^\prime_r}g^{\nu_1{\nu}^\prime_1}...\,
g^{\nu_s{\nu}^\prime_s}
\frac{\mathrm{D}A^{\mu'_1\dots \mu'_r }_{\quad \nu'_1\dots
\nu'_s }}{\dd \xi}\right]. \label{dualpap:fulla}
\end{align}\end{subequations}
The $\xi$-dependent fields $g_{\mu\nu}$ are defined as
$g_{\mu\nu} =e_\mu^{~a}e_\nu^{~b}\eta_{ab}$ where $e_\mu^{~a}=(e^{-h(\xi)})_\mu^{~a}$. The appearance of the
Lorentz metric $\eta_{ab}$ with $\eta_{11}=-1$ is a consequence of the
temporal involution $\Omega_1$. The metric $g_{\mu\nu}$ allows a
switch from  the Lorentz  indices  $(a,b)$ of the fields appearing  in
(\ref{dualpap:positive}) to
$\mathrm{GL}(D)$ indices $(\mu,\nu)$.  $\mathrm{D}/\dd \xi$ is a  covariant derivative generalising
$\dd /\dd \xi$ through  non-linear terms arising from  non-vanishing
commutators  between  positive  step operators. There evaluation is group dependent\footnote{As an example we give the covariant derivative of the $6$-form appearing at level $2$ in the decomposition of $E_8^{+++}$ (see Table \ref{tab:levdece11}): $\tfrac{\mathrm{D}}{\dd \xi} A_{a_1 a_2 \dots a_6} = \tfrac{\dd}{ \dd \xi} A_{a_1 a_2 \dots a_6}  +10\, A_{[a_1 a_2 a_3} \tfrac{ \dd}{ \dd \xi} A_{a_4 a_5 a_6]}$.}. The parameter $\lambda$ is the
generalisation of the scale parameter
$-\varepsilon_A a_A/2$ to all roots.
\subsection{Weyl reflections and signatures of $\gppp$-invariant action} \label{sec:tempinvolutintro}
We will now  review the effect of Weyl reflection on the space-time signature of a $\gppp$-invariant theory. First, recall that a Weyl transformation of a generator $T$ of a Lorentzian algebra $\mathfrak{g}^{+++}$ can be expressed as a conjugation by a group element $U_W$ of $\mathrm{G}^{+++}$: $T \longrightarrow  U_W\,  T \, U^{-1}_W $. The non-commutativity
of the temporal involution $\Omega_1$ with the Weyl reflections
\be \label{eqn:nocomu69}
  U_W\, (\Omega_1 T)\,  U_W^{-1}= \Omega '\,  (U_W T U^{-1}_W)\, , 
  \ee
implies that different Lorentz signatures $(t,s)$ (where $t(s)$ is the number of time (space) coordinates)   are related between themselves \cite{Keurentjes:2004bv,Keurentjes:2004xx,Englert:2004ph}. The analysis of signature changing has been done for all $\mathfrak{g}^{+++}$ that are very-extensions of a simple split Lie algebra $\mathfrak{g}$ \cite{deBuyl:2005it, Keurentjes:2005jw}. In these cases, Weyl reflections with respect to a root of gravity line\footnote{The gravity line is the set of the simple roots of the $\mathfrak {sl}(D, \mathbb R)$-part of $\mathfrak{g}^{+++}$ (see Figure \ref{fig:kmalgebrasgood}).} do not change the global Lorentz signature  but it changes only the identification of the time coordinate. In fact, only Weyl reflections with respect to roots not belonging to the gravity line can change the global signature of the theory.
 
\subsubsection{Weyl reflection with respect to roots of the gravity line}

To illustrate the consequence of the non-commutativity of Weyl transformations with  the temporal involution $\Omega_1$, let us consider the simple example of the Weyl reflection $s_1$ defined by \eqref{weylreflection}  associated to the simple root $\alpha_1$ (see Figure \ref{fig:kmalgebrasgood}). Only simple roots $\alpha_1$ and $\alpha_2$ are modified by this reflection: 
\be \begin{split}
s_1(\alpha_1)&= -\, \alpha_{1} \, ,\\
s_2(\alpha_{2})&= \alpha_{2} \,+ \alpha_{1}\, .
\end{split} \ee
The other simple roots are unchanged under this reflection ($s_1(\alpha_i)=\alpha_{i}$ for
i$\neq$ 1, 2).
Let us consider now, the action of this Weyl reflection on the generators associated to these roots. The positive generators associated to the simple roots of the gravity line are expressed as: 
\begin{eqnarray}
e_i= K^{i}_{\ i+1} \qquad i=1,\ldots,D-1\, .
\end{eqnarray}
The generators associated to the roots $\alpha_1$ and $\alpha_2$ are then modified respectively as
\be \begin{split}
U \, K^{1}_{\ 2} \, U^{-1}&= \delta \, K^{2}_{\ 1}\, ,\\
U \, K^{2}_{\ 3} \, U^{-1}&= \rho \, \lbrack K^{1}_{\ 2},K^{2}_{\
3}  \rbrack=\rho\, K^{1}_{\ 3}\, ,
\end{split} \ee
where $\delta$ and $\rho$ are plus or minus signs which arise as positive generators are representations of the Weyl group up to signs. We will see that such signs always cancel in the determination of $\Omega'$ because they are identical in the Weyl transform of corresponding positive and negative roots, as their commutator is in the Cartan subalgebra which forms a true representation of the Weyl group.

Before the Weyl reflection, the index $1$ is time index and the others are space-like indices because of the choice of the temporal involution $\Omega_1$. We recall that the action on the involution $\Omega$ on the generators $K^{a}_{\ b}$ (see \eqref{tabeq:inomegk}) is
\begin{eqnarray}
\Omega \, K^{a}_{\ b}= -\, \epsilon_a \, \epsilon_{b}\, K^{b}_{\ a}\, ,
\end{eqnarray}
where $\epsilon_a$ is equal to $1(-1)$ if the index $a$ is space (time). If we apply \eqref{eqn:nocomu69}, we find the action of $\Omega'$ on the two generators:
\be \begin{split} 
\delta \, \Omega'\, K^{1}_{\ 2}&= \Omega' \,U\,K^{2}_{\
1}\,U^{-1}=U\, \underbrace{\Omega_1\,
K^{2}_{\ 1}}_{K^{1}_{\ 2}}\,U^{-1}\, ,\\
&= \delta \,  K^{2}_{\ 1}\, ,\\
\rho \, \Omega'\, K^{2}_{\ 3}&= \Omega' \,U\,K^{1}_{\
3}\,U^{-1}=U\, \underbrace{\Omega_1\,
K^{1}_{\ 3}\,}_{K^{3}_{\ 1}}U^{-1}\, ,\\
&=\rho\,  K^{3}_{\ 2}\, .
\end{split}\ee
One gets
\be \begin{split}\label{tab:inv12}
\Omega' \, K^{1}_{\ 2}= K^{2}_{\ 1}  \quad &\mathrm{while} \quad   \Omega_1 \, K^{1}_{\ 2}=
K^{2}_{\ 1}\, ,\\
\Omega' \, K^{2}_{\ 3}= K^{3}_{\ 2}  \quad &\mathrm{while} \quad  \Omega_1 \, K^{2}_{\ 3}=-\,
K^{3}_{\ 2}\, .
\end{split} \ee
The other positive generators associated to simple roots are invariant under this Weyl transformation. This implies
\begin{eqnarray}\label{tab:invnn}
\Omega' \, K^{n}_{\ n+1}=\Omega_1 \, K^{n}_{\ n+1}= -\, K^{n+1}_{\ n}\,, 
\qquad n=3,\ldots,D-1\, .
\end{eqnarray}
The content of \eqref{tab:inv12} and \eqref{tab:invnn} is represented in Table \ref{dualpap:gravity}. The signs below the generators of the gravity line indicate the sign in front of the negative step operator $K^{a+1}_{\ a}$ obtained by the involution $\Omega_1\,(\Omega')$ : a minus sign is in agreement with the Chevalley involution and indicates that the indices in $K^{a}_{\ a+1}$ are both either space or time indices while a plus sign indicates that one index must be time and the other space.
\begin{table}[top]
\begin{center}
\begin{tabular}{|c|ccccc|c|}
\hline
 &$K^1_{\ 2}$&$K^2_{\ 3}$&$K^3_{\ 4}$&$\cdots$&$K^{D-1}_{\
D}$&Time coordinate\\
\hline
$\Omega_1$&$+$&$-$&$-$&$-$&$-$&1\\
$\,\Omega^\prime$&$+$&$+$&$-$&$-$&$-$&2\\
\hline
\end{tabular}
\caption{\sl \small Involution switch from $\Omega_1$ to
$\Omega^\prime$ due to the Weyl reflection $s_{1}$.}
\label{dualpap:gravity}
\end{center}
\end{table}
The Table \ref{dualpap:gravity} shows that the time coordinate after the Weyl reflection (see the line $\Omega'$) must be identified either with $2$, or with all indices $\neq 2$. We choose the first description, which leaves unaffected coordinates attached to the planes invariant under the Weyl transformation\footnote{The Weyl reflection $s_{\alpha}$ fixes the hyperplane orthogonal to $\alpha$ (see \eqref{eqn:planinvweyl}).}. \\

It is easy to generalize this computation for all Weyl reflections with respect to roots of the gravity line \cite{Englert:2004ph}. We get that any  Weyl reflection $s_{\alpha_a}$ generated by  $\alpha_a$ a simple root of the gravity line exchanges the index $a$ and $a+1$ along with the space-time nature of the corresponding coordinates.

\subsubsection{Weyl reflection with respect to electric or magnetic roots}

Weyl reflections generated by simple roots not belonging to the gravity line
relate step operators of different levels. As a consequence, these may potentially induce changes of signature far less trivial than the simple exchange of the index identifying the time coordinate\cite{Keurentjes:2004bv,Keurentjes:2004xx,deBuyl:2005it,Keurentjes:2005jw}.  
 
We will now set up the notations and recollect some formulas necessary for some discussions of signature changes which will appears in Part II (see Appendix \ref{appendixthese})  and in Part III (see Section \ref{sec:signa}). In order to address the action of an  involution on a generic step operator
 $R^{a_1 \dots a_r}$ of level greater than zero in a given irrreducible representation of $A_{D-1}$, we introduce some notations.
 First, given an involution $\tilde \Omega$, one defines
 $\mathrm{sign}(\tilde \Omega X)$ for any given positive step operator  $X$  in  the following way
 \begin{equation}
 \tilde \Omega \, X \equiv \mathrm{sign}(\tilde \Omega X) \, {\bar X},
 \label{dualpap:sidef}
 \end{equation}
 where $\bar X$ designates the negative step operator conjugate to $X$.
Second, we also introduce a sign associated to a given  positive step operator of level greater than zero\footnote{As far as the action of the involution is concerned the symmetry properties of a step operator given by its Dynkin labels do not play any role.}
 $R^{a_1 \dots a_r}$
in the following way
\begin{subequations} \label{eqn:signtheomega}\begin{align}
+\quad : &\quad \mathrm{sign}(\tilde \Omega R^{a_1 \dots a_r})= -\epsilon_{a_1}\dots\epsilon_{a_r}  \label{dualpap:signa}\\
- \quad : &\quad \mathrm{sign}(\tilde \Omega R^{a_1 \dots a_r})=+\epsilon_{a_1}\dots\epsilon_{a_r}, \label{dualpap:signb}
\end{align}\end{subequations}
where $\epsilon_a=-1$ if $a$ is a timelike index and
$\epsilon_a=+1$ if $a$ is a space-like index, the space-time
nature of the coordinate labelled by the index $a$ being defined
by the action of $\tilde \Omega$ on the $K^a_{\ b}$. The $+$ sign
defined in (\ref{dualpap:signa}) will lead to the positive kinetic
energy term for the corresponding field in the action while
the $-$ sign defined in (\ref{dualpap:signb}) will lead to a negative
kinetic energy term.\\

The precise analysis of the different possible signatures has been studied from the algebraic point
 of view in great details for $E_{11}$  (and more generally for $E_n$) in
 \cite{Keurentjes:2004bv,Keurentjes:2004xx}\footnote{Some algebraic considerations in this context for others groups $ \g$ are presented in the Appendix of the second reference.}. In this case the corresponding maximally oxidised theory is the bosonic sector of the  low effective action of M-theory and  the signatures found was $(1,10,+)$, $(2,9,-)$, $(5,6,+)$, $(6,5,-)$ and $(9,2,+)$.  These signatures match perfectly with the signatures changing dualities and the exotic phases of M-theories \cite{Hull:1998vg,Hull:1998fh,Hull:1998ym}. Indeed, we can interpret the Weyl transformation $s_{\alpha_{11}}$ generated by the simple root $\alpha_{11}$ (see Figure
\ref{fig:kmalgebrasgood}) as a double T-duality in the direction $9$ and $10$ followed by an exchange of these directions \cite{Elitzur:1997zn,Obers:1998rn,Banks:1998vs,Englert:2003zs}. \\

In the reference \cite{deBuyl:2005it} we have extend to all
$\gppp$-theories the analysis of signature changing.
 We find for all the
$\gppp$-theories all the possible signature $(t,s, \pm)$, where $t$
(resp. $s$) is the number of time-like (resp. space-like)
directions, related by Weyl reflections of $\gppp$ to the signature
$(1,D-1,+)$ associated to the theory corresponding to the
traditional maximally oxidised theories. Along with the different
signatures the signs $\pm$ of the kinetic terms of the relevant fields
are also discussed. We started the analysis with $A_{D-3}^{+++}$
corresponding to pure gravity in $D$ dimensions then we extend the
analysis to the other $\gppp$, first to the simply laced ones and then to the
non-simply laced ones. Each $\gppp$ algebra contains a $A_{D-3}^{+++}$
subalgebra, the signatures of $\gppp$ should thus includes the one of
$A_{D-3}^{+++}$. This is indeed the case, but some $\gppp$ will
contain additional signatures. If one want to restrict our focus
on  string theory, the special cases of $D_{24}^{+++}$ and
$B_{8}^{+++}$ are interesting, the former being related to the
low-energy effective action of the bosonic string (without
tachyon) and the latter being related to the low-energy effective
action of the heterotic string (restricted to one gauge field).
The existence of signature changing dualities  are related to the
magnetic roots and suggests that these transformations correspond
to a generalisation of the S-duality existing in these two
theories \cite{Sen:1994fa,Sen:1994wr}.

 \setcounter{equation}{0}
 \section{From $\gppp$ to $\gpp$-invariant actions} \label{sec:gppinvactionthese2}
 
 To make connections between this new formalism and the covariant
space-time theories, it is interesting to analyse the several
actions invariant under overextended Kac-Moody algebra $\agpp$. The Dynkin diagram of $\agpp$ is obtained by deleting the root $\alpha_1$ in the Dynkin diagram of  $\agppp$ (see Figure \ref{fig:kmalgebrasgood}).  
It has been shown  that for each very-extended algebra $\gppp$, the $\gppp$- invariant theory encompasses
two distinct theories invariant under the overextended Kac-Moody subalgebra $\gpp$ \cite{Englert:2004ph}.
The  $\gpp_C$- invariant action $\mathcal S_{\gpp_C}$ describes a motion in a coset    $\gpp/\mathrm K^{++}$ and carries a Euclidean signature while the second theory described by a different embedding of $\gpp$ in
$\gppp$, referred as $\gpp_B$,  describes a motion in a different coset  $\gpp/\mathrm K^{*++}$. In contradistinction  with
the $\gpp_C$ case, the $\gpp_B$-theory
carries various Lorentzian signatures which are revealed through various equivalent
formulation related by Weyl transformations \cite{Englert:2004ph}.

We now recall quickly  the construction of these two theories. There are obtained both by a truncation of a infinity of $\agppp$-fields, putting to zero all the fields multiplying generators involving the deleted root $\alpha_1$ in Figure \ref{fig:kmalgebrasgood}. This truncation is realised consistently
with all equations of motion, i.e. it implies
that all the solutions of the equations of motions of
$\mathcal{S}_{\gpp}$ are also solutions of the equations of motion
of $\mathcal{S}_{\gppp}$  \cite{Englert:2004ph}.

 \subsection{Cosmological model} \label{subsec:cosmomodelth}
 The  study of the properties of cosmological solutions in the vicinity of space-like singularity  revealed an overextend symmetry $\gpp_C$ \cite{Damour:2000hv,Damour:2001sa,Damour:2002et}.
The overextended algebra $\gpp_C$ is obtained from the very-extended
algebra $\gppp$  by deleting the  node labelled 1 from the Dynkin
diagrams of $\gppp$ depicted in Figure \ref{fig:kmalgebrasgood}.
The action $\mathcal S_{\gpp_C}$ describing the $\gpp_C$ theory is obtained from $\gppp$ by performing the following
consistent truncation. One puts to zero in  the coset representative (\ref{dualpap:positive}) the
field multiplying the Chevalley generator
$h_1=  K_{~1}^1- K_{~2}^2$ and all the
fields multiplying the positive step operators associated to  roots whose
decomposition in terms of simple roots contains  the deleted root
$\alpha_1$. Performing this truncation (see \cite{Englert:2004ph} for more details) one obtains the following action
\begin{equation}
\label{dualpap:fullp} \mathcal S_{\gpp_C} =\mathcal S_{\gpp_C}^{(0)}+\sum_B
\mathcal S_{\gpp_C}^{(B)}\, ,
\end{equation} with
\begin{subequations}\begin{align} 
 \mathcal S_{\gpp_C}^{(0)} &=\frac{1}{2}\int \dd t
\frac{1}{n(t)}\left[\frac{1}{2}(g^{\hat{\mu}\hat{\nu}}g^{\hat{\sigma}\hat{\tau}}-
g^{\hat{\mu}\hat{\sigma}}g^{\hat{\nu}\hat{\tau}})\frac{\dd g_{\hat{\mu}
\hat{\sigma}}}{\dd t}
\frac{\dd g_{\hat{\nu}\hat{\tau}}}{\dd t}+
\frac{\dd \phi}{\dd t}\frac{\dd \phi}{\dd t}\right] , \label{dualpap:fullop}\\
\mathcal S_{\gpp_C}^{(B)}&=\frac{1}{2 r! s!}\int \dd t
\frac{ e^{- 2\lambda
\phi}}{n(t)}\left[
\frac{\mathrm{D}B_{\hat{\mu}_1\dots \hat{\mu}_r}^{\quad \hat{\nu}_1 \dots
\hat{\nu}_s}}{\dd t} g^{\hat{\mu}_1{\hat{\mu}}^\prime_1}...\,
g^{\hat{\mu}_r{\hat{\mu}}^\prime_r}g_{\hat{\nu}_1{\hat{\nu}}^\prime_1}
...\     g_{\hat{\nu}_s{\hat{\nu}}^\prime_s}
\frac{\mathrm{D}B_{{\hat{\mu}}^\prime_1\dots {\hat{\mu}}^\prime_r}^{\quad
{\hat{\nu}}^\prime_1\dots {\hat{\nu}}^\prime_s}}{\dd t}\right] . \label{dualpap:fullc}
\end{align} \end{subequations}
 Hatted indices $\hat\mu \,  (\hat \mu=2,\dots, D)$ have been introduced, the remaining $A$-fields have been denoted by $B$ and  $\xi$  has been interpreted  as $t$.

This theory describes a motion on the coset
$\gpp/\mathrm K^{++}$ where  $\mathrm K^{++}$ is the maximal compact subgroup of $\gpp$. The action \eqref{dualpap:fullp} is then equivalent to
\be
\mathcal{S}_{\gpp/ \mathrm K^{++}} = \int \dd t \frac{1}{n(t)} \, (\mathcal{P}(t)|\mathcal{P}(t)),
\ee
where  $n(t)^{-1}= \sqrt{h}h^{1\, 1}$ (see \eqref{eqn:SigmaModelAction}) and the involution $\imath$ used to construct the subalgebras of $\agpp$ is the Chevalley involution $\omega$. Since  $\mathrm K^{++}$ is defined by the Chevalley
involution which commutes with the Weyl reflection of $\gpp$ this coset admits only the Euclidean signature.\\

This `cosmological' action $\mathcal S_{\gpp_C}$ generalises to
all $\gpp$  the $\mathZ E_{10}\, (\equiv \mathZ E_8^{++})$ action of reference \cite{Damour:2002cu, Damour:2004zy} proposed in the context of M-theory and cosmological billiards. Considering the action $\mathcal{S}_{\mathZ E_{10}/ \mathrm K^{++}}$, it has been shown that the action restricted to a definite number of lowest fields is equal to the eleven-dimensional supergravity in which the fields depend only on the time coordinate \cite{Damour:2002cu}. More precisely this restriction to lowest level fields  ($\ell < \bar{\ell}$) is done by vanishing all covariant derivative operating on a field $B_{\hat{\mu}_1\dots \hat{\mu}_r}^{\quad \hat{\nu}_1 \dots \hat{\nu}_s}$ which belongs to the subspace $\ag _{\ell} \ \mathrm{with} \ \ell > \bar{\ell}$. In the same way, considering the expansion of the algebra valued function $\mathcal{P}$, this truncation consists by setting to zero all the coefficients $\mathcal{P}_{(\ell)}$ when $\ell > \bar{\ell}$ in \eqref{eqn:generalP}. This restriction on some low level fields is a consistent truncation of the model $\mathrm{G}^{++}/ \mathrm{K}^{++}$ because it constitutes a solution of the equations of motion of $\mathcal{S}_{\mathrm{G}^{++}/ \mathrm{K}^{++}}$ (see \eqref{eqn:levelzeromotion} and \eqref{eqn:higherlevelmotion}). One gets a perfect match between the equations of motion of the $\mathZ{E}_{10}/ \mathrm{K}^{++}$ $\sigma$-model truncated up to level $\ell=3$ and height $29$ (see Section \ref{subsec:leveldece11}) and the equations of motion of eleven-dimesional supergravity in the vicinity of space-like singularity as well as between the Hamiltonian constraints. It is therefore possible to establish a 'dictionary' and   relate the truncated fields of the coset $\mathZ{E}_{10}/ \mathrm{K}^{++}$ and the degrees of freedom of eleven-dimensional supergravity (see Table \ref{tabtab:cosmo}). For the precise dictionary, see \cite{Damour:2002cu,Damour:2004zy}. 

\begin{table}[bottom]  \centering
\scalebox{0.83}{
\begin{tabular}{cccc}

&\emph{Fields belonging  to $\mathcal{S}_{\mathZ{E}_{10}/ \mathrm K^{++}}$}  & & \emph{Fields of supergravity depending on time} \\
&& &  \\
level 0& $g_{\hat{\mu} \hat{\nu}}(t)= f(h_a^{\ b})$  &  $\leftrightsquigarrow $  & metric \\ 
level 1&  $A_{\hat{\mu}\hat{\nu} \hat{\rho}}(t)$ &  $\leftrightsquigarrow$  & 3-form electric potential \\
level 2&  $A_{\hat{\mu}_1...\hat{\mu}_6}(t)$ &  $\leftrightsquigarrow$  & 6-form magnetic potential (dual of the 3-form) \\ 
level 3&  $A_{\hat{\mu}_1...\hat{\mu}_8,\hat{\nu}}(t)$ &  $\leftrightsquigarrow$  & 'dual' of the metric
\end{tabular}
} 
\caption{Link between the lowest level up to level $3$ and height $29$ of the cosmological $\mathZ{E}_8^{++}$-invariant theory and the fields of the eleven-dimensional supergravity.}\label{tabtab:cosmo}
\end{table}

We can now ask the question: what is about the higher level fields? There is a conjecture, not yet checked, in \cite{Damour:2002cu} that some higher level fields associated to roots of $E_9 \subset E_{10} \subset E_{11}$ (namely with $m_1=m_2=0$) contain spatial derivatives of the lowest fields of level $1, 2$ and $3$. These higher level fields at level $\ell=3k+1$, $\ell=3k+2$ and $\ell=3k+3$ are associated to the following $\mathfrak{sl}(11, \mathbb R)$ Dynkin labels
\be\label{eqn:dynkinlabel475}
[0\, k\, 0\, 0\, 0\, 0\,0\, 1\, 0\, 0\, ], \quad [0\, k\, 0\, 0\, 1\, 0\,0\, 0\, 0\, 0\, ], \quad [0\, k\, 1\, 0\, 0\, 0\,0\, 0\, 0\, 1\, ],
\ee
 with lowest weights $\Lambda^2=2$. The higher level fields correspond to the three infinite towers of fields
 \be\label{eqn:highfield475}
 A_{a_1\, a_2\, a_3}^{\quad s_1\, s_2\, \dots s_k }, \quad A_{a_1\, a_2\, \dots a_6}^{\quad s_1\, s_2\, \dots s_k }, \quad A_{a_1\, a_2\, \dots a_8|b}^{\quad s_1\, s_2\, \dots s_k } 
 \ee
where the upper indices are completely symmetric and the lower indices have the same symmetries as the level $1,2$ and $3$ representations. These infinite tower fields have precisely the right index structure to be interpreted as spatial gradiant of the lowest fields of level $1, 2$ and $3$
\be
\partial^{s_1}\ldots \partial^{s_k}  A_{a_1\, a_2\, a_3}(t), \quad \partial^{s_1}\ldots \partial^{s_k}  A_{a_1\, a_2\, \dots a_6}(t), \quad \partial^{s_1}\ldots \partial^{s_k}  A_{a_1\, a_2\, \dots a_8|b}.
\ee
We will see that in the case of the brane model, the higher level fields \eqref{eqn:highfield475} have been identified to BPS solutions of eleven-dimensional supergravity  or of its exotic counterparts (see \cite{Englert:2007qb} and Part II). Theses results suggest that the 'gradient conjecture' of \cite{Damour:2002cu} must be reconsidered.

 \subsection{Brane model} \label{subsec:branemodelth}
 The $\gpp_B$-theory is obtained by performing the same truncation
as for the $\mathrm{G}_C^{++}$-model, namely we equate as before to zero
all the fields in the $\gppp$-invariant action (\ref{dualpap:full}) which
multiply generators  involving the root $\alpha_1$. But this
truncation is performed {\em after} a $\gppp$ Weyl transformation $s_{\alpha_1}$
which
transmutes the time index 1 to a space index (see Section \ref{sec:tempinvolutintro}). This gives an action $\mathcal S_{\gpp_B}$ which is
formally identical to the one given by (\ref{dualpap:fullp})
 but with a Lorentz signature for the
metric, which in the flat coordinates amounts to a negative sign for the
Lorentz metric component $\eta_{22}$, and with $\xi$ identified to the missing
space coordinate  instead of $t$.

The $\mathcal S_{\gpp_B}$ action is thus characterised by a signature $(1,D-2,+)$
where  the sign $+$ means that (\ref{dualpap:signa}) is fulfilled for all the simple positive step operators
implying that all the kinetic energy terms in the action are positive.
The theory describes a motion on the coset
$\gpp/\mathrm K^{*++}$ where  $\mathfrak{k}^{*++}$ is the subalgebra of $\agpp$
invariant under the  time involution $\Omega_2$ defined as in
(\ref{tabeq:inomegk}) with 2 as the time coordinate.
As this involution will not generically commute with Weyl reflections, the same coset can be
described by actions $\mathcal S_{{\mathrm G}{}_{(i{}_1i{}_2\dots i{}_t)}^{++}}^{(t,s,\varepsilon)}$ where
the global signature is $(t,s,\varepsilon)$ with
$\varepsilon$ denoting  a set of  signs, one for each
simple step operator which does not belong to the gravity line, defined by (\ref{eqn:signtheomega}) and $i{}_1i{}_2\dots i{}_t$ are the time indices. The equivalence of the different actions has been
shown by deriving differential equations relating the fields parametrising the different coset
representatives \cite{Englert:2004ph}. \\

The 'brane' action admits exact solutions identical to intersecting extremal brane solutions of the corresponding maximally oxidised theory smeared in
all direction but one \cite{Englert:2004ph}. These solutions  provide a laboratory to understand the significance of the higher level fields and to check whether or not the Kac-Moody theory can
described uncompactified theories \cite{Damour:2002cu,Englert:2004ph,Damour:2005zb}. Furthermore the intersection rules  are neatly encoded in the $\mathrm G^{++}$ algebra through orthogonality conditions between the real positive roots corresponding to the branes in the configuration \cite{Argurio:1997gt,Englert:2004it}.\\

\begin{table}[top] 
\centering
\scalebox{.83}{
\begin{tabular}{cccc} 
&\emph{Fields belonging  to $\mathcal{S}_{\mathZ{E}_{10}/\mathrm K^{*++}}$}  & & \emph{Branes of M-theory} \\
&& &  \\
level 0& $g_{\hat{\mu} \hat{\nu}}(x)$&  $\leftrightsquigarrow$  & $KK$-wave (0-brane) \\
level 1&  $A_{\hat{\mu}\hat{\nu} \hat{\rho}}(x)$ + $g_{\hat{\mu} \hat{\nu}}(x)$&  $\leftrightsquigarrow$  & $M2$ (2-brane) \\
level 2&  $A_{\hat{\mu}_1...\hat{\mu}_6}(x)$ + $g_{\hat{\mu} \hat{\nu}}(x)$&  $\leftrightsquigarrow$  & $M5$ (5-brane) \\ 
level 3&  $A_{\hat{\mu}_1...\hat{\mu}_8,\hat{\nu}}(x)$ + $g_{\hat{\mu} \hat{\nu}}(x)$ &  $\leftrightsquigarrow$  & $KK6$-monopole 
\end{tabular}
}
\caption{Link between the lowest level up to level $3$ of the brane $\mathZ{E}_{10}$-invariant theory and the branes of M-theory.} \label{tabtab:brane}
\end{table}
The $\sigma$-model $\mathZ{E}_{10}/ \mathrm K^{*++}$ provides a natural framework for studying static solutions. It yields for the lowest levels, all the basic BPS solutions of eleven-dimensional supergravity \cite{Englert:2003py}, namely the KK-wave, the M2 brane, the M5 brane and the KK6-monopole, smeared in all space dimensions but one, as well as their exotic counterparts (see Chapter \ref{chap:basicbps} and Table \ref{tabtab:brane} ). We can again ask the question of the interpretation of the higher level fields. In Part II, we elucidate this interrogation for the real roots of $E_{10}$ belonging to the regular affine subalgebra $E_9 \subset E_{10}$. We construct an infinite $E_9$ multiplet of BPS states and for each positive real root of $E_9$ whose Dynkin labels are in the form \eqref{eqn:dynkinlabel475},
  we obtain a BPS solutions of eleven dimensional supergravity, or of its exotic counterparts, depending on two non-compact transverse variables \cite{Englert:2007qb}. All these solutions are related by U-dualities realised via $E_9$ Weyl transformations.

\part{An $E_9$ multiplet of BPS states}\label{parttwo}
 \setcounter{equation}{0}
\markboth{Introduction of Part II}{Introduction of Part II}
\vspace{6cm}
Elucidating the role of the huge number of  $E_{10}$ and $E_{11}$
generators is an important problem in the Kac--Moody approach to
M-theory. In this part of the thesis we focus on the  real roots of $E_{10}$, and
in particular to those belonging to its regular affine subalgebra
$E_9$.  We consider fields parametrising  the Borel representatives of
the coset space $ \mathZ{E}_{10}/\mathrm{K}^{*++} \equiv \mathZ{E}_{10}/\mathrm{K}_{10}^-$ (see Section \ref{subsec:branemodelth}), that is the Cartan and the positive 
generators of $E_{10}$. The algebra $E_{10}$ is taken to be embedded regularly
in $E_{11}$. The Dynkin diagrams of $E_{10}$ and $E_9$ are obtained
from the Dynkin diagram of Figure \ref{ffirst} by deleting successively the nodes 1
and~2. We find that each positive real root determines a BPS state in
space-time\footnote{This is in line with the analysis of
  \cite{West:2004st} where BPS states are associated with $E_{11}$
  roots.}, where a BPS state is defined by the no-force condition
allowing superposition of configurations centred at different
space-time points. We obtain explicitly an  infinite $E_9$ multiplet
of BPS static solutions of eleven-dimensional  supergravity depending on two
non-compact  space variables. They are all related  by U-dualities
realised  by the $E_9$ Weyl transformations. 

\begin{figure}[h]
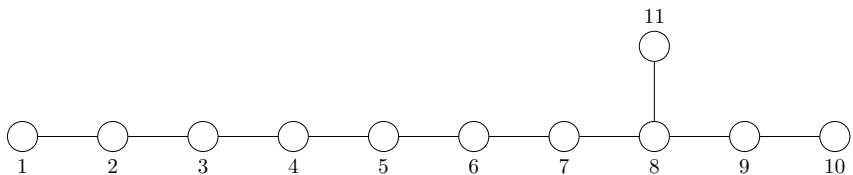
 
\begin{center}
\scalebox{.8}{
\begin{pgfpicture}{15cm}{-0.5cm}{1cm}{2.5cm}
\pgfnodecircle{Node1}[stroke]{\pgfxy(1,0.5)}{0.25cm}
\pgfnodecircle{Node2}[stroke]
{\pgfrelative{\pgfxy(1.5,0)}{\pgfnodecenter{Node1}}}{0.25cm}
\pgfnodecircle{Node3}[stroke]
{\pgfrelative{\pgfxy(1.5,0)}{\pgfnodecenter{Node2}}}{0.25cm}
\pgfnodecircle{Node4}[stroke]
{\pgfrelative{\pgfxy(1.5,0)}{\pgfnodecenter{Node3}}}{0.25cm}
\pgfnodecircle{Node5}[stroke]
{\pgfrelative{\pgfxy(1.5,0)}{\pgfnodecenter{Node4}}}{0.25cm}
\pgfnodecircle{Node6}[stroke]
{\pgfrelative{\pgfxy(1.5,0)}{\pgfnodecenter{Node5}}}{0.25cm}
\pgfnodecircle{Node7}[stroke]
{\pgfrelative{\pgfxy(1.5,0)}{\pgfnodecenter{Node6}}}{0.25cm}
\pgfnodecircle{Node8}[stroke]
{\pgfrelative{\pgfxy(1.5,0)}{\pgfnodecenter{Node7}}}{0.25cm}
\pgfnodecircle{Node9}[stroke]
{\pgfrelative{\pgfxy(1.5,0)}{\pgfnodecenter{Node8}}}{0.25cm}
\pgfnodecircle{Node10}[stroke]
{\pgfrelative{\pgfxy(1.5,0)}{\pgfnodecenter{Node9}}}{0.25cm}
 \pgfnodecircle{Node11}[stroke]
{\pgfrelative{\pgfxy(0,1.5)}{\pgfnodecenter{Node8}}}{0.25cm}

\pgfnodebox{Node12}[virtual]{\pgfxy(1,0)}{$1$}{2pt}{2pt}
\pgfnodebox{Node13}[virtual]{\pgfxy(2.5,0)}{$2$}{2pt}{2pt}
\pgfnodebox{Node14}[virtual]{\pgfxy(4,0)}{$3$}{2pt}{2pt}
\pgfnodebox{Node15}[virtual]{\pgfxy(5.5,0)}{$4$}{2pt}{2pt}
\pgfnodebox{Node16}[virtual]{\pgfxy(7,0)}{$5$}{2pt}{2pt}
\pgfnodebox{Node17}[virtual]{\pgfxy(8.5,0)}{$6$}{2pt}{2pt}
\pgfnodebox{Node18}[virtual]{\pgfxy(10,0)}{$7$}{2pt}{2pt}
\pgfnodebox{Node19}[virtual]{\pgfxy(11.5,0)}{$8$}{2pt}{2pt}
\pgfnodebox{Node20}[virtual]{\pgfxy(13,0)}{$9$}{2pt}{2pt}
\pgfnodebox{Node21}[virtual]{\pgfxy(14.5,0)}{$10$}{2pt}{2pt}
\pgfnodebox{Node22}[virtual]{\pgfxy(11.5,2.5)}{$11$}{2pt}{2pt}

\pgfnodeconnline{Node1}{Node2} \pgfnodeconnline{Node2}{Node3}
\pgfnodeconnline{Node3}{Node4}\pgfnodeconnline{Node4}{Node5}
\pgfnodeconnline{Node5}{Node6} \pgfnodeconnline{Node6}{Node7}
\pgfnodeconnline{Node7}{Node8}\pgfnodeconnline{Node8}{Node9}
\pgfnodeconnline{Node9}{Node10} \pgfnodeconnline{Node8}{Node11}

\end{pgfpicture}  }
 \caption { \sl \small The Dynkin diagram of $E_{11}$ and its gravity
 line. The roots 3 and 2 extend $E_8$ to $E_{10}$ and the additional
 root 1 to $E_{11}$. The Dynkin diagram of the  $A_{10}$  subalgebra
 of $E_{11}$,  represented in the figure by the horizontal line is its
 `gravity line'.} 
 \label{ffirst}
\end{center}
\end{figure} 

An obvious question at this point concerns the relation between our results
and the conjecture of~\cite{Hull:1994ys}, according to which the BPS solitons
of the toroidally compactified theory should transform under the arithmetic
group $ \mathZ{E}_9({\mathbb Z})$. Like~\cite{Obers:1998rn}, we here consider
only the Weyl group of $\mathZ{E}_9$, realized as a subgroup of $\mathZ{E}_9({\mathbb
  Z})$. More specifically, 
we consider the Weyl groups of various affine $A_1^+ \equiv A_1^{(1)}$
subgroups of $\mathZ{E}_9$. Each one of these is then found to act via inversions
and integer shifts on a certain analytic function characterising the given
BPS solution, cf. (\ref{Moebact}) ---  much like the modular group
$\mathrm{SL}(2,{\mathbb Z})$ (see Appendix \ref{affapp} for details on the embedding of
the Weyl group into the Kac--Moody group itself). There is also an
action on the conformal factor, which is, however, more complicated
and cannot be interpreted in this simple way (and which can be locally
undone by a conformal coordinate transformation). The remaining Weyl
reflections `outside' the given $A_1^+$, and associated to the simple roots of the $E_9$, are realized as permutations.  However, we should like
to stress that the full $\mathZ{E}_9({\mathbb Z})$ contains many more transformations
than those considered here.

We show that the BPS states we find admit an equivalent description as
solutions of effective actions, where the infinite set of $E_9$ fields
in the Borel representative of $\mathZ{E}_{10}/\mathrm{K}_{10}^-$ play the role of
matter fields in a `dual' metric. We label the description in terms of
eleven-dimensional  supergravity fields the `direct' one and the description in terms
of Borel fields the `dual' one. Comparing the direct with the dual
description sheds some light on the significance of the $E_9$ real
roots. We see indeed that the Borel fields corresponding to these
roots are related to the supergravity fields by an infinite set of
`non-closing dualities' generalizing the Hodge duality.  This is a
feature which is well-known in the context of the Geroch symmetry of
standard $D=4$ gravity reduced to two
dimensions~\cite{Geroch:1970nt,Geroch:1972yt,Breitenlohner:1986um,Nicolai:2001cc}. 
There one also defines an infinite set of dual potentials from the
standard fields and uses the infinite Geroch group to generate new
solutions. The Geroch group is the affine $A_1^+$ extension of $\mathrm{SL}(2)$
and we will make use of various subgroups of $\mathZ{E}_9\subset \mathZ{E}_{10}$
isomorphic to this basic affine group. 

One of the original motivations for the present work was to get a
better understanding of the significance of the higher level fields
in $E_{10}$ and maybe expand the known dictionary beyond level three (rather
height $29$) (see Section  \ref{subsec:branemodelth}). Since the role of some of these higher levels is
understood in the two-dimensional $E_9$ context in terms of the
generalised dual potentials one can 
anticipate that they will play a similar role in 
the one-dimensional brane $\s$-model $\mathcal{S}^{brane}$ and cosmological $\s$-model  $\mathcal{S}^{cosmo}$. Indeed
this is what we find in the BPS case but we
have not obtained an analytic form for these non-closing dualitites in
the form of an extended dictionary.
Still it is clear that their definition is not restricted to BPS
states.   

Borel fields attached to roots of $E_{10}$ that are not roots of $E_9$
define BPS states depending on one non-compact space variable. These
may not admit a direct description in terms of eleven-dimensional  supergravity fields
but their dual description is well defined. The explicit BPS solution
attached to a level 4 root of $E_{10}$ is obtained, in agreement with
reference \cite{West:2004st}, and describes the M9 solution in 11
dimensions which is the `uplifting' of the D8-brane of Type IIA string
theory. 

\noindent The results presented in this part have been published in the reference \cite{Englert:2007qb}.

  \chapter{ Basic BPS states in $E_{10} \subset E_{11}$}\label{chap:basicbps}

In this chapter, we review the construction of the $\sigma$-model
$\mathcal{S}^{brane}$ and its relation to the basic BPS solutions of eleven-dimensional 
supergravity, namely the KK-wave, the M2 brane, the M5 brane and the KK6-monopole, smeared in all space dimensions but one, as well as their exotic counterparts.

\setcounter{equation}{0}
\section{From $E_{11}$ to $E_{10}$ and the coset space $\mathZ{E}_{10}/\mathrm{K}^-_{10}$}\label{sec:frome11toe10}
We recall that the Kac-Moody algebra $E_{11}$ is
entirely defined by the commutation relations of its Chevalley generators and
by  the Serre relations (see Chapter \ref{chap:math}). The Chevalley presentation
consists of the generators $e_m,f_m$ and $h_m, \ m=1,2,\dots ,11$~with
commutation relations 
\begin{eqnarray}  \label{chevalley}\begin{split}
   [h_m,h_n]&=0~~, ~~&&[h_m,e_n]=A_{mn}e_n \, ,  \\
[e_m,f_n]&=\delta_{mn} h_m ~~ ,~~&& [h_m,f_n]=-A_{mn} f_n\, ,
\end{split}\end{eqnarray}
where $A_{mn}$ is the Cartan matrix which can be
expressed in terms of scalar products of the simple roots $\alpha_m$
as  
\begin{equation}
\label{cartan}  
A_{mn}=2\frac{(\alpha_m|\alpha_n)}{(\alpha_m|
\alpha_m )}\, \cvp
\end{equation}
The Cartan subalgebra is generated by $h_m$, while  the positive
(negative) step operators are  the  $e_m$ ($f_n$) and their
multi-commutators, subject to the Serre relations 
\begin{equation}
\label{serre} [e_m,[e_m,\dots,[e_m,e_n]\dots]]=0,\quad
[f_m,[f_m,\dots,[f_m,f_n]\dots]]=0\, ,
\end{equation} 
where the number of $e_m$ ($f_m$) acting on $e_n$     ($f_n$) is
given by $1-A_{mn}$. The Cartan matrix of $E_{11}$ is encoded in its
Dynkin diagram depicted in  Figure \ref{ffirst}. Erasing the node 1 defines the
regular embedding of its $E_{10}$ hyperbolic subalgebra and erasing
the nodes 1 and 2 yields the regular embedding of the affine
$E_9$.

$E_{11}$ contains a  subalgebra $\mathfrak{gl}(11)$ such that $\mathfrak{sl}(11)\equiv A_{10}
\subset \mathfrak{gl}(11)\subset
E_{11}$. The generators of the $\mathfrak{gl}(11)$ subalgebra are taken to be
$K^a_{~b}\ (a,b=1,2,\ldots ,11)$   with commutation relations
\begin{equation}
\label{Kcom} [K^a_{~b},K^c_{~d}]   =\delta^c_b
K^a_{~d}-\delta^a_dK^c_{~b}\,  .
\end{equation}
The relation  
between the
commuting generators $K^a{}_a$ of 
$\mathfrak{gl}(11)$ and   the Cartan generators $h_m$ of $E_{11}$ in the   
Chevalley basis
follows from comparing the commutation relations (\ref{chevalley})  
and
(\ref{Kcom}) and from the identification of the simple roots of
$E_{11}$. These are
$e_m=\delta_m^{a} K^a{}_{a+1},\ m=1,\ldots ,10$ and  $e_{11}= R^{\, 9\,  
10\,11}$
where
$R^{abc}$ is a generator in
$E_{11}$ that is a third rank anti-symmetric tensor under $A_{10}$.   
One gets
\begin{eqnarray}
\label{aa} h_m&=&\delta_m^{a}(K^a{}_a-K^{a+1}{}_{a+1})~~~~  
m=1,\dots,10\\
\label{el} h_{11}&=&-\frac{1}{3}(K^1{}_1+\ldots +K^8{}_8) +\frac{2} 
{3}(K^9{}_9+
K^{10}{}_{10}+K^{11}{}_{11})\,.
\end{eqnarray}
The positive (negative) step operators in the $A_{10}$ subalgebra are
the $K^a{}_b$ with $b>a$ ($b<a$). The adjoint representation of
$E_{11}$ can be written as an infinite direct sum of representations
of the $\mathfrak{gl}(11)$ generated by the $K^a{}_b$. This is known as the
$A_{10}$ level decomposition of
$E_{11}$ (see Section \ref{subsec:leveldece11}).
The $K^a{}_b$ define the level zero positive
(negative) step  operators.  The  positive (negative) level $\ell$  step
operators are defined by the number  of times the root $\alpha_{11}$
appears in the decomposition of the adjoint representation of $E_{11}$  into  
irreducible representations of  $A_{10}$.  At level 1, one has the
single representation spanned by the anti-symmetric tensor
$R^{abc}$. At each level the number of irreducible representations of
$A_{10}$ is finite and the symmetry properties of the irreducible
tensors $R^{c_1\dots c_r}$ ($ R_{c_1\dots c_r}$) are fixed by the
Young tableaux of the representations. 
In what follows, positive level step operators will always be denoted
with upper indices and negative level ones with lower ones. For
positive level $\ell$ the number of indices on a generator $R^{c_1\dots
  c_r}$ is $r=3\ell$.

The Borel group formed by the Cartan generators and the positive level
0 generators can be taken as representative of the coset space
$\mathrm{GL}(11)/ \mathrm{SO}(1,10)$ and hence the parameters of the Borel group can be
used to define in a particular gauge the 11-dimensional metric
$g_{\mu\nu}$  which spans this coset space at a given space-time
point. We recall that the subgroup 
$\mathrm{SO}(1,10)$ of $\mathrm{GL}(11)$ is the subgroup invariant under a temporal
involution $\Omega_1$ (defined in Section \ref{PERTABsec:tempinv})  which generalizes the Chevalley involution by
allowing the identification of the tensor index  1 to be the time
index. This suggests to define in general all $E_{11}$ fields as parameters of the coset space $\mathZ{E}_{11}/\mathrm K^{*+++} \equiv \mathZ{E}_{11}/\mathrm K^-_{11} $ where $ K^-_{11}$ is invariant under the temporal involution $\Omega_1$.

For the regular embedding of $E_{10}$ in $E_{11}$ obtained by deleting
the node 1 in Figure \ref{ffirst}, the description in terms of $\mathfrak{gl}(10)$ follows from
the description of $E_{11}$ in terms of $\mathfrak{gl}(11)$. Omitting the
generators $K^1{}_2,K^2{}_1$ and $K^1{}_1$, the relations
(\ref{Kcom}), (\ref{aa}) remain valid and  (\ref{el}) becomes  
\begin{equation}
\label{el10}
 h_{11}\to -\frac{1} {3}(K^2{}_2+\ldots +K^8{}_8) +\frac{2}
{3}(K^9{}_9+
K^{10}{}_{10}+K^{11}{}_{11})\,.
\end{equation}
The temporal involution $\Omega_i$, $\varepsilon_i=+1$ for
$a=2,3\dots,11$ reduces to the Chevalley involution acting on the
$E_{10}$ generators.  It leaves invariant a subalgebra 
$\mathfrak{k}^+_{10}$ of $E_{10}$. The Borel representative of the coset space
$\mathZ{E}_{10}/ \mathrm{K}^+_{10}$ is now parametrized by $A_9$ tensor fields in the
Euclidean metric $\mathrm{GL}(10)/ \mathrm{SO}(10)$.  Taking these fields to be
functions of the remaining time coordinate~1, one can built a
$\sigma$-model on this coset space (see Section \ref{subsec:cosmomodelth}).  This model has been used mainly
to study cosmological solutions
\cite{Damour:2002et,Kleinschmidt:2005gz} and we shall label the action
of this $\sigma$-model~\cite{Damour:2002cu} as $\mathcal{S}^{cosmo}$.  
 

We could of course have chosen 2 instead of 1 as time coordinate in
$\mathfrak{gl}(11) \subset E_{11}$.  This change of time coordinate can be
obtained by performing  the $E_{11}$ Weyl reflexion $ s_{\alpha_1}$
sending $\alpha_1\to -\alpha_1$ and $\alpha_2 \to \alpha_1 +\alpha_2$.
Choosing the gravity line of Figure \ref{ffirst} to be the reflected one, one finds
that its time coordinate has switched from 1 to  2 (see Section \ref{sec:tempinvolutintro}).  This results from
the fact that the temporal involution does not in general commute
with Weyl reflexions
\cite{Keurentjes:2004bv,Keurentjes:2005jw,Englert:2004ph}, a property
that has far reaching consequences, as reviewed in
Section~\ref{sec:tempinvolutintro}
. Deleting the node 1 in Figure \ref{ffirst} we obtain the Dynkin
diagram of $E_{10}$ endowed with the temporal involution
$\Omega_{2}$. This involution
leaves invariant a subalgebra $\mathfrak{k}_{10}^-$ of $E_{10}$.  The coset space
$\mathZ{E}_{10}/\mathrm{K}^-_{10}$ accommodates the Lorentzian metric  $\mathrm{GL}(10)/
\mathrm{SO}(1,9)$. Performing  products of $E_{10}$ Weyl reflexions on the
gravity line $s_{\alpha_i},\, i=2,\dots, 10$, one obtains ten possible
different identifications of the time coordinate from
$\Omega_{i}, \, i=2,3\dots ,11$ (see Section \ref{sec:tempinvolutintro}).   The
$\sigma$-model build upon the coset $\mathZ{E}_{10}/\mathrm{K}^-_{10}$  can be
constructed for any choice of $i$ in $\Omega_{i}$.
These formulations of the $\sigma$-model are all equivalent up to the
field redefinitions by $E_{10}$ Weyl transformations and we shall
label them by the generic notation $\mathcal{S}^{brane}$, leaving implicit the
choice of the time coordinate $i$. 

For sake of completeness, we recall the construction of $\mathcal{S}^{brane}$ which was performed in Section \ref{subsec:branemodelth}. We take as representatives of $\mathZ{E}_{10}/\mathrm{K}^-_{10}$ the elements of
the Borel group of $\mathZ{E}_{10}$ which we write
as\footnote{\label{borelfn}We recall that this type of Borel gauge for the coset
  space is not always accessible since the denominator $\mathrm{K}^-_{10}$ is
  not the compact subgroup of the split $\mathZ{E}_{10}$ and therefore the
  Iwasawa decomposition theorem fails. A simple finite-dimensional
  example is the coset space $\mathrm{SL}(2)/\mathrm{SO}(1,1)$ which will also play a
  role below. In this space  one cannot find  upper triangular
  representatives for matrices of the form
  \begin{equation*}
  {\cal V} = \left[\begin{array}{cc}a&b\\a&c\end{array}\right]
  \quad\quad\quad \text{with}\quad\big[a(c-b)=1\big],
  \end{equation*}
  since the lightlike first column vector cannot be Lorentz boosted
  to a spacelike one.}
\begin{equation}
\label{positive} {\cal V}(x^1)= \exp \left[\sum_{a\ge b}
h_b^{~a}(x^1)K^b_{~a}\right]\exp \left[\sum
\frac{1}{r!} A_{ a_1\dots a_r}(x^1) R^{
a_1\dots   a_r} +\cdots\right]\, ,
\end{equation}
where from now on all indices run from 2 to 11. The first exponential
contains only  level zero  operators  and the second one the positive
step operators of $E_{10}$ of levels strictly greater 
than zero. Following Section \ref{app:SigmaModel}, we define
\begin{equation}
\label{sym}
\mathcal{P}(x^1)=\frac{1}{2} \bigg(\frac{\dd {\cal V}}{\dd x^1} {\cal V}^{-1} - \Omega_{\lambda}   \big(   \frac{\dd {\cal V}}{\dd x^1} {\cal V}^{-1}   \big )        \bigg)\, ,
\end{equation}
with $\lambda$ equal to the  chosen time coordinate. 
Using the invariant scalar product $(.|.)$ for
$E_{10}$ one obtains a $\sigma$-model constructed on the coset
$\mathZ{E}_{10}/\mathrm{K}^-_{10}$  
\begin{equation}
\label{action} \mathcal{S}^{brane}=\int \dd x^1 \frac{1}{n(x^1)} (\mathcal{P}(x^1) | \mathcal{P}(x^1))\, ,
\end{equation} where
$n(x^1)$ is an arbitrary lapse function ensuring reparametrisation
invariance on the world-line. Explicitly, defining 
\begin{equation}
\label{vielbein}
e_\mu^{~m}=(e^{-h})_\mu^{~m} \qquad g_{\mu\nu}
=e_\mu^{~m}e_\nu^{~n}\eta_{mn} \, , 
\end{equation}
where $\eta_{mn}$ is the Lorentz metric with
$\eta_{\lambda\lambda}=-1$, one writes (\ref{action}) as 
\begin{equation}
\label{full} \mathcal{S}^{brane} =\mathcal S_0+\sum_A
\mathcal S_A\, ,
\end{equation} with
\begin{subequations} \begin{align}
 \mathcal S_0 &=\frac{1}{4}\int \dd x^1
\frac{1}{n(x^1)}(g^{\mu\nu}g^{\sigma\tau}-
g^{\mu\sigma}g^{\nu\tau})\frac{dg_{\mu\sigma}}{\dd x^1}
\frac{dg_{\nu\tau}}{\dd x^1}, \label{fullzero}\\
\mathcal S_A&=\frac{1}{2 r!}\int \dd x^1
\frac{ 1}{n(x^1)}\left[
\frac{DA_{\mu_1\dots \mu_r}}{\dd x^1} g^{\mu_1{\mu}^\prime_1}...\,
g^{\mu_r{\mu}^\prime_r}
\frac{DA_{{\mu}^\prime_1\dots {\mu}^\prime_r}}{\dd x^1}\right]. \label{fulla}
\end{align}\end{subequations} 
Here $D/\dd x^1$ is a group covariant derivative and all indices run from
2 to 11. Note that in (\ref{vielbein}) one may extend the range of
indices to include $\mu,\nu,m,n =1$, using the embedding relation
$E_{10}$ in $E_{11}$ which reads~\cite{Englert:2003zs} 
\begin{equation}
\label{embedding}
h_1^{~1}=\sum_{a=2}^{11} h_a^{~a}\qquad, \qquad  h_1^{~a}=0\qquad
a=2,3,\dots 11\, . 
\end{equation}

Up to level 3 and height 29, the fields in (\ref{action}) can be identified with fields of eleven-dimensional 
supergravity~\cite{Damour:2002cu}. They
can be used, as shall now be recalled, to characterise its basic BPS
solutions depending on one non-compact space variable. 
\setcounter{equation}{0}
 \section{The basic BPS solutions in 1 non-compact dimension}\label{sec:basiconedim}
\subsection{Generalities and Hodge duality}
The basic BPS solutions of eleven-dimensional  supergravity are the 2-brane (M2) and
its magnetic counterpart the 5-brane (M5), and in the pure gravity
sector the Kaluza--Klein wave (KK-wave) whose magnetic counterpart is
Kaluza-Klein monopole (KK6-monopole). These are static solutions
which, wrapped on tori, leave respectively 8, 5, 9 and 3 non-compact
space dimensions.  
It is convenient to express the magnetic solutions in terms of an
`electric' potential with a time index. This is done  by trading {\em
  on the equations of motion} the field strength for its Hodge
dual. For the M5 the field
$F_{\mu{}_1\mu{}_2\mu{}_3\mu{}_4}=4\,\partial_{[\mu{}_1}
  A_{\mu{}_2\mu{}_3\mu{}_4]}$ has as dual the field 
 $\widetilde
F_{\nu{}_1\nu{}_2\nu{}_3\nu{}_4\nu{}_5\nu{}_6\nu{}_7}=7\,\partial_{[\nu{}_1}
  A_{\nu{}_2\nu{}_3\nu{}_4\nu{}_5\nu{}_6\nu{}_7]}$ defined by 
\begin{equation}
\label{Pdual1}
\sqrt{- g}\widetilde
F^{\nu{}_1\nu{}_2\nu{}_3\nu{}_4\nu{}_5\nu{}_6\nu{}_7}=\frac{1}{4!}\epsilon^{\nu{}_1\nu{}_2\nu{}_3\nu{}_4\nu{}_5\nu{}_6\nu{}_7\mu{}_1\mu{}_2\mu{}_3\mu{}_4}
F_{\mu{}_1\mu{}_2\mu{}_3\mu{}_4}\, , 
\end{equation}
For the KK6-monopole the KK-potential $A_\mu^{(\nu)}$ in terms of the
vielbein $e_\mu{}^n$ is given by  $A_\mu^{(\nu)} = -e_\mu{}^n 
(e^{-1})_n{}^\nu$ where $\mu$ labels the non-compact directions,
$\nu$ is the Taub-NUT direction in coordinate indices and $n$ is
the Taub-NUT direction in flat frame indices and there is no summation
on $n$. The  field strength is
$F_{\mu{}_1\mu{}_2}^{(\nu)}=\,\partial_{\mu{}_1} A_{\mu{}_2}^{(\nu)}
-\partial_{\mu{}_2} A_{\mu{}_1}^{(\nu)}$. Its dual is $\widetilde
F_{\nu{}_1\nu{}_2\nu{}_3\nu{}_4\nu{}_5\nu{}_6\nu{}_7\nu{}_8
  \nu{}_9\vert\nu{}_9}=9\,
\partial_{[\nu{}_1}A_{\nu{}_2\nu{}_3\nu{}_4\nu{}_5\nu{}_6\nu{}_7\nu{}_8
    \nu{}_9]\vert\nu{}_9}$ where 
\begin{equation}
\label{Pdual2}
\sqrt{- g }\widetilde
F^{\nu{}_1\nu{}_2\nu{}_3\nu{}_4\nu{}_5\nu{}_6\nu{}_7\nu{}_8
  \nu{}_9\vert\nu{}_9}=\frac{1}{2}\epsilon^{\nu{}_1\nu{}_2\nu{}_3\nu{}_4\nu{}_5\nu{}_6\nu{}_7\nu{}_8 \nu{}_9 \mu{}_1\mu{}_2} F_{\mu{}_1\mu{}_2}^{(\nu{}_9)}\, .   
\end{equation}
These BPS solutions depend on space variables in the non-compact
dimensions only. One may further compactify on tori some of these
`transverse' directions. This increases accordingly the number of
Killing vectors and one obtains in this way new `smeared' solutions
depending only on the space variables in the remaining non-compact
dimensions\footnote{\label{foot:smearing} In the string language, the smearing process
  amounts to introducing image branes in the compact dimensions and
  averaging them over the torus radii (or equivalently considering in
  the non-compact dimensions distances large compared to these
  radii). Compact dimensions which cannot be `unsmeared'  are labelled
  `longitudinal'. It will be convenient  in what follows to take the
  time dimension as compact (and longitudinal). Decompactifying
  longitudinal space-time dimensions does not affect the field
  dependence of the solutions. However longitudinal dimensions cannot
  always be decompactified, as exemplified by the Taub-NUT direction
  of the KK6-monopole. This feature will be studied in detail in
  Section~\ref{chargesec}.}. 
As we shall see the smearing process is straightforward, except for
the smearing of magnetic solutions to one non-compact dimensions,
which can only be performed in the electric language, hinting on the
fundamental significance of the dual formulation as will indeed be
later confirmed.  

As recalled below, all BPS solutions, smeared up to one non-compact
dimensions are solutions of the    $\sigma$-model  $\mathcal{S}^{brane}$ given
by (\ref{full}) where $x^1$ is identified with the non-compact
space dimension \cite{Englert:2003py,Englert:2004ph}. 
\subsection{Levels 1 and 2: The 2-brane and the 5-brane}

For the 2-brane (M2) solution of eleven-dimensional  supergravity wrapped in the
directions 10 and 11, we choose 9 as the time coordinate, so that the
only non vanishing component  of the 3-form potential is
$A_{9\,10\,11}$. For the 5-brane (M5) wrapped in the directions
4,5,6,7,8, we choose 3 as time coordinate so that the 3-form potential
is still $A_{9\,10\,11}$ and the Hodge dual (\ref{Pdual1}) is
$A_{3\,4\,5\,6\,7\,8}$. One gets for these BPS solutions the following
metric and fields  
\be  \label{M2metric} \begin{split}
{\rm M2} \qquad: \qquad g_{11}&=g_{22}=H^{1/3}\,,\\
 g_{33}&=g_{44}=\dots
 =g_{88}=H^{1/3}\,,\\
  -g_{99}&=g_{10\,10}=g_{11\,11}=H^{-2/3}\,,\\ 
A_{9\,10\, 11} &=\frac{1}{H}\, ,\\
\end{split} \ee
\be  \label{M5metric} \begin{split}
{\rm M5}\qquad:\qquad g_{11}&=g_{22}=H^{2/3}\,,\\
-g_{33}&=g_{44}=\dots
 =g_{88}=H^{-1/3}\,,\\
  g_{99}&=g_{10\,10}=g_{11\,11}=
 H^{2/3}\,,\\ 
A_{3\,4\, 5\,6\,7\,8} &=\frac{1}{H}\, .
\end{split}\ee

Here $H$ is a harmonic function of the non-compact space dimensions
[(1,2,3,4,5,6,7,8) for M2 and  (1,2,9,10,11) for M5] with
$\delta$-function singularities at the location of the
branes. Smearing simply reduces the number of variables in the
harmonic functions to those labelling the remaining non-compact
dimensions. For instance a single M2 (M5) brane smeared to two
non-compact  dimensions, located at the origin of the coordinates
$(x^1,x^2)$, is described by $H= (q /2\pi)\ln r =(q
/2\pi)\ln\sqrt{(x^1)^2 +(x^2)^2}$ where $q$ is an electric (magnetic)
charge density. When smeared to one non-compact dimension one gets $H=
(q /2)\, \vert x^1\vert$. We see that in the one-dimensional case only
the electric dual description of the magnetic M5 brane is available,
as the duality relating the 6-form (\ref{Pdual1}) to the 3-form
supergravity potential   requires at least two transverse
dimensions. Note that this one-dimensional solution is still a
solution of eleven-dimensional  supergravity because the replacement in the equations
of motion  of the 4-form field strength by the 7-form dual is valid as
long as the Chern-Simons term contributions vanish, as is indeed the
case for the above brane solutions.  Actually the electric description
of the M5 (\ref{M5metric}), is a solution of the following
effective action (in any number of transverse non-compact dimensions)  
\begin{eqnarray}
\label{M5dual}
{\cal S}^{(11)}_{M5} =\frac{1}{16\pi G_{11}}\int \dd ^{11}x
\sqrt{-g^{(11)}}\left[R^{(11)}-  \frac{1}{2\cdot7!  
}\widetilde F_{\mu_1\mu_2 \ldots \mu_7}\widetilde
  F^{\mu_1\mu_2\mu_3 \ldots \mu_7}\right] . 
\end{eqnarray}
 
We recall \cite{Englert:2003py} how the M2 solution of eleven-dimensional 
supergravity smeared over all dimensions but one can be obtained as a
solution of the  the $\sigma$-model  by putting in $\mathcal{S}^{brane}$
(\ref{full}) all non-Cartan fields to zero except  the level 1
3-form component $A_{9\,10\,11}$ with time in 9. Similarly the M5
solution smeared in all directions but one solves the  equations of
motion of this  $\sigma$-model  by retaining for the non-Cartan fields
only the component $A_{3\,4\,5\,6\,7\,8}$ of the level 2 6-form  with
time in 3. These are respectively parameters of the Borel generators  
\begin{eqnarray}
\label{3ge}
 R^{[3]}_1&\equiv&R^{9\,10\,11}\\
 \label{6ge}
 R^{[6]}_2&\equiv& R^{345678}\, ,
\end{eqnarray}
where the subscripts denote the $A_{9}$ level in the decomposition of
the adjoint representation of~$E_{10}$. As we will see in more detail
below, the roots corresponding to the elements $R^{[3]}_1$ and
$R^{[6]}_2$ have scalar product $-2$ and thus give rise to an
infinite-dimensional affine $A_1^+$ subalgebra.

These two solutions of the $\sigma$-model are characterised by Cartan fields
 $ h_a^{~a}, a=2,3,\dots,11$. 
One has\footnote{\label{intconsfn}For simplicity we have chosen zero for the
 integration constants in the solutions of the equations of motion of
 the  fields $A_{9\,10\, 11}$ and $A_{3\,4\, 5\,6\,7\,8}$.}
 \cite{Englert:2003py} 
\be \label{M2} \begin{split}
{\rm M2}:\qquad h_a^{~a} &=\left\{\frac{-1}{6}\,\Big|\,
\frac{-1}{6},\frac{-1}{6},\frac{-1}{6},\frac{-1}{6},\frac{-1}{6},\frac{-1}{6},\frac{-1}{6},\frac{1}{3},\frac{1}{3},\frac{1}{3}\right\}
  \ln H\,,\\
   h_a^{~b}&=0 ~\hbox {for} ~a\neq b\,  \\ 
h_a^{~a}K^a_{~a}&=\frac{1}{2}\ln H \,\cdot\,h_{11}\,,\\
  A_{9\,10\, 11} &=\frac{1}{H}\, ,
\end{split}\ee

\be\label{M5} \begin{split}
{\rm M5}: \qquad h_a^{~a}&=\left\{\frac{-1}{3}\,\Big|\,
\frac{-1}{3},\frac{1}{6},\frac{1}{6},\frac{1}{6},\frac{1}{6},\frac{1}{6},\frac{1}{6},\frac{-1}{3},\frac{-1}{3},\frac{-1}{3}\right\}\ln
H\,,\\
 h_a^{~b}&=0 ~\hbox {for} ~a\neq b, \\ 
h_a^{~a}K^a_{~a}&=\frac{1}{2}\ln H \,\cdot\,(-h_{11}-K^2{}_2\,)\,,\\
 A_{3\,4\, 5\,6\,7\,8} &=\frac{1}{H}\, ,
\end{split} \ee
where we used (\ref{el10}) and $H(x^1)$ is a harmonic function.  To
the left of the $\vert$ symbol in the first line of (\ref{M2}) and
(\ref{M5}), we have added the field $h_1^{~1}$, evaluated from the
embedding of $E_{10}$ in $E_{11}$ (\ref{embedding}). All other
quantities in these equations are defined in $E_{10}$. From
(\ref{vielbein}), the $\sigma$-model results (\ref{M2}) and
(\ref{M5}) are equivalent to the supergravity results
(\ref{M2metric}) and (\ref{M5metric}) for branes smeared to one
space dimension. 

\subsection{Levels 0 and 3: The Kaluza-Klein wave and the
  Kaluza-Klein monopole} \label{subsec:kkwkkm}

We now examine the BPS solutions involving only gravity. 

First consider  the KK-wave solution. The supergravity solution with
time in 3 and torus compactification in the 11 direction, is  
\be
 \label{G0metric1}
 \dd s^2= -H^{-1}(\dd x^3)^2+(\dd x^1)^2 +(\dd x^2)^2+(\dd x^4)^2+\dots +(\dd x^{10})^2+
 H[\dd x^{11}- A_{3}^{(11)}\dd x^3]^2\, , 
\ee
where the electric potential  $A_{3}^{(11)}$ is related to the harmonic
function $H$ in nine space dimensions with suitable source
$\delta$-function singularities by   
\begin{equation}
\label{potential}
A_{3}^{(11)} = (1/H)-1\, .
\end{equation}
Smearing over any number of space dimensions results in taking  $H$ as
a harmonic function of only the remaining non-compact space
variables. For non-compact space dimension $d>2$, the constant $-1$ in
(\ref{potential}) makes the potential vanish in the asymptotic
Minkowskian space-time if the limit of $H$ at spatial infinity is
chosen to be one. For $d=2$ or 1 space is not asymptotically flat and
we  keep for convenience the non vanishing constant in
(\ref{potential}) to be one. 

Smearing the KK-wave to  one non-compact dimension, the above
supergravity solution is recovered from the  $\sigma$-model
(\ref{full}) by putting to zero all fields parametrising the
positive roots in the Borel representative (\ref{positive})  except
the level 0 field $h_{3}^{~11}(x^1)$, taking 3 as the time
coordinate. To see this, it is convenient to rewrite
\begin{equation}
\label{wava}
{\cal V}_1= \exp \big[h_3{}^3(K^3{}_3-K^{11}{}_{11})+\,h_3{}^{11}K^3{}_{11} \big]\, ,
\end{equation}
by disentangling the Cartan generators and the
level zero positive step operators in two separate exponentials:
\begin{equation}
\label{wavb}
{\cal V}_2= \exp \big[h_3{}^3(K^3{}_3-K^{11}{}_{11})\big] \exp \big[A_{3}^{~(11)} \,K^3{}_{11}\big]\, .
\end{equation}
Therefore, we want to determine $A_{3}^{~(11)}$ in \eqref{wavb}. In terms of the $\mathrm{SL}(2)$ matrices defined in (\ref{repsl}), $  K^3{}_3-K^{11}{}_{11}=h_1$ and $  K^3{}_{11}=e_1$.  One has
\begin{equation}
\label{nu1}
{\cal V}_1=\sum_{n=0}^{\infty}  \frac{1}{n!} \left[\begin{array}{cc}h_3{}^3&h_3{}^{11}\\0&-h_3{}^3\\\end{array}\right]^n =\left[\begin{array}{cc}e_{11}{}^{11}&-e_3{}^{11}\\0&e_3{}^3\\\end{array}\right]
\end{equation}
\begin{equation}
\label{vielkk}
e_3{}^3=e^{-h_3{}^3} \qquad e_{11}{}^{11}=e^{h_3{}^3}\qquad e_3{}^{11}=-  \frac{h_3{}^{11}}{2  h_3{}^3}  \,  ( e^{h_3{}^3}-  e^{-h_3{}^3})\, ,
\end{equation}
where the $e_{\mu} {}^{n}$ are the vielbein (\ref{vielbein}) characterising the KK-solution (see also Appendix B of \cite{Englert:2003py}). On the other hand, from (\ref{vielkk}), one gets
\begin{equation}
\label{nu2}
{\cal V}_2=\left[\begin{array}{cc}e_{11}{}^{11}&e_{11}{}^{11} \, A_{3}^{~(11)}\\0&e_3{}^3\\\end{array}\right]\, .
\end{equation}
Equating ${\cal V}_1$  (\ref{nu1}) and ${\cal V}_2$ (\ref{nu2}), we get
\begin{equation}
\label{kka}
A_{3}^{~(11)}=-e_3{}^{11} (e_{11}{}^{11})^{-1}\, .
\end{equation}
and the metric corresponding to the representative (\ref{wava}) with time in 3  is thus
\begin{equation} \begin{split}
\dd s^2=&\quad(\dd x^1)^2+(\dd x^2)^2+(\dd x^4)^2+\dots +(\dd x^{10})^2 -(e_3{}^3)^2 (\dd x^3)^2\\
& \quad   + (e_{11}{}^{11})^2\left[ \dd x^{11} -A_{3}^{~(11)} \dd x^3 \right]^2\, .
  \label{ametric}
\end{split}\end{equation}
Taking into account the
embedding relation (\ref{embedding}), the solution can be written as
 \be \label{GO} \begin{split}
{\rm KKW} :  \qquad \ h_a^{\ a} &=\left\{0\Big|
0,\frac{1}{2},0,0,0,0,0,0,0,\frac{-1}{2}\right\}\ln H\,, \\ 
h_a^{~a}K^a_{\ a}&=\frac{1}{2}\ln H \,\cdot\,
(K^3{}_3-K^{11}{}_{11})\,,\\ 
A_{3}^{(11)} &=\frac{1}{H}-1\, ,
\end{split}
\ee
which, using (\ref{vielbein}) and (\ref{kka}), is
equivalent to the KK-wave solution (\ref{G0metric1}) and
(\ref{potential}) of general relativity. \\

Consider now the KK6-monopole solution. In 11 dimensions, it has 7
longitudinal dimensions (see footnote \ref{foot:smearing} page \pageref{foot:smearing}). Taking 11 as the Taub-NUT direction and 4 as
the timelike direction, the general relativity solution reads 
\begin{equation} \begin{split}
\label{G3metric1}
\dd s^2= &\  H\left[(\dd x^1)^2+(\dd x^2)^2
+(\dd x^3)^2\right] -(\dd x^4)^2+(\dd x^5)^2+\dots+(\dd x^{10})^2\\
& \quad +H^{-1}\left[ \dd x^{11}
  -\sum_{i=1}^3   A_i^{(11)} \dd x^i\right]^2 
\end{split}\end{equation}
and
\begin{equation}
\label{F}
F_{ij}^{(11)}\equiv \partial_i A_j^{(11)}-\partial_j
A_i^{(11)}=-\varepsilon_{ijk}\, \partial_k H\, , 
\end{equation}
where $H(x^1,x^2,x^3)$ is the harmonic function. It can be smeared to
2 spatial dimensions by taking the index $j$ in (\ref{F})  to label
a compact dimension, say 3,  
\begin{equation}
\label{F2}
\partial_i A_3^{(11)}=\varepsilon_{ik}\, \partial_k H\qquad i,k = 1,2
\end{equation}
and $H(x^1,x^2)$ is now  harmonic in two dimensions. The metric
(\ref{G3metric1}) becomes 
\begin{equation} \begin{split}
\label{G3metric2}
\dd s^2=& \ H\left[(\dd x^1)^2+(\dd x^2)^2+(\dd x^3)^2\right]
  -(\dd x^4)^2+(\dd x^5)^2+\dots+(\dd x^{10})^2\\
  &\quad +H^{-1}\left[ \dd x^{11} - A_3^{(11)}  \dd x^3\right]^2\, 
\end{split}\end{equation}
and is a solution of Einstein's equations.

The smearing to one space dimension is more subtle. As for the
magnetic 5-brane, it requires a dual formulation which in this case is
defined by the duality relation (\ref{Pdual2}). However the
reformulation of the supergravity action is now less
straightforward. To understand the dual formulation we first show how
to use it for the unsmeared KK6-monopole given by
(\ref{G3metric1}) and (\ref{F}).  
We rewrite the metric by setting $e_3^{~11}$, hence $A_3^{(11)}$, to
zero and substitute for it the field dual to $A_3^{(11)}$, defined
with field strength
$\widetilde
F_{\nu{}_1\nu{}_2\nu{}_3\nu{}_4\nu{}_5\nu{}_6\nu{}_7\nu{}_8
  \nu{}_9\vert\nu{}_9}=9\,
\partial_{[\nu{}_1}A_{\nu{}_2\nu{}_3\nu{}_4\nu{}_5\nu{}_6\nu{}_7\nu{}_8
    \nu{}_9]\vert\nu{}_9}$, where the dual field strength $\widetilde
F$ is defined by (\ref{Pdual2}). The dual (diagonal) metric and the
dual potential read 
\be \begin{split} 
 \label{KKMmetric}
{\rm KK6M} : \qquad \qquad \qquad  g_{11}&=g_{22}= g_{33}= H\,,\\
  -g_{44}&=g_{55}\dots
 =g_{10\,10}=1\,,\\
  g_{11\,11}&= H^{-1}\, ,\\ 
 A_{4\,5\,6\,7\,8\,9\,10\, 11\vert 11} &=\frac{1}{H}\,  .
\end{split} \ee
One verifies that the dual description  of the KK6-monopole given by
(\ref{KKMmetric}) can be derived from an effective action  in
analogy with  the action (\ref{M5dual}) for the M5. Here the dual
field plays the role of a matter field. In the gauge considered here
one takes as effective action action 
\begin{equation}
\label{KKMdual}
{\cal S}^{(11)}_{KK6} =\frac{1}{16\pi G_{11}}\,\int \dd ^{11}x \sqrt{-g^{(11)}}\left[R^{(11)}- {1\over 2  
}\widetilde F_{i\,4\,5\,6\,7\,8\,9\,10\, 11\vert 11}\widetilde F^{i\,4\,5\,6\,7\,8\,9\,10\, 11\vert 11}\right] \, ,
\end{equation}
where $i$ runs over the three non-compact dimensions $1,2,3$. In this
dual description we may trivially smear the KK-monopole to two or to
one non-compact space dimensions by letting the index $i$  in
(\ref{KKMdual}) run over the remaining non-compact dimensions. In
two non-compact space dimensions, one obtains the dual of the
description (\ref{F2}) and (\ref{G3metric2}) and in one dimension
one gets  in this way  a definition of the smeared KK6-monopole which
inherits its charge and mass from the parent one with 3 non-compact
dimensions.  

The charge carried by the KK6-monopoles in three or less non-compact
dimensions can be obtained in the dual formulation  from the equations
of motion of the field
$A_{\nu{}_2\nu{}_3\nu{}_4\nu{}_5\nu{}_6\nu{}_7\nu{}_8
  \nu{}_9\vert\nu{}_9}$. From (\ref{KKMdual}), one gets 
\begin{equation}
\label{current}
\sum_{i,j=1}^2
\partial_i\big(\sqrt {-g}
g^{ij}g^{44}g^{55}g^{66}g^{77}g^{88}g^{99}g^{10\,10}(g^{11\,11})^2\partial_j 
A_{4\,5\,6\,7\,8\,9\,10\, 11\vert 11}\big)=0 \, , 
\end{equation}
and using (\ref {KKMmetric}) one finds
\begin{equation}
\label{currentH}
\sum_{i=1}^2 \partial_i\partial_i H =0\, ,
\end{equation}
outside the source singularities of the harmonic function $H$. If the
latter yields in (\ref{currentH}) only $\delta$-function
singularities $\sum_k q_k\delta(\vec r-\vec r_k)$ located in
non-compact space points $\vec r_k$, one may extend
(\ref{currentH}) to the whole non-compact space. We write  
\begin{equation}
\label{h}
\sum_{i=1}^2
\partial_i\partial_i H\propto \sum_k q_k \delta(\vec r-\vec r_k)\, ,
 \end{equation}
where $q_k$ is the charge of the monopole located  at $\vec r_k$. For
instance a single KK6-monopole located at the origin in 2 non-compact
space is described by $H= (q/2\pi)\ln r$.  
 
For 3 or 2 non-compact dimensions, writing (\ref{F}) or
(\ref{F2}) as $F_{\mu\nu}^{(11)}$ one recovers from (\ref{h}) the
charge of the monopoles from the conventional surface integral  
\begin{equation}
\label{chargeKKM}
\int F \equiv\int (1/2) F_{\mu\nu}^{(11)}\,\dd x^\mu \wedge \dd x^\nu \propto
\sum_k q_k\, , 
\end{equation}
where the surface integral enclosed the charges
$q_k$. The equation (\ref{chargeKKM}) is equivalent to (\ref{h}). For
KK6-monopoles smeared to one non-compact dimension, the surface
integral loses its meaning but the direct definition of charge
(\ref{h}) is still valid. In this way the magnetic KK6-monopole in
any number of non-compact transverse dimensions is suitably described
(as  the magnetic 5-brane by (\ref{M5dual})) by a dual effective
action  (\ref{KKMdual}). As expected for a BPS solution, the charge
of the KK6-monopole, smeared or not, is equal to its tension evaluated
in string theory,  as recalled in Chapter \ref{chap:dualfor} and in
Appendix~\ref{appm3} where the mass of the KK6-monopole is derived
from T-duality.

The above KK6-monopole solution smeared over all dimensions but one
can again be obtained as a solution of the $\sigma$-model  by putting
in $\mathcal{S}^{brane}$,  with time in 4, all non-Cartan fields to zero except
the level 3 component $A_{4\,5\,6\,7\,8\,9\,10\,11\vert
  11}$\cite{Englert:2003py}. This is the parameter of the Borel
generator\footnote {We use the bar symbol to distinguish within the
  same irreducible level 3 $A_9$ representation the generator $\bar
  R^{[8,1]}_3$ corresponding to a real $E_{10}$ root from the
  generator  $R^{[8,1]}_3=[R^{[3]}_1,R^{[6]}_2]$ pertaining to the
  degenerate   null root, see also below in Section~\ref{m2m5sec}.}  
\begin{equation}
\label{81ge}
\bar R^{[8,1]}_3 \equiv R^{4\,5\,6\,7\,8\,9\,10\,11\vert 11}\, ,
\end{equation}
where the subscript labels the $A_{9}$ level in the decomposition of
the adjoint representation of $E_{10}$. The solution is, taking into
account the embedding relation (\ref{embedding}) 
\be
 \label{G3} \begin{split}
{\rm KK6M}: \qquad \qquad  h_a^{~a} &=\left\{\frac{-1}{2}\,\Big|\,
  \frac{-1}{2},\frac{-1}{2},0,0,0,0,0,0,0,\frac{1}{2}\right\}\ln H\,,\\ 
h_a^{~a}K^a_{~a} &=\frac{1}{2}\ln H
  \,\cdot\,(-K^2{}_2-K^3{}_3+K^{11}{}_{11})\,,\\ 
A_{4\,5\dots 10 \, 11\vert 11} &=\frac{1}{H}\,,
\end{split}\ee
with $H(x^1)$ harmonic. From (\ref{vielbein}) one indeed recovers
the KK6-monopole solution of the dual action (\ref{KKMdual})
displayed in (\ref{KKMmetric}). 

\subsection{The exotic BPS solutions}
As a consequence of the non-commutativity of Weyl reflexions with the
temporal involution, the $\mathZ{E}_{10}$ $\sigma$-model $\mathcal{S}^{brane}$ living on
$\mathZ{E}_{10}/\mathrm{K}_{10}^-$ was expressed in 10 different ways according to the
choice of the time coordinate in (\ref{full}) in the global
signature (1,9). These are related through Weyl transformations of
$E_{10}$ from roots of the gravity line (see Section \ref{sec:tempinvolutintro}). Adding the Weyl reflexion
$s_{\alpha_{11}}$ one gets in addition equivalent expressions for
$\mathcal{S}^{brane}$ where the signatures  in
(\ref{full}) are globally different \cite{Englert:2004ph}.  This equivalence realises in
the action formalism  the general analysis of Weyl transformations by
Keurentjes \cite{Keurentjes:2004bv, Keurentjes:2004xx}. Starting with
the global signature (1,9) in 10 dimensions, or (1,10) in 11
dimensions, one reaches  different  signatures $(t,s,\pm)$ in 11
dimensions where $t$ is the number of timelike directions, $s$ is the
number of spacelike directions and $\pm$ encodes the sign of the
kinetic energy term of the level 1 field in the action
(\ref{fulla}), $+$ being the usual one and $-$  the `wrong'
one. These are \footnote{ The analysis of the
  different possible signatures related by Weyl reflexions has been
  extended to all $\mathfrak{g}^{++}$ in \cite{deBuyl:2005it,
    Keurentjes:2005jw}.}: $(1,10,+)$, $(2,9,-)$, $(5,6,+)$, $(6,5,-)$
and $(9,2,+)$ \cite{Englert:2004ph}. The signature changes  under the Weyl transformations
used in the following sections are presented in Appendix~\ref{appw}.

The results obtained in the   $\sigma$-model (\ref{full}) are in
complete agreement with the interpretation of the Weyl reflexion  
$s_{\alpha_{11}}$ as a double T-duality in the direction $9$ and $10$
plus exchange of the two directions   
\cite{Elitzur:1997zn, Obers:1998rn, Banks:1998vs,
  Englert:2003zs}. Indeed, it has been shown that T-duality involving
a timelike direction changes 
the signature of space-time leading to the exotic phases of M-theory
\cite{Hull:1998vg, Hull:1998ym}.  The signatures found by Weyl
reflexions are thus in perfect agreement with the analysis  of
timelike T-dualities.  

The brane scan of the exotic phases has been studied
\cite{Hull:1998fh, Argurio:1998ad}. Their different BPS branes depend
on the signature and the sign of the kinetic term. The number of
longitudinal timelike directions for a given brane is constrained. As
an example if we consider the so-called ${\rm M}^*$ 
phase characterised by the signature $(2,9, -)$, the wrong  sign of
the kinetic energy term implies that the exotic M2 brane must have
even number of timelike directions. There are thus two different M2
branes in ${\rm M}^*$ theory denoted $(0,3)$ and $(2,1)$ where the
first entry is the number of timelike longitudinal directions and the
second one the number of spacelike longitudinal directions. 
For instance, the metric of a $(2,1)$ exotic M2 brane with timelike
directions $10$ , $11$, and spacelike direction $9$ is  
\be
 \label{M2metricE} \begin{split}
{\rm M2}^*:\quad  &g_{11}=g_{22}=H^{1/3},\\
 &g_{33}=g_{44}=\dots
=g_{88}=H^{1/3}\,,\\ 
&g_{99}=-g_{10\,10}=-g_{11\,11}= H^{-2/3}, 
\end{split}\ee
where $H$ is  the harmonic function in the transverse non-compact dimensions. 

When smeared in all directions but one this metric 
is also a solution of the  $\sigma$-model  $\mathcal{S}^{brane}$ with the
correct identification in (\ref{full}) of time components and sign
shifts in kinetic energy terms. More generally, all the exotic branes
smeared to one non-compact  space dimension are solutions of this
$\sigma$-model living on the coset $\mathZ{E}_{10}/\mathrm{K}_{10}^-$
\cite{Englert:2004ph}.

  \chapter{$E_9$-branes and the infinite U-duality group} \label{chap:infiniteudualgroup}
In this chapter we construct an infinite set of BPS solutions  of eleven-dimensional 
supergravity  and of its exotic counterparts depending on two
non-compact space variables. They are  related by the Weyl group of
$E_9$ to the basic ones reviewed in Chapter \ref{chap:basicbps} and constitute an
infinite multiplet of U-dualities viewed as Weyl transformations. 
\setcounter{equation}{0}
\section{The working hypothesis}

Our working hypothesis is that the fields describing  BPS solutions of
eleven-dimensional  supergravity  depending on two non-compact space variables $(x^1,
x^2)$ are coordinates in  the coset $\mathZ{E}_{10}/\mathrm{K}_{10}^-$, in the regular
embedding $E_{10}\subset E_{11}$. The coset representatives are taken
in the Borel gauge, subject of course  to the remark in
footnote~\ref{borelfn} page \pageref{borelfn}.   

We first express in this way the basic solutions of Chapter \ref{chap:basicbps}, smeared
to two space dimensions.  Consider the M2 and M5 branes.  Their Borel representatives are
\begin{eqnarray}
\label{2M2} {\rm M2}&:& {\cal V}_1= \exp \left [\frac{1}{2}\ln H\,
  h_{11}\right]\, \exp\left [\frac{1}{H}\, R^{[3]}_1\right]\\ 
\label{2M5} {\rm M5} &:&{\cal V}_2= \exp \left[\frac{1}{2}\ln H\,
  (-h_{11} -K^2{}_2)\right]\, \exp \left[\frac{1}{H}\,
  R^{[6]}_2\right]\, . 
\end{eqnarray}
Here  $R^{[3]}_1$ and  $R^{[6]}_2$ are defined in (\ref {3ge}) and
(\ref {6ge}), respectively, and $h_{11}$ was defined in
(\ref{el10}). The Cartan fields and the potentials 
$A_{9\,10\,11}(x_1,x_2)$ for the M2 and
$A_{3\,4\,5\,6\,7\,8}(x_1,x_2)$ for the M5 are given by
(\ref{M2}), (\ref{M5}) with $H$ now a function of the two
variables $x^1,x^2$. Their metric (\ref{M2metric}) and
(\ref{M5metric})  are encoded in (\ref{vielbein}) giving the
relation  of the Cartan fields to the vielbein  and in
(\ref{embedding}) expressing the embedding of $E_{10}$ in
$E_{11}$. The Hodge duality relations 
\begin{eqnarray}
\label{dual1}
\sqrt{\vert g \vert}g^{11}g^{99}g^{10\,10} g^{11\, 11}\partial_1
A_{9\,10\,11}&=&\partial_2  A_{3\,4\,5\,6\,7\,8} \\ 
\label{dual2}
\sqrt{\vert g \vert}g^{22}g^{99}g^{10\,10} g^{11\, 11}\partial_2
A_{9\,10\,11}&=&-\partial_1  A_{3\,4\,5\,6\,7\,8}\, , 
\end{eqnarray}
 reads both for M2 and M5, using (\ref{vielbein}), (\ref{M2}) and
 (\ref{M5}),  
\begin{eqnarray}
\label{duality1}
\partial_1 H&=& \partial_2 B\\
\label{duality2}
\partial_2 H&=&-\partial_1 B \, ,
\end{eqnarray}
where $B=A_{3\,4\,5\,6\,7\,8}$ for the M2 and $B=A_{9\,10\,11}$ for
 the M5. In this way, due to the particular choice we made for the
 tensor components defining the branes, the fields $A_{9\,10\,11}$ and
 $A_{3\,4\,5\,6\,7\,8}$ are interchanged  between the M2 and the M5
 when their common value switches from $1/H(x^1,x^2)$ to\footnote{We
 chose zero for the  integration constants  of the dual  fields
 $A_{9\,10\, 11}$ and $A_{3\,4\, 5\,6\,7\,8}$ (cf footnote~\ref{intconsfn} page \pageref{intconsfn}).} 
 $B(x^1,x^2)$.  Note however that for the M2 (M5) the time in
 $A_{3\,4\,5\,6\,7\,8}$ ($A_{9\,10\,11}$) is still 9 (3). The relations (\ref
 {duality1}) and (\ref{duality2}) are the Cauchy-Riemann relations for
 the analytic function
\begin{equation}
\label{analytic}
{\cal E}_{(1)}= H + iB\, ,
\end{equation}
and  $H$ and $B$ are thus conjugate harmonic functions. The duality
relations (\ref{duality1}) and (\ref{duality2})  allow for the
replacement of the Borel representatives ${\cal V}_1$ and ${\cal V}_2$
by 
\begin{eqnarray}
\label{boreln1}
{\rm M2}&:& {\cal V}^\prime_1= \exp \left [\frac{1}{2}\ln H\,
  h_{11}\right]\, \exp \left[B \,R^{[6]}_2\right]\\ 
\label{boreln2}
{\rm M5}&:& {\cal V}^\prime_2= \exp \left [\frac{1}{2} \ln H
  \,(-h_{11}-K^2{}_2)\right ] \, \exp \left[B\, R^{[3]}_1\right]\, . 
\end{eqnarray}
Note that the representative of the M5 in (\ref{boreln2}) is, as
the representative of the M2 in (\ref{2M2}), expressed in terms of
the supergravity metric and 3-form potential in two non-compact
dimensions.  

Consider now the purely gravitational BPS solutions. According to
(\ref{GO}) and (\ref{G3})  the Borel representatives are 
\begin{subequations} \begin{align}
{\rm KKW}: {\cal V}_0 &= \exp \left[\frac{1}{2}\ln H \,
  (K^3{}_3-K^{11}{}_{11})\right]\, \exp \left[(H^{-1}-1)\,
  K^3_{~11}\right] \label{bgr1} \\
  {\rm KK6M}: {\cal V}_3&= \exp \left[\frac{1}{2}\ln H
  \,(-K^2{}_2-K^3{}_3+K^{11}{}_{11})\right ] \, \exp \left[H^{-1}\,
  R^{4\,5\,6\,7\,8\,9\,10\,11\vert \,11}\right]\, , \label{bgr2}
\end{align}\end{subequations}
with $H=H(x^1,x^2)$. Using the duality relations (\ref{Pdual2})
between the [8,1]-form and  $A_3^{(11)}$, on may express the Borel
representative of the KK6-monopole, as the representative of the
KK-wave (\ref{bgr1}), in two space dimensions in terms of the
11-dimensional metric.  

Transforming  by $E_9$ Weyl transformations  the Borel representatives
of the basic solutions given here, we shall obtain for all  $E_9$
generators associated to its real positive roots,  representatives
expressed in terms of the harmonic functions $H=H(x^1,x^2)$. These
will be transformed through dualities and compensations to Borel
representatives expressed in terms of new level 0 fields and 3-form
potentials $A_{9\,10\,11}$. This potential and the metric encoded in
the level 0 fields through the embedding relation (\ref{embedding})
and (\ref{vielbein}) will be shown to solve the equations of
motions of eleven-dimensional supergravity. In this way we shall find an infinite set
of $E_9$ BPS solutions related to the M2, M5, KK-waves and
KK6-monopoles  by U-duality, viewed as $E_9$ Weyl transformations. 

\setcounter{equation}{0}
\section{The M2 - M5 system} \label{m2m5sec}
\subsection{The group-theoretical setting}

Generalizing the previous notation for generators to all levels by
using subscripts denoting the $A_9$ level, we write $R^{[3]}_1\equiv
R^{9\,10\,11}$ , $R^{[3]}_{-1}\equiv R_{9\,10\,11}$ and
$R^{[6]}_2\equiv R^{3\,4\,5\,6\,7\,8} $ , $R^{[6]}_{-2}\equiv
R_{3\,4\,5\,6\,7\,8}$. One has 
\begin{align} \label{geroch}
 [R^{[3]}_1,R^{[3]}_{-1}]&=h_{11},&&[h_{11},R^{[3]}_1]=2\, R^{[3]}_1, && \nn \\ 
[R^{[6]}_2, R^{[6]}_{-2}]&=-h_{11}-K^2_{~2}, && [h_{11}, R^{[6]}_2]=-2\, R^{[6]}_2, &&\\
[R^{[3]}_1, R^{[6]}_{-2}]&=0\,. &&   && \nn
\end{align}
These commutation relations form a Chevalley presentation of a group
with Cartan  matrix 
\begin{equation}
\label{gerochC}
{\bf A}=\left[\begin{array}{cc}
2&-2\\-2&2 \\
\end{array}\right]\, ,
\end{equation}
and one verifies that the Serre relations are satisfied. The matrix \eqref{gerochC} is the Cartan matrix of  
the affine $A_1^+$ algebra. Its corresponding group is  isomorphic to the standard
Geroch group which is also the affine extension of $\mathrm{SL}(2)$. The
central charge  is $k$ and derivation $d$, whose eigenvalues define 
the affine level, are here given by  the embedding of the $A_1^+$ in
$E_{10} $ as 
\begin{equation}
\label{brol}
k=-K^2_{\ 2}\quad , \quad d= -\frac {1}{3} K^2_{\
  2}+\frac{2}{9}(K^4_{\ 4} +\dots +K^9_{\ 9}) - \frac{1}{9}(K^3_{\ 3}
+K^{10}_{\ 10} +K^{11}_{\ 11})\, . 
\end{equation}
The level counting operator $d$ is not fixed uniquely by the present
embedding. 

The multicommutators satisfying the Serre relations form three
towers. The  positive  generators, normalized to one, are 
\begin{align}
\label{tower1}
R^{[3]}_{1+3n} &=  2^{-n}\,\,
\left[R^{[3]}_1\left[R^{[3]}_1\left[R^{[6]}_2\left[R^{[3]}_1\dots\left[R^{[6]}_2\left[R^{[3]}_1,R^{[6]}_2\right]\right]\dots\right]\right.\right.\right.\qquad 
      &&n\ge 0\\ 
\label{tower2}
R^{[8,1]}_{3n}
 &= 2^{-(n-1/2)}\,\,\left[R^{[3]}_1\left[R^{[6]}_2\left[R^{[3]}_1\dots\left[R^{[6]}_2\left[R^{[3]}_1,R^{[6]}_2\right]\right]\dots\right]\right.\right.
    \qquad &&n>0\\ 
\label{tower3}
R^{[6]}_{-1+3n} &=  2^{-(n-1)}\,\,
\left[R^{[6]}_2\left[R^{[3]}_1\dots\left[R^{[6]}_2\left[R^{[3]}_1,R^{[6]}_2\right]\right]\dots\right]\right.\qquad 
  &&n>0\, , 
\end{align} 
where the affine level $n$ is equal to the number of $R^{[6]}_2$ in
the tower\footnote{Shifts in the affine level by one unit corresponds
  to shifts in $A_9$ levels by three units, see
  also~\cite{Kleinschmidt:2006dy}. In what follows, when the 
  term level is left unspecified, we always mean the $A_9$
  level.}. The  $R^{[8,1]}_{3n} $ tower correspond to the null
roots $n\,\delta$ where 
\begin{equation}
\label{delta}
\delta = \alpha_3 + 2\alpha_4 +3\alpha_5 +4\alpha_6 +5\alpha_7
+6\alpha_8 +4\alpha_9 +2\alpha_{10} +3\alpha_{11}
\end{equation}
is the null root of level $\ell=3$ (see Table \ref{tab:levdece11}). This root
has the following properties
\begin{equation}
\label {products}
(\delta|\delta) =0,\qquad (\delta| \alpha_i
)=0\quad i=3,\dots,11\, . 
\end{equation}
In particular $R^{[8,1]}_{3} $ is a linear combination
$R^{3\,4\,5\,6\,7\,8\,[9\,10\, , 11]}$ of level  3 tensors with all
indices distinct. Its height is 30 and thus exceeds the `classical'
limit 29 of \cite{Damour:2002cu}.  

Substituting $h_{11} =-h_{11}^\prime-K^2{}_2$ in (\ref{geroch}) one
obtains a presentation with $R^{[3]}_1$ and $R^{[6]}_2$
interchanged. While the $A_1^+$ group in the presentation
(\ref{geroch}) appears associated with the M2 brane, one could
associate the alternate presentation with the M5 brane. The two
presentations differ by  shifts in the affine level but not by the
$A_9$ level. To avoid complicated notations we keep for the complete
M2-M5 system the description given by (\ref{geroch}) which is
labelled explicitly in terms of  the $A_9$ level. The generators of
the $A_1^+$ group pertaining to the real roots of the M2-M5 system
appear in Figure \ref{ger1fig}a and in Figure \ref{ger1fig}b. 

All  the real roots of $E_9$  can be reached by $E_9$ Weyl
transformations acting on (say) $\alpha_{11}$ defining the generator
$R^{[3]}_1$. We shall find convenient for our construction of the
infinite set of the 
$E_9$ BPS-branes to generate all the real roots from two different
real roots, namely $\alpha_{11}$ and $-\alpha_{11}+\delta$
characterising respectively the generators $R^{[3]}_1$ and
$R^{[6]}_2$.\footnote{We note that in $A_1^+$ not all real roots are
  Weyl equivalent but there are two distinct orbits as we will see in
  more detail below.}
 In this section we obtain the generators of the
$R^{[3]}_{1+3n}$ and $R^{[6]}_{-1+3n}$ towers (\ref{tower1}) and
(\ref{tower3}) and their negative counterparts  by the Weyl reflexions
$s_{\alpha_{11}}$ and $s_{-\alpha_{11}+ \delta}$ acting in alternating
sequences, starting  from their action on the generators $R^{[3]}_1$
and $R^{[6]}_2$. This is depicted in Figure \ref{ger1fig}a.  
The Weyl reflexions $s_{\alpha_{11}}$ and $s_{-\alpha_{11}+ \delta}$
generate the Weyl group of the affine subgroup
$A_1^+$ of $E_9$ depicted in Figure \ref{fig:a1+++}. Its  formal structure is discussed in Appendix \ref{affapp}.

For the simply laced algebra considered here, with real roots normed
to square length 2,  the Weyl reflexion  $s_\alpha(\beta)$ (see Section \ref{subsec:weylgroupthese}) in the
plane perpendicular to the real root $\alpha$ acting on the arbitrary
weight $\beta$ is given by 
\begin{equation}
\label{defWeyl}
s_\alpha(\beta)= \beta -( \beta | \alpha)\,  \alpha.
\end{equation}
Consider the Weyl reflexion $s_{- \alpha_{11}+\delta}$ acting on
$\alpha_{11} + n\,\delta$, which defines  $R^{[3]}_{1+3n}$ for $n\geq
0$ and $R^{[6]}_{1+3n}$ for $n<0$.  
One gets
\begin{equation}
\label{weyl1}
s_{-\alpha_{11}+\delta}(\alpha_{11}+n\,\delta) =
\alpha_{11}+n\,\delta -( \alpha_{11}+n\,\delta|  -\alpha_{11}
+\delta) (-\alpha_{11} +\delta) = -\alpha_{11} +(n+2)\, \delta
\, , 
\end{equation}
where we have used (\ref{products}). The real root $ -\alpha_{11}
+(n+2)\, \delta$ defines the generator $R^{[6]}_{-1+3(n+2)}$ for $n\ge
0$ and $R^{[3]}_{-1+3(n+2)}$ for $n<0$.  
The Weyl reflexion has induced an $A_9$ level increase of four `units'
since $\delta$ from (\ref{delta}) has three units of $A_9$ level and
$\alpha_{11}$ has one. 
Similarly, acting on the root $-\alpha_{11} +n\, \delta$, which
defines the generator $R^{[6]}_{-1+3n}$  for $n>0$  and
$R^{[3]}_{-1+3n}$ for $n\leq 0$, with the Weyl reflexion
$s_{\alpha_{11}}$  
one gets
\begin{equation}
\label{weyl2}
s_{\alpha_{11}}(-\alpha_{11}+n\,\delta)= -\alpha_{11}+n\,\delta
-( -\alpha_{11}+n\,\delta\, |\,   \alpha_{11}) \, \alpha_{11} =
\alpha_{11} +n\, \delta \, . 
\end{equation}
This Weyl reflexion induces an $A_9$ level increase of two units. Thus acting
successively with the Weyl reflexions (\ref{weyl1}) and
(\ref{weyl2}) on any root $\alpha_{11} + n\,\delta$ or in the reverse
order on the root $-\alpha_{11} +n\, \delta$ one induces a level
increase of six units, or equivalently of two affine levels. Of course
interchanging the initial and final roots and the order of the two
Weyl reflexions, one decreases the affine level by two units. Thus
starting from the roots $\alpha_{11}$ and $-\alpha_{11}+\delta$, we
obtain  the real roots defining all the generators of the
$R^{[3]}_{1+3n}$ and $R^{[6]}_{-1+3n}$ towers (\ref{tower1}) and
(\ref{tower3}) and their negative counterparts. These form two
sequences depicted in Figure \ref{ger1fig}a. The `M2 sequence' originates
from the $\alpha_{11}$  root (and thus from the generator $R^{[3]}_1$)
and reads 
\begin{figure}[h]
\begin{center}
{\scalebox{0.60}
{\includegraphics{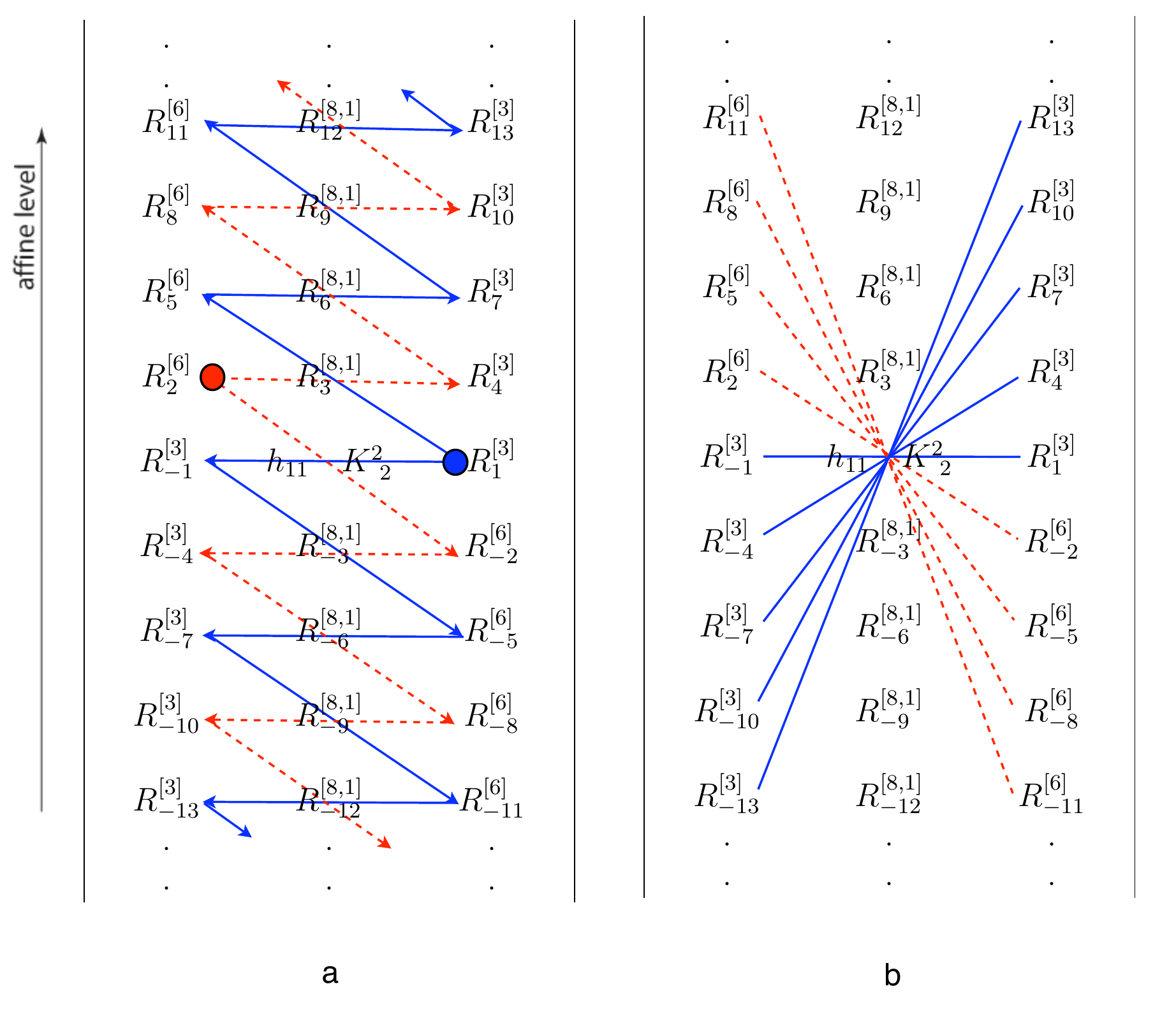}}}
  \end{center}
  \caption {\sl \small $A_1^+$ group for the M2-M5
 system. (a)  The Weyl 
 group: the M2 sequence is depicted by  solid lines and the M5
 sequence by dashed lines. Horizontal lines  represent Weyl reflexions
 $s_{\alpha_{11}}$ and diagonal ones $s_{-\alpha_{11}+ \delta}$. (b)
 The $\mathrm{SL}(2)$ subgroups:   the solid lines label $\mathrm{SL}(2)$ subgroups from
 the 3-tower (\ref{tower1})  and the dashed ones the $\mathrm{SL}(2)$
 subgroups from the 6-tower (\ref{tower3}).  } 
\label{ger1fig}
\end{figure}
\be
\label {M2seq}
\dots\stackrel{\alpha_{11}}{\longleftarrow} 
R^{[6]}_{-5}\stackrel{-\alpha_{11}+\delta}{\longleftarrow}
R^{[3]}_{-1}\stackrel{\alpha_{11}}{\longleftarrow}\
\framebox{$R^{[3]}_1$}\stackrel{-\alpha_{11} +Ê\delta}{\longrightarrow}
R^{[6]}_5\stackrel{\alpha_{11} }{\longrightarrow}
R^{[3]}_7\stackrel{-\alpha_{11} +Ê\delta}{\longrightarrow}
R^{[6]}_{11}\stackrel{\alpha_{11}}{\longrightarrow}\dots  
\ee
The `M5 sequence' originates from $-\alpha_{11}+\delta$ (and thus from
the generator $R^{[6]}_2$) and reads
\begin{equation}
\label {M5seq}
\dots\stackrel{-\alpha_{11}+\delta}{\longleftarrow}
R^{[3]}_{-4}\stackrel{\alpha_{11}}{\longleftarrow}
R^{[6]}_{-2}\stackrel{-\alpha_{11}+\delta}{\longleftarrow}\framebox{$R^{[6]}_2$}
\stackrel{\alpha_{11}
}{\longrightarrow} R^{[3]}_4\stackrel{-\alpha_{11}
  +Ê\delta}{\longrightarrow}
R^{[6]}_8\stackrel{\alpha_{11}}{\longrightarrow}
R^{[3]}_{10}\stackrel{-\alpha_{11} +Ê\delta}{\longrightarrow}\dots  
\end{equation}
Both sequences are represented in Figure \ref{ger1fig}a.

The Hodge duality relations (\ref{dual1}) and (\ref{dual2}) will
play an essential role in the determinations of the $E_9$
BPS-branes. The Hodge dual generators  $R^{[3]}_1$ and $R^{[6]}_2$
have commutation relations  
\begin{equation}
\label{dualcom}
\left[R^{[3]}_1,R^{[6]}_2\right]= R^{[8,1]}_3 \, .
\end{equation}
The roots of the $E_9$ subalgebra do not contain $\alpha_2$ when expressed in terms of simple roots. Hence from (\ref{products}) any Weyl transformation from a $E_9$ real root leaves invariant (possibly up to a sign) the right hand side of (\ref{dualcom}). Therefore, the image of the basic pair
$R^{[3]}_1,R^{[6]}_2$ by any such $E_9$ Weyl transformation are pairs whose $A_9$ level sum is equal to three and we have:

\noindent
{\bf Theorem 1}
\it 
{The set of  Weyl transformations in $E_9$ mapping the $A_1^+$ group
  ({\ref{geroch}) into itself either transforms the pair
    ($R^{[3]}_1,R^{[6]}_2$) into itself or into one of the pairs
    ($R^{[3]}_{1+3p} ,R^{[3]}_{-1-3(p-1)}$), ($R^{[6]}_{-1+3(p+1)}
    ,R^{[6]}_{1-3p}$) where $p$ is a positive integer.} \rm 

\noindent
This theorem applies  to the above Weyl transformations and is easily
checked from Figure \ref{ger1fig}a. 

The $A_1^+$ group  ({\ref{geroch}) admits two infinite sets of
  $\mathrm{SL}(2)$ subgroups 
\begin{align}
\label{sla}
\left[R^{[3]}_{1+ 3p},R^{[3]}_{-1-3p}\right]&=h_{11}-p K^2{}_2 ,
&& \left[h_{11}-p K^2{}_2 \, , R^{[3]}_{\pm(1+ 3p)}\right]  =\pm 2R^{[3]}_{\pm (1+
  3p)}\,,\\ 
\label{slb}
\left[R^{[6]}_{-1+ 3p},R^{[6]}_{1-3p}\right]&=- h_{11}-p K^2{}_2 ,
  &&\left[- h_{11}-p K^2{}_2\, , R^{[6]}_{\pm (-1+ 3p)}\right]
   =\pm 2R^{[6]}_{\pm (-1+ 3p)} \,, 
\end{align}
where $p\geq 0$ in \eqref{sla} and $p >0$ in \eqref{slb}. As all Weyl reflexions send opposite roots to opposite transforms, one has

\noindent
{\bf Theorem 2}
\it {The set of  Weyl transformations in $E_9$ mapping the $A_1^+$
  group ({\ref{geroch}) into itself exchanges the $\mathrm{SL}(2)$ subgroups
    between themselves.}\rm

\noindent
The subgroups (\ref {sla}) and (\ref {slb}) are depicted in Figure \ref{ger1fig}b. 
\subsection{The M2 sequence}\label{subsec:meseqth}

We take as representatives of the M2 sequence all the Weyl transforms
of the M2 representative (\ref {2M2}). The time coordinate is
9. Following  in Figure \ref{ger1fig}a the solid line towards positive step
generators, we encounter  Weyl transforms of the $\mathrm{SL}(2)$ subgroup
generated by $(h_{11}\, , R^{[3]}_1\, ,R^{[3]}_{-1})$ represented by a
solid line in Figure \ref{ger1fig}b. Theorem 2 determines from  (\ref{sla}) and
(\ref{slb}) the Weyl transform of the Cartan generators of
(\ref{2M2})  and we write 
\begin{align}
\label{borel3n}
{\cal V}_{1+6n}&= \exp \left[\frac{1}{2} \ln H \, (h_{11}-2n K^2{}_2)\right]\, \exp\left [\frac{ 1}{H}\, R^{[3]}_{1+6n}\right]\qquad &&n\ge 0\\
\label{borel6n}
{\cal V}_{-1+6n}&= \exp \left[\frac{1}{2} \ln H\, (-h_{11}-2n K^2{}_2)\right] \, \exp \left[\frac{1}{H}\, R^{[6]}_{-1+6n}\right] \qquad &&n>0\, .
\end{align}
We shall trade the tower fields $A^{[3]}_{1+6n}, A^{[6]}_{-1+6n}$
parametrising the generators in (\ref{borel3n}) and
(\ref{borel6n}) in favour of the supergravity potential
$A_{9\,10\,11}$ and construct from them BPS solutions of eleven-dimensional 
supergravity. We have not indicated sign shifts induced by
the Weyl transformations in the tower fields from the sign of the
lowest level $n=0$ field which is taken to be $(+1/H)$. This is here
the only relevant sign,  as  discussed below. 

Let us consider explicitly the first two steps. These will introduce the two essential features of our construction: compensation and signature changes.

\bigskip
\noindent
{\bf $\bullet$ From level 1 to level 5: compensation}

Following in Figure \ref{ger1fig}a the solid line towards positive step generators, we first encounter the Weyl reflexion $s_{-\alpha_{11}+\delta}$ sending the level 1 generator $R^{[3]}_1$ to the level 5 generator $R^{[6]}_5$.  (\ref{borel6n}) for $n=1$ reads
\begin{equation}
\label{five}
{\cal V}_5= \exp\left [\frac{1}{2} \ln H\,(-h_{11}-2K^2{}_2)\right] \, \exp\left [\frac{1}{H}\, R^{[6]}_5\right]\, .
\end{equation}
As can be seen in Figure \ref{ger1fig}a,  the Weyl reflexion $s_{- \alpha_{11}+\delta}$ sending $R^{[3]}_1$ to $R^{[6]}_5$ sends its dual  $
R^{[6]}_2$ to $ R^{[6]}_{-2}$, in accordance with Theorem 1.  We call
$ R^{[6]}_{-2}$ the dual generator of $R^{[6]}_5$ and we get by acting
with  $s_{- \alpha_{11}+\delta}$ on  the dual representative for the
M2 (\ref {boreln1}), the `dual' representative of (\ref{five}), 
\begin{equation}
\label{prime5}
{\cal V}^\prime_5= \exp  \left[\frac{1}{2} \ln
  H\,(-h_{11}-2K^2{}_2)\right] \,\exp \left[-B\, R^{[6]}_{-2}\right]\,
. 
\end{equation}
The  $-$ sign in front of $B$ arises as follows. 
The   generators $-h_{11}-K^2{}_2\, ,\, R^{[6]}_2\, ,\,R^{[6]}_{-2}$
form an $\mathrm{SL}(2)$ group, depicted by a dashed line in Figure \ref{ger1fig}a. We use the
representation\footnote{ $h_1$, $e_1$ and $f_1$ are the Chevalley
  generators of an  $\mathfrak{sl}(2)\subset E_9$.} 
\begin{equation}
\label{repsl}
h_1= \left[\begin{array}{cc}
1&0\\0&-1 \\
\end{array}\right]\,,\quad e_1=\left[\begin{array}{cc}
0&1\\ 0&0\\
\end{array}\right]\,,\quad f_1=\left[\begin{array}{cc}
0&0\\1&0\\
\end{array}\right]\, ,\quad K^2{}_2=\left[\begin{array}{cc}
1&0\\0&1\\
\end{array}\right]\, ,
\end{equation}
with  $ -h_{11}-K^2{}_2=h_1, R^{[6]}_2=e_1\, ,R^{[6]}_{-2}=f_1$, where
we have also included a representation for the central element $K^2{}_2$.
The Weyl reflexion $s_{- \alpha_{11}+\delta}$  on a generator of the algebra is generated by the
group conjugation matrix $U_5$ of  $\mathrm{SL}(2)$  \cite{Kac:book} 
\begin{equation}
\label{weyl5}
U_5=\exp R^{[6]}_{-2} \, \exp\,( - R^{[6]}_2 )\, \exp R^{[6]}_{-2}\, ,
\end{equation}
which can be represented by
\begin{equation}
\left[\begin{array}{cc}
1&0\\1&1\\
\end{array}\right]\left[\begin{array}{cc}
1&-1\\0&1\\
\end{array}\right]\left[\begin{array}{cc}
1&0\\ 1 &1\\
\end{array}\right] =\left[\begin{array}{cc}
0&-1\\1&0\\
\end{array}\right]\, ,
\end{equation}
and thus 
\begin{equation}
\label{trans5}
U_5 R^{[6]}_2 U^{-1}_5 =\left[\begin{array}{cc}
0&0\\-1&0\\
\end{array}\right] =-R^{[6]}_{-2}\, .
\end{equation}
One may verify that the conjugation matrix $U_5$ acting on the Cartan
generator of the representative of the M2 in dual form
(\ref{boreln1}) yields the same Cartan generator in the dual
representative ${\cal V}'_5$ in (\ref{prime5}) as in the direct
form ${\cal V}_5$, which was obtained from the level 1 representative
of the M2 (\ref{2M2}). 

We now write (\ref{prime5}) as an $\mathrm{SL}(2)$ matrix times a factor
coming from the $K^2{}_2$ contribution. One has  
\begin{equation}
\label{matrix5}
{\cal V}^\prime_5= H^{-1/2}\left[\begin{array}{cc}
H^{1/2}&0\\0&H^{-1/2}\\
\end{array}\right]\left[\begin{array}{cc}
1&0\\ -B &1\\
\end{array}\right]\, ,
\end{equation}
The negative root in (\ref{prime5}) can be transferred to the
original Borel gauge by a compensating element of $\mathrm{K}_{10}^{-}$. To
this effect we multiply on the left the matrix (\ref{matrix5}) by a
suitable element of the group  $\mathrm{SO}(2)= \mathrm{SL}(2)\cap \mathrm{K}_{10}^{-}$. For a
well chosen $\theta$ we get 
\be \label{comp5} \begin{split}
\overline{\cal V}^\prime_5&=
H^{-1/2}\left[\begin{array}{cc}\cos\theta&\sin\theta\\-\sin\theta&\cos\theta\\ 
\end{array}\right]\left[\begin{array}{cc}
H^{1/2}&0\\0&H^{-1/2}\\
\end{array}\right]\left[\begin{array}{cc}
1&0\\ -B &1\\
\end{array}\right]\\
&= H^{-1/2}\left[\begin{array}{cc}
\displaystyle
\left(\frac{H}{H^2 + B^2}\right)^{-1/2}&0\\0& \displaystyle
\left(\frac{H}{H^2 + B^2}\right)^{1/2}\\
\end{array}\right]\left[\begin{array}{cc}
1& \displaystyle
\frac{-B}{H^2 + B^2}\\ 0 &1\\
\end{array}\right]\\
&=\exp \left[-\frac{1}{2}\ln (H^2 + B^2)\, K^2{}_2\right] \, \exp
\left[\frac{1}{2} \ln \frac{H}{H^2 + B^2}\, h_{11}\right]\, \exp
\left[\frac{-B}{H^2 + B^2}\, R^{[6]}_2\right]\, . 
\end{split} \ee
Using the embedding relation (\ref{embedding}) we get from
(\ref{vielbein}) the metric\footnote{This solution of eleven-dimensional 
  supergravity has been derived previously in a different context
  \cite{Lozano-Tellechea:2000mc}.} encoded in the representative
(\ref{comp5})   
\be\label{metric5} \begin{split}
{\rm Level }\ 5 :\qquad g_{11}&=g_{22}=(H^2+B^2) \widetilde H^{1/3}\\
g_{33}&=g_{44}=\dots =g_{88}= \widetilde H^{1/3}\\ 
-g_{99}&=g_{10\,10}=g_{11\,11}=\widetilde H^{-2/3}\, ,
\end{split}\ee
where
\begin{equation}
\widetilde H =\frac{H}{H^2+B^2}\, .
\end{equation}
Using this metric and the duality equations (\ref{dual1}) and
(\ref{dual2}), one obtains the supergravity 3-form potential $A_{9\,
  10\,11}$ dual to  $A_{3\,4\,5\,6\,7\,8}=-B/(H^2+B^2)$  
\begin{equation}
\label{potential5}
A_{9\, 10\,11}= \frac{1}{\widetilde H}\, ,
\end{equation}
and the dual representative of $\overline{\cal V}^\prime_5$, expressed
in terms of the potential (\ref{potential5}) is  
\begin{equation}
\label{fin5}
\overline{\cal V}_5=\exp\left [-\frac{1}{2}\ln (H^2 + B^2)\,
  K^2{}_2\right] \, \exp \left[\frac{1}{2} \ln \widetilde H\,
  h_{11}\right]\,  \exp\left [\frac{1}{\widetilde H}\,
  R^{[3]}_1\right] \, . 
\end{equation}
The dual pair $\widetilde H$ and $\widetilde B\equiv -B/(H^2+B^2)$ are
the conjugate harmonic functions defined by the analytic function
${\cal E}_2$ given by 
\begin{equation}
\label{analytic5}
{\cal E}_2= \frac{1}{{\cal E}_1}=\widetilde H + i\widetilde B\, .
\end{equation}
We shall verify later that the metric (\ref{metric5}) and the
potential (\ref{potential5}) solve the equations of motion of eleven-dimensional 
supergravity. 

\bigskip
\noindent
{\bf $\bullet$ From level 5 to level 7: signature change}

Pursuing further in Figure \ref{ger1fig}a the solid line towards positive step
generators, we encounter the Weyl reflexion $s_{\alpha_{11}}$  leading
from level 5 to the level 7 root $R^{[3]}_7$.  The equation (\ref{borel3n})
reads 
\begin{equation}
\label{seven}
{\cal V}_7= \exp \left[\frac{1}{2} \ln H \, (h_{11}-2
  K^2{}_2)\right]\, \exp \left[\frac{1}{H}\, R^{[3]}_7\right]\, . 
\end{equation}
In the computation of the level 5 solution we have  
followed the sequence of duality transformations and the compensation
depicted on the second horizontal line  in  Figure \ref{second}. The sequence of
operations required to transform ${\cal V}_7$ to a representative
expressed in terms of the supergravity 3-form potential parametrizing
$R^{[3]}_1$ is depicted in the third line of Figure \ref{second}. As discussed below
in more details in the analysis of the full M2 sequence, all steps
appearing in the figure on the same column at  
levels 5 and 7 are related by the same Weyl transformation
$s_{\alpha_{11}}$. Hence one may short-circuit the first two dualities
and the first compensation and evaluate directly the Borel
representative pertaining to the last column of the level 5 line in
Figure \ref{second}. This amounts to take  
the Weyl transform by  $s_{\alpha_{11}}$
  of $\overline{\cal V}_5$ given in (\ref{fin5}). The generators
  $h_{11}\, , R^{[3]}_{-1}\, ,R^{[3]}_1$ generate as above an $\mathrm{SL}(2)$
  group represented here by a solid line in Figure \ref{ger1fig}b. The Weyl
  conjugation matrix $U_7$ sending $R^{[3]}_1$ to $R^{[3]}_{-1}$ is 
  \begin{equation}
\label{weyl7}
U_7=\exp R^{[3]}_{-1}\, \exp\,( - R^{[3]}_1)\,  \exp R^{[3]}_{-1}\, ,
\end{equation}
which yields the result corresponding to (\ref{trans5}), namely
\begin{equation}
\label{trans7}
U_7\,R^{[3]}_1\, U^{-1}_7  =- R^{[3]}_{-1}\, .
\end{equation}
One gets
\begin{equation}
\label{in7}
\overline{\cal V}_7=\exp \left[-\frac{1}{2}\ln (H^2 + B^2)\,
  K^2{}_2\right] \, \exp \left[-\frac{1}{2} \ln \widetilde H\,
  h_{11}\right]\,  \exp\left[-\frac{1}{\widetilde H}\,
  R^{[3]}_{-1}\right]\, . 
\end{equation}
To convert the negative root into a positive one  we have to perform a
second compensation. Here a new phenomenon appears: the space-time
signature changes because  the temporal involution $\Omega_9$ does not
commute with the  above Weyl transformation. As explained in
Appendix~\ref{appw1}, the signature $(1,10, +)$ becomes $(2,9, -)$
with time coordinates 10 and 11 and negative kinetic energy for the
field strength. 
The signature change affects the compensation matrix. The intersection
of the $\mathrm{SL}(2)$ group generated by $h_{11}\, , R^{[3]}_{-1}\,
,R^{[3]}_1$ with $\mathrm{K}_{10}^{-}$ is not $\mathrm{SO}(2)$ but $\mathrm{SO}(1,1)$. Indeed the
transformed involution $\Omega^\prime$ resulting from the action of
the Weyl reflexion $\alpha_{11}$ yields $\Omega^\prime R^{[3]}_1= +
R^{[3]}_{-1}$, and the combination of step operators invariant under
$\Omega^\prime $ is the hermitian, hence {\em non-compact generator}
$R^{[3]}_1+R^{[3]}_{-1}$, implying   $\mathrm{SL}(2)\cap \mathrm{K}_{10}^{-}=\mathrm{SO}(1,1)$.  

Recall that the field $1/\widetilde H$ in (\ref{in7}) is inherited
from (\ref{fin5}) which was obtained from (\ref{comp5}) using
Hodge duality. The latter is a differential equation and we have
hitherto chosen for simplicity the integration constant to be
zero. This choice would lead at level 7 to a singular compensating
matrix (see footnote~\ref{borelfn} page \pageref{borelfn}) and we therefore will use instead
the field $ 1/\widetilde H-1$ 
(we could keep an arbitrary constant $\gamma\ne 0$ instead of $-1$ but that
would unnecessarily complicate notations). Using the matrices
(\ref{repsl}) with  $h_{11}=h_1\, , R^{[3]}_1=e_1\,
,R^{[3]}_{-1}=f_1$, we get for  a suitable choice of $\eta$ the
compensated representative 
\be \begin{split}
\overline{\overline{\cal V}}_7&= (H^2+B^2)^{-1/2}\left[\begin{array}{cc}
\cosh\eta&\sinh\eta\\ \sinh\eta&\cosh\eta\\
\end{array}\right]\left[\begin{array}{cc}
\widetilde H^{-1/2}&0\\0&\widetilde H^{1/2}\\
\end{array}\right]\left[\begin{array}{cc}
1&0\\-\widetilde H^{-1}+1 &1\\
\end{array}\right]\\
&= (H^2+B^2)^{-1/2}\left[\begin{array}{cc}
(2-\widetilde H )^{1/2}&0\\0&(2-\widetilde H)^{-1/2}\\
\end{array}\right]\left[\begin{array}{cc}
1&1- (2-\widetilde H)^{-1}\\ 0 &1\\
\end{array}\right]\, ,
\end{split}\ee
or
\begin{equation}
\label{fin7}
\overline{\overline{\cal V}}_7=\exp \left[-\frac{1}{2}\ln (H^2+B^2)\,
  K^2{}_2 \right]\, \exp\left[ \frac{1}{2}\ln (2-\widetilde H )\,
  h_{11}\right]\, \exp \left[(1- \frac{1}{2-\widetilde H}) \,
  R^{[3]}_1\right]\, . 
\end{equation}
As shown below, this yields a solution of the supersymmetric exotic
partner of eleven-dimensional  supergravity with times in 10 and 11 and a negative
kinetic energy term. One has 
\begin{equation}
\label{potential7}
A_{9\,10\,11}=-\frac{1}{2-\widetilde H }\, ,
\end{equation}
where we dropped in the 3-form potential the irrelevant constant
$(-1)$. The metric encoded in the representative (\ref{fin7}) is,
from (\ref{embedding}) and (\ref{vielbein}), 
\be \label{metric7} \begin{split}
{\rm Level }\ 7 :\qquad g_{11}&=g_{22}=(H^2+B^2) \widetilde{\widetilde H}^{1/3}\, ,\\
g_{33}&=g_{44}=\dots =g_{88}= \widetilde{\widetilde H}^{1/3}\, ,\\
g_{99}&=-g_{10\,10}=-g_{11\,11}=\widetilde{\widetilde H}^{-2/3}\, ,
\end{split}\ee
where
\begin{equation}
\widetilde{\widetilde H} =2-\widetilde H\, .
\end{equation}
From the duality relations (\ref{dual1}) and (\ref{dual2}), we see
that the field  $\widetilde{\widetilde B}$ dual to
$\widetilde{\widetilde H}$ is equal to $ -\widetilde B$. The dual pair
$\widetilde{\widetilde H}$ and $\widetilde{\widetilde B}$ are
conjugate harmonic functions associated to the analytic function 
\begin{equation}
\label{analytic7}
{\cal E}_3= 2-{\cal E}_2=\widetilde{\widetilde H} + i\widetilde
{\widetilde B}\, . 
\end{equation}

\bigskip
\noindent
{\bf $\bullet$ The complete M2 sequence}

The  M2 sequence is characterized by the generators $R^{[6]}_{-1+ 6n},
n>0$ and $R^{[3]}_{1+ 6n}, n\ge 0$. These are reached by following in
Figure \ref{ger1fig}a the solid line starting from $R^{[3]}_1$ towards the positive
roots. The representatives are given in (\ref{borel3n}) and
(\ref{borel6n}).  
 
 \begin{figure}[h]
   \centering
   \includegraphics[width=13cm]{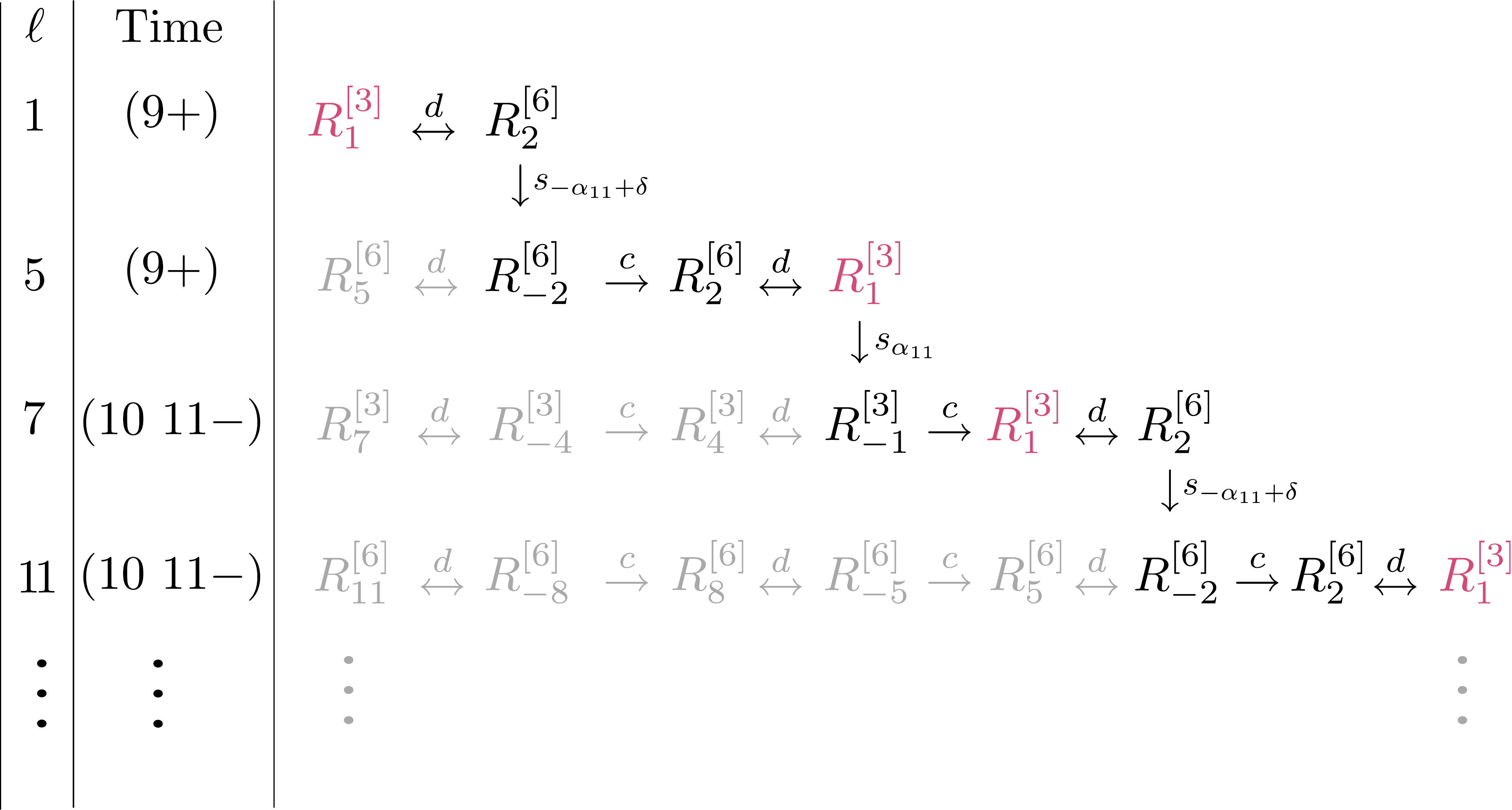}
 \caption {\sl \small The construction of BPS states for the M2
 sequence. The superscript $d$ labels a duality and the superscript
 $c$ labels a compensation. The horizontal rows are connected by Weyl
 reflexions as indicated.} 
 \label{second}
\end{figure}

A glance on Figure \ref{second} shows that at any level of this sequence one may
 trade this representative either by performing the Weyl
 transformation $s_{\alpha_{11}}$ on the preceding level
 representative 
 expressed in terms of the generator $R^{[3]}_1$ (as was done to reach
 the level 7 from the level 5 representative), or  by  the Weyl
 transformation $s_{- \alpha_{11}+\delta}$ on its dual representative
 in terms of $R^{[6]}_2$ (as was done to reach the level 5 from the
 level 1 representative). In this way, one bypasses all the steps
 depicted on horizontal lines in Figure \ref{second} which involve complicated
 duality relations and compensations  to solely perform a single
 compensation by a $\mathrm{SO}(2)$ or $\mathrm{SO}(1,1)$ matrix  and a known Hodge
 duality defined by (\ref{dual1}) and (\ref{dual2}), as
 exemplified by the detailed analysis of the two first levels of the
 sequence.  The nature of the compensation needed  in the construction
 of the M2 sequence alternates at each step on a horizontal line of
 Figure \ref{second} between $\mathrm{SO}(2)$ and $\mathrm{SO}(1,1)$ as shown in
 Appendix~\ref{appw1}. The representatives of the M2 sequence in terms
 of the supergravity fields are easily written in terms of complex
 potentials ${\cal E}_n$. These are obtained by operating on the
 analytic function
 \be
 {\cal E}(z)_1= H(z,\bar z) + iB(z,\bar z),
 \ee
 with $z=
 x^1+ix^2, \bar z= x^1-ix^2$, successively by  inversions  and
 translations according to 
 \begin{equation}
\label{Moebact}
{\cal E}(z)_{2n}= \frac{1}{{\cal E}(z)_{2n-1}}\,,\qquad{\cal
  E}(z)_{2n+1}= 2-{\cal E}(z)_{2n} \qquad n >0\, . 
\end{equation}
These formul\ae{} summarise the action of the $E_9$ Weyl group on BPS
states which are largely characterised by the harmonic functions
${\cal E}$. From (\ref{Moebact}) one sees that the action consists
of inversion and shift in a way very similar to the modular group
$\mathrm{SL}(2,{\mathbb Z})$. In order to see that there is more than just the
action of an $\mathrm{SL}(2,{\mathbb Z})$ one must consider the action of the
transformations on the full metric, including in particular the
conformal factor.
For the full solution one gets, defining ${\cal
  F}_{2n-1}={\cal E}_{1}{\cal E}_{3}\dots {\cal 
  E}_{2n-1} $ for $n >0$ (and ${\cal F}_{-1} \equiv 1$), the
straightforward generalisation of  (\ref{2M2}), (\ref{fin5}) and
(\ref{fin7}), 

\begin{align}
\label{seqM21}
{\cal V}_{1+6n}&=\exp \left[-\frac{1}{2}\ln ({\cal F}_{2n-1}
  \bar{\cal F}_{2n-1})\, K^2{}_2\right] \, \exp\left[\frac{1}{2} \ln
  {\cal R}e\, {\cal E}_{2n+1}\, h_{11}\right]\,  \exp
\left[\frac{(-1)^{n}}{{\cal R}e\, {\cal E}_{2n+1}}\,
  R^{[3]}_1\right]\\ 
\label{seqM22}
{\cal V}_{-1+6n}&=\exp \left[-\frac{1}{2}\ln( {\cal F}_{2n-1}
  \bar{\cal F}_{2n-1})\, K^2{}_2\right] \, \exp \left[\frac{1}{2} \ln
  {\cal R}e\, {\cal E}_{2n}\, h_{11}\right]\,  \exp
\left[\frac{(-1)^{n+1}}{{\cal R}e\, {\cal E}_{2n}}\,
  R^{[3]}_1\right],  
\end{align}
with $n\geq 0$ in \eqref{seqM21} and  $n>0$ in \eqref{seqM22}.
The quantity ${\cal \bar{F}}$ denotes the complex conjugate of ${\cal F}$ and
the signatures in (\ref{seqM21})  are $(1,10,+)$ with time
in 9 for $n$ even and $(2,9, -)$ with times in 10 and 11 for $n$ odd. In
(\ref{seqM22}) the signatures are $(2,9, -)$ with times in 10 and 11
for $n$ even  and $(1,10,+)$  with time in 9 for $n$ odd. The detailed
analysis of the signatures for the M2 sequence is done in
Appendix~\ref{appw1} and the final results are summarised in Table \ref{tab:M2seq}.  
To interchange the role of even and odd $n$ in the above signatures,
one simply builds another M2 sequence starting from the  exotic M2 of
the theory $(2,9,-)$ whose metric is given in (\ref{M2metricE}). It
has two longitudinal times in 10 and 11 and one longitudinal spacelike
direction 9.

The coset representatives (\ref{seqM21}) and (\ref{seqM22})  yield from (\ref{vielbein})
and (\ref{embedding}) the metric and three-form potential 
\begin{align}
\label{M2odd}
\dd s^2_{[1+6n]}&= {\cal F}_{2n-1}\bar{\cal
 F}_{2n-1}H_{2n+1}^{1/3}[(\dd x^1)^2+(\dd x^1)^2]
+H_{2n+1}^{1/3}[(\dd x^3)^2+\dots+(\dd x^8)^2] \nn\\ 
 &\quad + H_{2n+1}^{-2/3}[(-1)^{(n+1)}
 (\dd x^9)^2+(-1)^n(\dd x^{10})^2+(-1)^n(\dd x^{11})^2]  \\
 A_{9\,10\,11}&=\frac{(-1)^{n}}{H_{2n+1}} \nonumber\\ 
\nonumber\\
\dd s^2_{[-1+6n]}&= {\cal F}_{2n-1}\bar{\cal
 F}_{2n-1}H_{2n}^{1/3}[(\dd x^1)^2+(\dd x^1)^2]
 +H_{2n}^{1/3}[(\dd x^3)^2+\dots+(\dd x^8)^2]\label{M2even}  \nn \\ 
&\quad + H_{2n}^{-2/3}[(-1)^n (\dd x^9)^2+(-1)^{n+1}(\dd x^{10})^2+(-1)^{(n+1)}
 (\dd x^{11})^2] \\ 
A_{9\,10\,11}&=\frac{(-1)^{n+1}}{H_{2n}}\nonumber
\end{align}
with $n\geq 0$ in \eqref{M2odd}, $n>0$ in \eqref{M2even} and   $H_p={\cal R}e\, {\cal E}_p$. We stress that an important effect
of the action of the affine Weyl group on the BPS solutions is the
change in the conformal factor which is expressed through ${\cal
  F}_{2n-1}\bar{\cal  F}_{2n-1}$. 

For each level on the M2-sequence, these equations satisfy the
equations of motion of eleven-dimensional  supergravity or of its exotic counterpart
outside the singularities of the functions $H_p$.  There, the factor
${\cal F}_{2n-1}\bar{\cal F}_{2n-1}$ can indeed be eliminated by a
change of coordinates and the functions $H_p$ are still harmonic
functions in the new coordinates. The solutions (\ref{M2odd}) and (\ref{M2even})
have then the same dependence on $H_p$ as the M2 metric and 3-form
have on $H_1\equiv H$ and differ thus from the M2 solution only
through the choice of the harmonic function. They therefore solve the
Einstein equations. This is also discussed more abstractly in
Section~\ref{analyticsec}. 
 
To obtain these results, we have chosen a particular path to reach
from level 1 the end of any horizontal line in Figure \ref{second}. Along this path,
all signs of the fields in the representatives were fixed by the
choice $+1/H$ at level 1 and by the Hodge duality relations
(\ref{dual1}) and (\ref{dual2}).  Thus, we do not have to
explicitly take into account signs which might affect higher level
tower fields in (\ref{borel3n}) and (\ref{borel6n}). 

The consistency of the procedure used in this chapter to obtain
solutions related by U-dualities viewed as Weyl transformations rests
however on the arbitrariness of the path chosen to reach from level 1
the end of any horizontal line in Figure \ref{second}. Dualities for levels $l>2$
are in principle defined by the Weyl transformations.  Consistency is
thus equivalent to commuting  Weyl transformations with
compensations. Compensations and Weyl transformations do indeed
commute, as proven in Appendix \ref{appc}.  
\subsection{The M5 sequence}\label{subsec:m5seqth}

We follow the same procedure as for the M2 sequence. We take as
representatives of the M5 sequence all the Weyl transforms of the M5
representative (\ref {2M5}) with fields (\ref{M5}) and time
coordinate 3. Following  in Figure \ref{ger1fig}a the dashed line towards positive
step generators, we encounter  Weyl transforms of the $\mathrm{SL}(2)$ subgroup
generated by $(-h_{11}- K^2{}_2\, , R^{[6]}_2\, ,R^{[6]}_{-2})$
represented by a dashed line in Figure \ref{ger1fig}b.  Theorem 2 determines from
(\ref{sla}) and (\ref{slb}) the Weyl transform of the Cartan
generators of (\ref{2M5})  and we write 
\begin{align}
\label{6Nborel}
{\cal V}_{2+6n}&= \exp \left[\frac{1}{2} \ln H \, (-h_{11}-(
2n+1) K^2{}_2)\right]\, \exp \left[\frac{1}{H}\,
  R^{[6]}_{2+6n}\right]\qquad &&n\ge 0\\ 
\label{3Nborel}
{\cal V}_{-2+6n}&= \exp\left[\frac{1}{2} \ln H\, (h_{11}-(2n-1)
  K^2{}_2)\right] \, \exp \left[\frac{1}{H}\, R^{[3]}_{-2+6n}\right]
\qquad &&n>0\, . 
\end{align}
We shall trade the tower fields $A^{[6]}_{2+6n}=A^{[3]}_{-2+6n}$  in
favour of the supergravity potential $A_{9\,10\,11}$ and construct
from them BPS solutions of eleven-dimensional  supergravity. 

For the M5 itself, we take the dual representative  $ {\cal
  V}^\prime_2$ expressed in terms of the 3-form potential, which is
given in (\ref{boreln2}). As previously the first two steps, levels
4 and 8, contain the essential ingredients of the whole sequence. 
\begin{figure}[h]
   \centering
   \includegraphics[width=13cm]{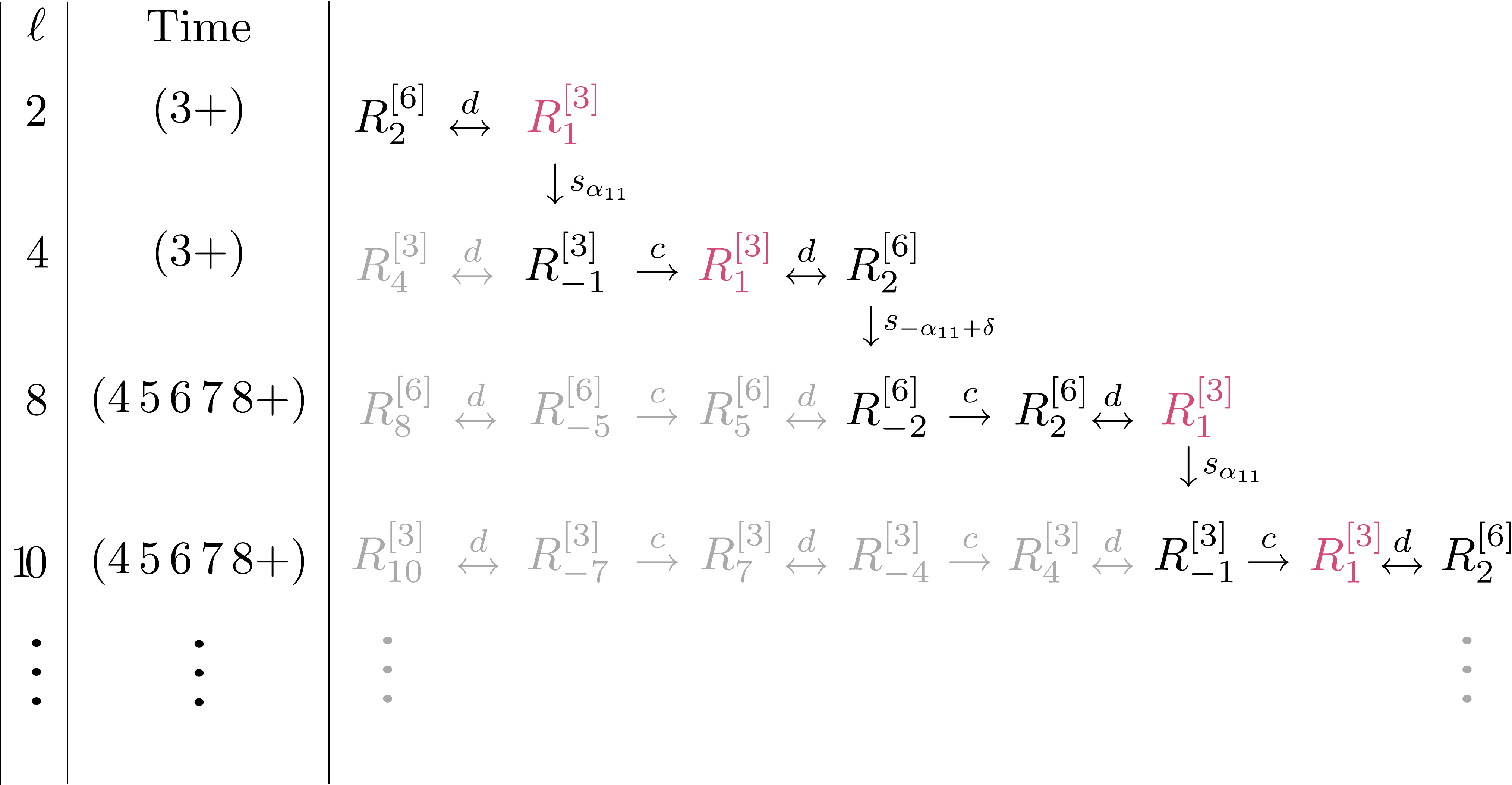}
 \caption {\sl \small The construction of BPS states for the M5
 sequence. The superscript $d$ labels a duality and the superscript
 $c$ labels a compensation.} 
 \label{fig:sequenceb}
\end{figure}

Following in Figure \ref{ger1fig}a the dashed line towards positive step generators,
we first encounter the Weyl reflexion $\alpha_{11} $ sending the level
2 generator $R^{[6]}_2$ to the level 4 generator $R^{[3]}_4$, or
equivalently, as exhibited in Figure \ref{fig:sequenceb}, the dual generator $R^{[3]}_1$ to
the generator $R^{[3]}_{-1}$. Applying this Weyl reflexion to
(\ref{boreln2}) and performing an $\mathrm{SO}(2)$ compensation we get the
representative 
\begin{equation}
\label{fin4}
\overline{\cal V}_4=\exp \left[-\frac{1}{2}\ln (H^2 + B^2)\,
  K^2{}_2\right] \, \exp \left[-\frac{1}{2} \ln \widetilde H\,
  (h_{11}+ K^2{}_2)\right]\,  \exp \left[\widetilde B\,
  R^{[3]}_1\right]\, , 
\end{equation} 
which yields $A_{9\, 10\,11}= \widetilde B$ and the
metric\footnote{This solution of 11 D supergravity has already been
  derived in a different context   \cite{Lozano-Tellechea:2000mc}.} 
\be \label{metric4} \begin{split}
{\rm Level }\ 4 :\qquad g_{11}&=g_{22}=(H^2+B^2) \widetilde H^{2/3}\, ,\\
-g_{33}&=g_{44}\dots =g_{88}= \widetilde H^{-1/3}\, ,\\ 
g_{99}&=g_{10\,10}=g_{11\,11}=\widetilde H^{2/3}\, ,
\end{split} \ee
with $\tilde{H}$ and $\tilde{B}$ as before.
The level 4 results are in  agreement with the interpretation of the
Weyl reflexion $s_{\alpha_{11}}$ as a double T-duality in the
directions 9 and 10 plus interchange of the two directions
\cite{Elitzur:1997zn, Obers:1998rn, Banks:1998vs, Englert:2003zs}. We
recover indeed  the level 4 metric and  3-form by applying  Buscher's
duality rules to the  M5 smeared in the directions $9,10$ and
$11$. This is shown in Appendix~\ref{appb}. 
The next step leads to level 8. As for the computation of the level 7
representative in the M2 sequence, we may skip the two first dualities
and the first compensation indicated in the third line of Figure \ref{fig:sequenceb}. It
suffices to perform  the Weyl reflexion $s_{- \alpha_{11}+\delta}$ on
the dual representative of (\ref{fin4}) followed by a $\mathrm{SO}(1,1)$
compensation and a Hodge duality. One gets 
\begin{equation}
\label{fin8}
\overline{\overline{\cal V}}_8=\exp\left[-\frac{1}{2}\ln (H^2+B^2)\,
  K^2{}_2\right]\, \exp\left[ -\frac{1}{2}\ln \widetilde{\widetilde
    H}\, (h_{11}+ K^2{}_2)\right]\, \exp \left[- \widetilde{\widetilde
    B} \, R^{[3]}_1\right]\, , 
\end{equation}
which yields $A_{9\, 10\,11}=-\widetilde{\widetilde B}$ and the metric
\be  \label{metric8} \begin{split}
{\rm Level }\ 8 :\qquad g_{11}&=g_{22}=(H^2+B^2) \widetilde{\widetilde
  H}^{2/3}\, ,\\
   g_{33}&=-g_{44}\dots =-g_{88}= \widetilde{\widetilde
  H}^{-1/3}\, ,\\ 
g_{99}&=g_{10\,10}=g_{11\,11}=\widetilde{\widetilde H} H^{2/3}\, .
\end{split}\ee
As shown in Appendix~\ref{appw2} the signature is now $(5, 6,+)$ with
times in 4, 5, 6, 7 and 8. 

The full M5 sequence is characterized by the roots $R^{[3]}_{-2+ 6n},
n>0$ and $R^{[6]}_{2+ 6n}, n\ge 0$. These are reached by following in
Figure \ref{ger1fig}a the dashed line starting at $R^{[6]}_2$ towards the positive
roots. The representative is defined by the Cartan generator given in
(\ref{sla}) or (\ref{slb}) and by the field $1/H$ multiplying
the positive root.  
As for the M2 sequence, the generalisation to all levels to the lowest
ones (\ref{boreln2}), (\ref{fin4}) and (\ref{fin8}) is
straightforward. As indicated in Figure \ref{fig:sequenceb},  one obtains iteratively the
representatives in terms of the supergravity 3-form by solely
performing  a single compensation by a $\mathrm{SO}(2)$ or $\mathrm{SO}(1,1)$ matrix
and a known Hodge duality defined by (\ref{dual1}) and
(\ref{dual2}).  One alternates after two steps between
representatives with a single time in 3 and  exotic ones with times in
4, 5, 6, 7 and 8. The nature of the compensation changes at each step.  
One has
\begin{align}
{\cal V}_{2+6n}&=\exp \left[-\frac{1}{2}\ln ({\cal F}_{2n-1}
  \bar{\cal F}_{2n-1})\, K^2{}_2\right] \, \exp \left[-\frac{1}{2} \ln
  {\cal R}e\, {\cal E}_{2n+1}\, (h_{11}+K^2{}_2)\right]\nonumber \\  
\label{seqM51}
& \quad  \cdot \exp \left[ (-1)^{n} {\cal I}m\, {\cal E}_{2n+1}\,
  R^{[3]}_1\right] \quad \quad n\geq 0\\ 
\nonumber\\
{\cal V}_{-2+6n}&=\exp \left[-\frac{1}{2}\ln ({\cal F}_{2n-1}
  \bar{\cal F}_{2n-1})\, K^2{}_2\right] \, \exp\left[-\frac{1}{2} \ln
  {\cal R}e\, {\cal E}_{2n}\, (h_{11}+K^2{}_2)\right]\nonumber\\ 
\label{seqM52}
& \quad \cdot \exp\left[ (-1)^{n+1}{\cal I}m\, {\cal E}_{2n}\,
  R^{[3]}_1\right] \quad \quad 
n>0
\end{align}
where in (\ref{seqM51}) one has the signatures $(1,10,+)$ with time
in 3 for $n$ even and $(5,6, +)$ with times in 4, 5, 6, 7 and 8 for $n$ odd,
and in (\ref{seqM52}) the signatures are  $(1,10, +)$  with time in
3 for $n$ odd and $(5,6,+)$ with times in 4, 5, 6, 7 and 8 for $n$ even. As
previously it is always possible to interchange at each pair of levels
the two signatures by choosing an exotic M5 to initiate the
sequence. The detailed analysis of the signatures for the M5 sequence
and  of the compensations required is done in Appendix~\ref{appw2} and
summarised  in Table \ref{tab:signam5seq}. 

These representatives yield the metric and 3-form potential for all
states on the M5 sequence. We get from (\ref{seqM51}) and
(\ref{seqM52}) 
\begin{align}
\label{M5odd}
\dd s^2_{[2+6n]}&= {\cal F}_{2n-1}\bar{\cal
 F}_{2n-1}H_{2n+1}^{2/3}[(\dd x^1)^2+(\dd x^1)^2]+H_{2n+1}^{2/3}[
 (\dd x^9)^2+(\dd x^{10})^2+(\dd x^{11})^2] \nn \\
 & \quad + H_{2n+1}^{-1/3}[(-1)^{n+1}(\dd x^3)^2
 +(-1)^n(\dd x^4)^2\dots+(-1)^n(\dd x^8)^2] \\ 
 \nonumber A_{9\,10\,11}&= (-1)^{n} B_{2n+1}\\
\nonumber\\
\label{M5even} 
\dd s^2_{[2+6n]}&= {\cal F}_{2n-1}\bar{\cal
  F}_{2n-1}H_{2n}^{2/3}[(\dd x^1)^2+(\dd x^1)^2]
+H_{2n}^{2/3}[(\dd x^9)^2+(\dd x^{10})^2+
  (\dd x^{11})^2] \nn \\& \quad +H_{2n}^{-1/3}[(-1)^n(\dd x^3)^2+(-1)^{n+1}(\dd x^4)^2+\dots+(-1)^{n+1}(\dd x^8)^2]  \\ 
\nonumber
A_{9\,10\,11}&= (-1)^{n+1} B_{2n}
\end{align}
with $n\geq 0$ in \eqref{M5odd}, $n>0$ in \eqref{M5even} and $B_p={\cal I}m\, {\cal E}_p$. 

For each level on the M5-sequence, these equations satisfy   the
equations of motion of eleven-dimensional  supergravity or of its exotic counterpart
outside the singularities of the harmonic functions $H_p$ and $B_p$.
There, the factor ${\cal F}_{2n-1}\bar{\cal F}_{2n-1}$ can indeed be
eliminated by a change of coordinates and the functions $H_p$ and
$B_p$ are still conjugate harmonic functions of the new
coordinates. (\ref{M5odd}) and (\ref{M5even}) have then the same
dependence on $H_p$ and $B_p$ as the M5 metric and 3-form have on
$H_1\equiv H$ and $B_1\equiv B$ and differ thus from the M5 solution
only through the choice of the harmonic functions. They therefore
solve the Einstein equations.  
\setcounter{equation}{0}
\section{The gravity tower}\label{sec:gravitytowersec}

The affine $A_1^+$ group generated by $R^{[3]}_1\equiv R^{9\,10\,11}$ and
$R^{[6]}_2\equiv R^{3\,4\,5\,6\,7\,8} $ spans three towers of
generators. We found BPS solutions for each positive generator of the
3-tower (\ref{tower1}) and of the 6-tower (\ref{tower3}). All
these generators correspond to real roots while those in the third
tower (\ref{tower2}) generators correspond to null roots of
square length zero.  
Each generator of the third tower at level $3(n+1)\quad n\ge 0$
belongs to an irreducible representation of $A_8\subset E_9$ whose
lowest weight is the real root $\alpha_4 + 2\alpha_5
+3\alpha_6+4\alpha_7 +5\alpha_8 +3\alpha_9 +\alpha_{10} +3\alpha_{11}
+n\delta$. We now show that the  lowest weight generators belong to a
$A_1^+$ subgroup  of $E_9$ generated by $R^{4\,5\dots 10 \, 11\vert
  11}$ which sits at level 3 and by $K^3{}_{11}$, which is defined by
the level 0 real root
$\alpha_3+\alpha_4+\alpha_5+\alpha_6+\alpha_7+\alpha_8+\alpha_9+\alpha_{10}$. 

These two generators are related as follows
\begin{align}
\label{ger1}
K^3{}_{11} &\leftrightarrow
\alpha_3+\alpha_4+\alpha_5+\alpha_6+\alpha_7+\alpha_8+\alpha_9+\alpha_{10}\equiv\lambda\,
,\\ 
\label{ger2}
R^{4\,5\dots 10 \, 11\vert 11} &\leftrightarrow \alpha_4 + 2\alpha_5
+3\alpha_6+4\alpha_7 +5\alpha_8 +3\alpha_9 +\alpha_{10}
+3\alpha_{11}=-\lambda+\delta \, , 
\end{align}
where the last equality in (\ref{ger2}) is easily checked using
(\ref{delta}). They are  
Weyl transforms of the generators $R^{9\,10\,11}$ and
$R^{3\,4\,5\,6\,7\,8}$ of the the $A_1^+$ group defined in
(\ref{geroch}). To see this, first perform the Weyl transformation
interchanging 9 and 3. The $A_1^+$ generators  (\ref{geroch}) are
transformed to (all Weyl transforms of step generators are written up
to a sign) 
\be \begin{split}
R^{9\,10\,11}\quad&\to  \quad R^{3\,10\,11}\\
R^{3\,4\,5\,6\,7\,8} \quad &\to \quad  R^{4\,5\,6\,7\,8\,9}\, ,
\end{split} \ee
defined by the roots $\alpha_3 +
\alpha_4+\alpha_5+\alpha_6+\alpha_7+\alpha_8 +\alpha_{11}$ and
$\alpha_4+2\alpha_5+3\alpha_6+4\alpha_7+5\alpha_8+4\alpha_9
+2\alpha_{10} +2\alpha_{11}$. 
Then perform the Weyl reflexion  $s_{\alpha_{11}}$ to get the generators
\be \begin{split}
R^{3\,10\,11} \quad &\to  \quad  K^3{}_9\\
R^{4\,5\,6\,7\,8\,9}  \quad&\to \quad  R^{4\,5\,6\,7\,8\,9\,10\,11\vert 9}\, ,
\end{split}\ee
defined by the roots $\alpha_3 +
\alpha_4+\alpha_5+\alpha_6+\alpha_7+\alpha_8$ and  
$\alpha_4+2\alpha_5+3\alpha_6+4\alpha_7+5\alpha_8+4\alpha_9
+2\alpha_{10} +3\alpha_{11}$. Finally perform the Weyl
transformation exchanging 9 and 11 to get 
\be \begin{split}
K^3{}_9  \quad &\to  \quad  K^3{}_{11}\\
R^{4\,5\,6\,7\,8\,9\,10\,11\vert \,9}  \quad&\to \quad
R^{4\,5\,6\,7\,8\,9\,10\,11\vert \,11}\, , 
\end{split}\ee
whose defining roots are $\lambda$ and $-\lambda +\delta$. 
The transformed Cartan generators are $K^3{}_3 -K^{11}{}_{11}$ and
$-K^2{}_2 +K^{11}{}_{11}-K^3{}_3$. Under these transformations, the
M2-brane generator is mapped onto the  Kaluza-Klein wave generator in
the direction 11. The M5-brane generator is mapped to the dual
Kaluza-Klein monopole generator $R^{4\,5\,6\,7\,8\,9\,10\,11\vert
  \,11}$. These generate the `gravity $A_1^+$ group' conjugate in $E_9$
to the `brane $A_1^+$ group' (\ref{geroch}). 

We now find the BPS solutions of eleven-dimensional  pure gravity (which are of course
solution of eleven-dimensional  supergravity) associated to each positive real root of
the gravity $A_1^+$ group. One could redo  the analysis of the M2-M5
system starting from the representatives of the KK-wave and
KK6-monopole given in (\ref{bgr1}) and  (\ref{bgr2}) and the
duality relations (\ref{Pdual2}).  It is however simpler to take
advantage of the Weyl mapping of the two $A_1^+$ subgroups of $E_9$ 
\begin{eqnarray}
\label{bragramap}
R^{9\,10\,11} &\leftrightarrow&   K^3{}_{11} \\
\label{bragra1}
R^{4\,5\,6\,7\,8\,9} &\leftrightarrow&   R^{4\,5\,6\,7\,8\,9\,10\,11\vert 11}\\
 \alpha_{11} \leftrightarrow   \lambda\quad &, & \quad\delta\leftrightarrow\delta\, .
  \label{bragra2}
\end{eqnarray}

The generators $R^{[3]}_{1+3n}$ of the 3-tower (\ref{tower1}) are
mapped to generators of level $3n$. We label these generators
$R^{[0]}_{3n}$ ($R^{[0]}_{0}\equiv K^3{}_{11}$). The generators
$R^{[6]}_{-1+3n}$ of the 6-tower (\ref{tower3}) are also mapped to
generators of level $3n \, (n>0)$. We label these generators $\bar
R^{[8,1]}_{3n}$ ($\bar R^{[8,1]}_{3}\equiv
R^{4\,5\,6\,7\,8\,9\,10\,11\vert 11}$).  In the mapping the signature
changes as shown in Appendix~\ref{appsg}. In particular, the KK-wave
$R^{[0]}_{0}$ yields a single time in 3 and  the KK6-monopole $\bar
R^{[8,1]}_{3}$ becomes exotic with two times 9 and 10. This mapping of
the M2-M5 sequences of Figure \ref{ger1fig}a to the gravity sequences is illustrated
in Figure \ref{fig:mappingbranegrav}. 
To the M2 sequence corresponds a wave sequence starting with
the KK-wave and to the M5 sequence a monopole sequence staring with
the (exotic) KK6-monopole. Note that there is a duplication in each
sequence of states with the same level $3n$ for $n>0$. We shall show
that this duplication is spurious in the sense that the two states are
related by a switch of coordinates.  

\begin{figure}[h]
\begin{center}
{\scalebox{0.60}
{\includegraphics{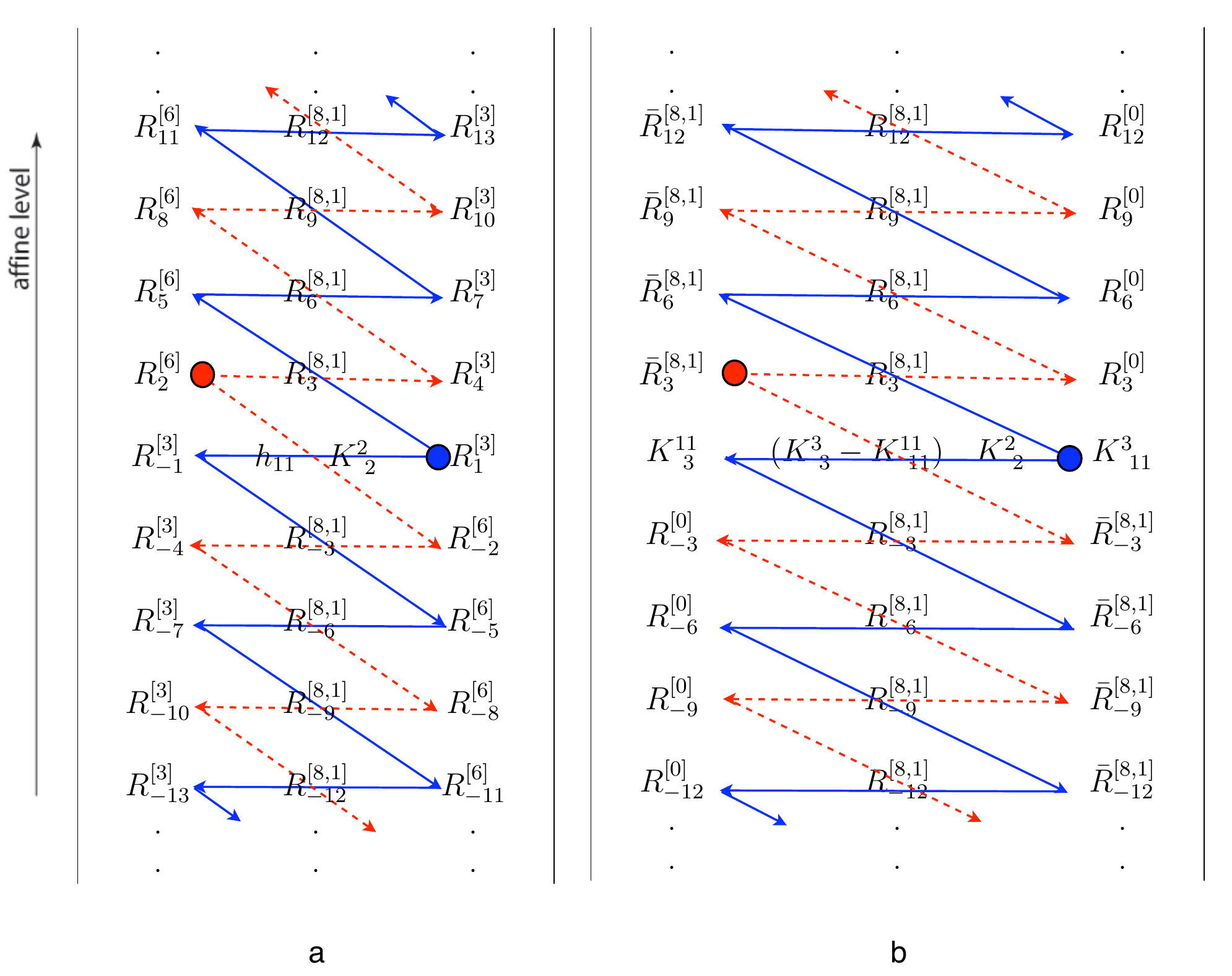}}}
\end{center}
 \caption {\sl \small Mapping of the brane $A_1^+$ group (Figure \ref{fig:mappingbranegrav}a) to the
 gravity $A_1^+$ group (Figure \ref{fig:mappingbranegrav}b). M2 and wave sequences are depicted by
 solid lines, M5 and monopole sequences by dashed lines. In Figure \ref{fig:mappingbranegrav}a
 (Figure \ref{fig:mappingbranegrav}b) horizontal lines represent Weyl reflexions by
 $s_{\alpha_{11}}$ ($s_\lambda$), diagonal lines by
 $s_{-\alpha_{11}+\delta}$ ($s_{-\lambda+\delta}$).} 
 \label{fig:mappingbranegrav}
\end{figure}

From the correspondence we immediately get from the representatives of
the M2 sequence (\ref{borel3n}), (\ref{borel6n}), and of the M5
sequences (\ref{6Nborel}) and (\ref{3Nborel}), the representatives of
the KK-wave sequence  (\ref{KW1}), (\ref{KW2}) and of the
KK-monopole sequence (\ref{KM1}), (\ref{KM2}) in terms of the
$R^{[0]}_{3n}$ and $\bar R^{[8,1]}_{3n'}$ generators 
\begin{align}
\label{KW1}
{\cal V}_{6n}&= \exp \left[\frac{1}{2} \ln H \,
  (K^3{}_3-K^{11}{}_{11} -2n K^2{}_2)\right]\, \exp
\left[\frac{1}{H}\, R^{[0]}_{6n}\right]\qquad &&n\ge0 \\ 
\label{KW2}
{\cal V}_{6n'}&= \exp \left[\frac{1}{2} \ln H\,
  (-K^3{}_3+K^{11}{}_{11} -2n' K^2{}_2)\right] \, \exp
\left[\frac{1}{H}\, \bar R^{[8,1]}_{6n'}\right]\qquad &&n'>0\\ 
\label{KM1}
{\cal V}_{3+6n'}&= \exp \left[\frac{1}{2} \ln H \,
  (-K^3{}_3+K^{11}{}_{11} -(2n'+1) K^2{}_2)\right] \, \exp
\left[\frac{1}{H}\, \bar R^{[8,1]}_{6n'+3}\right]\ &&n'\ge0\\ 
\label{KM2}
{\cal V}_{-3+6n}&= \exp \left[\frac{1}{2} \ln H\,
  (K^3{}_3-K^{11}{}_{11} -(2n-1) K^2{}_2)\right]\, \exp
\left[\frac{1}{H}\, R^{[0]}_{6n-3}\right] \quad &&n>0\ \, . 
\end{align}
In these equations we distinguish the representatives of the [0]-tower
depicted in the right column of Figure \ref{fig:mappingbranegrav}b from those of the [8,1]-tower
depicted in the left column by labelling the former by $n$ and the
latter by $n'$.

To get the representatives for the wave sequence in terms of the
gravitational potential $A_{3}^{(11)}$ given by (\ref{kka}),
we apply the mapping (\ref{bragramap}), (\ref{bragra1}) and
(\ref{bragra2}) to the representative of the M2 sequence in terms of
the supergravity 3-form potential\footnote{We have added an
  integration constant $-1$ to the field $A_3^{(11)}$ as in the
  discussion below (\ref{potential}).} (\ref{seqM21}) and
(\ref{seqM22}) 
\begin{eqnarray}
\label{seqG01}
{\cal V}_{6n}&=&\exp \left[-\frac{1}{2}\ln ({\cal F}_{2n-1} \bar{\cal
 F}_{2n-1})\, K^2_{\ 2}\right] \, \exp \left[\frac{1}{2} \ln {\cal
 R}e\, {\cal E}_{2n +1}\, (K^3{}_3-K^{11}{}_{11})\right] \nonumber \\ 
 && \exp \left[(-1)^{n} \bigg(\frac{1}{{\cal R}e\, {\cal
 E}_{2n+1}}-1\bigg)\, K^{3}_{\ 11}\right] \qquad n\geq 0\\ 
\nonumber\\
\label{seqG02}
{\cal V}_{6n^\prime}&=&\exp \left[-\frac{1}{2}\ln ({\cal F}_{2n'-1}
  \bar{\cal F}_{2n'-1})\, K^2_{\ 2}\right] \, \exp [\frac{1}{2} \ln
  {\cal R}e\, {\cal E}_{2n^\prime}\,(K^3{}_3-K^{11}{}_{11})] 
\nonumber\\  
&& \exp \left[ (-1)^{n^\prime+1} \bigg( \frac{1}{{\cal R}e\, {\cal
  E}_{2n^\prime}} -1 \bigg)\, K^{3}_{\ 11}\right] \qquad n^\prime>0 
\end{eqnarray}
where in (\ref{seqG01}) one has the signatures $(1,10,+)$ with time
in 3 for $n$ even and in 11 for $n$ odd, and in (\ref{seqG02}) the
signatures are  $(1,10,+)$ with time in 3 for $n^{\prime}$ odd and in
11 for $n^\prime$ even (see Appendix~\ref{appsg}). 

It is proven in Appendix~\ref{apprg} that the KK-wave sequence
contains a redundancy of the solutions for $n>0$, namely  
(\ref{seqG01}) and (\ref{seqG02})  lead to identical metric up
to interchange of the time coordinates 3 and 11. The full wave
sequence for $n>0$ has metric: 
\begin{align}
\label{KKW}
\dd s^2_{[6 n^\prime]}&= {\cal F}_{2n'-1}\bar{\cal F}_{2n'-1} \Big
 [(\dd x^1)^2+(\dd x^2)^2 \Big ] + (-1)^{n^\prime}
 H^{-1}_{2n^\prime}(\dd x^3)^2+\Big [(\dd x^4)^2\dots+(\dd x^{10})^2 \Big ]
 \nonumber\\ 
 &\quad + (-1)^{n^\prime +1} H_{2n^\prime} \Big [ \dd x^{11} -  \Big(
 (-1)^{n^\prime +1}   H^{-1}_{2n^\prime}+ (-1)^{n^\prime} \Big)\dd x^3
 \Big]^2  \, , 
 \end{align}
where $H_p={\cal R}e\, {\cal E}_p$. For $n=0$ it is given by
(\ref{G0metric1}). All metrics in the KK-wave are solutions of eleven-dimensional 
supergravity.  The factor ${\cal F}_{2n'-1}\bar{\cal F}_{2n'-1}$ can
again be eliminated by a (singular) coordinate change, preserving the
harmonic character of $H_p$.  

From the representatives of the M5 sequence in terms of the
supergravity 3-form potential, (\ref{seqM51}) and (\ref{seqM52}),
we get the representatives for the monopole sequence in terms of the
gravitational potential $A_{3}^{~(11)}$ 
\begin{align}
\label{seqG31}
{\cal V}_{3+6n'}&=\exp \left[-\frac{1}{2}\ln ({\cal F}_{2n'-1}
  \bar{\cal F}_{2n'-1})\, K^2_{\ 2}\right] \, \exp \left[-\frac{1}{2}
  \ln {\cal R}e\, {\cal E}_{2n'+1}\,
  (K^3{}_3-K^{11}{}_{11}+K^2_2)\right]\nonumber\\ 
&\quad  \   \exp \left[ (-1)^{n'} {\cal I}m\, {\cal E}_{2n'+1}\,  K^{3}_{\
  11}\right]\qquad n'\geq0\\ 
\nonumber\\
\label{seqG32}
{\cal V}_{-3+6n}&=\exp \left[-\frac{1}{2}\ln ( {\cal F}_{2n-1}
  \bar{\cal F}_{2n-1})\, K^2_{\ 2}\right] \, \exp \left[-\frac{1}{2}
  \ln {\cal R}e\, {\cal E}_{2n}\,
  (K^3{}_3-K^{11}{}_{11}+K^2_2)\right]\nonumber\\  & \quad \  \exp
\left[(-1)^{n+1 } {\cal I}m\, {\cal E}_{2n}\,  K^{3}_{\ 11}\right] 
\qquad n>0
\end{align}
where in (\ref{seqG31}) one has  the signatures  $(2,9, -)$  with
time in 9 and 10 for $n'$ even and $(5,6,+)$ with time in 4, 5, 6, 7 and 8 for
$n'$ odd, and in (\ref{seqG32}) the signatures are    $(2,9, -)$
with time in 9 and 10 for $n$ odd and $(5,6,+)$ with time in 4, 5, 6, 7 and 8 for
$n$ even (see Appendix~\ref{appsg}). 

In analogy with the KK-wave sequence,  the metric in
(\ref{seqG31}) and (\ref{seqG32}) are equivalent up to a
redefinition of the time coordinates (see Appendix~\ref{apprg}).
There is thus only one gravity tower, the left and the right tower of
Figure \ref{fig:mappingbranegrav}b are equivalent, each of them contains the full wave and
monopole sequences. 

The full monopole sequence has the metric: 
\begin{align}
\dd s^2_{[3+ 6n']} &= {\cal F}_{2n'-1}\bar{\cal F}_{2n'-1}
 H_{2n'+1}\Big [(\dd x^1)^2+(\dd x^2)^2 \Big ] + H_{2n'+1}(\dd x^3)^2 \nn\\
 &\quad +(-1)^{n'}
 \Big [(\dd x^4)^2\dots+(\dd x^{8})^2 \Big ] +(-1)^{n'+1}  \Big [(\dd x^9)^2+(\dd x^{10})^2 \Big ]  \nn \\ 
 & \quad +
 H^{-1}_{2n'+1} \Big [ \dd x^{11} -  \Big( (-1)^{n'}   B_{2n'+1}\Big)\dd x^3
 \Big]^2 \, , \label{mono} 
 \end{align}
 where $B_p = {\cal I}m\, {\cal E}_p$.\\
 Again the metric of the monopole sequence solve the Einstein equations.
 
 The generators $R_{1+3p}^{[3]}$, $R_{2+3p}^{[6]}$, $\bar
 R_{3+3p}^{[8,1]}$, $p\ge 0$, and $K^3{}_{11}$ span the M2, M5 and
 gravity towers for positive real roots and define distinct BPS
 solutions. All positive real roots of $E_9$ can be reached from these
 by permuting coordinate indices in $A_8$ or equivalently by
 performing Weyl transformations $s_{\alpha_i}$ from the  gravity line
 depicted in Figure \ref{ffirst} with nodes 1 and 2 deleted. In this way we reach
 all $E_9$ positive real roots and the related BPS solutions. In what
 follows we shall keep the above notation for all  towers of positive
 real roots differing by $A_8$ indices, and specify the coordinates
 when needed. 
 \setcounter{equation}{0}
\section{Analytic structure of BPS solutions and the Ernst potential} \label{analyticsec}
 
We have obtained an infinite U-duality multiplet of $E_9$ BPS solutions of
eleven-dimensional  supergravity depending on two non-compact space variables. This was
achieved by analysing various $A_1^+ \equiv A_1^{(1)}$
subalgebras of $E_9$, which allow us to reach all positive roots within such a
subalgebra from sequences of Weyl reflexions starting from basic BPS
solutions reviewed in Chapter \ref{chap:basicbps}. A striking feature of the method is that
each solution is determined by a pair of conjugate harmonic functions $H_p$
and $B_p$ which can be combined into an analytic function ${\cal E}_p
=H_p+iB_p$, where $p$ characterises the level of the solution. This feature
emerges from the action of the affine $A_1^+$ subgroup on the
representatives and is clearly not restricted to supergravity. In this
section, we establish the link with another
$A_1^+$ subgroup of $E_9$, namely the Geroch group of general relativity. As
is well known  \cite{Geroch:1970nt, Geroch:1972yt, Breitenlohner:1986um},
the latter acts on stationary axisymmetric (or colliding plane wave)
solutions in four space-time dimensions (which can be embedded consistently
into eleven-dimensional  supergravity) via `non-closing dualities'
generating infinite towers of higher order dual potentials. Here
we explain the action of the Geroch group on BPS solutions, for which
the so-called {\em Ernst potential} (see (\ref{Ernst}) below)
is an analytic function, and hence is entirely analogous to the function
${\cal E}$ encountered above. As we will see this action
(so far not exhibited in the literature to the best of our knowledge)
`interpolates' between free field dualities and the full non-linear
action of the Geroch group on non-analytic Ernst potentials --- exactly
as for the M2-M5 sequence discussed in Section~\ref{m2m5sec}. To keep
the discussion simple we will restrict attention to four-dimensional
Einstein gravity with two commuting Killing vectors, that is,
depending only on two (spacelike) coordinates.

Before we specialise to the case of BPS solutions we present the more
general formalism. The general line element in this case is of the
form 
\be\label{linelem}
\dd s^2 = H^{-1} e^{2\s} (\dd x^2+\dd y^2) + (-\rho^2H^{-1} +H \mmp^2)\dd t^2
+2H \mmp\dd t\,\dd z +H\dd z^2.
\ee
Here, $\p_t$ and $\p_z$ are Killing vectors, hence the metric coefficients
depend only on the space coordinates $(x,y)\equiv (x^1,x^2)$.
Furthermore, we have adopted a conformal frame for the $(x,y)$ 
components of the metric, with conformal factor $e^{2\s}$. $\mmp$ is 
the called the Matzner--Misner potential and related to the Ehlers
potential $B$ through the duality relation
\be\label{dualrel}
\eps_{ij}\p_j B = \rho^{-1}H^2 \p_i \mmp,
\ee
where $i,j=1,2$. There is no need to raise or lower indices, as the
metric in $(x,y)$ space is the flat Euclidean metric, with
$\eps_{12}=\eps^{12}=1$. Therefore the inverse duality relation 
is $\eps_{ij}\p_j \mmp = - \rho H^{-2}\p_i B$.

The vacuum Einstein equations for the line element (\ref{linelem})
in terms of the Matzner--Misner potential $\mmp$ read
\be\label{eommm} \begin{split}
H\p_i (\rho \p_iH) &= \rho \left(\p_iH \p_iH -\rho^{-2}H^4\p_i
  \mmp\p_i \mmp\right)\\
\rho H^{-1}\p_i (\rho \p_i \mmp) &= 2 \rho \p_i \left(\frac\rho H\right) \p_i
\mmp \, , 
\end{split}\ee
Rewritten in terms of the Ehlers potential $B$ these give, using
(\ref{dualrel}),
\be\label{eome} \begin{split}
H \p_i(\rho\p_i H) &= \rho \left(\p_iH \p_iH - \p_iB \p_iB \right)\\
H\p_i(\rho \p_i B) &= 2 \rho \p_iH \p_iB \, ,
\end{split}\ee
where the two sets of equations (\ref{eommm}) and (\ref{eome}) are related
by the so-called Kramer--Neugebauer transformation $B\leftrightarrow \mmp,
H \leftrightarrow \rho/H$. In addition, there are equations for $\rho$ 
and the conformal factor $\s$. These are two (compatible) first order
equations for the conformal factor  
\be\label{constraint} \begin{split}
\rho^{-1}\p_{(x} \rho \p_{y)}\s &= \frac14 (H^{-1}\p_x H)(H^{-1}\p_y H)
  +\frac14 (H^{-1}\p_x B)(H^{-1}\p_y B),\\
\rho^{-1}\p_x \rho \p_x\s - \rho^{-1}\p_y \rho \p_y\s &=
       + \frac14 (H^{-1}\p_xH)^2 + \frac14 (H^{-1}\p_x B)^2 \\
       &\quad  -\frac14 (H^{-1}\p_y H)^2 -  \frac14 (H^{-1}\p_y B)^2 \, ,
\end{split}\ee
while $\rho$ satisfies the two-dimensional Laplace equation without source
\be
\p_i\p_i\rho =0\, .
\ee
A second order equation for $\s$ can be deduced by varying $\rho$, or
alternatively from the constraints and the dynamical equations for the 
metric (\ref{eommm}) [or (\ref{eome})]; it reads
\be\label{conf}
\p_i\p_i \s = -\frac14\rho H^{-2} \big( \p_i H\p_iH + \p_i B \p_i B \big),
\ee
If $\rho$ is different from a constant (as is the case generally with
axisymmetric stationary or colliding plane wave solutions), we can 
integrate the first order (\ref{constraint}), which determine
the conformal factor up to one integration constant; the second order
equation (\ref{conf}) is then automatically satisfied as a consequence 
of the other equations of motion. On the other hand, as we 
will see below, the BPS solutions are characterized by $\rho = \rm constant$, 
for which the l.h.s. of (\ref{constraint}) vanishes identically (whence 
the r.h.s. must also vanish identically). In this case, we are left with
the second order (\ref{conf}), and the conformal factor is only
determined modulo a harmonic function in $(x,y)$. 

The equations of motion (\ref{eome}) can be rewritten conveniently
in terms of the {\em complex Ernst potential} [cf .(\ref{analytic})]
\be\label{Ernst}
\cE = H + i B,
\ee
satisfying the Ernst equation
\be\label{ernsteq}
H \p_i (\rho \p_i \cE) = \rho \p_i\cE\p_i\cE.
\ee
As we will see below this equation is trivially satisfied for BPS
solutions in the sense that both sides vanish identically.

\subsection{BPS solutions}

In our analysis of the 11-dimensional gravity tower, the lowest level
BPS solution is the KK-wave   (\ref{G0metric1}). It stems from the
generator $K^3_{~11}$ with time in 3. In 4D gravity with $x^1,x^2$ as
non compact space variables, the corresponding wave solution is
associated to the Chevalley generator $K^3_{~4}$ depicted  in Figure \ref{fig:a1+++} by
the node 3. Taking the timelike direction to be 3, we get
[$(x,y)\equiv (x^1,x^2)$, $(t,z) \equiv (x^3,x^4)$] 
\be
\label{kksol2d}
\dd s^2 =  \dd x^2 +\dd y^2 + (H-2)\dd t^2 - 2 (1-H)\dd t\,\dd z + H\dd z^2.
\ee
Here, $H=H(x,y)$ is a {\em harmonic function} in $x,y$, which, for the
brane with a  source at $x=y=0$ we choose to be $H=\frac12 
\ln (x^2+y^2) = \ln |\z|$ in terms of the complex coordinate
\be
\z = x + i y.
\ee 
Comparing (\ref{kksol2d}) with (\ref{linelem}),  
we see that for this  BPS solution the general fields $\mmp,\s$
and $\rho$ are expressed 
in terms of $H$ as\footnote{The constant $b$ for $\mmp$ should be chosen
as $1$ in order to obtain an asymptotically flat solution in more
than four space-time dimensions. From the point of view of the
two-dimensional reduction,
however, it does not matter and can be chosen arbitrarily. Note also,
that constant shifts of $\mmp$ are part of the Matzner-Misner $\mathrm{SL}(2)$,
see below.}
\be\label{kksol2da}
 e^{2\s}=H,\quad\quad \mmp= b - H^{-1},\quad\quad
\rho=1.
\ee
Using the duality relation (\ref{dualrel}) one obtains the Ehlers
potential $B$ up to an integration constant. Indeed, as already
mentioned above, with (\ref{kksol2da}), 
the duality relations (\ref{dualrel}) just become the  Cauchy--Riemann 
equations for the Ernst potential (\ref{Ernst}), to wit
\be
\p_x B = -\p_y H \quad , \quad \p_y B = \p_x H \, ,
\ee
or, in short notation, $\eps_{ij}\p_jH = - \p_i B$. Conversely,
for (\ref{dualrel}) to reduce to the Cauchy--Riemann relations, we must
have $\mmp= 1/H + {\rm constant}$ and $\rho$ constant. {\it Therefore, the
Cauchy--Riemann equations for the Ernst potential are equivalent to
the BPS (`no force')  
condition and may thus be taken as the defining equations for BPS 
solutions}. In a supersymmetric context these (first order) equations
would be equivalent to the Killing spinor conditions defining the
BPS solution.

For $H =  \frac12 \ln (x^2 + y^2)$, we immediately obtain
\be
B = \arctan\left(\frac{y}{x}\right) + {\rm constant} = \arg(\z) + {\rm constant} \, ,
\ee
whence the Ernst potential is simply
\be
\cE(\z) = \ln |\z| + i \arg(\z) + {\rm constant} = \ln \z + {\rm constant} \, ,
\ee
and so is an {\em analytic} function of $\z$. It is then easy
to see that the equations of motion and the constraint equations are 
satisfied for {\em any} analytic Ernst potential $\cE$ if $\rho$ is
constant (and in particular, with $\rho=1$). Namely both the Ernst equation 
(\ref{ernsteq}) as well as (\ref{constraint}) reduce to the identity 
$0=0$ for all such solutions. Because (\ref{constraint}) is void, the 
conformal factor must then be determined from the second order equation 
(\ref{conf}). For holomorphic $\cE$, the equation (\ref{conf}) can be
rewritten as
\be\label{conf1}
\p_\z \p_{\bar{\z}} \s = -\frac12\rho\,
\frac{\p_\z \cE \p_{\bar{\z}} \bar{\cE}}{(\cE + \bar{\cE})^2}
\ee
and only in this case the solution to this equation can be given in
closed form. It reads 
\be
\s(\z,\bar{\z}) = \frac12\ln \, (\cE + \bar{\cE})\, .
\ee
The ambiguity involving harmonic functions left by (\ref{conf})
is related to the covariance of the equations of motion under 
{\it conformal analytic coordinate transformations} of the complex 
coordinate $\z=x+iy$, which leave the 2-metric in diagonal form, viz.
\be
\z \longrightarrow \z'=f(\z) \;\; .
\ee
As is well known, the conformal factor transforms as
\be
\s(\z,\bar\z) \longrightarrow \s (\z,\bar\z) 
    + \frac12 \ln \big| f(\z)\big|^2 \, ,
\ee
under such transformations, where the second term on the r.h.s. is
indeed harmonic. We already used this fact  when we removed the
conformal factors $ {\cal F}_{2n-1}\bar{\cal F}_{2n-1}$  to prove that
the metric in (\ref{M2odd}), (\ref{M2even}), (\ref{M5odd}),
(\ref{M5even}),  (\ref{KKW}) and (\ref{mono}) solve the Einstein
equations outside the singularities of the harmonic functions. 
\subsection{Action of Geroch group}

The Geroch group for $(3+1)$-dimensional gravity in
the stationary axi-symmetric case discussed here is affine 
$\widehat{\mathrm{SL}(2)}$ with central extension ($\equiv A_1^+$); this is the same
structure 
we encountered for the M2 and M5 towers. It is depicted in Figure \ref{fig:a1+++} by
the Dynkin diagram formed by the nodes 3 and 4. Extending the diagram
with node 2 to the overextended $A_1^{++}$, we may as previously
identify the central charge with a Cartan generator of $A_1^{++}$
($-K^2_{~2}$ in $E_{10}\equiv E_8^{++}$). Adding the node 1 leads to
the very extended $A_1^{+++}$ which is the pure gravity counterpart
(in $ D=4$) of $E_{11}$. 
\begin{figure}[h]
\centering
 \includegraphics [width=4cm]{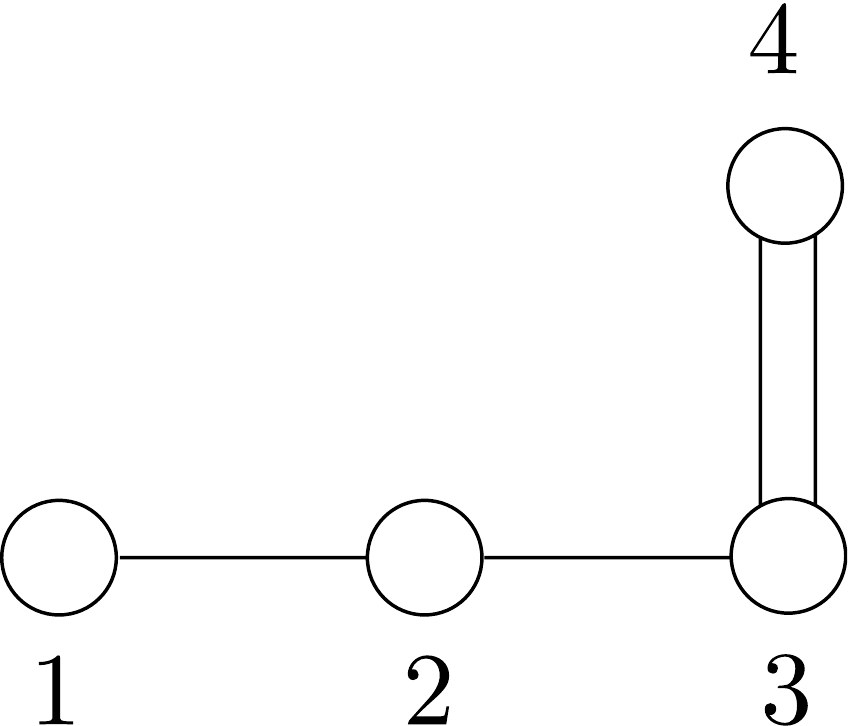}
\caption{\sl \small Dynkin diagram of $ A_1^+$ and
  its `horizontal' extensions $A_1^{++}$ and $A_1^{+++}$. The
  Matzer--Misner $\mathrm{SL}(2)$ is represented by the node 3  and the Ehlers
  $\mathrm{SL}(2)$ by the node 4.}  
  \label{fig:a1+++}
\end{figure}

The two distinguished $\mathrm{SL}(2)$ subgroups of $A_1^+$
corresponding to the Matzner--Misner and the Ehlers cosets appear
in different real forms. The timelike Killing vector
$\p_t$ turns the Matzner--Misner coset into $\mathrm{SL}(2)/\mathrm{SO}(1,1)$ (with
non-compact denominator group), whereas the Ehlers coset is
$\mathrm{SL}(2)/\mathrm{SO}(2)$ (with compact denominator group).
With 3 or 4 as time direction,  the temporal involution (\ref{tabeq:inomegk})
leaves indeed invariant  the Lorentz generator $\Omega\, (
K^3{}_4+K^4{}_3)=K^3{}_4+K^4{}_3$ in the Matzer--Misner $ \mathrm{SL}(2)$, but
preserves the rotation generator 
 $\Omega \,(R^{4 \vert 4}-R_{4 \vert 4})=R^{4 \vert 4}-R_{4 \vert 4}$
in the Ehlers  $ \mathrm{SL}(2)$. Here $R^{4 \vert 4}$  is  the simple
positive step operator corresponding  in four dimensions to the
11-dimensional generator   $R^{4\,5\,6\,7\,8\,9\,10\,11\vert \,11}$ of
Section \ref{sec:gravitytowersec}.

For the underlying split real algebra $A_1^+$ we use a simple
Chevalley--Serre basis consisting of $e_i,f_i, h_i$ ($i=3,4$) and
derivation $d$ (see also Appendix~\ref{affapp}). The central element 
is $c=h_3+h_4$. The index `4' refers to the Ehlers
$\mathrm{SL}(2)$ and the index `3' refers to the Matzner--Misner
$\mathrm{SL}(2)$, as depicted in Figure \ref{fig:a1+++}.

\medskip
\noindent
{\bf $\bullet$ Ehlers group}

The Ehlers $\mathrm{SL}(2)$ acts by M\"obius transformations on the
(analytic) Ernst
potential. Using the standard notation for M\"obius
transformation generators\footnote{Note that these are M\"obius
  transformations on $\cE$ and {\em not} on the complex coordinate $\z$.
  The system admits an additional
  {\em conformal symmetry} acting on the complex coordinate $\z$
  \cite{Julia:1996nu}, which shows again that {\em any} analytic Ernst
  potential solves the equations of motion.}
  \be
L_{-1} = - \frac{\p}{\p\cE},\quad\quad L_0 = -\cE \frac{\p}{\p\cE},
\quad\quad L_{1} = - \cE^2 \frac{\p}{\p\cE} \, ,
\ee
the Ehlers generators are\footnote{The factors $i$ are understood
from studying the invariant bilinear form for $(L_{-1},L_0,L_1)$
which is non-standard from the Kac--Moody point of view.}
\be
e_4 = i L_{-1},\quad\quad h_4 = 2 L_0,\quad\quad f_4 = i L_1 \, ,
\ee
with the resulting transformation of the real and imaginary
components of $\cE$
\be \begin{split} \begin{aligned}
f_4 H &= 2 H B,&\quad\quad f_4 B = & B^2-H^2,\\
h_4 H &= -2 H,&\quad\quad  h_4 B = & -2 B,\\
e_4 H &= 0,&\quad\quad e_4 B =& -1.
\end{aligned} \end{split} \ee

\noindent
{\bf $\bullet$ Matzner--Misner group}

The infinitesimal action of the Matzner--Misner group on general
$\mathrm{SL}(2)/\mathrm{SO}(1,1)$ coset fields $H$ and $\mmp $ is (for $\rho=1$)
\be \label{mmtrm} \begin{split} \begin{aligned}
f_3 H &= -2 H \mmp,&\quad\quad f_3 \mmp = & \mmp^2+H^{-2},\\
h_3 H &= 2 H,&\quad\quad  h_3 \mmp = & -2 \mmp,\\
e_3 H &= 0,&\quad\quad e_3 \mmp =& -1.
\end{aligned} \end{split} \ee

In order to compute its action on the Ernst potential $\cE$, we have to
exploit the duality relation (\ref{dualrel}) between the potentials
$\mmp$ and $B$. It is straightforward
to work out the action of $e_2$ and $h_2$, with the result
\be
h_3 B = 2 B,\quad\quad e_3 B =0,
\ee
where some constants have been fixed from the commutation
relations. Finally, the action of $f_2$ on $B$ follows from
\be
\eps_{ij}\p^j (f_3 B) =f_3 (H^2 \p_i \mmp) =
   - 2 H^2 \mmp \p_i\mmp + H^2 \p_i (H^{-2}).
\ee
Using  $\mmp = b - H^{-1}$ and the Cauchy--Riemann equation, this
yields
\be
\eps_{ij}\p^j (f_3 B) = -2b \p_i H = -2 b \eps_{ij}\p^j B.
\ee
Therefore we find
\be\label{f2B1}
f_3 B = -2b B,
\ee
setting an integration constant equal to zero.
We note that in order to satisfy $[e_3,f_3]=h_3$ on
$B$ we need to have a non-trivial action of $e_3$ on the
integration constant $b$, namely $e_3b =-1$ which is consistent
with the general shift property of the Matzner--Misner group
(\ref{mmtrm}) and $\mmp=b-H^{-1}$ for this solution. In a sense,
one can view the Matzner--Misner group as acting via M\"obius transformations
on the variable $b$. However, closing this action with the Ehlers
M\"obius transformations on $\cE$ leads one to introduce new
constants (notably in $f_4 \mmp$) which transform non-trivially
under the remaining generators. For
completeness we note the transformation rules
\be
e_4 \mmp = 0,\quad\quad h_4 \mmp = 2 \mmp,
\ee
which are true generally and
\be
f_4 \mmp = 2 B H^{-1} +\l \, , 
\ee
where $\l$ is an example of a new constant. This last relation is
true {\em on the solution} $\mmp=b-H^{-1}$; generally the result would
be some non-local expression.

Combining (\ref{mmtrm}) and (\ref{f2B1}),
the action of $f_2$ on the full Ernst potential is
\be
\label{finact}
f_3 \cE =  -2b\cE - 2 \, .
\ee
This shows that the action of $f_3$ on $\cE$ does not yield a new
transformation, but simply a linear combination of previous ones
(to wit, $e_4$ and $h_4$). Hence, under the action of the Matzner--Misner
$\mathrm{SL}(2)$, a BPS solution will remain a BPS solution.\footnote{The
  constant parameters $b,\l,\ldots$ do not influence the analyticity of
  the solution, although they are essential for the action of the
  Geroch group.} The formula  (\ref{finact}) agrees with our
findings in (\ref{Moebact}).

This almost `trivial' action of the Geroch group on the BPS solutions --
which essentially acts only via M\"obius transformations on the Ernst 
potential $\cE$ -- confirms our previous finding for M2 and M5 branes
(\ref{Moebact}), but is in marked contrast to its action on non-BPS
stationary axisymmetric 
solutions \cite{Breitenlohner:1986um}. There $\cE(x,y)$ is {\em not}
analytic, and $\rho(x,y)$ is a non-constant function, often identified
with a radial coordinate (so-called Weyl canonical coordinates).
When starting from the vacuum solution to obtain say, the Schwarzschild
or Kerr solution, the $(x,y)$ dependence of the Ernst potential is
precisely the one induced by the $(x,y)$ dependence of the spectral
parameter whose coordinate dependence, in turn, hinges on the coordinate
dependence of $\rho$. Since we have $\rho=1$ for BPS solutions, this
mechanism does not work, confirming our conclusion that the action
of the Geroch group cannot turn an analytic Ernst potential into
a non-analytic one, hence leaves the class of BPS solutions stable.
 
The results of this section can be summarised by saying that the Weyl
group of $A_1^+$ acts via shifts and inversions on the complex Ernst
potential and at the same time transforms the conformal factor but
leaves invariant the set of analytic Ernst potentials.

  \chapter{Dual formulation of the $E_9$ multiplet} \label{chap:dualfor}

We showed in Chapter \ref{chap:basicbps} that the basic magnetic BPS solutions of eleven-dimensional 
supergravity (M5 and KK6-monopole) smeared in all directions but one
are expressible in terms of the dual potentials $A_{3\,4\,
  5\,6\,7\,8}$ and  $A_{4\,5\,6\,7\,8\,9\,10\,11\vert 11}$
parametrising the Borel generators $R^{[6]}_2$ and $ \bar
R^{[8,1]}_3$. In higher non-compact transverse space dimensions these
potentials are related by Hodge duality to the supergravity fields
$A_{9\,10\,11}$ and $A_{i}^{~(11)}$. The dual potentials take on the
solutions, up to an integration constant, the same value $1/H$ as do
the fields $A_{9\,10\,11}$  and $A_3^{~(11)}$ for the basic electric BPS
solutions, namely the M2-branes at level 1 and the KK-waves at level
0. $H$ is, in any number of non-compact transverse spacelike
directions, a harmonic function with $\delta$-function singularities
at the location of the sources.  

In Chapter \ref{chap:infiniteudualgroup} we  constructed  BPS solutions of eleven-dimensional  supergravity in two
transverse spacelike directions for all $E_9$ positive real roots. We
shall label such description of the BPS states in terms of the
supergravity metric and 3-form the `direct' description.  
Each solution was obtained by relating through dualities and
compensations the `generalised dual potential' $1/H$ parametrising an
$E_9\subset E_{10}$ positive root in the Borel representative of
$E_{10}$ to the supergravity metric and 3-form. 

In this chapter, we will present a space-time description of the BPS
states directly in terms of the generalised dual potentials. We 
label it  the `dual' description. We shall show that the dual
description of the BPS solutions can be derived from gauge fixed
effective actions 
\begin{equation}
\label{gendual}
{\cal S}^{(11)}_{\{q\}} =\frac{1}{16\pi G_{11}}\,\int \dd ^{11}x
\sqrt{\vert g\vert}\left(R^{(11)}- \epsilon{1\over 2   
} F_{i\,\{q\}} F^{i\,\{q\}}\right)\, ,
\end{equation}
where $i$ runs over the two non-compact dimensions $1,2$ and
$\epsilon$ is $+1$ if the action involves a single time coordinate (or
an odd number of time coordinates) and $-1$ if the number of time
coordinates is even. $F_{i\,\{q\}}=\partial_i A_{\{q\}}$ where $\{q\}$
stands for the tensor indices of the $A_p^{[N]}$ potential
multiplying $R_p^{[N]}$ in the Borel representative. Here $p$ is the
level and $[N]$  labels a tower $[3]$, $[6]$, $[0]$ or $[8,1]$ for any
set of $A_8\subset E_9$ tensor indices.  The set of indices is fixed
by the  $A_1^+$ group selected by a Weyl transformed in $A_8$ of the
$A_1^+$ subgroup of $E_{10}$ chosen in  (\ref{geroch}). \\

We will also define and compute the masses in the string theory context. We will see that the dual formalism introduced in this chapter will be a convenient tool to analyze the charge and mass content of the $E_9$ BPS states. 

We will enclose this chapter by considering some $E_{10}$ field associated to real roots which are not in $E_9$. Theses ones may not admit  a direct description but the dual description will be still well defined.
 \setcounter{equation}{0}
 \section{Effective actions}\label{sec:effectaction}

We first consider the M2-M5 system of Section \ref{m2m5sec}. Explicitly, for the
3-tower depicted in the right column of Figure \ref{ger1fig}a   we have
$A_{1+3n}^{[3]}= A_{9\,10\,11\, ,\, [3\,4\,\dots \,10\,11].n}$ and for the 6-tower
depicted in the left column we have $A_{2+3n}^{[6]}=
A_{3\,4\,5\,6\,7\,8\, ,\, [3\,4\,\dots \,10\,11].n}$. Here the symbol $n$  is the number of times the
antisymmetric set of indices  $[3\,4\,\dots \,10\,11]$ must be taken.
That this is the correct tensor structure follows from the structure
of the `gradient representations'
in~\cite{Damour:2002cu,Nicolai:2003fw} (see also Section \ref{sec:gppinvactionthese2}).

For all BPS states in the M2-M5 system, we take 
\begin{equation}
\label{Hfield}
A_p^{[N]}=1/H\, .
\end{equation}
The metric associated to $A_p^{[N]}$ is encoded in the Borel
representatives of the M2 sequence (\ref{borel3n}),
(\ref{borel6n}) and the M5 sequence (\ref{6Nborel}),
(\ref{3Nborel}). We combine (\ref{borel3n}) and (\ref{3Nborel}) to
form the 3-tower (\ref{tower1}) and (\ref{borel6n}) and
(\ref{6Nborel}) to form the 6-tower (\ref{tower3}). We have 
\begin{align}
\label{bor3}
{\cal V}_{1+3n}&= \exp \left[\frac{1}{2} \ln H \, (h_{11}-n
  K^2{}_2)\right]\, \exp \left[\frac{1}{H}\,
  R^{[3]}_{1+3n}\right]\qquad &&n\ge 0 \\ 
\label{bor6}
{\cal V}_{2+3n}&= \exp \left[\frac{1}{2} \ln H\, (-h_{11}-(n+1)
  K^2{}_2)\right] \, \exp \left[\frac{1}{H}\, R^{[6]}_{2+3n}\right]
\qquad &&n\ge 0 \, , 
\end{align}
which yield for the 3-tower (\ref{bor3}) the metric
(\ref{metric3}) and for the 6-tower (\ref{bor6}) the metric
(\ref{metric6})  
\be \begin{split}
 \label{metric3}
\vert g_{11}\vert&=\vert g_{22}\vert=H^{1/3 +n}\, ,\\ 
3-\mathrm{tower}: \qquad \vert g_{33}\vert &=\vert g_{44}\vert =\dots =\vert g_{88}\vert
 =H^{1/3}\, , \\
 \vert g_{99}\vert& =\vert g_{10\,10}\vert =\vert
 g_{11\,11}\vert = H^{-2/3}\, ,
\end{split} \ee

\be \begin{split}
\label{metric6}
\vert g_{11}\vert&=\vert g_{22}\vert=H^{2/3 +n}\, ,\\
6-\mathrm{tower}: \qquad \vert g_{33}\vert &=\vert g_{44}\vert =\dots =\vert g_{88}\vert
 =H^{-1/3}\, , \\
 \vert g_{99}\vert &=\vert g_{10\,10}\vert =\vert
 g_{11\,11}\vert = H^{2/3}\, . 
\end{split} \ee

For the time components of the metric one multiplies the absolute
values of the metric components by a minus sign. The time components
for the 3- and 6-towers are specified in Section \ref{m2m5sec} for the Weyl
orbits initiated by the M2 with time in 9 and by the M5 with time in
3. Note that we could as well take the  Weyl orbit initiated by the M2
with times in 9 and 10 and by the M5 with times in 4, 5, 6, 7 and 
8. Alternatively we could mix the two orbits to avoid for instance at
all level exotic solutions and have always time in 9 for the M2
sequence and in 3 for the M5 sequence depicted in Figure \ref{ger1fig}a. Note that in
all cases, climbing the 3-tower or the 6-tower by steps of one unit of
$n$ amounts to alternate between BPS states on the M2 sequence and the
M5 sequence. We shall comment on this feature in the summary of this part.

We now verify that the matter term in (\ref{gendual}) solves the
Einstein equations with  $A_p^{[N]}=1/H$ and  the metric given in
(\ref{metric3}) and (\ref{metric6}). For the 3-tower the matter
Lagrangian reads 
 \begin{align}
\label{matter3}
{\cal L} &=-\epsilon\frac{1}{2} \sqrt{\vert g\vert }  F_{i\,\{q\}}
F^{i\,\{q\}}\, ,\\ &= 
-\epsilon \frac{1}{2}\sum_{i=1}^2  \sqrt{\vert g\vert}  g^{ii} g^{99}
g^{10\,10} g^{11\,11} [g^{33}g^{44}\dots g^{11\,11}]^n
\left(\partial_i A_{9\,10\,11\, ,\, [3\,4\,\dots
    \,10\,11].n}^{[3]}\right)^2\, , \nn
\end{align} 
while for the 6-towers one gets
\be
\label{matter6}
{\cal L}
=-\epsilon\frac{1}{2}\sum_{i=1}^2  \sqrt{\vert g\vert }  g^{ii} g^{33}
g^{44} g^{55}g^{66}g^{77}g^{88} [g^{33}g^{44}\dots g^{11\,11}]^n
\left(\partial_i A_{3\,4\,5\,6\,7\,8\, ,\, [3\,4\,\dots
    \,10\,11].n}^{[6]}\right)^2\, . 
\ee
One computes from  (\ref{matter3}) and (\ref{matter6}) the
energy-momentum tensors for the 3-tower 
\be \begin{split}
\label{tensor3}
T^1_1&=-T^2_2= -\frac{1}{4} H^{-7/3 -n} \,\left[(\partial_1 H)^2
  -(\partial_2 H)^2\right]\\
  T^1_2&=T^2_1 = -\frac{1}{2} H^{-7/3
  -n} \,[\partial_1 H\partial_2 H]\\
T^3_3&=T^4_4=\dots
T^8_8=-\frac{2n-1}{4} H^{-7/3 -n}\, [(\partial_1H)^2
  +(\partial_2H)^2]\\ 
 T^9_9&=T^{10}_{10}=T^{11}_{11}=-\frac{2n+1}{4} H^{-7/3 -n} \,
\left[(\partial_1H)^2 +(\partial_2H)^2\right],
\end{split}\ee
and for the 6-tower
\be \begin{split}
\label{tensor6}
T^1_1&=-T^2_2=- \frac{1}{4} H^{-8/3 -n} \,\left[(\partial_1 H)^2
  -(\partial_2 H)^2\right]\\
T^1_2&=T^2_1 = -\frac{1}{2} H^{-8/3 -n}
\,[\partial_1 H\partial_2 H]\\
T^3_3&=T^4_4=\dots T^8_8=-\frac{2n+1}{4} H^{-8/3 -n}\, 
  \left[(\partial_1H)^2  +(\partial_2H)^2\right]\\ 
 T^9_9&=T^{10}_{10}=T^{11}_{11}=-\frac{2n-1}{4} H^{-8/3 -n} \, 
  \left[(\partial_1H)^2 +(\partial_2H)^2\right] .
\end{split}\ee
The fact that $T^\mu_\nu$ in (\ref{tensor3}) and (\ref{tensor6})
does not depend on $\epsilon$ results from the cancellation between
negative signs arising from the kinetic energy term in the action
(\ref{gendual})  and from the concomitant even numbers  of time
metric components. From the metric (\ref{metric3}) and
(\ref{metric6}) one easily verifies that the Einstein equations 
\begin{equation}
\label{einstein}
R^\mu_\nu-\frac{1}{2}\delta^\mu_\nu R=T^\mu_\nu\, ,
\end{equation}
with $T^\mu_\nu$ given by (\ref{tensor3}) and (\ref{tensor6}), are
satisfied. This result holds for  $\epsilon=\pm1$ because the left
hand side of (\ref{einstein}) turns out to be independent of
$\epsilon$. \\

Using the mapping (\ref{bragramap}), (\ref{bragra1}) and
(\ref{bragra2}) of the brane $A_1^+$ group (\ref{geroch}) to the
gravity towers $A_1^+$ group, one maps the 3-tower to the 0-tower and
the 6-tower to the [8,1]-tower depicted respectively on the right and
left columns of Figure \ref{fig:mappingbranegrav}b.  
We combine (\ref{KW2}) and (\ref{KM1}) to form the [8,1]-tower. One has
\begin{equation}
{\cal V}_{3+3n}= \exp \left[\frac{1}{2} \ln H \,
  (-K^3{}_3+K^{11}{}_{11} -(n+1) K^2{}_2)\right] \, \exp
\left[\frac{1}{H}\, \bar R^{[8,1]}_{3n+3}\right] \qquad n \ge 0\, . 
\end{equation}
The corresponding dual metric is
\be \begin{split}
 \label{metric0}
\vert g_{11}\vert&=\vert g_{22}\vert=H^{n+1}\\ 
\vert g_{33}\vert &=H\\
 \vert g_{11\,11}\vert &= H^{-1}\\
\vert g_{aa}\vert &=1\qquad a\neq 1,2,3,11
\, .
\end{split}\ee
As expected, up to the interchange of the coordinates 3 and 11, the
same result holds for the redundant 0-tower (except for the level 0 of
$E_{10}$, which is the KK-wave). The verification of the Einstein
equations derived from the actions (\ref{gendual})  duplicates that
of the M2-M5 system. Note that in this generalized dual formulation,
the metric of the gravity tower are all diagonal. 

The above results for the 3- and 6-towers and for the gravity tower
have been established for a chosen set of tensor indices, namely the
set determined by the choice of the $A_1^+$ subgroup of $E_{10}$ for
which $R^{[3]}_1$ is identified with $R^{9\, 10 \, 11}$. The validity
of the effective action (\ref{gendual}) for all $E_9\subset E_{10}$
fields associated to its positive real roots follows from permuting
the $A_8$ tensor indices, that is from performing   Weyl
transformations of the gravity line. 
\setcounter{equation}{0}
\section{Charges and masses} \label{chargesec}
The charge content of the $E_9$ BPS solutions is easier to analyse in
the dual description because the dual potential is not mixed with
fields arising from the compensations. {\em Outside the sources}\/,
the equation of motion for the dual field $A_{\{q\}}$ is from
(\ref{gendual}) 
\begin{equation}
\label{genfield}
\sum_{i=1}^2\partial_i \left(\sqrt{\vert g \vert}g^{\{q\}}\partial_i
A_{\{q\}}\right)=0 \, , 
\end{equation}
where 
\begin{align}
\label{lap3}
 g^{\{q\}}&=g^{ii}g^{99} g^{10\,10} g^{11\,11} [g^{33}g^{44}\dots
  g^{11\,11}]^n &&\hbox {3-tower : level  (1+3$n$)}\\ 
\label{lap6}
 g^{\{q\}}&=g^{ii}g^{33} g^{44} g^{55}g^{66}g^{77}g^{88}
      [g^{33}g^{44}\dots g^{11\,11}]^n && \hbox {6-tower :
	level  (2+3$n$)} \\ 
\label{lapKK}
 g^{\{q\}}&=g^{ii} g^{44} g^{55}\dots g^{10\,10}
(g^{11\,11})^2[g^{33}g^{44}\dots g^{11\,11}]^n  &&\hbox{[8,1]-tower
  : level (3+3$n$)}\, , 
\end{align}
with $n\ge 0$. The appearance of the $n$-fold blocks of antisymmetric
metric factors $g^{33}g^{44}\cdots g^{11\,11}$ is again due to the
embedding of $E_9$ in $E_{10}$, see~\cite{Kleinschmidt:2006dy}.  
From (\ref{metric3}), (\ref{metric6}),
(\ref{metric0}), and from the embedding relation in $E_{11}$
(\ref{embedding}), we get  for all towers and hence, by permutation
of tensor indices in $A_8$, for all $E_9$ BPS states $\sqrt{\vert
  g\vert }g^{\{q\}} =\pm H^2$. As, up to an integration constant, one
has always $A_{\{q\}}=1/H$, the field equation (\ref{genfield})
reduces to 
\begin{equation}
\label{allH}
\sum_{i=1}^2\partial_i \partial_i H=0 \, .
\end{equation}
Here, as for the KK-monopole  discussed in Section \ref{sec:basiconedim}, the equation
(\ref{allH}) is valid outside the sources and the latter are
determined by fixing the singularities of the function $H$. Labelling
the positions of the smeared M2 by $x^1_k,x^2_k$ and their charges by
$q_k$, one takes  
\begin{equation}
\label{source2}
H(x^1,x^2)=\sum_k \frac{q_k}{2\pi}\ln  \sqrt{(x^1-x^1_k)^2+(x^2-x^2_k)^2}\, ,
\end{equation}
and,  in analogy with (\ref{h}), the extension of (\ref{allH})
including the sources  reads\footnote{We  fix the M5 charges by the
  Weyl transformation relating the M2 to the M5, which for convenience
  was not explicitly used in our general derivation of the $E_9$ BPS
  solutions in Chapter \ref{chap:infiniteudualgroup}.} {\em for all $E_9$ BPS solutions} 
\begin{equation}
\label{fullH}
\sum_{i=1}^2\partial_i \partial_i H=\sum_k \frac{q_k}{2\pi}\delta
({\bf x - x_k})\, . 
\end{equation}
Thus we obtain, for all BPS states, the same charge value as for the
M2, as expected from  U-dualities viewed as $E_9$ Weyl reflexions.  
Our identification of $q_k$ with a charge is however not the
conventional one as long as our solutions with 2 non-compact space
dimensions have not been identified with static solutions in 2+1
space-time dimensions. This raises the question whether we are allowed
to decompactify the time. This will be examined in the following
Sections \ref{subsec:froml4tol6} and \ref{subsec:highlsol}. 
We wish to stress that decompactification of time or space dimensions
is {\em not} the same as `unsmearing'. The latter term refers to the
undoing of the smearing process by which the dependence of the
harmonic functions characterising our BPS solutions is reduced by one
or more variables through compactification. Thus unsmearing implies
decompactification of space dimensions but the converse is not
necessarily true. When it is true we call the decompactified
dimensions `transverse'.  For the basic BPS solution smeared to two
space dimensions, unsmearing is of course possible up to the space
dimensions of the defining solution given in Section \ref{sec:basiconedim} (8 for the
M2, 5 for the M5, 9 for the KK-wave and 3 for the KK-monopole). {\em
  For all higher level BPS states in 2 non-compact space dimensions,
  unsmearing is impossible.} It is indeed straightforward to show that
the Einstein equation (\ref{einstein}) is not satisfied if the
harmonic function $H$ entering the right hand side of the equation is
extended to three dimensions. 

One verifies that in the dual formulation all $E_9$ BPS states can be
smeared to one space dimension, the charge being still defined by
(\ref{fullH}) with $i$ equal to 1. These solutions are also
solutions of the  $\sigma$-model  $\mathcal{S}^{brane}$ (\ref{full}). 

A criterion for decompactification of longitudinal spacelike
directions  and of timelike directions will be obtained from the
requirement that the tensions should be finite. These quantities will
be evaluated in the string context from string dualities and uplifting
to eleven dimensions.  For each BPS state characterised by a dual
potential $A_p^{[N]}$ 
we define an action $\cal A$ given in Planck units by the product of
all spatial and temporal compactification radii, each of them at a
power equal to the number of times the corresponding index occurs in
$A_p^{[N]}$. 
One gets from (\ref{lap3}), (\ref{lap6}) and (\ref{lapKK}) the
action ${\cal A}_{\ell}$ of the level $\ell$ solution\footnote{Similar actions
  were considered also in~\cite{Brown:2004jb}.} 
\begin{align}
\label{act3}
{\cal A}_{1+3n}&=\frac{1}{l_p^{9n+3}}\, R_9 R_{10} R_{11} [R_3
  R_4\dots R_{11}]^n && \hbox {3-tower : level  (1+3$n$)}\\ 
\label{act6}
 {\cal A}_{2+3n}&=\frac{1}{l_p^{9n+6}}\,R_3 R_4R_5R_6R_7R_8 [R_3 R_4\dots R_{11}]^n && \hbox{ 6-tower : level  (2+3$n$)} \\
\label{actKK}
 {\cal A}_{3+3n}&=\frac{1}{l_p^{9(n+1)}}\,R_4R_5\dots R_{10}(R_{11})^2 [R_3 R_4\dots R_{11}]^n && \hbox{ [8,1]-tower : level  (3+3$n$)} \, ,
\end{align}
where $l_p$ is the 11-dimensional Planck constant ($l_p^9=8\pi
G_{11}$). We identify for non-exotic states $\cal A$ to $MR_t$ where
$M$ is the mass of the source and $R_t$ the compactification time
radius. We derive in Appendix \ref{appm}  the actions $\cal A$,
(\ref{act3}), (\ref{act6}) and (\ref{actKK}), from the
interpretation in the context of string theory of the Weyl reflexions
used to construct the BPS solutions, both for exotic and non-exotic
states. Requiring finiteness of the action density implies that
$A_{\{q\}}$ be linear in the radii for those directions, spatial or
temporal, which can be decompactified. For non-exotic states this is
equivalent to requirement of finite tension.

It is immediately checked that for the basic branes $n=0$  one obtains
the correct mass formula for the M2 (time in 9) (\ref{massl1}), the
M5 (time in 3) (\ref{mm5}) and the KK-monopole (time in 4)
(\ref{mkk}). The KK-monopole mass is in agreement with the
calculation of the ADM mass of the unsmeared KK6-monopole
\cite{Bombelli:1986sb}. 
Our criterion confirms that time and all longitudinal space dimensions
can be taken to be non-compact except for the Taub-Nut direction 11 of
the KK6-monopole which occurs quadratically in (\ref{actKK}) and
hence cannot be decompactified. This fact is in agreement with the
fact that the  KK monopole solution in 3 transverse space dimensions,
characterised by an harmonic function $H=1+q/r$ where $r^2 \equiv
(x^1)^2+(x^2)^2+(x^9)^2$,  has, in order to avoid a conical
singularity, its radius $R_{11} \propto q$ \cite{Hawking:1976jb,
  Gibbons:1979xm, Sen:1997js} and hence finite. 

We now examine further the nature of the BPS solutions for level
higher than 3, that is outside the realm of the basic BPS solutions of
Section \ref{sec:basiconedim}. 

\subsection{From level 4 to level 6}\label{subsec:froml4tol6}

The actions ${\cal A}_{\ell}$ defined in (\ref{act3}), (\ref{act6}) and
(\ref{actKK}) are in agreement with the computation of the masses for
level 4 (\ref{ml4}) with time in 3, level 5 (\ref{massl5}) with
time in 9, level 6 (\ref{massl6}) with time in 3 obtained in
Appendix \ref{appm} in the string context. Thus time occur linearly in $\cal A$
and can be non-compact. Examining the dependence of $\cal A$ in the
spatial radii, we see that 
the spacelike directions that can be decompactified are $(4,5,6,7,8)$
for $\ell=4$, $(10,11)$ for $\ell=5$ and none for $\ell=6$.  

U-duality requires that the dimensionally reduced metric in 2+1
dimensions  be identical for  these solutions and in addition be
equivalent to the (2+1)-dimensional metric for the basic BPS solutions
of Section \ref{sec:basiconedim}.\footnote {\label{bpsfn}All BPS solutions for $\ell\le 6$ should then
  form a multiplet  of $E_8$ which is the  symmetry of eleven-dimensional 
  supergravity reduced to (2+1) dimensions.}  We now show that this
requirement is fulfilled, both in the direct and in the dual
formalism. 

To perform the dimensional reduction we write in general the
11-dimensional metric in the following form 
\begin{equation}
 \dd s^{2}= g_{\mu \nu} \dd x^\mu \, \dd x^\nu + \sum_r h_{rr}(\dd x^r)^2 \, ,
 \end{equation}
 where $\mu, \nu=1,2,a$ and $r\neq 1,2,a$, labelling $a$ as the time
 coordinate. To find the canonical Einstein action in $d= 3$
 dimensions, the components of the reduced metric have to be Weyl
 rescaled  
\begin{eqnarray}\label{reduc}
 \widetilde{g}_{\mu \nu} = g_{\mu \nu}\,  h^{\frac{1}{d-2}} =g_{\mu
 \nu}\,  h\, , 
 \end{eqnarray}
 where $h= \det\, h_{rs}$.
 
We first consider the level 4 state. In the direct formulation the
level 4 metric is given by 
(\ref{metric4}) and for the dimensional reduction to $3$ dimensions
we get using (\ref{reduc}) with $h= \widetilde{H}^{1/3}$ and time
in 3 
\begin{eqnarray}
\dd s^2_{l=4,3D}=  \widetilde{g}_{\mu \nu} \dd x^\mu \, \dd x^\nu= H [(\dd x^1)^2+
  (\dd x^2)^2]- (\dd x^3)^2\, . 
 \end{eqnarray}
The same metric is obtained in the reduction of the dual metric
(\ref{metric3}) with $n=1$, where now  $h= {H}^{-1/3}$. Similarly
the level 5 solution with time in 9, given in the direct formulation
by (\ref{metric5}) and in the dual one by (\ref{metric6}) with
$n=1$ with respectively $h= \widetilde{H}^{2/3}$ and $h= {H}^{-2/3}$,
yields from (\ref{reduc}) 
\begin{equation}
\dd s^2_{\ell=5,3D}=  \widetilde{g}_{\mu \nu} \dd x^\mu \, \dd x^\nu= H [(\dd x^1)^2+
  (\dd x^2)^2]- (\dd x^9)^2\, . 
 \end{equation}
The level 6 solution with the timelike direction 3 is given in the
direct formulation by (\ref{KKW}) with $n^\prime=1$, namely 
\begin{align}
\label{level6sol}
\dd s^2_{\ell=6}&= (H^2+B^2)  [(\dd x^1)^2+(\dd x^2)^2  ] - \widetilde
 H^{-1}(\dd x^3)^2+[(\dd x^4)^2\dots+(\dd x^{10})^2  ] \nonumber\\ 
 &\quad + \widetilde H  [ \dd x^{11} -  (  \widetilde  H^{-1}-1 )\dd x^3 ]^2  \, .
 \end{align}
  and in the dual formulation by (\ref{metric0}) with $n=1$, that is
\begin{equation}
\label{level6soldu}
\dd s^2_{\ell=6}= H^2 [(\dd x^1)^2+(\dd x^2)^2  ] - H(\dd x^3)^2+
[(\dd x^4)^2\dots+(\dd x^{10})^2  ]+H^{-1} ( \dd x^{11} )^2  \, , 
 \end{equation}
Reducing the level $6$ metric (\ref{level6sol}) and
(\ref{level6soldu}) to $3$ dimensions, we again find 
\begin{equation}
\dd s^2_{\ell=6,3D}=  \widetilde{g}_{\mu \nu} \dd x^\mu \, \dd x^\nu= H [(\dd x^1)^2+
  (\dd x^2)^2]- (\dd x^3)^2 \,. 
 \end{equation}
One easily checks that the same 3-dimensional metric  (with suitable
time coordinate)  are recovered for all basic BPS solutions recalled
in Chapter \ref{chap:basicbps}. 

We have thus verified that all the BPS solutions of eleven-dimensional  supergravity,
for levels $\ell \leq 6$  are equivalent in 2+1 dimensions.  These
solutions constitute an $E_8\subset E_9$ multiplet of branes (see
footnote~\ref{bpsfn} page \pageref{bpsfn}).  The  $E_8$ multiplet is the same as the one studied
some time ago algebraically as a consequence of M-theory compactified
on $T^8$ for which the masses of the different BPS states of the
multiplet has been  
derived (see \cite{Elitzur:1997zn} and in particular Table 11), and
their  space-time interpretation  was  obtained in reference
\cite{Lozano-Tellechea:2000mc}. We recover here these results from the
Weyl group  of $E_9$ endowed with the temporal involution,  and from
the interpretation of these Weyl transformation in the context of
string theory. This $E_9$ containing a timelike direction  is the
correct setting to describe all  the BPS solutions with two unsmeared
spacelike directions  in a group theoretical language. We summarise in the Table \ref{tab:mass}
 the mass content\footnote{For the
  level 3 KK-monopole potential we have put the time in 4 as in
  Chapter \ref{chap:basicbps} instead of 9, 10 in the general metric (\ref{mono}) to
  avoid here exotic states.}  for all the levels $0< \ell\le 6$ in the form ${\cal A}_{\ell} =M R_t$, where
${\cal A}_{\ell}$ is the level $\ell$ action (\ref{act3}), (\ref{act6})
and (\ref{actKK}), $M$  the mass and $R_t$ the time radius (which can
be taken to $\infty$),  to exhibit their striking relation with the
$E_9$ dual potentials. 
\begin{table}[h]
\begin{center}
\begin{tabular}{|c|c|c|c|}
\hline
$\ell$ & time& $E_9$ field& ${\cal A}_{\ell} =MR_t$ \\ \hline \hline
1&9& $A_{9\,10\,11}$ & $R_9R_{10}R_{11}/ l_p^3$\\
2&3& $A_{3\dots8}$&$R_3 \dots R_8/l_p^6$\\
3&4& $A_{4\dots 10\,11,11}$&$R_4 \dots R_{11}.R_{11}/l_p^9$\\
4&3& $A_{3\dots11, 9\,10\,11}$&$R_3 \dots R_{11}.R_9R_{10}R_{11}/l_p^{12}$\\
5&9& $A_{3\dots11, 3\dots 8}$& $R_3 \dots R_{11}.R_3\dots R_8/l_p^{15}$\\
6&3& $A_{3\dots11,4\dots11,  11}$&$R_3 \dots R_{11}.R_4\dots R_{11}. R_{11}/l_p^{18}$\\
\hline
\end{tabular} \caption{\sl \small Mass content for the levels $0< \ell\le 6$ in the form  ${\cal A}_{\ell} =M R_t$, where
${\cal A}_{\ell}$ is the level $\ell$ action, $M$  the mass and $R_t$ the time radius.  }
\label{tab:mass} 
\end{center}
\end{table}
These BPS states  already reach levels beyond the classical levels $\ell
 \leq 3$ for which a dictionary between $E_{10}$ fields  and
 space-time fields depending on one coordinate
 exists~\cite{Damour:2002cu}. In the next subsection we discuss the
 solutions for $\ell>6$.

\subsection{The higher level solutions}\label{subsec:highlsol}

The action formulae (\ref{act3}), (\ref{act6}) and (\ref{actKK})
are all in agreement with the evaluation of $\cal A$ in the string
context, as seen from Appendix \ref{appm} no time or longitudinal space
radius occurs linearly in $\cal A$. Thus time is compact and the only
non-compact space radii are the transverse two dimensions.  

All these states can  only be reached from basic non-exotic solutions
through a timelike T-duality. 

U-duality does no more imply that the dimensionally reduced metric in
(2+1) dimensions should be identical to that of the basic ones and one
indeed verifies  that they are distinct.  However metric and the
induced dilaton field are expected to be identical when reduced to 2
dimensions. This is indeed the case as we now show.  

Reducing all the solutions down to three dimensions\footnote{The
  spacelike or timelike nature of $x^r$ is irrelevant for this
  argument.} 1, 2  and $r$, with $r\in 3, \dots ,11$, we
get\footnote{In the dual formalism this formula holds with
  $\widetilde{g}_{11}$ depending only on $H$.} 
\begin{equation}
\label{metriclg}
\dd s^2_{\ell,3D}=  \widetilde{g}_{11}(H,B)  [(\dd x^1)^2+
  (\dd x^2)^2]+(\dd x^r)^2.
 \end{equation}
Because for all the solutions we have unity in front of $(\dd x^r)^2$,
reducing on $x^r$ down to two dimensions and performing a Weyl
rescaling all the metric reduce to a flat two-dimensional space with
zero dilaton field.

\setcounter{equation}{0}
\section{ Transcending eleven-dimensional supergravity}

We have seen in Sections \ref{sec:effectaction} and \ref{chargesec} that all $E_9$ BPS solutions (including the
basic ones discussed in Section \ref{sec:basiconedim}) can be smeared to one space
dimension in the dual formalism. There is in fact no such description
in the direct formalism, except for the M2 and the KK-wave. More
generally, the dual formalism in one non-compact space dimension is
equivalent to the  $\sigma$-model (\ref{full})  restricted to a
single  $\ell> 0$ field. In that case indeed  the matter term is the same
in both formalisms as no covariant derivative arises in the
$\sigma$-model \cite{Damour:2002cu, Englert:2004ph} and the Einstein
term  \cite{Banks:1998vs} coincides with its  level zero. 

We now consider BPS solutions obtained from positive real roots of
$E_{10}$ not present in $E_9$. These solutions can be obtained by
performing Weyl transformations on 
$E_9$ BPS solutions smeared to one dimension. They  depend on one
non-compact space variable and may have no counterpart in eleven-dimensional 
supergravity. We illustrate the construction of such solutions by one
example. 
\begin{table}[h]
\begin{center}
\begin{tabular}{|cc|cccccccccc|c|}
\hline
$IIA$&$M$-theory&1&2&3&4&5&6&7&8&9&10&11\\
\hline
D6 & KK6& \,&$\times$&$\times$&$\times$&$\times$&$\times$&$\times$&$\times$&\, & \,& $\times$ \\
D8 & M9& \,&$\times$&$\times$&$\times$&$\times$&$\times$&$\times$&$\times$&$\times$&$\times$&$\times$ \\
\hline
\end{tabular}
\caption{\sl \small The D8 brane obtained from the D6 brane by
  performing the Weyl  
reflexion $s_{\alpha_{11}}$, i.e. a double T-duality plus exchange of
the directions $x^9$ and $x^{10}$.} 
\label{tab:m9}
\end{center}
\end{table}

We shall obtain the M9 (namely the `uplifting'\footnote{There is a
  no-go theorem stating that massive 11-dimensional supergravity does
  not exist \cite{Bautier:1997yp}. The concept of uplifting the D8
  seems thus puzzling. However  a definition of  `massive eleven-dimensional 
  supergravity' has been proposed \cite{Bergshoeff:1997ak} for a
  background with an isometry generated by a spatial Killing
  vector. In that theory the M9 solution does exist
  \cite{Bergshoeff:1998bs}. Existence of the M9 is also suggested by
  the study of the central charges of the M-theory superalgebra
  \cite{Townsend:1997wg, Hull:1997kt}.}  of the D8 brane of massive
Type IIA supergravity)\footnote{The metric of the D8 brane has
  previously been discussed in the context of $E_{11}$ in
  \cite{West:2004st}.} by performing a Weyl transformation on the
KK6-monopole smeared in all directions but one described in Section
\ref{subsec:kkwkkm}. We start with a D6 
along the spatial directions $2,3,\dots,7$ and we choose 8 as time
  coordinate.  The D6 is smeared in the directions 9 and 10 and thus
  depends only on the non-compact variable $x^1$. 
The uplifting to M-theory of the D6 yields a KK6-monopole with
  Taub-NUT direction 11 (see Table \ref{tab:m9}). 
To obtain the D8 and its uplifting M9, we  perform the Weyl reflexion
  $s_{\alpha_{11}}$ which may be viewed as a double T-duality in the
  directions 9 and 10 plus exchange of the two radii  
\cite{Elitzur:1997zn, Obers:1998rn, Banks:1998vs, Englert:2003zs}.

The KK6-monopole, in the longitudinal directions $2,\dots ,7$ with
timelike direction 8 and Taub-NUT direction 11,  smeared in  9 and 10,
is described in the  $\sigma$-model by a  
level 3 generator\footnote{This level 3 step operator contains the
  index 2 and belongs to a  $E_9$ conjugate in $E_{10}$ to the $E_9$
  we used in the previous sections.}  $R^{2345678\,11\vert11}$. The
solution is given in (\ref {G3}) up to a permutation of
indices. Defining $p^a = -h^{~a}_a$ one has 
\be \begin{split}
\label{modkkmono}p^1&=p^9=p^{10} = {1\over 2} \ln H(x^1) \\
p^{11}&=-{1 \over 2} \ln 
H(x^1)\\  p^i&=0, \qquad i=2 \dots 8 \\
A_{2345678\,11\vert 11}&={1 \over H(x^1)} .
\end{split}\ee
The level 3 root $\alpha^{(3)}$  corresponding to
$R^{2345678\,11\vert11}$  is
$\alpha^{(3)}=\alpha_2+2\alpha_3+3\alpha_4+4\alpha_5+5\alpha_6+6\alpha_7+7\alpha_8+4\alpha_9+\alpha_{10}+3\alpha_{11}$.
Performing the Weyl 
reflexion $s_{\alpha_{11}}$, we get
$s_{\alpha_{11}}(\alpha^{(3)})\equiv \alpha^{(4)}
=\alpha^{(3)}+\alpha_{11}$. 
The root $\alpha^{(4)}$  is the lowest weight of the $A_9$ irreducible
representation of level 4 \cite{Nicolai:2003fw} whose 
Dynkin labels are $[0\, 0\,0\,0\,0\,0\,0\,0\,2]$. This $A_9$ representation in
the decomposition of the adjoint representation of $E_{10}$  is not in
$E_9$. The action of    $s_{\alpha_{11}}$ on the Cartan fields is
\cite{Englert:2003zs} 
\be \begin{split}
 p^{\prime a}&= p^a+ {1 \over 3} (p^9+p^{10}+p^{11}) \qquad a=2 \dots
  8\noindent \\ 
p^{\prime a}&= p^a- {2 \over 3} (p^9+p^{10}+p^{11}) \qquad a=9,10,11.
\end{split}\ee
Using the embedding  relation (\ref{embedding}), the Weyl transform
of (\ref{modkkmono})  yields the solution 
\begin{eqnarray}
\label{mett1}
p^{\prime 1}&=&{4\over 6} \ln H(x^1)\\
\label{mett2}
 p^{\prime a}&=& {1\over 6} \ln H(x^1) \qquad a=2 \dots 10 \noindent \\
\label{mett3}
p^{\prime 11}&=& - {5\over 6} \ln H(x^1). \\ \label{complicated}
A_{23456789 \,10\,11\vert 11\vert 11}&=&{1 \over H(x^1)}\, .
\end{eqnarray}
The level 4 field (\ref{complicated}) contains the antisymmetric
 set of indices $2,3\dots ,11$. These are not apparent in its $A_9$
 Dynkin labels given above but are needed in the 11-dimensional metric
 stemming from the embedding $E_{10}\subset E_{11}$ encoded in 
 (\ref{embedding}). The $A_{10}\subset E_{11} $ Dynkin labels of
 this  field in this embedding are indeed
 $[1\,0\,0\,0\,0\,0\,0\,0\,0\,2]$ (see Table \ref{tab:levdece11}). From (\ref{vielbein}) one gets the
 11-dimensional metric  
\begin{equation}
\label{M9}
\dd s^{2}_{M9}= H^{4/3} (\dd x^1)^2+H^{1/3} [(\dd x^2)^2+\dots
  -(\dd x^8)^2+(\dd x^9)^2+(\dd x^{10})^2] 
+H^{-5/3} (\dd x^{11})^2\, .
\end{equation}
One verifies the validity of the dual formalism equation
(\ref{genfield}) for the field (\ref{complicated}), namely  
\begin{equation}
\label{tenfield}
\frac{\dd}{\dd x^1} \left(\sqrt{\vert g \vert}g^{22}g^{33}\dots g^{10\,10}[
  g^{11\,11}]^3\frac{d}{\dd x^1}A_{23456789 \,10\,11\vert 11\vert
  11}\right)= \frac{\dd ^2}{\dd (x^1)^2}H(x^1)=0 \, . 
\end{equation}
The metric (\ref{M9}) describes the   M9  \cite{Bergshoeff:1998bs},
which reduced to 10 dimensions gives the D8~\cite{Bergshoeff:1996ui}
of  massive type IIA. The computation of the M9-mass in the string
context (\ref {1110}) agrees with our general action formula 
\begin{equation}
{\cal A}_{M9}= R_8 M_{M9}={R_2R_3R_4R_5R_6R_7R_8R_9  R_{10} R^3_{11}
  \over l_p^{12}}\, , 
\end{equation}
indicating that $R_{11}$ is compact.

\chapter*{} \section*{Summary and comments of Part II}
\markboth{Summary and comments of Part II}{Summary and comments  of Part II}

We have constructed an infinite $E_9$ multiplet of  BPS solutions of
eleven-dimensional  supergravity and of its exotic counterparts depending on  two
non-compact variables. These solutions  are related by U-dualities
realised as  Weyl transformations of the $E_9$ subalgebra of $E_{11}$
in the regular embedding $E_9\subset E_{10}\subset E_{11}$. Each BPS
solution stems from an $E_9$ potential $A^{[N]}_p$ multiplying the
generator  $R^{[N]}_p$ in the Borel representative of the coset space
$\mathZ{E}_{10}/\mathrm{K}^-_{10}$ where $\mathrm{K}^-_{10}$ is invariant under a temporal
involution.  $A^{[N]}_p$ is related to the supergravity 3-form and
metric through dualities and compensations. This $E_9$ multiplet of
states split into three classes according to the  level $\ell$.  For
$0\le \ell \le 3$ we recover the basic BPS solutions, namely, the
KK-wave, the M2, the M5 and the KK-monopole. For $4\le \ell \le 6$, the
solutions have 8 longitudinal space dimensions. We argue that for
higher levels, all 9 longitudinal directions, including time, are
compact.  Each  BPS solution can be mapped to a solution of a dual
effective action of gravity coupled to matter expressed in terms of
the $E_9$ potential $A^{[N]}_p$. In the dual formulation the BPS
solutions can be smeared to one non-compact space dimension and
coincides then with solutions of the $E_{10}$ $\sigma$-model build
upon $E_{10}/\mathrm{K}^-_{10}$. The $\sigma$-model yields in addition an
infinite set of BPS space-time solutions corresponding to all real
roots of $E_{10}$ which are not roots of $E_9$. These appear to
transcend eleven-dimensional  supergravity, as exemplified by the lowest level $(\ell=4)$
solution which is identified to the M9.

The relation between the  eleven-dimensional  supergravity 3-form and metric, and the
$E_9$ potentials  $A^{[N]}_p$  has significance beyond the realm of
BPS solutions. To see this we first recall that the  $E_9$ potentials
can be organised in towers defined by decomposing $E_9\subset E_{10}$
into  $A_1^+$ subgroups with central charge $-K^2{}_2\in E_{10}$.  
Each of these $A_1^+$ subgroup contains  two `brane' towers $[N]=[3],[6]$
or one `gravity tower' $[N]=[8,1]$ of real roots (the two gravity towers
$[N]=[8,1]$ and [0] are redundant except for the lowest level representing
the KK-wave). We first examine the brane towers. 
 
The recurrences of the 3-tower $A^{[3]}_{1+3n}$ alternate in nature at
each step: they switch from states on the M2 sequence to those on the
M5 sequence. This feature is illustrated in  Figure \ref{ger1fig}a where the 3-tower
recurrences are depicted on the right column. On the other hand, each
recurrence of the 3-tower is related by {\em duality-compensation
  pairs} to the supergravity 3-form potential which we denote by
$(A^{[3]}_1)_q$,  as seen in the horizontal lines of both Figure \ref{second} and
Figure \ref{fig:sequenceb}. Here we designated by the integer $q$ the number of
duality-compensation pairs needed to reach the field $(A^{[3]}_1)_q$
from $A^{[3]}_{1+3n}$. In the realm of BPS states studied in this
paper each field $A^{[3]}_{1+3n}$ defines a different BPS solution of
eleven-dimensional  supergravity defined by $(A^{[3]}_1)_q$ and the related
metric. Comparing Figure \ref{ger1fig}a with Figure \ref{second} and Figure \ref{fig:sequenceb} we see that  $q$ is equal to
$n$. This relation expresses the fact that {\em the number of steps
  needed to climb the 3-tower up to the field 
$A^{[3]}_{1+3n}$ is equal to the number of duality-compensation pairs
  needed to reach from $A^{[3]}_{1+3n}$ the eleven-dimensional  supergravity 3-form
  defining the corresponding BPS solution}. However the relation $q=n$
does not rely on the BPS character of the solution and hence has
general significance, which can be pictured as follows.  

Were all the compensations matrices put equal to unity  no new
solutions could be generated by Weyl transformations from any solution
of eleven-dimensional  supergravity defined by its  3-form $A^{[3]}_1$ (or from its
Hodge dual) and metric. Indeed, because of the Weyl equivalence of all
dualities depicted in Figure \ref{second} and Figure \ref{fig:sequenceb}, one would simply get  
\be \label{comp01} \begin{split}
A^{[3]}_{1+3n} \quad &\stackrel{\rm I}{=} \quad  A^{[3]}_1\qquad n \,{\rm even}\\
A^{[3]}_{1+3n} \quad  &\stackrel{\rm I}{=}  \quad A^{[6]}_2\qquad n \,{\rm odd}\, ,
\end{split}
\ee
where $A_2^{[6]}$ is taken to be the Hodge dual of $A_1^{[3]}$  and
the superscript I means that all compensations have been formally
equated to unity. A similar analysis of the 6-tower depicted in the
left column of Figure \ref{ger1fig}a would yield 
\be \label{comp02} \begin{split}
A^{[6]}_{2+3n} \quad  &\stackrel{\rm I}{=}  \quad  A^{[6]}_2\qquad n \,{\rm even}\\
A^{[6]}_{2+3n}  \quad &\stackrel{\rm I}{=}  \quad   A^{[3]}_1\qquad n \,{\rm odd}\, .
\end{split}
\ee
The same phenomenon would appear in the gravity tower, where it is
somewhat hidden in the redundant 0-tower of Figure \ref{fig:mappingbranegrav}. 

The non-trivial content of the $E_9$ tower potentials is entirely
due to compensations. These prevent `duals of duals' to be equivalent
to unity and one may view the $E_9$ towers as defining `non-closing
dualities', familiar from the standard Geroch group. They translate
through the compensation process the genuine non-linear structure of
gravity.

  \part{Finite and infinite-dimensional symmetries of pure $\mathcal N=2$ supergravity in $D=4$}\label{partthree}
  \setcounter{equation}{0}
\markboth{Introduction of Part III}{Introduction  of Part III}

\mbox{}\\

\vspace{6cm}

In this part of the thesis, we study the symmetries of four-dimensional pure ${\cal N}=2$ supergravity which is of interest for several reasons. The bosonic sector of this model consists of gravity coupled to a single Maxwell field. It admits half-BPS solutions like the extremal Reissner-Nordstr\"om black hole. In fact, a general half-BPS solution can be described by four charges $m$, $n$, $q$ and $h$ subject to the constraint $m^2+n^2=q^2+h^2$~\cite{Kallosh:1994ba,Argurio:2008zt}. The first two charges are gravitational mass and NUT charge and the latter two correspond to the electric and magnetic charges under the Maxwell field. In addition, it has been known for a long time that the solutions of this model with one Killing vector transform under the group $\mathrm G=\mathrm{SU}(2,1)$~\cite{Kinnersley}. This symmetry group is not in its split real form (which would be $\mathrm{SL}(3,\mathbb{R}))$ and one of our motivations for this part was to investigate whether the conjectures discussed in Chapter \ref{chap:sugrareformulatedth} apply also in this case (see~\cite{HenneauxJulia,deBuyl:2003ub,Henneaux:2007ej,Riccioni:2008jz} for related work). In particular, the theory of real forms for the extended infinite-dimensional symmetries $\mathrm G^{++}$ and $\mathrm G^{+++}$ is not as well-developed as for finite-dimensional groups but see~\cite{Rousseau1989,Rousseau1995,BenMessaoud} for some mathematical results. Since the symmetry $\mathrm{SU}(2,1)$ mixes the two gravitational charges one can study the question of gravitational dualities~\cite{Hull:1997kt,Hull:2000zn,Henneaux:2004jw,Bunster:2006rt,Boulanger:2008nd,Bergshoeff:2009zq,Argurio:2008zt,Nieto:1999pn} analogous to electromagnetic duality in this simple model.\\

\noindent The results presented in this part have been published in the reference \cite{Houart:2009ed}.

   \chapter{On finite-dimensional symmetries of pure $\mathcal{N}=2$ supergravity } \label{chap:finitesymsu215}

\setcounter{equation}{0}
\section{Symmetries and BPS solutions of pure $\mathcal{N}=2$  supergravity}
\label{sec:EinsteinMaxwell}

Pure $\mathcal{N}=2$ supergravity in four dimensions is the natural supersymmetric completion of Einstein-Maxwell theory. To set the scene, we shall in this section present our conventions for this theory, and in particular discuss its underlying symmetries in the presence of a space-like or a time-like  Killing vector. The presence of these Killing vectors is equivalent to performing a Kaluza-Klein reduction of the theory on a space-like or a time-like circle, respectively. This process reveals a hidden global symmetry in $D=3$, described by the group SU(2,1) \cite{Kinnersley,Breitenlohner:1987dg}.
\subsection{Einstein--Maxwell in $D=4$}

The field content of four-dimensional $\, \mc{N}=2$ supergravity consists of a gravity multiplet, with a graviton $g_{\alpha \beta}$, two gravitino $\Psi^a_\alpha$ ($a=1,2$) and a Maxwell field $A_\alpha$. The bosonic part of the theory is then described by the standard Einstein-Maxwell Lagrangian:
\begin{eqnarray}
\label{eqn:SugraAction4d}
\mathcal{L}_{4d}=  \frac{1}{4}\, \sqrt{- g} \, \Big(R- F_{\alpha\beta} F^{\alpha \beta}\Big)\,, 
\end{eqnarray}
such that the Maxwell field is minimally coupled to gravity, and where $F_{\alpha \beta} = 2\,  \p_{\, [ \alpha} \,A_{\beta ]}$, locally. We will take space-time $\mathcal{M}_4$ to be Lorentzian with signature $(-, +, +, +)$.\footnote{Regarding index notation, Greek letters $\alpha, \beta...$ will indicate the four-dimensional curved space-time indices, $\mu, \nu...$ the three-dimensional curved indices, $A,B,...$ flat space-time indices, and $a,b...$ flat space indices.} The equations of motion derived from (\ref{eqn:SugraAction4d}), written in flat coordinates, are for the metric
\begin{equation}
\label{eqn:Einstein}
R_{AB} + \frac{1}{2} \eta_{AB} F_{CD} F^{CD} - 2 F_{AC}{F_B}^{C} = 0, 
\end{equation}
and for the Maxwell field
\begin{equation}
\label{eqn:Maxwell}
D^A F_{AB} = 0 .
\end{equation}
Here $D$ is the covariant derivative with respect to the spin connection. From the symmetry properties of the fields, we can derive the following Bianchi-identities for the Riemann-tensor and the field strength
\begin{eqnarray}
\label{eqn:MaxwellBianchi}
\epsilon^{ABCD}D_B F_{CD} &= &0 ,\\
\label{eqn:RiemannBianchi1}
\epsilon^{ABCD} R_{BCDE} &=& 0 .
\end{eqnarray}

For the analysis of finite symmetries, we will mainly be concerned with space-times preserving some subgroup of the diffeomorphism group of $\mathcal{M}_4$. These residual symmetries are described by the existence of Killing vectors $\kappa$. The Maxwell field will also preserve this symmetry if its Lie derivative with $\kappa$ vanishes. The dynamics of such solutions can be described from a three-dimensional perspective, formally reducing (\ref{eqn:SugraAction4d}) on the orbits of $\kappa$. This three-dimensional theory is then living on an orbit space $\mathcal{M}_3 =  \mathcal{M}_4/\Sigma$, where $\Sigma$ is the exponentiated action of $\kappa$ on $\mathcal{M}_4$.\footnote{Note that generally $\kappa$ will vanish on certain submanifolds of $\mathcal{M}_4$, and when defining its orbit space, we choose a component of $\mathcal{M}_4$ where $\kappa$ is non-vanishing and connected to infinity.}

In three dimensions vector fields have only one propagating degree of freedom and are hence equivalent to scalar fields. One can therefore dualise a vector -- using the Hodge star on the corresponding field strength -- to a scalar by explicitly imposing its Bianchi-identity and consequently write the three-dimensional Lagrangian only in terms of a metric and a set of scalars (see Section \ref{subsec:compactificationd3th}). For example, a four-dimensional stationary Maxwell-field will in three dimensions be described by two scalars (one from the component of the potential in the time-direction, and one from dualisation). As we will see below, this will concretely realize electromagnetic duality as well as a gravitational duality such as the Ehlers symmetry. As a consequence, the three-dimensional theory allows for a big set of symmetries, acting on the set of solutions preserving the given Killing vector. In fact, the moduli space of solutions  (almost) realizes a generally non-linear representation of this symmetry group. We will discuss this in more detail in Section \ref{sec:GroupAction}.

\subsection{$\mathrm{SU}(2,1)$ and coset models}
\label{SU(2,1)SigmaModel}

In the following, the group $\mathrm{SU}(2,1)$ and some of its subgroups will play an important role since $\mathrm{SU}(2,1)$ is the global symmetry group of Einstein-Maxwell theory in the presence of a Killing vector~\cite{Kinnersley}. We briefly discuss its definition and the construction of coset models with $\mathrm{SU}(2,1)$ symmetry, relegating more details and explicit expressions of the generators to the Sections \ref{app:SigmaModel} and  \ref{subsec:su21andextension}.

In our conventions the group $\mathrm{SU}(2,1)$ is defined by the set of all unit-determinant complex $(3\times 3)$ matrices $g$ that preserve a metric $\eta$ of signature $(2,1)$;
\begin{equation}
{\mathrm {SU}}(2,1) = \left\{ g \in \mathrm{SL}(3, \mathbb{C}) \,:\, g^\dagger \eta g = \eta \right\}\quad\text{with}\quad
\eta = \left(\begin{array}{ccc}0&0&-1\\0&1&0\\-1&0&0\end{array}\right)\,,
\end{equation}
and we denote the associated Lie algebra by $\mathfrak{su}(2,1)$. This Lie algebra is a real form\footnote{We refer the reader to the Section \ref{sec:realkm} for introductions on real forms of complex Lie algebras.} of $\mathfrak{sl}(3,\mathbb{C})$ which may be described via the Tits-Satake diagram shown in Figure~\ref{fig:su21}. The labelling of nodes in Figure~\ref{fig:su21} is chosen to leave room for the extension of $\mathfrak{su}(2,1)$ to the Kac-Moody algebra $\mathfrak{su}(2,1)^{+++}$  discussed in Chapter~\ref{chap:infinitedimsymsu21}. 

\begin{figure}[t]
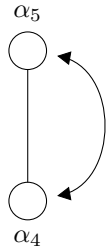

\begin{center}
\begin{pgfpicture}{1cm}{0cm}{1cm}{3cm}

\pgfnodecircle{Node1}[stroke]{\pgfxy(1,0.5)}{0.25cm}
\pgfnodecircle{Node2}[stroke]
{\pgfrelative{\pgfxy(0,2)}{\pgfnodecenter{Node1}}}{0.25cm}

\pgfnodebox{Node6}[virtual]{\pgfxy(1,0)}{$\alpha_{4}$}{2pt}{2pt}
\pgfnodebox{Node7}[virtual]{\pgfxy(1,3)}{$\alpha_{5}$}{2pt}{2pt}
\pgfnodeconnline{Node1}{Node2} 

\pgfsetstartarrow{\pgfarrowtriangle{4pt}}
\pgfsetendarrow{\pgfarrowtriangle{4pt}}
\pgfnodesetsepend{5pt}
\pgfnodesetsepstart{5pt}
\pgfnodeconncurve{Node2}{Node1}{-10}{10}{1cm}{1cm}

\end{pgfpicture}
\caption {\label{fig:su21} \sl \small The Tits-Satake diagram of the real form $\asu$. This real form is in one to one correspondence with a conjugation $\sigma$ of the complex Lie algebra $A_2=\mathfrak{sl}(3,\mathbb{C})$ which fixes completely $\asu$. The Tits-Satake diagram precisely gives  the action of  $\sigma$ on the simple roots of $A_2$. The presence of the double arrow means that $\sigma(\alpha_4)= \alpha_5$ and $\sigma(\alpha_5)= \alpha_4$. See Section~\ref{subsec:su21andextension}  for more details.} 
\end{center}
\end{figure}
The Lie algebra $\mathfrak{su}(2,1)$ has two maximal subalgebras that will play a central role in what follows. The first one is the maximal compact subalgebra $\mathfrak{k}=\mathfrak{su}(2)\oplus \mathfrak{u}(1)$, defined by the subset of generators which are pointwise fixed by the so-called Cartan involution $\theta$:
\begin{equation}
\mathfrak{k}=\mathfrak{su}(2)\oplus \mathfrak{u}(1)=\{ x\in \mathfrak{su}(2,1)\,: \, \theta(x)=x\}.
\end{equation}
The other (non-compact) maximal subalgebra is $\mathfrak{k}^{*}=\mathfrak{sl}(2,\mbb{R})\oplus \mathfrak{u}(1)$, which is similarly defined with respect to a ``temporal involution'' $\Omega_4$. The two involutions $\theta$ and $\Omega_4$ are discussed in more detail in Sections \ref{app:k} and  \ref{sec:carttemporalin}.  We recall that the Cartan involution induces the following Cartan decomposition in terms of vector spaces\footnote{In comparison with Section \ref{subsec:su21andextension}, we remove from now the subscript $_r$ of $\mathfrak{k}_r$ and $\mathfrak p_r$ to reduce the notations.} 
\begin{equation}
\mathfrak{su}(2,1)=\mathfrak{k}\oplus \mathfrak{p},
\end{equation}
where $\mathfrak{p}$ is the subspace which is anti-invariant under $\theta$, corresponding to the orthogonal complement of $\mathfrak{k}$ with respect to the Killing form on $\mathfrak{su}(2,1)$. Similarly, the temporal involution $\Omega_4$ induces the analogous decomposition
\be \mathfrak{su}(2,1)=\mathfrak{k}^{*}\oplus \mathfrak{p}^{*}.
\ee 
Note that $\mathfrak{p}$ and $\mathfrak{p}^{*}$ transform respectively  in representations of $\mathfrak{k}$ and $\mathfrak{k}^{*}$ but are not subalgebras of $\mathfrak{su}(2,1)$. For later reference, let us recall another useful decomposition of $\mathfrak{su}(2,1)$, known as the algebraic Iwasawa decomposition in terms of vector spaces
\begin{equation}
\mathfrak{su}(2,1)=\mathfrak{k}\oplus \mathfrak{a} \oplus \mathfrak{n}_+,
\label{Iwasawa}
\end{equation}
where $\mathfrak{a}$ is the non-compact part of the Cartan subalgebra $\mathfrak{h}\subset \mathfrak{su}(2,1)$ and $\mathfrak{n}_+$ is the nilpotent subspace of upper-triangular matrices. The subspace $\mathfrak{b}_+=\mathfrak{a} \oplus \mathfrak{n}_+$ is known as the standard Borel subalgebra.

At the group level, we then have the corresponding maximal compact subgroup  
\begin{equation}
\mathrm{K} = \mathrm{SU}(2)\times \mathrm{U}(1)\, 
\end{equation}
and non-compact subgroup 
\begin{equation}
\mathrm{K}^* = \mathrm{SL}(2, \mathbb{R})\times \mathrm{U}(1)\, .
\end{equation}
Similarly, the subspaces $\mathfrak{p}$ and $\mathfrak{p}^{*}$ correspond to the two coset spaces
\begin{equation}
\mathcal{C} = \mathrm{G}/\mathrm{K} = \frac{\mathrm{SU}(2,1)}{\mathrm{SU}(2)\times \mathrm{U}(1)}
\quad\text{and}\quad
\mathcal{C}^* = \mathrm{G}/\mathrm{K}^* = \frac{\mathrm{SU}(2,1)}{\mathrm{SL}(2, \mathbb{R})\times \mathrm{U}(1)}\,.
\end{equation}
Physically, $\mathcal{C}$ and $\mathcal{C}^{*}$ arise, respectively, as the moduli spaces of scalars upon reduction to $D=3$ of the Einstein-Maxwell Lagrangian on a space-like or a time-like circle. 

The coset space $\mathcal{C}=\mathrm{G}/\mathrm{K}$ is a Riemannian symmetric space of dimension $\dim(\mathrm{G})-\dim(\mathrm{K})=4$, matching the combined number of degrees of freedom contained in the metric and the Maxwell field in $D=4$. To describe a coset model on this space one can choose a map $\mathcal{V}: \mathcal{M}_3  \rightarrow \mathrm{G}/\mathrm{K}$ in a fixed gauge, that transforms under global transformations $g\in\mathrm{G}$ as $\mathcal{V}(x) \to k(x) \mathcal{V} g^{-1}$, where $k(x)\in\mathrm{K}$ is a local compensating transformation required to restore the chosen gauge for the coset representative.

A manifestly SU(2,1)-invariant Lagrangian can now be constructed using the Maurer-Cartan form $\mathrm{d}\mathcal{V} \mathcal{V}^{-1}$ as follows. Its projection
\begin{equation}
\mathcal{P}= \frac{1}{2} \big(\dd\mathcal{V} \mathcal{V}^{-1}\,  - \,  \theta(\dd\mathcal{V} \mathcal{V}^{-1})\big) \ \in \mathfrak p
\end{equation}
 along the coset transforms $\mathrm{K}$-covariantly under the global $\mathrm{G}$ action as $\mathcal{P}\to k \mathcal{P}k^{-1}$ and the (invariant) trace of its square can be used as a $\mathrm{G}$-invariant Lagrangian that is second order in derivatives: 
 \begin{equation}
 \mathcal{L}=\sqrt{|h|}h^{\mu\nu}  (\mathcal{P}_{\mu} | \mathcal{P}_{\nu}) ,
 \end{equation}
 where $h_{\mu\nu}$ is the metric on $\mathcal{M}_3$. To make this concrete, we shall extend the decomposition (\ref{Iwasawa}) to the group level using the Iwasawa theorem, so that
\begin{equation}
\mathrm{SU}(2,1)=\mathrm{KAN},
\end{equation}
where $\mathrm A$ is the abelian group with the Lie algebra $\mathfrak{a}$ and $\mathrm N$ is the nilpotent group corresponding to the subspace $\mathfrak{n}_+$.
This ensures that we may choose a coset representative $\mathcal{V}\in \mathrm{AN}$ of upper-triangular matrices, traditionally referred to as the ``Borel gauge''. As a consequence of this gauge choice, the coset element  $\mathcal{V}$ can be parametrized using four scalar fields (coordinates on G/K), to be illustrated in detail below.

On the other hand, the coset space $\mathcal{C}^*=\mathrm{G}/\mathrm{K}^*$ is a homogeneous space of dimension $\dim(\mathrm{G})-\dim(\mathrm{K}^*)=4$, but is no longer a Riemannian manifold. Rather, it has signature $(2,2)$, and is usually referred to as a ``pseudo-Riemannian'' symmetric space \cite{Breitenlohner:1987dg}. The general construction of an SU(2,1)-invariant coset model discussed above is still applicable, although in this case the global Iwasawa decomposition is no longer valid (see the discussion in Section \ref{app:SigmaModel}). 

For both choices of subgroup, $\mathrm{K}$ and $\mathrm{K}^*$, there is a Noether current 
\begin{equation}
\label{eqn:Noether}
\mathcal J^{\mu} =  \sqrt{|h|} h^{\mu \nu} \mathcal{V}^{-1} \mathcal P_{\nu} \mathcal{V} ,
\end{equation}
associated to the global $\mathrm{G}$ symmetry. We will see later that its values ``at infinity''  for $\mathcal{V}$ describing a half-BPS solutions can be related to the four charges describing the most general such solution.

\subsection{Solutions with space-like Killing vector}
\label{sec:SpacelikeKillingVector}

To give a flavor of the relevance of coset models, we will first quickly consider the case of solutions for which the preserved Killing vector $\kappa$ is space-like. After choosing suitable coordinates so that $\kappa = \partial_x$, the reduction of the Einstein-Maxwell Lagrangian (\ref{eqn:SugraAction4d}) yields in three dimensions, after dualisation, an Einstein plus scalar Lagrangian, with scalar part given by:
\be \label{eqn:lag3ds} \begin{split}
\mathcal{\tilde{L}}_{\mathrm{scal}}= -\,  \frac{\sqrt{ -h}}{4}\,  \Big( \frac{1}{2}\,  \p_{ \m} \phi\,  \p^{\m} \phi &+ 2\,  e^{ \phi} (\p_{ \m} \chi_e\,  \p^{\m} \chi_e+ \p_{ \m} \chi_m\,  \p^{\m} \chi_m) \\ 
&+ e^{2 \phi} (\p_{ \m} \psi +  \sqrt{2}\,  \chi_m \, \p_{\mu} \chi_e - \sqrt{2} \, \chi_e\, \p_{\mu} \chi_m)^2 \Big).
\end{split}
\ee

This reduced Lagrangian contains a dilaton $\phi$ and three axions: $\chi_e$ coming from the component $A_x$ of the Maxwell vector potential, $\chi_m$ coming from the dualisation of the Maxwell vector potential in three dimensions and $\psi$ arising from the dualization of the graviphoton. 
These four scalar fields parametrize the coset  space $\mathcal{C} = \su/(\mathrm{SU}(2)\times \mathrm{U}(1))$ \cite{Kinnersley,Julia:1980gr}. More concretely, the scalar dynamics given by (\ref{eqn:lag3ds}) is equivalent to a non-linear $\sigma$-model describing maps $\mathcal{V}$ from $\mathcal{M}_3$ to the homogeneous space $\mathcal{C}$ as described above. 
This map $\mathcal{V}$ is the composition of a map into the tangent space of $\mathcal{C}$ together with the exponential map from this tangent space to the coset. Naturally parametrized by the four scalar fields, it is therefore given by the expression
\be
\label{eqn:CosetElement} \begin{split}
\mathcal{V}&= e^{\frac{1}{2} \phi\, \mathbf{h_4}}\, e^{\sqrt{2} \,\chi_e\,    \mathbf{e_4} + \sqrt{2} \,\chi_{m} \,   \mathbf{e_5} + \sqrt{2}\, \psi \,   \mathbf{e_{4,5}}} \\
&= \left(\begin{array}{ccc}e^{\frac{\phi}{2}} & e^{\frac{\phi}{2}}(\sqrt{2} \, \chi_e + i \, \sqrt{2}\,  \chi_m) & e^{\frac{\phi}{2}}(\chi_e^{\  2} +  \chi_m^{\  2 }+ i\,  \sqrt{2} \, \psi) \\0 & 1 & \sqrt{2} \, \chi_e - i \, \sqrt{2}\,  \chi_m \\0& 0 & e^{- \frac{\phi}{2}} \end{array}\right)\, ,
\end{split}
\ee
where $   \mathbf{h_4},\,    \mathbf{e_4},\,    \mathbf{e_5}$ and  $   \mathbf{e_{4,5}}$ (see (\ref{eqn:su21gens})) are the generators of the Borel subalgebra $\mathfrak{b}_+=\mathfrak{a}\oplus \mathfrak{n}_+$.\footnote{Recall from the previous section that the Borel subalgebra $\mathfrak{b}$ of a non-split real form does not contain the full Cartan subalgebra $\mathfrak{h}$ but only its non-compact part $\mathfrak{a}$ (see Section \ref{subsec:restiwa} for more details). }  The scalar fields here depend on the coordinates $x^\mu$ of $\mathcal{M}_3$.

\subsection{Solutions with time-like Killing vector}
\label{sec:StationarySolutions}
Let us now repeat this discussion in a little bit more detail in a case that will be more interesting for us, namely BPS solutions. These solutions  preserve a time-like Killing vector, which is most easily seen by considering the fact that a BPS-solution necessarily preserve a Killing spinor $\epsilon$. Forming the supersymmetry generator $Q = Q^{\mu}\epsilon_{\mu}$, the supersymmetry algebra -- in which $Q$ squares to the generator of time-translation -- shows that the solution must be preserved under time-translation. The set of single centered BPS-solutions is a subset of a general set of generalized Reissner-Nordstr\"om black hole-like solutions to the equations of motion (\ref{eqn:Einstein}) and (\ref{eqn:Maxwell}) with mass $m$, NUT charge $n$, and electric and magnetic Maxwell charges $q$ and $h$. This solution can be written as~\cite{Argurio:2008zt}
\be\label{eqn:TaubNUTsolution} \begin{split}
\dd s^2 =& - \frac{{\tilde r}^2-n^2- 2\,m\,{\tilde r}+q^2+h^2}{{\tilde r}^2+n^2}\,  ( \, \dd t + 2 \,n \,\cos \theta\, \dd \phi \,)^2  \\ 
&+ \frac{\tilde{r}^2+n^2}{{\tilde r}^2-n^2- 2\,m\, {\tilde r}+q^2+h^2}\, \dd {\tilde r}^2 + ({\tilde r}^2+n^2)(\dd \theta^2+ \sin^2 \theta \, \dd \phi^2)\, , 
\end{split}
\ee
\begin{eqnarray} \label{eqn:TaubNUTsolution2}
& A_t& = \frac{q\,{\tilde r} +n\,h}{{\tilde r}^2+n^2}, \hspace{1cm}   A_{\phi}= \frac{2\, n\, q\,  {\tilde r} - h\, ({\tilde r}^2-n^2)}{{\tilde r}^2 +n^2}\, \cos \theta \, .
\end{eqnarray}
For our purpose it is convenient to introduce isotropic coordinates defined by ${\tilde r}= r+ m$, in which  the metric in (\ref{eqn:TaubNUTsolution}) becomes
\begin{equation}
\dd s^2=-\frac{\lambda(r)}{R^2(r)} (\dd t+2\,n \, \cos \theta \, \dd \phi \, )^2+ \frac{R^2(r)}{\lambda(r)} \dd r^2 +
R^2(r) (\dd \theta^2+ \sin^2 \theta \, \dd \phi^2)\, .
\ee
Here the functions $\lambda(r)$ and $R^2(r)$ are given by 
\begin{eqnarray}
&&\lambda(r)= r^2-\Delta^2, \qquad \Delta^2 \equiv m^2+n^2-q^2-h^2 \label{lambda}, \\
&& R^2(r)=r^2+2mr+m^2+n^2 \label{R2} .
\end{eqnarray}

We will be interested in the subclass of solutions (\ref{eqn:TaubNUTsolution}) which are BPS, namely the ones which preserve 1/2 of the supersymmetry. These solutions are characterised by the following  BPS condition among the four charges \cite{Kallosh:1994ba}  (see also \cite{Argurio:2008zt}):
\begin{equation}
\label{BPScondp}
\Delta^2 =0   \quad \Leftrightarrow \quad m^2+n^2=q^2+h^2.
\end{equation}
Using  (\ref{lambda}) and (\ref{BPScondp}), we have $\lambda_{BPS}=r^2$ and the BPS metric is:
\begin{equation}
\label{Isom}
\dd s^2_{BPS}=-\frac{r^2}{R^2(r)} (\dd t+2\,n \, \cos \theta \, \dd \phi \, )^2 + \frac{R^2(r)}{r^2} \sum_{a=1}^{3}\dd x_a^2,
\end{equation} 
where the $x_a$'s are the flat Euclidean space coordinates.
In the particular case $n=0$, one finds again the extremal Reissner-Nordstr\"om black hole (or an extremal $0$ brane in four dimensions) characterised by the harmonic function $1+\frac{m}{r}$ (see for instance \cite{Argurio:1998cp}). Note that upon dimensional reduction of (\ref{eqn:TaubNUTsolution}) on time the three-dimensional Euclidean metric is
\begin{equation}
\dd s^2_{3D}= \dd r^2 + (r^2-\Delta^2) \dd \Omega_2^2.
\label{redu3d}
\end{equation}
When the BPS condition (\ref{BPScondp}) is fulfilled, the line element (\ref{redu3d}) is just the line element for three-dimensional Euclidean flat space in spherical coordinates.

This solution is an example of a solution to the Einstein-Maxwell system which allows for a time-like Killing vector. 
These solutions are referred to as {\it stationary}. More generally, choosing suitable coordinates, in this case so that $\kappa = \partial_t$, a convenient metric ansatz for this type of solutions is\footnote{In the original work on time-like reductions to $D=3$ \cite{Breitenlohner:1987dg}, it was noted that by further assuming spherical symmetry for the three-dimensional metric $g_3$, the Einstein-scalar Lagrangian in $D=3$ describes the geodesic motion of a fiducial particle moving on (a cone over) the moduli space $\mathcal{C}^{*}=\mathrm{SU(2,1)}/(\mathrm{SL}(2,\mathbb{R})\times \mathrm{U(1)})$. The dynamics of the particle on $\mathcal{C}^{*}$ thus corresponds to motion in the space of stationary, spherically symmetric solutions to Einstein-Maxwell theory. This point of view has been extended recently in \cite{Gaiotto:2007ag,Bergshoeff:2008be} in the context of solution-generating techniques. Moreover, for the special case of BPS solutions this philosophy was elaborated upon in \cite{Pioline:2006ni,Neitzke:2007ke}, where it has been shown that the classical phase space of the particle dynamics coincides with the six-dimensional coset space $\mathcal{Z}=\mathrm{SU(2,1)/(U(1)}\times \mathrm{U(1)})$ (known as the twistor space of $\mathcal{C}$, see e.g. \cite{Gunaydin:2007qq}). This result has been used as a starting point for quantizing BPS black hole solutions by (radial) quantization of the particle dynamics on $\mathcal{C}^{*}$ \cite{Gunaydin:2007bg}.} 
\begin{eqnarray}
\label{redukill}
\dd s^{2}_{4D} = - e^{- \phi} (\dd t + \omega)^2+ e^{\phi} \dd s_{3D}^2\, ,
\end{eqnarray}
where $\dd s^2_{3D}$ is the invariant line element on $\mathcal{M}_3$ corresponding to the pullback metric $g_3$ of the four dimensional metric $g$ to $\mathcal{M}_3$.\footnote{In the case of extremal solutions this will give a flat $g_3$ metric, and for extremal solutions with horizon this will generally give $\mathcal{M}_3$ a topology homeomorphic to $\mathbb{R}^3\backslash \{0\}$.} Here we see explicitly the origin of the graviphoton, given as the 1-form $\omega$.  The dynamics of these particle-like solutions is now governed by (\ref{eqn:SugraAction4d}) reduced on the orbits of the Killing vector $\kappa$. Explicitly
\be \label{eqn:3dlagt} \begin{split}
\mathcal{L}_{3D}' =  \frac{1}{4} \sqrt{g_3} \,\Big[\ &\big(R_3 - \tfrac{1}{2}  \p_{ \m} \phi\,  \p^{\m} \phi+ \tfrac{1}{4} e^{-2 \phi} \, F^2_{(2)} \big)\\
&-\, \big(e^{- \phi} \tilde{F}^2_{(2)} - 2\,  e^{ \phi}\, \p_{ \m} \chi_e\,  \p^{\m} \chi_e  \big)   \   \Big]\, ,
\end{split}\ee
where $F_{(2)}= \dd\omega$ and $\tilde{F}_{(2)}= \dd A- \dd\chi_e\wedge \omega  $. After dualization of the two field strengths,
\begin{eqnarray}
\label{eqn:DualisingMaxwell}
\tilde{F}^{\lambda \nu }&=& - \frac{ \epsilon^{\mu \lambda \nu}}{\sqrt{g_3}}\, e^{\phi}\, \p_{\mu}\chi_m\, , \\
\label{eqn:DualisingEinstein}
F^{\lambda \nu }&=&  \frac{ \epsilon^{\mu \lambda \nu}}{\sqrt{g_3}}\, e^{2 \phi}\, \Big( 2\, (\chi_m\, \p_{\mu} \chi_e - \chi_e\,\p_{\mu} \chi_m ) + \sqrt{2} \, \p_{\mu} \psi \Big)\, ,
\end{eqnarray}
with $ \epsilon^{r \theta \phi} =-1 $, we can  rewrite the three-dimensional Lagrangian (\ref{eqn:3dlagt}) using only the metric in three dimensions and the scalar fields $\phi, \chi_e, \chi_m$ and $\psi$. We hence get
\be  \label{eqn:3dStationaryLagrangian} \begin{split}
 \mathcal{L}_{3D}=  \frac{1}{4}\, \sqrt{g_3} \, \Big[ \ R_3 &- \tfrac{1}{2}\,  \p_{ \m} \phi\,  \p^{\m} \phi + 2\,  e^{ \phi}\, (\p_{ \m} \chi_e\,  \p^{\m} \chi_e+ \p_{ \m} \chi_m\,  \p^{\m} \chi_m) \\
&- e^{2 \phi}\, (\p_{ \m} \psi +  \sqrt{2}\,  \chi_m \, \p_{\mu} \chi_e - \sqrt{2} \, \chi_e\, \p_{\mu} \chi_m)^2 \, \Big]\, .
\end{split}
\ee
One sees directly that $\chi_e$ and $\chi_m$ appear completely symmetrically. The duality between the two gravitational scalars is less apparent, but is in fact present as we will see in Section \ref{sec:GroupAction}. Note the change of signs in front of the kinetic terms for the Maxwell scalars in comparison with \eqref{eqn:lag3ds}, revealing that the two scalar actions for space-like and time-like reductions are related by a ``Wick rotation'' of the Maxwell scalars $\chi_e$ and $\chi_m$.

The scalar part of (\ref{eqn:3dStationaryLagrangian}) can now be identified with a non-linear $\sigma$-model constructed on the coset $\mathcal{C}^* = \su/(\mathrm{SL}(2, \mathbb{R})\times \mathrm{U}(1))$, where the change in the quotient group has its origin in the different kinetic terms. Hence the construction of this theory is the same as the one used for space-like reduction, except that when deriving the coset Lagrangian we replace the Cartan involution $\theta$ by the temporal involution $\Omega_4$ (defined in Section \ref{app:k*}), having as fixed subalgebra $\mathfrak{k}^* = \mathfrak{sl}(2,\mathbb{R}) \oplus \mathfrak{u}(1)$. \\

Generally we still write maps $\mathcal{V}$ in the Borel gauge by using the expression (\ref{eqn:CosetElement}). The solution (\ref{eqn:TaubNUTsolution}) and (\ref{eqn:TaubNUTsolution2})  in terms of our four scalar fields can then be rewritten as:
\be \begin{split}\label{eqn:BPSscalars}
\phi(r)&= \ln \Big(\frac{(r+m)^2 + n^2 }{ r^2}\Big)\, , \\
\chi_e(r)&= \frac{h n + q (m+r)}{(r+m)^2 + n^2 } \,,  \\
\chi_m(r)&= \frac{nq - h(m+r)}{(r+m)^2 + n^2}\, , \\
\psi(r)&=   -\frac{\sqrt{2} n r}{(r+m)^2 + n^2}\, .
\end{split}
\ee
Using the Noether current of the coset Lagrangian, we can now relate the three-dimensional ``conserved'' $\sigma$-model quantities to the four-dimensional charges. At infinity ($r\to\infty$) the coset element $\mathcal{V}$ parametrized by (\ref{eqn:BPSscalars}) tends towards the identity element ${\bf{1}}$ of $\mathrm{SU}(2,1)$, implying that $\mathcal{J} \rightarrow \mathcal{P}$ when $r\to\infty$. 

Furthermore, one can compute using \eqref{eqn:Noether} (see also \cite{Bossard:2009at})
\begin{equation}\label{eqn:charges} \begin{split} \
Q &= \int_{S^{\infty}} \mathcal{J} =  \int_{S^{\infty}} \mathcal{P}\\
&=  - m\,  \mathbf{h_4}+ n\, ( \mathbf{e_{4,5}}-  \mathbf{f_{4,5}})- \frac{q}{\sqrt{2}} \, ( \mathbf{e_4}-  \mathbf{f_4}) +\frac{h}{\sqrt{2}}\, ( \mathbf{e_5}+ \mathbf{f_5}) \\
&= \left(\begin{array}{ccc}-m & \frac{i(h+iq)}{\sqrt{2}} & in \\\frac{ih + q}{\sqrt{2}} & 0 & \frac{-ih-q}{\sqrt{2}} \\-in & \frac{-ih+q}{\sqrt{2}} & m\end{array}\right),
\end{split}
\end{equation} 
where $m,n,q$ and $h$ are the four-dimensional charges. For the derivation of the elements in $\mathfrak{p}^*$ which is the orthogonal complement of $\mathfrak{k}^*$ with respect to the Killing form, see Section \ref{app:k*}.

The form (\ref{eqn:charges}) is preserved by coset transformations belonging to $\mathrm{K}^*$ since $\mathcal{J}=\mathcal{P}\in \mathfrak{p}^*$ and the reductive homogeneous space decomposition ensures that $[\mathfrak{k}^*,\mathfrak{p}^*]\subset\mathfrak{p}^*$. As we will argue below, the transformations from $\mathrm{K}^*$ preserve the asymptotic conditions on the BPS solution and therefore act (linearly) on the four BPS charges. The transformations that belong to $\mathrm{K}^*$ also preserve the asymptotic condition $\mathcal{V}\to \bf{1}$ as $r\to\infty$. The $\mathrm{SU}(2,1)$ transformations that are not part of $\mathrm{K}^*$ violate this asymptotic condition on the coset element and also map the Noether current $\mathcal{J}$ out of $\mathfrak{p}^*$. This makes the identification of the four BPS charges from the Noether current less evident. However, as the physical fields are related to the scalar fields of the coset $\mathcal{C}^*$ (mostly) by duality relations these transformations do preserve the asymptotic conditions on the physical fields. In fact, we will show below that  Iwasawa decomposable  transformations of $\mathrm{SU}(2,1)$ outside $\mathrm{K}^*$ act as gauge transformations on the physical fields and do not change the BPS charges. For this reason it will turn out to be sufficient to use the Noether charges from (\ref{eqn:charges}) and their transformation under $\mathrm{K}^*$ to find the orbits of BPS solutions under $\mathrm{SU}(2,1)$.

\setcounter{equation}{0}
\section{Action of $\mathrm{SU}(2,1)$ on BPS solutions}
\label{sec:GroupAction}

Let us now proceed to discuss the action of SU(2,1) on the stationary solution (\ref{eqn:TaubNUTsolution}). The $\sigma$-model is $\mathrm G$-invariant by construction and acting with $\su$ on the coset with the natural action from the right, gives an action on the maps $\mathcal{V}$. We thus generate new solutions when lifting the transformed $\mathcal{V'}$ back to four dimensions, using the explicit form (\ref{eqn:CosetElement}) of $\mathcal{V}$ and the dualisation relations (\ref{eqn:DualisingMaxwell}) and (\ref{eqn:DualisingEinstein}). Furthermore, we know that every single centered extremal solution is uniquely fixed by the values of the scalar fields at infinity, in terms of the four charges $m,\, n,\, q$ and $h$. This induces a representation of $\su$ on these four charges. Now, as the coset space $\mathcal{C}^*$ in the case of stationary solutions is not a Riemannian symmetric space, there is not a single coordinate system covering the whole coset \cite{Bossard:2009at}. However, our $\sigma$-model describes maps to a given coordinate patch. If the action of $\mathrm G$ takes us outside of this patch, we have no way of relating the new $\mathcal{V}'$ to the four-dimensional fields. Constructing our coordinate system on $\mathcal{C}^*$ via the Borel gauge (i.e. treating $\mathcal{V}$ as the composition of the exponential map and a map from $\mathcal{M}_3$ to $\mathfrak{a}\oplus \mathfrak{n}_+$), we will only consider the subspace of $\mathrm G$ where the elements are decomposable in Iwasawa form. These elements are exactly the ones that preserve our coordinate patch. Hence we can consider the action of $\su$ in three different cases, one for each of the subgroups $\mathrm N,\mathrm A$ and $\mathrm K^*$ in the local Iwasawa decomposition.\footnote{A similar analysis was recently done for five-dimensional minimal supergravity which gives rise to a $\mathrm G_2$ $\sigma$-model when  this theory is reduced on two commuting Killing vectors \cite{Compere:2009zh}.} Our four-dimensional interpretation will differ in all of these cases. Solution generation in the case of Einstein-Maxwell theory has been considered also in~\cite{Neugebauer:1969wr,Kinnersley,Bossard:2009at}.

\subsection{Action of the nilpotent generators}

Let us begin with the analysis of the nilpotent group $\mathrm N$. As elements in $\mathrm N$ do not take us outside of the Borel gauge, the analysis of how the scalar fields change is simply done by multiplying $\mathcal{V}$ described by (\ref{eqn:CosetElement}) from the right by elements in the group $\mathrm N$ of nilpotent elements, i.e. if $n \in \mathrm N$, $\mathcal{V} \rightarrow \mathcal{V}' = \mathcal{V}n^{-1}$.

As described in Section \ref{subsec:su21andextension}, the Lie algebra $\mathfrak{n}_+$ of $\mathrm N$ is generated by the three nilpotent generators $ \mathbf{e_4}, \mathbf{e_5}$ and $ \mathbf{e_{4,5}}$. Under the three corresponding nilpotent 1-parameter subgroups (with real parameters $\alpha,\beta$ and $\gamma$), the scalar fields given by (\ref{eqn:BPSscalars}) transform as follows; Under the group generated by $ \mathbf{e_4}$: 
\be \begin{split} 
\chi_e &\rightarrow \ \chi_e - \frac{1}{\sqrt{2}} \alpha\, ,\\
\psi &\rightarrow \ \psi - \alpha \chi_m .
\end{split}\ee
Under the group generated by $ \mathbf{e_5}$,
\be \begin{split} 
\chi_m &\rightarrow \ \chi_m - \frac{1}{\sqrt{2}} \beta \, ,\\
\psi &\rightarrow \  \psi + \beta \chi_e .
\end{split}\ee
Finally, under the group generated by $ \mathbf{e_{4,5}}$,
\begin{equation}
\psi \rightarrow \ \psi -  \frac{1}{\sqrt{2}}\gamma .
\ee
Looking at the dualisation relations (\ref{eqn:DualisingMaxwell}) and (\ref{eqn:DualisingEinstein}) we see that these transformations simply vanish when lifting the fields back to four dimensions. We can hence interpret the symmetry group $\mathrm N$ as appearing from realizing an inherent redundancy in the formulation of the three-dimensional theory, and is therefore not visible in four dimensions. Equivalently, the action of the nilpotent group $\mathrm N$ corresponds to gauge transformations.

\subsection{Action of the non-compact Cartan generator}

The action of the abelian group $\mathrm A$, with Lie algebra $\mathfrak{a}$, is generated by the non-compact Cartan generator $ \mathbf{h_4} \in \mathfrak{p}^*$. The action of $\mathrm A$, parametrized by $d \in \mathbb{R}$ is
\be \begin{split} 
\phi &\rightarrow \  \phi-2d \, , \\
\chi_e& \rightarrow \ e^{d} \chi_e \, , \\
\chi_m& \rightarrow \  e^{d} \chi_m \, , \\
\psi &\rightarrow \  e^{2d} \psi\, .
\end{split}\ee
Lifting this transformation back into the four-dimensional metric and Maxwell potential, we see that it is just a coordinate transformation coming from a rescaling of the time and space coordinates $t \rightarrow e^d t$ and $r \rightarrow e^{-d} r$. The solution is therefore unchanged. 

\subsection{Action of $\mathrm K^*$}
\label{sec:ActionOnCharges}

We have now discussed from a physical perspective why the generators in $\mathrm A \mathrm N \subset \mathrm G$ act trivially on a given solution. By restricting to transformations that stay in our coordinate patch on $\mathrm G/\mathrm K^*$, what is left to consider is now the non-compact group $\mathrm K^*$. It turns out that it is $\mathrm K^*$ that realizes electromagnetic and gravitational duality. From the expression of the Noether charge (\ref{eqn:charges}), we see that $\mathrm K^*$ transforms non-trivially the set of conserved four-dimensional charges, and it is natural to ask precisely how this action is realized. This is done by extracting the transformed charges as the coefficients in front of the generators of $\mathfrak{p}^*$ just as in the expression (\ref{eqn:charges}). The algebra $\mathfrak{k}^*$ of $\mathrm{K}^*$ is generated by the elements $u$, and $t_i$, $i=1,2,3$, where the $t_i$'s generate an $\mathfrak{sl}(2,\mathbb{R})$, commuting with $u$. The definition of $\mathfrak{k}^*$ is described in Section \ref{app:k*}. Treating these four Lie algebra generators separately, as we did in the case of $\mathrm N$ above, we find that the 1-parameter subgroup generated by $u$, with parameter $a$ generates the transformation
\begin{equation}
\left(\begin{array}{c}m \\n \\q \\h\end{array}\right) \rightarrow \left(\begin{array}{cccc}\cos(a) & \sin(a) & 0 & 0 \\-\sin(a) & \cos(a) & 0 & 0 \\0 & 0 & \cos(a) & \sin(a) \\0 & 0 & -\sin(a) & \cos(a)\end{array}\right) \left(\begin{array}{c}m \\n \\q \\h\end{array}\right) \, ,
\end{equation}
under finite transformations generated by $t_1$ with parameter $b$,
\begin{equation}
\left(\begin{array}{c}m \\n \\q \\h\end{array}\right) \rightarrow \left(\begin{array}{cccc}\cos(b) & -\sin(b) & 0 & 0 \\\sin(b) & \cos(b) & 0 & 0 \\0 & 0 & \cos(b) & \sin(b) \\0 & 0 & -\sin(b) & \cos(b)\end{array}\right) \left(\begin{array}{c}m \\n \\q \\h\end{array}\right) \, ,
\end{equation}
under $t_2$ with parameter $c$,
\begin{equation}
\left(\begin{array}{c}m \\n \\q \\h\end{array}\right) \rightarrow \left(\begin{array}{cccc}\cosh(c) & 0 & -\sinh(c) & 0 \\0 & \cosh(c) & 0 & -\sinh(c) \\-\sinh(c) & 0 & \cosh(c) & 0 \\0 & -\sinh(c) & 0 & \cosh(c)\end{array}\right) \left(\begin{array}{c}m \\n \\q \\h\end{array}\right) \, ,
\end{equation}
and finally under $t_3$ with parameter $d$,
\begin{equation}
\left(\begin{array}{c}m \\n \\q \\h\end{array}\right) \rightarrow \left(\begin{array}{cccc}\cosh(d) & 0 & 0 &\sinh(d)  \\0 & \cosh(d) & -\sinh(d)  & 0\\ 0 & -\sinh(d)  & \cosh(d) & 0 \\\sinh(d) \ &0 & 0 & \cosh(d)\end{array}\right) \left(\begin{array}{c}m \\n \\q \\h\end{array}\right) .
\end{equation}

We see here that $\mathrm K^*$ realizes a linear representation $R$ on the charges, identified with $R= {\bf 2}_{1} \oplus {\bf 2}_{-1}$ (decomposed with respect to $\mathfrak{sl}(2,\mathbb{R})\oplus \mathfrak{u}(1)$, where the subscript indicates the charge under $\mathrm{U}(1)$), acting as two boosts and two rotations. In particular we see that $u + t_1= \tfrac{2}{3}\, \mathbf{h_5}$ acts as a rotation of electric and magnetic charges. This is in agreement with  discussion in Section \ref{sec:levdecompo}, considering the commutation relations (see (\ref{comutenc45}))
\begin{equation}
[\mathbf{h_5},r^a] =3\, \tilde{r}_a \quad, \quad [\mathbf{h_5},\tilde{r}^a] =- 3\, r_a\, ,
\end{equation} 
and the identification in the dictionary (Table \ref{tab:Dictionary}) stating that the generators $r^a$ and $\tilde{r}^a$ correspond to the electric and magnetic parts of the Maxwell field. We can also see that $u-t_1$ acts as gravitational duality rotation (see for instance \cite{Bunster:2006rt,Argurio:2008zt}).

\subsection{Describing $\mathrm K^*$ as a subgroup of $\mathrm{SO}(2,2)$}
\label{sec:SO22Embedding}

From group theoretic considerations one can derive the above conclusions using rather general arguments. Let $\Delta^2 :\mathbb{R}^4 \rightarrow \mathbb{R} $ be the homogeneous quadratic form  defined  by
\begin{equation}
\Delta^2(v) = m^2+n^2-q^2-h^2,
\end{equation}
for $v=(m,n,h,q)\in\mathbb{R}^4$, and let 
\begin{equation}
B = \{ v \in \mathbb{R}^4 \backslash \{0\}; \Delta^2(v) = 0\}
\end{equation}
be the set of zeros of $\Delta^2$. We know from Section \ref{sec:StationarySolutions} that $B$ is isomorphic to the set of single centered BPS-solutions via the maps (\ref{eqn:BPSscalars}). The set $B$ is by definition preserved by the group $\mathrm S = \mathrm{SO}(2,2)$. The corresponding algebra is the real form $\mathfrak{so}(2,2)$ of the complex algebra $D_2= A_1 \oplus A_1$. It is defined as the set of matrices
\begin{equation}
X  = \left(\begin{array}{cc}X_1 & X_2 \\{X_2}^T & X_3\end{array}\right),
\end{equation}
where all $X_i$ are real $2\times2$ matrices, and $X_1$ and $X_3$ are skew symmetric \cite{Helgason:1978}. It is therefore spanned by the six generators
\begin{eqnarray}
b_1 = &\left(\begin{array}{cccc}0 & -1 & 0 & 0 \\1 & 0 & 0 & 0 \\0 & 0 & 0 & 1 \\0 & 0 & -1 & 0\end{array}\right)  & b_2 = \left(\begin{array}{cccc}0 & 0 & -1 & 0 \\0 & 0 & 0 & -1 \\-1 & 0 & 0 & 0 \\0 & -1 & 0 & 0\end{array}\right) \nonumber\\
b_3 = & \left(\begin{array}{cccc}0 & 0 & 0 & 1 \\0 & 0 & -1 & 0 \\0 & -1 & 0 & 0 \\1 & 0 & 0 & 0\end{array}\right) & b_4 = \left(\begin{array}{cccc}0 & 1 & 0 & 0 \\-1 & 0 & 0 & 0 \\0 & 0 & 0 & 1 \\0 & 0 & -1 & 0\end{array}\right) \\
b_5 =&\left(\begin{array}{cccc}0 & 0 & 0 & -1 \\0 & 0 & -1 & 0 \\0 & -1 & 0 & 0 \\-1 & 0 & 0 & 0\end{array}\right) & b_6 = \left(\begin{array}{cccc}0 & 0 & 1 & 0 \\0 & 0 & 0 & -1 \\1 & 0 & 0 & 0 \\0 & -1 & 0 & 0\end{array}\right) \nonumber \\
\nonumber
\end{eqnarray}
and the choice of base here is to streamline the  analysis afterwards. In fact, $\mathfrak{so}(2,2) \cong \mathfrak{sl}(2,\mathbb{R}) \oplus  \mathfrak{sl}(2,\mathbb{R})$, which is easily seen in this basis as $b_1,b_2$ and $b_3$ generate one $\mathfrak{sl}(2,\mathbb{R})$ summand, and $b_4,b_5$ and $b_6$ the other. The two compact generators are $b_1$ and $b_4$.\\

The question of how $\mathrm K^*$ acts on the set of charges can then be transformed into the question of how $\mathrm K^*$ embeds into $\mathrm S$ as we know from the previous section that $\mathrm K^*$ preserves the BPS-condition. The Lie algebra isomorphism $\mathfrak{so}(2,2) \cong \mathfrak{sl}(2,\mathbb{R}) \oplus \mathfrak{sl}(2,\mathbb{R})$ induces the Lie group isomorphism $ \mathrm{SO}(2,2)_0 \cong \mathrm{SL}(2,\mathbb{R}) \times \mathrm{SL}(2,\mathbb{R})$. Here  $\mathrm{SO}(2,2)_0$ indicates the component connected to the identity. We also know that $ \mathfrak{sl}(2,\mathbb{R})$ contains two non-compact generators and one compact. Comparing with $\mathrm K^*$ whose Lie algebra we know contains two compact generators and two non-compact generators, forming $\mathfrak{k}^* = \mathfrak{sl}(2,\mathbb{R}) \oplus \mathbb{R}u$, the embedding $I : \mathrm K^* \hookrightarrow \mathrm S$ is therefore given by lifting the natural (up to automorphisms) differential $\mathrm{d}I_e : \mathfrak{k}^* \rightarrow \mathfrak{so}(2,2)$ at the identity $e\in\mathrm{K}^*$, mapping compact generators to compact generators. More concretely, if $b_{i=1,...6}$ are the generators of $\mathfrak{so}(2,2)$ 
\be \begin{split}
\mathrm{d}I_e(t_i)& = b_i \, ,\\
\mathrm{d}I_e(u) &= b_4 \, ,
\end{split}\ee
where $b_1,b_2$ and $b_3$ generate one $\mathfrak{sl}(2)$ summand in $\mathfrak{so}(2,2)$, and $b_4$ is the compact generator in the other. The normalization is not fixed, but is up to redefinition of the generators of the two Lie algebras. By looking at the action of $I(\mathrm K^*)$, we now see that $b_1-b_4$ generates an Ehlers $\mathrm{U}(1)$-group rotating $m,n$ into each other, $b_1+b_4$ generates a $\mathrm{U}(1)$ rotating $q,h$ and the non-compact $b_2$ and $b_3$ act as boosts. This is in complete agreement with the analysis in Section \ref{sec:ActionOnCharges} above.

\subsection{The space of BPS solutions}

Now as we know how $\su$, or more precisely, how $\mathrm K^*$ acts on the set of  single centered BPS-solutions we can ask the question about duality orbits. Namely, starting with one BPS-solution, can we generate all the others by acting with $\mathrm K^*$? If $\su$ is to be considered as a spectrum generating group \cite{Cremmer:1997xj}, this must clearly be the case. Here the fact that $\mathrm K^*$ is non-compact will be of crucial importance. In fact, we have the following result.

\begin{theorem}
\label{thm:modulispace}
The group $\mathrm K^*$ acts transitively on the set of single centered BPS-solutions, so that $B \cong \mathrm K^*/\mathbb{R}$.
\end{theorem}
\begin{proof}
Consider the set $B$. As it is defined by the homogeneous form $\Delta^2$, we can consider the projective descendant of $B$, namely $\mathrm{P}B =  \{ v \in \mathrm{P}\mathbb{R}^3 ; \Delta^2(v) = 0\}$, where $\mathrm{P}\mathbb{R}^3$ is the three-dimensional projective space. In analogy with the isomorphism $\mathrm{SO}(2,2)_0 \cong \mathrm{SL}(2,\mathbb{R}) \times \mathrm{SL}(2,\mathbb{R})$, we get the isomorphism $\mathrm{P}B \cong \mathrm{P}\mathbb{R}^1\times  \mathrm{P}\mathbb{R}^1$ via a bijection $F :  \mathrm{P}\mathbb{R}^1\times  \mathrm{P}\mathbb{R}^1 \rightarrow \mathrm{P}B$ given by the expression
\begin{equation}
F([x_0,x_1], [y_0,y_1]) = [x_0 y_0+x_1 y_1, y_0x_1-y_1x_0,x_0y_0-x_1y_1, x_0y_1+x_1y_0] .
\end{equation}
The action of $I(\mathrm K^*)$ on $B$ descends to an action on $\mathrm{P}B$, and hence to an action on $\mathrm{P}\mathbb{R}^1\times  \mathrm{P}\mathbb{R}^1$ by $F$. Furthermore, $\mathrm{P}\mathbb{R}^1 \cong S^1$ via the map $f([x_0,x_1]) = \arctan(x_0/x_1)$ (schematically), so that we have a diffeomorphism $\mathrm{P}B \cong S^1 \times S^1$. In fact, the Lie subgroup $\mathrm{U}(1) \times \mathrm{U}(1) \subset \mathrm K^*$, generated by the subalgebra $\mathbb{R} t_1 \oplus \mathbb{R} u$, acts transitively on these two circles by complex multiplication. We conclude that $\mathrm K^*$ acts transitively on $\mathrm{P}B$. Furthermore $\mathrm{U}(1) \times \mathrm{U}(1)\subset \mathrm K^*$ contains an element acting as $v \rightarrow -v$ for $v \in B$. Let us now turn to the action of $\mathrm K^*$ on $B$. Due to the above analysis, it is sufficient to consider charge vectors with all charges positive and equal. As $\mathrm K^*$ contains non-compact generators it is now in fact possible to reach all these charge vectors, being given one. The explicit 1-parameter Lie subgroup is $S(\lambda) = \exp\big(-\lambda b_2\big)$, acting so that $(k,k,k,k) \mapsto  \exp{\lambda} (k,k,k,k)$. This proves our assertion, noting that the 1-parameter subgroup stabilizing a diagonal vector  $(k,k,k,k)$ is $\mathrm{Stab}(c) \equiv \exp{c (b_1+b_3)} \cong \mathbb{R}$.
\end{proof}

If we consider this proof from the physical point of view, it may seem surprising that $\mathrm K^*$ is transitive on $\mathrm{P}B$ by only using $b_1$ and $b_4$ as these do not mix gravitational and electromagnetic degrees of freedom. This is in fact true as for four charges $m,n,q,h$ to fulfill the BPS-condition we need both non-zero gravitational and non-zero electromagnetic charges and to generate new solutions we can treat these two sectors separately.  Furthermore, we can compare the result of Theorem \ref{thm:modulispace} with the expression for the $1/2$-BPS strata in \cite{Bossard:2009at}, (equation (5.5)) and see that the two results are in full agreement.

\subsection{The quantum moduli space and string theory}

Our analysis so far has been performed purely at the classical level. In the full quantum theory it is expected that the classical moduli space  is affected by quantum corrections. These can be both of perturbative and of non-perturbative origin and they are not understood generally. The only exceptions are cases where there are additional duality symmetries that constrain them.

In general, electric and magnetic charges are subject to quantization in the sense of Dirac. For example, the $28+28$ electric and magnetic charges in type II string theory on a six-torus break the  classical continuous $\mathZ E_{7(7)}$ symmetry group to the discrete subgroup $\mathZ E_{7(7)}(\mbb{Z})$ \cite{HullTownsend}
\be
\mathZ E_{7(7)}(\mbb{Z})=\mathZ E_{7(7)}(\mbb{R})\cap\mathrm{Sp}(56, \mbb{Z})\,,
\ee
where $\mathrm{Sp}(56, \mbb{Z})$ is the symmetry group of the 56-dimensional symplectic lattice of electric and magnetic charges, associated with the 28 abelian vector fields in $D=4$. 

It has furthermore been speculated that after further reduction of this maximal supergravity theory on a space-like circle $S^1$ to $D=3$, the duality group should be enhanced to some discrete subgroup $\mathrm G(\mbb{Z})$ of the classical hidden symmetry group $\mathZ E_{8(8)}(\mbb{R})$ \cite{HullTownsend}. However, in three dimensions it is by no means clear how to define the group $\mathrm G(\mbb{Z})$, since there are no vector fields whose associated charge lattice provides a natural integral structure. Moreover, in $D=3$ one is forced to take into account gravitational effects since the moduli space includes components of the four-dimensional metric.
It was recently argued that the three-dimensional duality groups that arise in this way do not act nicely on the gravitational part of the moduli space, and there is therefore no natural candidate for a discrete subgroup $\mathrm G(\mbb{Z})$ which should be preserved in the quantum theory in $D=3$~\cite{Bossard:2009at}.

Returning to the ${\cal N}=2$ theory discussed in this paper, the situation is not very different at face value. However, we propose that the $c$-map \cite{CecottiFerraraGirardello,Behrndt,VandorenVroome} improves the situation. The $c$-map can be thought of as a type of $T$-duality in $D=3$, where it exchanges the moduli space obtained from the reduction of the Einstein-Maxwell sector with that obtained by the reduction of a universal hypermultiplet sector that can be added to the ${\cal N}=2$ theory in $D=4$ and that is present in any Calabi-Yau reduction of type IIA superstring theory \cite{CecottiFerraraGirardello}.\footnote{We ignore, i.e. set to zero, the effects of the other hyper- and vectormultiplets that arise in the reduction.} The point here is that the universal hypermultiplet in $D=4$ is also described classically by a coset space $\mathrm{SU}(2,1)/\mathrm{SU}(2)\times\mathrm{U}(1)$. The quantum corrections to this universal hypermultiplet moduli space are not fully understood, but recently~\cite{UHPaper} it has been proposed that a promising candidate for the discrete group $\mathrm{G}(\mathbb{Z})$ in this case is the so-called the Picard modular group $\mathrm{SU}(2,1;\mathbb{Z}[i])$, whose generators can be given an intuitive physical interpretation in terms of Peccei-Quinn symmetries, electric-magnetic duality and S-duality. Assuming this to be the correct quantum duality group of the universal hypermultiplet and the validity of the $c$-map at the quantum level would imply that the correct moduli space and quantum symmetry group of Einstein-Maxwell theory with one Killing vector is also encoded in the Picard group. A further verification of these claims is outside the scope of this work. See \cite{UHPaper} for more detailed discussions of these issues.

    \chapter{On infinite-dimensional symmetries of pure $\mathcal{N}=2$ supergravity }\label{chap:infinitedimsymsu21}

So far we have analyzed the role of the duality group SU(2,1) for understanding BPS solutions in $\mathcal{N}=2$ supergravity in four dimensions. This was done by performing a dimensional reduction to three dimensions, where the Lagrangian corresponds to Einstein gravity coupled to scalars parametrizing a Riemannian coset space $\mathcal{C}$ in the case of space-like reduction, and a pseudo-Riemannian coset space $\mathcal C^*$ in the case of time-like reduction.\\
 
Motivated by this, it is interesting to assume that the Einstein-Maxwell theory exhibits a hidden nonlinearly realized Kac-Moody symmetry group SU(2,1)$^{+++}$, formally arising in the reduction to zero dimensions~\cite{Julia:1982gx}, but as a conjectured symmetry of the full model~\cite{West:2001as}. The associated Kac-Moody algebra $\mathfrak{su}(2,1)^{+++}$ can be obtained by adding three nodes $\alpha_1, \alpha_2$ and $\alpha_3$ to the Tits-Satake diagram of $\asu$ displayed in Figure \ref{fig:su21} (see also Section \ref{subsec:su21andextension}). The Tits-Satake diagram of  $\mathfrak{su}(2,1)^{+++}$ is given in Figure \ref{fig2:su21+++}.\\

 \setcounter{equation}{0}
\section{On $\asuppp$} \label{sec:su21+++}

\begin{figure}[h!]
\begin{center}
\begin{pgfpicture}{0cm}{-1.5cm}{9cm}{2.5cm}
\pgfnodecircle{Node1}[stroke]{\pgfxy(2,0.5)}{0.25cm}
\pgfnodecircle{Node2}[stroke]
{\pgfrelative{\pgfxy(1.5,0)}{\pgfnodecenter{Node1}}}{0.25cm}
\pgfnodecircle{Node3}[stroke]
{\pgfrelative{\pgfxy(1.5,0)}{\pgfnodecenter{Node2}}}{0.25cm}
\pgfnodecircle{Node4}[stroke]
{\pgfrelative{\pgfxy(1.5,1)}{\pgfnodecenter{Node3}}}{0.25cm}
\pgfnodecircle{Node5}[stroke]
{\pgfrelative{\pgfxy(1.5,-1)}{\pgfnodecenter{Node3}}}{0.25cm}

\pgfnodebox{Node6}[virtual]{\pgfxy(2,0)}{$\alpha_{1}$}{2pt}{2pt}
\pgfnodebox{Node7}[virtual]{\pgfxy(3.5,0)}{$\alpha_{2}$}{2pt}{2pt}
\pgfnodebox{Node8}[virtual]{\pgfxy(5,0)}{$\alpha_{3}$}{2pt}{2pt}
\pgfnodebox{Node9}[virtual]{\pgfxy(6.5,-1)}{$\alpha_{4}$}{2pt}{2pt}
\pgfnodebox{Node10}[virtual]{\pgfxy(6.5,2)}{$\alpha_{5}$}{2pt}{2pt}

\pgfnodeconnline{Node1}{Node2} \pgfnodeconnline{Node2}{Node3}
\pgfnodeconnline{Node3}{Node4} \pgfnodeconnline{Node3}{Node5} 
\pgfnodeconnline{Node4}{Node5}
\pgfsetstartarrow{\pgfarrowtriangle{4pt}}
\pgfsetendarrow{\pgfarrowtriangle{4pt}}
\pgfnodesetsepend{5pt}
\pgfnodesetsepstart{5pt}
\pgfnodeconncurve{Node4}{Node5}{-10}{10}{1cm}{1cm}
\end{pgfpicture}
\caption {  \small {\sl Tits-Satake diagram of $\asuppp$.} }
\label{fig2:su21+++}
\end{center}
\end{figure}
In this section, we will use  the definitions and  properties of $\asuppp$ developed in Section \ref{subsec:su21andextension}  and we will  explain the construction of a non-linear $\s$-model on the infinite- dimensional coset space $\suppp/\mathrm K^{*\,+++}$, generalizing the finite-dimensional $\sigma$-model on G/K$^{*}$ considered in Section \ref{SU(2,1)SigmaModel}. Here $\mathrm K^{*\,+++}$ is the subgroup of $\suppp$ consisting of those generators which are pointwise fixed by the temporal involution $\Omega_1$, defined such that we may identify the index $1$ by a time coordinate (see Section \ref{PERTABsec:tempinv}). To this end we will use the level decomposition of the adjoint representation of $\asuppp$ into representation of an $\mathfrak{sl}(4, \mbb R)$ subalgebra defined by the nodes $\alpha_1$, $\alpha_2$ and $\alpha_3$ in Figure \ref{fig2:su21+++}. This level decomposition was performed in details in Section \ref{sec:levdecompo} and it is represented in Table \ref{tab2:levdecsu} up to level $4$.  We will see in Section \ref{sec:CosmologicalModel} that this representation content up to level 2 where the generator $r^{a_1 a_2}$ is projected out, can be associated with the bosonic field content of pure $\mathcal N=2$ supergravity in $D=4$. 

\begin{table}[h]
\begin{center}
\begin{tabular}{|c|c|}
\hline
$L= \ell_1 +\ell_2$ &  Generator of $\asuppp$ \\
\hline \hline
$0$  & $K^a_{\ b}$\\
$0$  & $i \, T$\\
$1$  & $r^{\, a} = R^a + \tilde{R}^a$\\
$1$  & $\tilde{r}^{\, a} = i (R^a - \tilde{R}^a) $\\
$2$  & $s^{\,s_{1} s_{2}} = - 2\, i\, S^{\, s_{1} s_{2}}$\\
$2$  & $r^{\,a_{1} a_{2}}=  2\, R^{\, a_{1} a_{2}}$\\
$3$  & $r^{\, a_{0}|a_{1} a_{2}}$\\
$3$   & $\tilde{r}^{\, a_{0}|a_{1} a_{2}}$\\
$4$  & $r^{\, a_{0}|a_{1} a_{2} a_{3}}$\\
$4$  & $r^{\, a_{1} a_{2} | a_{3} a_{4}}$\\
$4$  & $r^{\,s_{1} s_{2} |a_{3} a_{4} }$\\
$\vdots$ & $\vdots$\\
\hline
\end{tabular}
\caption{\sl \small Level decomposition of $\asuppp$ under $\mathfrak{sl}(4,\mathbb{R})$ up to level $4$. The indices  $a_{i}$ are antisymmetric  while the indices  $s_{i}$ are symmetric. Note that the generators from the level $L=3$ with mixed Young symmetries are subject to constraints.  }
\label{tab2:levdecsu}
\end{center}
\end{table}

\subsubsection{ Cartan and  temporal involutions} \label{sec:carttemporalin}

For the $\sigma$-models to be constructed in this  section we need to fix a (local) subgroup of $\mathrm{SU}(2,1)^{+++}$. We require two different choices, denoted $\mathrm{K}^{+++}$ and $\mathrm{K}^{*+++}$, leading to different coset spaces and that are defined by appropriate involutions at the level of the $\mathfrak{su}(2,1)^{+++}$ Lie algebra. The level decomposition discussed above does not depend on the choice of this subalgebra but the $\sigma$-model to be studied below does.

The first choice of subalgebra, $\mathfrak{k}^{+++}$, is defined by the Cartan involution $\theta$. Its action on $\asuppp$ may be read off from the Tits-Satake diagram of $\asuppp$ (see (\ref{eqn:thetasu21+++})). It has the following action on the generators of $\asuppp$, 
\be \label{eqn:thetainv} \begin{split} \begin{aligned}
\theta (r^a)&= - r_a, &\quad \theta(r_a)&= -r^a,\\
\theta (\tilde{r}^a)&= \tilde{r}_a, &\quad \theta(\tilde{r}_a)&= \tilde{r}^a,\\
\theta (s^{ab})&= s_{ab}, &\quad \theta(s_{ab})&= s^{ab},\\
\theta (r^{ab})&=- r_{ab}, &\quad \theta(r_{ab})&=- r^{ab},
\end{aligned}   \end{split} \ee
while on level $L=0$ it has the familiar action
\be \label{eqn:thinkab}
\theta (K^a_{\  b}) = -\, K^b_{\  a}\quad \theta(iT) = iT\,.
\ee
The Cartan decomposition therefore reads
\be
\asuppp = \mathfrak{k}^{+++} \oplus \mathfrak{p}^{+++},
\ee
where the subalgebra $\mathfrak{k}^{+++}$ is defined as the fixed point set under the Cartan involution, while $\mathfrak{p}^{+++}$ contains the generators which anti-invariant under $\theta$.
The generators of $\mathfrak{k}^{+++}$ reads
\be \label{eqn:k3p} \begin{split} 
\mathfrak{k}^{+++}&= \{ x \in \asuppp\,  :\,  \theta(x)= x \}
\\
&= \{ iT,\,j^{ab},\, (r^a- r_a),\,  (\tilde{r}^a+ \tilde{r}_a),\, (s^{ab}+ s_{ab}),\, (r^{ab}- r_{ab}), \ldots  \}\, ,
\end{split}
\ee
where $j^{ab}= K^a_{\ b }- K^{b}_{\ a }$, and those of $\mathfrak{p}^{+++}$ are
\be \label{eqn:p3p} \begin{split} 
\mathfrak{p}^{+++}&= \{ x \in \asuppp \, :\,  \theta(x)= - x \}
\\
&= \{ k^{ab},\,  (r^a+ r_a),\,  (\tilde{r}^a-\tilde{r}_a),\,  (s^{ab}- s_{ab}),\, (r^{ab}+ r_{ab}), \ldots  \}\, ,
\end{split}
\ee
where $k^{ab}= K^a_{\ b }+ K^{b}_{\ a }$.\\

The second choice of subalgebra, $\mathfrak{k}^{*+++}$, is introduced via the so-called temporal involution \cite{Englert:2003py}.
The possible existence of a Kac-Moody symmetry $\mathrm G^{+++}$ motivated the construction of a Lagrangian formulation explicitly invariant under $\mathrm G^{+++}$ . This Lagrangian $\mathcal{S}_{\suppp}$ is defined in a reparametrisation invariant way on a world-line parameter $\xi$, apriori unrelated to space-time, in terms of fields living in a coset $\suppp/ \mathrm K^{*\,+++}$. As the metric $g_{\mu \nu}$ at a fixed space-time parametrises the coset $\mathrm{GL}(D)/ \mathrm{SO}(1,D-1)$, the subgroup $ \mathrm K^{*\,+++}$ must contain the Lorentz group. As $\mathrm{SO}(1, D-1)$ is non-compact, we cannot use the Cartan involution $\theta$ to construct $\mathrm K^{*\,+++}$ that is now non-compact. Rather we will use the temporal involution $\Omega_i$ from which the required non-compact generators of $\mathrm K^{*\,+++}$ can be selected. The temporal involution $\Omega_i$ generalises the Cartan involution $\theta$ described in (\ref{eqn:thetainv}) and \eqref{eqn:thinkab}  to allow the identification of the index $i$ as a time coordinate. It is defined by
\be \label{eqn:temporal} \begin{split}
\Omega_i (iT) &= iT,\\
\Omega_i(K^a_{\ b })&= - \epsilon_a  \epsilon_b\,  K^b_{\ a}, \\
\Omega_i (r^a)&= - \epsilon_a\,  r_a, \\
\Omega_i (\tilde{r}^a)&= \epsilon_a\,  \tilde{r}_a,\\
\Omega_i  (s^{ab})&= \epsilon_a \epsilon_b\,  s_{ab}, \\
\ \Omega_i  (r^{ab})&=-\epsilon_a \epsilon_b \,  r_{ab}, 
\end{split}
\ee
with $\epsilon_a= -1$ if $a=i$ and $\epsilon_a= 1$ otherwise. \\

  \setcounter{equation}{0}
\section{On $\asupp \subset \asuppp$ and $\sigma$-models } \label{sec:su21++}
We now turn our attention to one-dimensional $\sigma$-models based on the group $\supp$. The $\mathfrak{g}^{++}$ content of the $\mathfrak{g}^{+++}$-invariant actions $\mathcal{S}_{\mathrm G_{+++}}$ has been analysed in reference \cite{Englert:2004ph} where it was shown that two distinct actions invariant under the overextended Kac-Moody algebra $\mathfrak{g}^{++}$ exist.  We will apply this analysis to $\mathfrak g=\asu$ and study the two actions invariant under $\supp$.\\

The first one $\mathcal{S}_{\supp_C}$ is called the cosmological $\sigma$-model and constructed from $\mathcal{S}_{\suppp}$ by performing a truncation putting consistently to zero some fields.  The corresponding $\asupp$  algebra is obtained from $\asuppp$ by deleting the node $\alpha_1$ from the Tits-Satake diagram of $\asuppp$ depicted in Figure \ref{fig2:su21+++}. The involution used to construct the action $\mathcal{S}_{\suppp}$  is the temporal involution $\Omega_1$ (defined in (\ref{eqn:temporal})) such that  coordinate $1$ is time-like. This implies that the truncated theory $\mathcal{S}_{\supp_C}$ carries a Euclidean signature in space-time. The $\mathcal{S}_{\supp_C}$ is the generalisation of the $\mathZ E_8^{++}=\mathZ E_{10}$ invariant action of reference \cite{Damour:2002cu} proposed in the context of M-theory and cosmological billiards. The parameter $\xi$ along the world-line will then be identified with the time coordinate and we will see in Section \ref{sec:CosmologicalModel} that this action restricted to a defined number of  levels is equal to the corresponding $\mathcal{N}=2$ supergravity in $D=4$ in which the fields depend only  on this time coordinate.

A second $\supp$-invariant action $\mathcal{S}_{\supp_B}$, called the brane model, is obtained from $\mathcal{S}_{\suppp}$ by performing the same consistent truncation {\it after} conjugation by the Weyl reflection $s_{\alpha_1}$ in $\asuppp$. Here, $s_{\alpha_1}$ is the Weyl reflection in the hyperplane perpendicular to  the simple root $\alpha_1$ corresponding to the node 1 of Figure \ref{fig2:su21+++}. The non-commutativity of the temporal involution $\Omega_1$ with the Weyl reflection \cite{Keurentjes:2004bv, deBuyl:2005it, Keurentjes:2004xx} implies that this second action is inequivalent to the first one (see Section \ref{sec:signa} where it is recalled the consequence of  $s_{\alpha_1}$ on the time identification).  In $\mathcal{S}_{\supp_B}$, $\xi$ is identified with a space-like direction.  For a generic $\mathrm G$, the $\mathrm G^{++}$-brane model describes intersecting extremal brane configurations smeared in all directions but one \cite{Englert:2003py, Englert:2004it}.

\subsection{Infinite-dimensional cosmological $\sigma$-model }
\label{sec:CosmologicalModel}

In this section we will analyze how well the suggestions in \cite{Damour:2002cu} apply to the pure $\mathcal{N}=2$ theory. More concretely, we will investigate what features of this theory can be described using a non-linear $\sigma$-models over an infinite-dimensional coset space, as a generalization of what we have seen in the case of the scalar Lagrangian (\ref{eqn:3dStationaryLagrangian}). In fact, we will find a correspondence between the supergravity fields and the parameters in a one-dimensional $\sigma$-model. For example, as we will see, the dynamics of some solutions to the supergravity equations of motion can be modelled by motion on a coset space $\supp/\mathrm K^{++}$, where $\mathrm K^{++}$ is the compact subgroup with Lie algebra $\mathfrak{k}^{++} \subset \mathfrak{su}(2,1)^{++}$. The results of this section will hence be a map between parts of the cosmological $\sigma$-model and parts of the supergravity. This confirms that the general conjecture describing supergravities with overextended Kac-Moody groups holds, to the same extent, also in the present case of pure $\mathcal{N}=2$ supergravity where the symmetry group is in a non-split form. In analogy with the discussion in Section \ref{sec:SpacelikeKillingVector} and \ref{sec:StationarySolutions}, the dynamics will be modelled by a non-linear $\sigma$-model of maps from $\mathcal{M}_1 \cong \mathbb{R}$ to $\supp/\mathrm K^{++}$. We will now formally construct this $\sigma$-model, and perform a check (as for example done in \cite{Damour:2004zy} in the case eleven-dimensional supergravity), to see if the corresponding equations of motion match with the dynamics given by the supergravity equations of motion (\ref{eqn:Einstein}) and (\ref{eqn:Maxwell}), when restricting to spatially constant solutions (in a sense to made more clear below). The action for the $\sigma$-model (given generally by (\ref{eqn:SigmaModelAction})) is
\begin{equation}
\mathcal{S}_{\supp_C} = \int_{\mathcal{M}_1} \frac{1}{n(t)} (\mathcal{P}(t) | \mathcal{P}(t)) \, \dd t\, ,
\end{equation}
where $n(t) = \sqrt{h}$ is the lapse function and necessary for reparametrization invariance on the world-line. The function $h(t)$ is the metric on the one-dimensional manifold $\mathcal{M}_1$ and $( \cdot | \cdot )$ is an invariant bilinear form of $\mathfrak{su}(2,1)^{++}$, formed for example by restriction from $\mathfrak{su}(2,1)^{+++}$. As described previously, $\mathcal{P}(t)$ is the component along the coset of the Maurer-Cartan form defined by maps into the coset. In the case of a one-dimensional base-manifold the $\sigma$-model equations of motion are (see e.g. (\ref{eqn:SigmaMotion}))
\begin{equation}
\label{eqn:cosmsigmamotion}
n\,  \partial_t (n^{-1} \, \mathcal{P}) + [\mathcal{Q}, \mathcal{P}] = 0 ,
\end{equation}
where $\mathcal P$ and $\mathcal Q$ are defined in \eqref{eqn:pandq}.
Now, as we are dealing with an infinite-dimensional coset space we cannot directly realize this model. What we will do is to use the level decomposition as described previously in Section \ref{sec:levdecompo}, and perform a truncation of the Kac-Moody algebra by throwing away all the generators above a certain level. This truncation can be shown to be a consistent truncation of the $\sigma$-model equations of motion \cite{Damour:2004zy}. Before performing this truncation however, we have to describe the level decomposition of $\asupp$ in terms of the decomposition of $\asuppp$, given in Section \ref{sec:levdecompo}. By defining $\mathfrak{su}(2,1)^{++}$ as a regular subalgebra of $\asuppp$, the level decomposition given in Table \ref{tab2:levdecsu} descends to $\asupp$ by restricting the indices to not run over $1$, or equivalently by generating the representations at every level by acting on the highest weight with the regular $\mathfrak{sl}(3, \mathbb{R})$ subalgebra of $\mathfrak{sl}(4, \mathbb{R})$. In this section, the $\mathfrak{sl}(4, \mathbb{R})$ indices $a,b...$ will therefore only take the values $2,\,3$ and $4$. 

We can hence realize a suitable truncation of $\asupp$ by setting all $\mathfrak{g}_L = 0$ for $|L| > 2$, and furthermore set the antisymmetric generator $r^{ab}$ at level $L=2$ to zero, as this generator has no clear interpretation in terms of supergravity quantities. We can therefore write general $\mathcal{P}$ and $\mathcal{Q}$ as
\begin{equation}
\mathcal{P} = \frac{1}{2} p_{ab} k^{ab} + \frac{1}{2} P_a (r^a + r_a) + \frac{1}{2} \tilde{P}_a (\tilde{r}^a - \tilde{r}_a) + \frac{1}{2} P_{ab}(s^{ab} - s_{ab})\,,
\end{equation}
and
\begin{equation}
\mathcal{Q} = \frac{1}{2} q_{ab} j^{ab} + \frac{1}{2} P_a (r^a - r_a) + \frac{1}{2} \tilde{P}_a (\tilde{r}^a + \tilde{r}_a) + \frac{1}{2} P_{ab}(s^{ab} + s_{ab}) ,
\end{equation}
where we have expanded in the basis given in (\ref{eqn:k3p}) and (\ref{eqn:p3p}), in parameters $p_{ab}, P_a$ and so on, depending only on the time coordinate $t$. We have chosen to put different parameters in front of the generators at level zero in the expressions for $\mathcal{P}$ and $\mathcal{Q}$ , considering that $k^{ab}$ and $j^{ab}$ are symmetric and anti-symmetric respectively. Using the commutation relations in Section \ref{subsec:su21andextension}, the equations of motion are now (given by inserting the expressions for $\mathcal{P}$ and $\mathcal{Q}$ in (\ref{eqn:levelzeromotion}) and (\ref{eqn:higherlevelmotion}))
\begin{eqnarray}
\label{eqn:pab}
n\, \partial_t (n^{-1} p_{ab} ) - q_{ca} {p^c}_b -  q_{cb}{p^c}_a + P_a P_b -\frac{1}{2} \delta_{ab} P_c P^c + \tilde{P}_a \tilde{P}_b \nonumber \\ 
- \frac{1}{2} \delta_{ab} \tilde{P}_c \tilde{P}^c - 2 \delta_{ab} P_{cd}P^{cd} + 4 P_{ac} {P_b}^c = 0\, ,
\end{eqnarray}
for the field $p_{ab}$, 
\begin{equation}
\label{eqn:Pa}
n\,  \partial_t (n^{-1} P_a) - p_{ac} P^c + q_{ac} P^c + 2 P_{ac} \tilde{P}^c = 0\, ,
\end{equation}
for the field $P_a$, 
\begin{equation}
\label{eqn:tildePa}
n \,\partial_t (n^{-1} \tilde{P}_a) - p_{ac} \tilde{P}^c + q_{ac} \tilde{P}^c - 2 P_{ac} P^c = 0\, ,
\end{equation}
for $\tilde{P}_a$ and finally
\begin{equation}
\label{eqn:Pab}
n\, \partial_t (n^{-1} P_{ab}) - 2 p_{ac} {P_b}^c +2 q_{ac} {P_b}^c = 0\, ,
\end{equation}
for $P_{ab}$. Regarding notation, we will in the following assume that indices are symmetrized or anti-symmetrized according to the tensor appearing linearly in expressions as these ones. For example, in (\ref{eqn:Pab}) the term $2 p_{ac} {P_b}^c$ is then short for $\frac{1}{2}(2 p_{ac} {P_b}^c + 2 p_{bc} {P_a}^c)$, the parameter $P_{ab}$ being a symmetric $\mathfrak{sl}(3, \mathbb{R})$ tensor.

\subsubsection{Dictionary} 

Let us now begin to compare the above $\sigma$-model dynamics with the dynamics given by our supergravity theory (\ref{eqn:SugraAction4d}). We will do this by choosing the parameters in $\mathcal{P}$ such that the $\sigma$-model equations of motion (\ref{eqn:pab})-(\ref{eqn:Pab}) match with the equations of motion on the supergravity side. Due to the construction of the $\sigma$-model, the natural framework for doing this is in the ADM-formalism where we will split the Einstein-Maxwell equations into dynamical equations and constraints/initial conditions. Concretely, we will only consider the dynamical supergravity equations. For the comparison it will be convenient to treat the spin connection $\omega_{ABC}$ and the field strength $F_{AB}$ as the fundamental fields of the Einstein-Maxwell theory and also considered as constant by letting them be space-independent. This is suitable because no spatial derivatives exist on the $\sigma$-model side. Locally we will also split the flat space coordinates $x^{a = 2,3,4}$ from the flat time coordinate $x^1$. The spin connection and the field strength transform under the Lorentz group $\mathrm{SO}(3,1)$ and we can use use this freedom, and a coordinate transformation, to put $\omega_{ABC}$ in a pseudo-Gaussian gauge by setting the metric shift functions to zero. This leads to $\omega_{ab1}$ being symmetric, and we also assume $\omega_{11c} = 0$, corresponding to space-independent gravity lapse $N$ . This gauge corresponds to a vielbein of the form
\begin{equation}
\label{eqn:vielbein}
{e_{\alpha}}^A = \left(\begin{array}{cc}N & 0 \\0 & {e_{\mu}}^a\end{array}\right)\,.
\end{equation}
In fact, the spin connection can be defined in terms of a tensor $\Omega_{ABC}$, called the anholonomy, by the relation
\begin{equation}
\label{eqn:anholonomy}
\omega_{ABC} = \frac{1}{2}(\Omega_{ABC} - \Omega_{BCA} +\Omega_{CAB}),
\end{equation}
and such that the anholonomy is given in terms of the vielbein by
\begin{equation}
{\Omega_{AB}}^C = 2\, {e_A}^{\alpha}{e_{B}}^{\beta}\partial_{[\alpha}{e_{\beta]}}^C.
\end{equation}
This pseudo-Gaussian gauge breaks the Lorentz group down to $\mathrm{SO}(3)$, acting on the spatial vielbein $ {e_{\mu}}^a$. Note also that we can rewrite the covariant derivative with respect to the spin connection using the vielbein, i.e.
\begin{eqnarray}
\label{eqn:SpinCovDerivative}
e^{-1} \partial_t (e \omega_{ab1})&  = & e^{-1}\,   \partial_t e\,  \omega_{ab1} +  \partial_t  \omega_{ab1}\nonumber \\
& = &  {e_{m}}^c \, \partial_t {e^{m}}_c\,  \omega_{ab1} +   \partial_t  \omega_{ab1}  \\
& = &  N {\omega^c}_{c1} \omega_{ab1} +  N \partial_1  \omega_{ab1} \nonumber .
\end{eqnarray}
Here we have used that $\partial_1 = N^{-1} \partial_t$. We will now consider the different parts of the supergravity equations of motion (\ref{eqn:Einstein}), (\ref{eqn:Maxwell}) and the Bianchi identity (\ref{eqn:MaxwellBianchi}), one at the time.

In addition to the gauges in the gravity sector, we also have to adopt a temporal gauge for the Maxwell field, corresponding with our choice of time coordinate to 
\be
A_1 = 0\,.
\ee

\begin{itemize}
\item{{\bf Ricci-tensor}}

First, let us consider the Ricci-tensor. Our Riemann-tensor written with flat indices is given in terms of the spin connection and the anholonomy by
\begin{equation} \begin{split}
\label{eqn:Riemann}
R_{ABCD} =\  & \partial_{A} \omega_{BCD} - \partial_{B} \omega_{ACD} + {\Omega_{AB}}^E \omega_{ECD} \\ &+ {\omega_{AC}}^E \omega_{BED} - {\omega_{BC}}^E \omega_{AED}\, .
\end{split} \end{equation}
From this expression we can derive the purely spatial Ricci-tensor appearing in (\ref{eqn:Einstein}) with our gauge choice,
\begin{equation}
\label{eqn:spatialRicci} 
R_{ab} = \partial_1 \omega_{ab1} +  \omega_{ab1} {\omega^c}_{c1} +  {\omega_{ab}}^d {\omega^c}_{dc} - \omega_{1ca} {\omega^c}_{b1} +  \omega_{ca1} {\omega_{1b}}^c +  {\omega^c}_{da} {\omega^d}_{bc}. 
\end{equation}
Now, to match with the $\sigma$-model equation (\ref{eqn:pab}), we make the ansatz $\omega_{ab1} = c_1\, p_{ab}$ and $\omega_{1ab} = c_2 \,q_{ab}$. Using this ansatz, one rewrites (\ref{eqn:spatialRicci}) as
\begin{equation}
\label{eqn:RewrittenSpatialRicci} \begin{split}
R_{ab} = \ &(Ne)^{-1} \partial_t ( e\, c_1 \,p_{ab}) - c_1\, c_2\, q_{ca}\,{p^c}_b -  c_1\, c_2 \,q_{cb}\,{p^c}_a \\ &+  {\omega_{ab}}^d {\omega^c}_{dc} + {\omega^c}_{da} {\omega^d}_{bc}\, . \end{split}
\end{equation}
Comparing with (\ref{eqn:pab}) we conclude that $c_1 = c_2 = N^{-1}$, and $n = e^{-1} N$ (multiply  (\ref{eqn:Einstein}) with $N^2$ to make the identification easier). Consider now the last term in  (\ref{eqn:RewrittenSpatialRicci}). Somehow we need to match $\omega_{abc}$ with the parameter $P_{ab}$. We do this by the ansatz
\begin{equation}
\label{eqn:omegaansatz}
\Omega_{abc} = c_3 \epsilon_{abd}{P^d}_c .
\end{equation}
There is here a mismatch in the number of degrees of freedom between these two tensors. From the symmetry of $P_{ab}$ we see that the anholonomy must obey a trace condition,
\begin{equation}
{\Omega_{ab}}^b = c_3 \epsilon_{abc} P^{cb} = 0 .
\end{equation}
This removes three of the nine degrees of freedom in the purely spatial anholonomy and with this condition its degrees of freedom equals the number of components of $P_{ab}$. It is generally assumed that this trace condition always can be imposed \cite{Damour:2002cu}. Observe that the tracelessness $\Omega_{abb} = 0$ is equivalent to $\omega_{bba} = 0$. Hence the second to last term in (\ref{eqn:RewrittenSpatialRicci}) vanishes. From the expression (\ref{eqn:anholonomy}) of the spin connection in terms of the anholonomy, the last term in (\ref{eqn:RewrittenSpatialRicci}) can be rewritten as
\begin{equation}
\label{eqn:SpinAndHolonomy}
{\omega^c}_{da} {\omega^d}_{bc} = \frac{1}{4} \Omega_{cda} {\Omega^{cd}}_b - \frac{1}{2} {\Omega_{da}}^c {\Omega^{d}}_{bc}- \frac{1}{2} \Omega_{adc} {\Omega_b}^{cd}.
\end{equation}
Let us take a closer look at the last term $\Omega_{adc} {\Omega_b}^{cd}$. The first two indices $a$ and $d$ in the first anholonomy has no matching index with the first two indices $b$ and $c$ of the second anholonomy. With our ansatz (\ref{eqn:omegaansatz}) this kind of index structure is impossible to match with  anything in the $\sigma$-model at low levels, as there is no such term in (\ref{eqn:pab}). This is a general phenomena when matching infinite-dimensional coset space $\sigma$-models with supergravity theories. In particular it is true also for the well studied case of eleven-dimensional supergravity. For example in \cite{Damour:2004zy} it is suggested that this term comes from terms in the $\sigma$-model that we in the current truncation have thrown away but the confirmation of this claim is still an open problem. Inserting our ansatz (\ref{eqn:omegaansatz}) in (\ref{eqn:SpinAndHolonomy}) we get (leaving the last term as it is)
\begin{equation}
{\omega^c}_{da} {\omega^d}_{bc} = - \frac{c^2}{2} \delta_{ab} P_{cd}P^{cd} + c^2 P_{ac} {P_b}^c - \frac{1}{2} \Omega_{adc} {\Omega_b}^{cd}. 
\end{equation}
Looking at (\ref{eqn:pab}) we get precise matching if ${c_3}^2 = 4N^{-2}$ and if we ignore the anomalous monomial in the anholonomy. The sign of $c_3$ remains unfixed, so we define $c_3 = 2N^{-1}c_3'$, where $|c_3'| = 1$. Let us now turn to the rest of the terms in the equation of motion (\ref{eqn:Einstein}) for the metric.

\item{$\mathbf{Maxwell\ field}$}

From the Einstein-Maxwell equation (\ref{eqn:Einstein}) we get when looking at the spatial part of the two monomials in the field strength (remembering that we multiplied with $N^2$),
\be \begin{split}
\frac{N^2}{2}\, \delta_{ab} F_{CD} F^{CD} - 2\, N^2 F_{aC}{F_b}^{C} = & \ \frac{N^2}{2} \delta_{ab}\, (-2 F_{1c}{F_1}^{c}  + F_{cd}F^{cd}) \\ 
&+ 2\, N^2 (F_{a1} F_{b1} -F_{ac} {F_b}^c). 
\end{split} \ee
A reasonable here ansatz is
\be
\label{eqn:fieldstrengthansatz} \begin{split}
F_{1c} &= c_4 P_c\, , \\
F_{ab} &= c_5 \epsilon_{abc} \tilde{P}^c\, ,
\end{split} \ee
giving
\be \begin{split} 
\frac{N^2}{2} \delta_{ab} F_{CD} F^{CD} - 2 N^2F_{aC}{F_b}^{C}  = &\quad   N^2 {c_4}^2(- \delta_{ab} P_c P^c + 2 P_a P_b)  \\
  & + N^2{c_5}^2(- \delta_{ab} \tilde{P}_c \tilde{P}^c + 2\tilde{P}_a \tilde{P}_b) .
\end{split} \ee
Comparing with (\ref{eqn:pab}) we see that we must put ${c_4}^2 = {c_5}^2 = \frac{1}{2N^2}$. As in the case for $c_3$, the signs of $c_4$ and $c_5$ is still unfixed, so we define $c_4 = \frac{1}{\sqrt{2}N} c_4'$ and $c_5 = \frac{1}{\sqrt{2}N} c_5'$ where again $|c_4'| = |c_5'| = 1$. Note that whether $F_{1c}$ should be proportional to $P_a$ or $\tilde{P}_a$ is up till now not fixed as they have appeared symmetrically so far. In fact, when we now turn to consider the equation of motion for $F_{AB}$ \eqref{eqn:Maxwell} and its Bianchi identity \eqref{eqn:MaxwellBianchi} it turns out that neither of these equations will fix this arbitrariness or the signs of the functions $c_i$. This is due to the symmetry between the roots $\alpha_4$ and $\alpha_5$ and can be interpreted physically as electromagnetic duality.

\item{{\bf Equations of motion and Bianchi identities for $F_{AB}$ }}

Finally we consider the equation of motion for the field strength $F_{AB}$ (\ref{eqn:Maxwell}). Explicitly the covariant derivative becomes
\begin{equation}
D^A F_{AB} = \partial^A F_{AB} - {\omega_{C}}^{AC}F_{AB} -  \omega_{ACB}F^{AC} = 0.
\end{equation}
Looking at the spatial dynamics, putting $B = b$ and splitting the sums over space and time we get
\begin{equation}
D^A F_{Ab} = -\partial_1 F_{1b} -  {\omega^e}_{e1} F_{1b} + \omega_{1cb}{F_1}^c -  \omega_{ab1}{F^a}_1 -  \omega_{acb}F^{ac} = 0 .
\end{equation}
Again we recognize the ``time'' covariant derivative from (\ref{eqn:SpinCovDerivative}). Hence we have
\begin{equation}
e^{-1} N^{-1} \partial_t (e  F_{1b})  = {\omega^e}_{e1} F_{1b} + \partial_1 F_{1b} . 
\end{equation}
Note also that $\omega_{acb}F^{ac} = \frac{1}{2} \Omega_{acb}F^{ac}$. With the expressions for the anholonomy (\ref{eqn:omegaansatz}), we derived above, we rewrite the equation for the field strength as
\begin{equation}
e^{-1} N^{-1} \partial_t (e N^{-1} {c_4}' P_b) + {c_4}' N^{-2} (q_{bc} P^c - p_{cb} P^c )+ 2{c_3}'{c_5}' N^{-2} P_{cb} \tilde P^c = 0. 
\end{equation}
This agrees with the $\sigma$-model equation (\ref{eqn:Pa}) if ${c_3}' {c_5}' = {c_4}'$.  Consider now the Bianchi-identity (\ref{eqn:MaxwellBianchi}). Letting $A=a$ we get
\be \begin{split}
\epsilon^{aBCD}D_B F_{CD} &=  \epsilon^{a1bc}D_1 F_{bc} + 2 \epsilon^{abc1} D_b F_{c1}\, , \\
&=  \epsilon^{abc} \partial_1 F_{bc} - 2 \epsilon^{abc}  \omega_{1db} {F^d}_c + 2\epsilon^{abc} \omega_{bd1} {F^d}_c - 2 \epsilon^{abc} \omega _{bdc} {F^d}_0    \\ 
& = 0.
\end{split}\ee
As $\epsilon^{abc} \omega_{bdc}  = \frac{1}{2} \epsilon^{abc} \Omega_{cbd}$,
\begin{equation} \nonumber
 \partial_1 ({c_5}' N^{-1} \tilde{P}_a) + {c_5}' N^{-2} ( q_{ab}\tilde{P}_b - p_{ab}\tilde{P}_b + N {w^b}_{b1}\tilde{P}_a )- 2 {c_3}' {c_4}' N^{-2} P_{ac} P_c =0 . 
\end{equation}
Again using (\ref{eqn:SpinCovDerivative}) and multiplying everything with $N^2$ we find
\begin{equation}
n\,  \partial_t(n^{-1} \tilde{P}_a) - p_{ab}\tilde{P}^b + q_{ab}\tilde{P}^b - 2\frac{{c_3}'{c_4}'}{{c_5}'} P_{ac} P^c = 0 .
\end{equation}
This is precisely the corresponding equation (\ref{eqn:tildePa}). 

\item{{\bf Riemann Bianchi}}

What remains to analyze is the last equation of the $\sigma$-model (\ref{eqn:Pab}), which is to be matched with the algebraic Bianchi identity (\ref{eqn:RiemannBianchi1}) for the Riemann tensor on the supergravity side. The component of (\ref{eqn:RiemannBianchi1}) to be considered is the symmetric purely spatial part. These turn out to be exactly equivalent in the current truncation, automatically by the above mapping of fields. This is a consistency check of our analysis.

\item{{\bf Summary}}

Hence, our analysis has given us an almost complete correspondence between the parameters of the truncated $\supp_C$-model and the dynamics of certain spatially constant solutions of pure $\mathcal{N} = 2$ supergravity. This result summarized in the Table \ref{tab:Dictionary} is what was expected from the structure of the low-lying $\mathfrak{sl}(3,\mathbb{R})$ representations. The map is similar to those already constructed for other supergravity theories, and succeeds and fails at the same points. We point out that in addition to the dynamical equations, there are in general constraint equations to be verified, for example the (spatial) diffeomorphism constraint and Gauss constraints. We expect that they are satisfied in the same way as for the maximally supersymmetric case~\cite{Damour:2007dt}.

\begin{table}[h]
\begin{center}
\begin{tabular}{|c|c|c|c|}
\hline
Level $L$ & Supergravity field & $\supp$ field & $ \asupp$ generator\\
\hline \hline
$0 $ &$ \omega_{abt} $ & $ p_{ab}$ & $k^{a b}$\\
$0 $ &$\omega_{tab}$ & $ q_{ab}$ & $j^{a b}$\\
$1$ &$F_{tc} $ & $ \frac{1}{\sqrt{2}} {c_4}' P_c $ & ${r^a}$\\
$1$ &$N F_{ab} $ & $\frac{1}{\sqrt{2}} {c_5}' \epsilon_{abc} \tilde{P}^c$ & $\tilde{r}^a$\\
$2$ &$N \Omega_{abc}$ & $2{c_4}' {c_5}' \epsilon_{abd}{P^d}_c $& $s^{a b}$\\
$-$ & $Ne^{-1}$& $n$ & -\\
\hline
\end{tabular}
\caption{\sl \small Correspondence between the bosonic fields in the supergravity theory and the Kac-Moody $\sigma$-model. The parameters ${c_4}'$ and ${c_5}'$ are unfixed and are $\pm 1$. All the supergravity quantities are assumed to be evaluated at a fixed spatial point.}
\label{tab:Dictionary}
\end{center}
\end{table}

\end{itemize}

 \subsection{The $\asuppp$ algebraic structure of BPS branes}

In this section, we show that the BPS solutions ({\ref{Isom}) which are upon dimensional reduction on time described in the $\asu$ $\sigma$-model by equation ({\ref{eqn:BPSscalars}), are in fact  completely algebraically described in $\asuppp$. In order to do so we now choose the time coordinate to be the direction $x_4$. More precisely, we show that the full space-time solution (\ref{Isom}) can be reconstructed (i.e. not only the part which correspond to scalars upon dimensional reduction) by demanding that the $\asu$ is regularly embedded\footnote{We recall that a subalgebra $\bar{\mathfrak{g}}\subset\mathfrak{g}$ is \emph{regularly embedded} in $\mathfrak{g}$ if the root vectors of $\bar{\mathfrak{g}}$ are root vectors of
$\mathfrak{g}$, and the simple roots of $\bar{\mathfrak{g}}$ are real roots of
$\mathfrak{g}$. Of particular relevance for our analysis is that, as a consequence, the Weyl group $\mathcal{W}(\bar{\mathfrak{g}})$ of $\bar{\mathfrak{g}}$ is a subgroup of $\mathcal{W}({\mathfrak{g}})$. For finite-dimensional Lie algebras the concept of a regular embedding was introduced by Dynkin in \cite{Dynkin:1957um}, and was subsequently extended to the infinite-dimensional case by Feingold and Nicolai \cite{Feingold}.} in $\asuppp$. The regular embedding is defined by erasing the nodes $\alpha_1, \alpha_2$ and $\alpha_3$  in Figure~\ref{fig2:su21+++}.

We first recall that we can describe the non-compact Cartan fields of $\asuppp$ in two bases, the $\mathfrak{gl}(4, \mathbb R)$ one described by the generators  $K^a_{\ a}$,   (see (\ref{eqn:kab})) and the Chevalley base given by the $\mathbf{h_m}, \ m=1 \dots 4$ (see (\ref{eqn:su21+++Chev})). The fields corresponding to the former are denoted $p_a$ and the ones corresponding to the latter denoted $q_a$ ($a=1 \dots4)$. The relation between these two bases is:
\begin{equation}
\sum_{a=1}^{4} \, p_a K^a_{\ a}=\sum_{a=1}^{4} \, q_a\,  \mathbf{h_a},
\label{basecartsu21}
\end{equation}
where the $p_a$'s encode the diagonal metric in $\asuppp$. We have indeed $p_a= \frac{1}{2} \ln g_{aa}$ where $g_{aa}$ is the four-dimensional metric. This follows for instance from \cite{Englert:2003zs} or also from the results of the preceding section, summarized in Table~\ref{tab:Dictionary}.

We are now in position to impose the regular embedding which amounts at the level of the Cartan to enforce
\begin{equation}
q_1=q_2=q_3 =0.
\label{embeddsu21}
\end{equation}
Using (\ref{basecartsu21}) the conditions (\ref{embeddsu21}) translate for the $p_a$'s into
\be
p_1=p_2=p_3=-p_4.
\ee
Consequently the regular embedding of $\asu$  in $\asuppp$ imply on the physical four-dimensional metric the following conditions:
\begin{equation}
\label{finalembsu21}
g_{1\, 1}=g_{2\, 2}=g_{3\, 3}=g_{4\, 4}^{-1},
\end{equation}
which is satisfied by the BPS metrics (\ref{Isom}). This completes the proof that the 
 four-dimensional BPS solutions are described  by the regular embedding of $\asu$  in $\asuppp$. It is worth noticing that this description is not valid for non-BPS solutions and it indicates again the special role played by BPS solutions in the $\mathfrak{g}^{+++}$ approach (see \cite{Englert:2003py}, \cite{Englert:2007qb}).

 \subsection{Weyl reflection in $\asuppp$ }
 In this section, we first discuss a definition of the Weyl group of $\asu$ and its action on BPS solutions of the $\mathcal{N}=2$ supergravity. Then, we will study  the Weyl group of $\asuppp$ and its possible consequence on the space-time signature.
 \subsubsection{Weyl reflection in $\asu$ }
 
First, we briefly recall how to construct the Weyl group of the complex  algebras $A_2$. The Weyl group $W$ of $A_2$ is generated by the two simple Weyl reflections $s_{\alpha_4}$, $s_{\alpha_5}$ associated respectively to   the simple roots $\alpha_4$ and $\alpha_5$ (see Figure \ref{fig2:su21+++}). The group contains six elements
\be
W_{A_2}=\{1,s_{\alpha_4},s_{\alpha_5},s_{\alpha_4}s_{\alpha_5},s_{\alpha_5}s_{\alpha_4},s_{\alpha_4}s_{\alpha_5}s_{\alpha_4}\}\, ,
\ee
and is isomorphic to the symmetric group $\mathrm{S}_3$ on three letters. We first note that among the six elements, three correspond to reflections: $s_{\alpha_4},s_{\alpha_5}$ and $s_{\alpha_4}s_{\alpha_5}s_{\alpha_4}$. The third transformation correspond to the Weyl reflection associated to the non-simple roots $\alpha_4+\alpha_5$ namely  $s_{\alpha_4}s_{\alpha_5}s_{\alpha_4}=s_{\alpha_4+\alpha_5}$. The action of $s_{\alpha_4+\alpha_5}$ on the simple roots of $A_2$ is:
\begin{equation} \label{eqn:weyl5} \begin{split}
s_{\alpha_4+\alpha_5}(\alpha_4)\, &= -\alpha_5\, , \\
s_{\alpha_4+\alpha_5}(\alpha_5)\, & = -\alpha_4\,  .
\end{split}
\end{equation}
The strategy used here to define the Weyl group of $\asu$ is to retain only the reflections of the Weyl group of $A_2$  associated to the roots which are invariant under the conjugation $\sigma$ fixing the real form $\asu$. Using (\ref{eqn:sigmaroots}) we deduce that the only invariant Weyl reflection is  $s_{\alpha_4+\alpha_5}$. Consequently,  we define the Weyl group of $\asu$ as being
$W_{ \asu } =\{ 1, s_{\alpha_4+\alpha_5} \}$. This is in agreement with the restricted root system describing $\asu$ given in Section \ref{app:restricted}. The restricted root system of $\asu$ is $(BC)_1$ \cite{Helgason:1978} and the Weyl group of $(BC)_1$ is generated by one restricted root $\lambda_2$  (see (\ref{eqn:restricted})) which precisely correspond to the root $\alpha_4+\alpha_5$ in $A_2$.\footnote{The fact that the restricted root system of $\mathfrak{su}(2,1)$ is of non-reduced type has interesting consequences for the behaviour of $D=4$ Einstein-Maxwell gravity in the vicinity of a space-like singularity (``BKL-limit''). For information on these aspects of Maxwell-Einstein gravity, we refer to \cite{HenneauxJulia,Henneaux:2007ej}.  }

We now determine the element $\mathcal{W}$ of  SU(2,1) corresponding to the Weyl transformation  $s_{\alpha_4+\alpha_5}$ and acting by conjugation on the coset element $\mathcal{V}$ namely: $\mathcal{V}^\prime=\mathcal{W} \, \mathcal{V} \, \mathcal{W}^{-1} $.  The conjugate action on  $\mathcal{V}$ implies a conjugate action on $\mathcal{P}=\tfrac{1}{2}\big( d\mathcal{V} \mathcal{V}^{-1}\,  - \,  \Omega_4(d\mathcal{V} \mathcal{V}^{-1}) \big)$, if  $\mathcal{W}$ pertains to the invariant subgroup under the temporal involution $\Omega_4$ namely $\mathrm{K}^{*} =\mathrm{SL}(2, \mathbb R) \times \mathrm{U(1)}$ (see (\ref{eqn:temporal}) and   Section \ref{app:k*}).  We will check below that it is indeed the case.  In order to find $\mathcal{W}$ we use  (\ref{eqn:weyl5}) which translate at the level of the $A_2$ algebra into 
\begin{equation} \label{eqn:weylal} \begin{split}
&\mathcal{W} \,  e_4 \, \mathcal{W}^{-1} = \epsilon\,  f_5,\\
&\mathcal{W} \, e_5 \, \mathcal{W}^{-1} = \epsilon\,  f_4,
\end{split}
\end{equation}
 while on the generators of $\asu$ we get (see \eqref{eqn:su21+++Chev})
\begin{equation} \label{eqn:weylalsu} \begin{split}
&\mathcal{W} \,  \mathbf{e_4} \, \mathcal{W}^{-1} = \epsilon\,  \mathbf{f_4}\, ,\\
&\mathcal{W} \, \mathbf{e_5} \, \mathcal{W}^{-1} =-  \epsilon\,  \mathbf{f_5}\, ,
\end{split}
\end{equation}
where $\epsilon$ is a plus or minus sign\footnote{\label{foot:sign}This arises since step operators are representations of the Weyl group up to signs.}.

 Demanding the equations (\ref{eqn:weylal}) to be satisfied and imposing  $\mathcal{W}^2=1$ determine  $\mathcal{W}$ univocally, we get:
\begin{equation}
\label{eqn:weylgf}
\mathcal{W}= \exp{[-\tfrac{\pi}{2} \, \mathbf{h_5}]} \, \exp{[\tfrac{\pi}{2} \, (\mathbf{e_{4,5}}+\mathbf{f_{4,5}})]},
\end{equation}
which fixes $\epsilon=-1$.
The generators $\mathbf{h_5}$ and $(\mathbf{e_{4,5}}+\mathbf{f_{4,5}})$ pertaining both to $\mathfrak{k^*}= \mathfrak{sl}(2,\mathbb{R})\oplus \mathfrak{u}(1)$  and  $\mathfrak{k}= \mathfrak{su}(2)\oplus \mathfrak{u}(1)$ (see Section \ref{app:k}), the element $\mathcal{W}$ belongs to both $\mathrm{K}^*$ and $\mathrm{K}$, ensuring the validity of the procedure to derive it.

We are know in the position to derive the effect of the Weyl transformation on the BPS solutions given by  (\ref{Isom}). Since the element $\mathcal{W} \in \mathrm{K}^*$, to see how the four charges transforms we can just conjugate by $\mathcal{W}$  the charge matrix  (\ref{eqn:charges}). We find that under $\mathcal{W}$  the charges transform as:
\begin{equation}
\label{eqn:transfochw}
(m,n,q,h) \ \stackrel{\mathcal W}{\longrightarrow}\ (-m,-n,q,h) .
\end{equation}
This Weyl transformation maps physical solutions with positive charges to unphysical solutions with negative charges.

  \subsubsection{Effect of Weyl reflections on space-time signature} \label{sec:signa}
  In this section we will focus on the Weyl group of $\asuppp$ and we will study the effect of Weyl reflections on the space-time signature $(1,3)$ of the $\mathcal N=2$ supergravity theory in $D=4$. First, recall that a Weyl transformation of a generator $T$ of a Lorentzian algebra $\mathfrak{g}^{+++}$ can be expressed as a conjugation by a group element $U_W$ of $\mathrm{G}^{+++}$: $T \longrightarrow  U_W\,  T \, U^{-1}_W $. Because of the non-commutativity of Weyl reflections with the temporal involution $\Omega_i$ (defined in (\ref{eqn:temporal}))
  \be \label{eqn:nocomu}
  U_W\, (\Omega_i T)\,  U_W^{-1}= \Omega '\,  (U_W T U^{-1}_W)\, , 
   \ee
 different Lorentz signatures $(t,s)$ (where $t(s)$ is the number of time (space) coordinates) can be obtained \cite{Keurentjes:2004bv, Englert:2004ph}. The analysis of signature changing has been done for all $\mathfrak{g}^{+++}$ that are very-extensions of a simple split Lie algebra $\mathfrak{g}$ \cite{deBuyl:2005it, Keurentjes:2005jw}. In these cases, Weyl reflections with respect to a root of gravity line\footnote{The gravity line is the set of the simple roots of the $\mathfrak {sl}(n, \mathbb R)$-part of $\mathfrak{g}^{+++}$.  It corresponds in the case of  $\asuppp$ to the roots $\alpha_1, \,\alpha_2$ and $\alpha_3$. } do not change the global Lorentz signature $(t,s)$ but it changes only the identification of the time coordinate. In fact, only Weyl reflections with respect to roots not belonging to the gravity line can change the global signature of the theory (see Section \ref{sec:tempinvolutintro}). We will now study the possible signature changing induced by Weyl reflections of the non-split real form $\asuppp$.\\
 
 The Weyl group of $\asuppp$ namely $W_{\asuppp}$ is generated by the Weyl reflection $s_{\alpha_4+\alpha_5}$ belonging to $W_{\asu}$ and by the simple Weyl reflections with respect to the roots of gravity line $s_{\alpha_1}, \, s_{\alpha_2},\, s_{\alpha_3}$. Because of the presence of the affine Weyl reflection $s_{\alpha_3}$,  the Weyl group $W_{\asuppp}$ becomes infinite-dimensional
 \be
 W_{\asuppp}=\{1, s_{\alpha_1}, \, s_{\alpha_2}, \, s_{\alpha_3},\,  s_{\alpha_4+ \alpha_5}, \ldots       \}\, .
 \ee
 
 \vspace{.3cm}
\noindent $ \bullet \ $  \textbf{The effect of the Weyl reflection } $\mathbf{s_{\alpha_1} }$

\vspace{.2cm}
 As is the case for split forms,  Weyl reflections with respect to the gravity line of $\asuppp$ will not change the global signature $(1,3)$ but it will only change the identification of time index. The roots of the gravity line are indeed not affected by arrows and  they are all non-compact roots as for split real form. Let us recall a simple example of the consequence of the Weyl reflection $s_{\alpha_1}$ on the space-time signature \cite{Englert:2004ph}. We start with the temporal involution $\Omega_1$ allowing the index $1$ to be the time index.  Applying (\ref{eqn:nocomu}) to
the Weyl reflexion $s_{\alpha_1}$  generates from 
$\Omega_i \equiv\Omega_1$ a new involution $\Omega' \equiv \Omega_2$ such that
\be \begin{split}\begin{aligned}
\label{permute}
U_1\, \Omega_1 K^2_{\ 1} \, U^{-1}_1&= \rho \, K^2_{\ 1} = \rho\, \Omega_2\, 
 K^1_{\ 2}, \\
U_1\, \Omega_1 K^1_{\ 3} \, U^{-1}_1&= \sigma \, K^3_{\ 2} =
\sigma \, \Omega_2\, 
 K^2_{\ 3} \, ,\\
U_1\, \Omega_1 K^i_{\ i +1} \, U^{-1}_1&= -\tau\,  K^{i+1}_{\ \, i} =
\tau\, \Omega_2 \,  K^i_{\ i +1}\quad i >2\, ,
\end{aligned} \end{split} \ee

\noindent where $\rho,\sigma,\tau$ are plus or minus signs (see footnote \ref{foot:sign}). The equations (\ref{permute}) illustrate the general result
that such signs always cancel in the determination of
$\Omega^\prime$ because they are identical in the Weyl transform of
corresponding positive and negative roots, as
their commutator is in the Cartan subalgebra which  forms a true
representation of the Weyl group. The content of (\ref{permute}) is
represented in Table \ref{tablin}. 
The signs below the generators of the gravity
line indicate the sign in front of the
 negative step operator obtained by the involutions $\Omega_1$ and $\Omega_2$ (see (\ref{eqn:temporal})): a
minus sign indicates that the indices in
$K^m_{\ m +1}$ are both either space or time indices while a plus sign
indicates that one index must be time and the other  space.
\begin{table}[t]
\begin{center}
\begin{tabular}{|c|ccc|c|}
\hline
&$K^1_{\ 2}$&$K^2_{\ 3}$&$K^3_{\ 4}$&time coordinate\\
\hline
$\Omega_1$&$+$&$-$&$-$&1\\
\hline$\,\Omega_2$&$+$&$+$&$-$&2\\
\hline
\end{tabular}
\caption{\sl \small Involution switches from $\Omega_1$  to
$\Omega_2$ in $\asuppp$ due to the Weyl reflection $s_{\alpha_1}$.}
\label{tablin}
\end{center}
\end{table}

The Table \ref{tablin} shows that
the  time coordinates in
$\asuppp$ must now be identified either with 2, or with all indices
$\neq 2$. We choose the first description, which leaves
unaffected coordinates attached to planes invariant under the Weyl
transformation. More generally, by Weyl reflections with respect to a root of the gravity line, it is possible to identify the time index to any $\mathfrak{sl}(4, \mathbb R)$ tensor index.

 \vspace{.3cm}
\noindent $ \bullet \ $  \textbf{The effect of the Weyl reflection } $\mathbf{s_{\alpha_4+\alpha_5} }$

\vspace{.2cm}

We will now study the effect of the particular Weyl reflection $s_{\alpha_4+\alpha_5}$ on the space-time signature $(1,3)$. We will first act with $s_{\alpha_4 +\alpha_5}$ on the generators of $A_2^{+++}$ to find then the transformation of the generators of $\asuppp$. Only the simple roots $\alpha_3, \alpha_4$ and $ \alpha_5$  are modified by this reflection. Its action  on the roots $\alpha_4$ and $\alpha_5$  is done in (\ref{eqn:weyl5}) while  on the root $\alpha_3$, it acts as
\be
s_{\alpha_4+\alpha_5} (\alpha_3)= \alpha_3 + 2 (\alpha_4+ \alpha_5) \,.
\ee
Note that the root $\alpha_3$ is transformed in a root of level $\ell=(2,2)$, the root $\alpha_4$ to a negative root of level $\ell=(0,-1)$ and the root $\alpha_5$ to a negative root of level $\ell=(-1,0)$ (see Table \ref{tab:levdeca2}).
The generators associated to roots $\alpha_3, \alpha_4$ and  $\alpha_5$ are modified respectively as 
\be \begin{split}
\mathcal{W}\, K^3_{\ 4} \mathcal{W}^{-1}&= \gamma\, \big[\overbrace{[K^3_{\ 4}, S^{44}   ]}^{2 S^{34}} ,S^{44} \big] \\
&= \gamma\, 2\,  R^{44|34}\, ,\\
\mathcal{W}\, R^{4} \mathcal{W}^{-1} &= \epsilon\, \tilde{R}_{4}\, ,\\
 \mathcal{W}\, \tilde{R}^{4} \mathcal{W}^{-1} &= \epsilon\, R_{4}\, .
\end{split} \ee
Using the Table \ref{tab2:levdecsu}, we find how the generators of $\asuppp$ transform under this Weyl reflection
\be \begin{split}
\mathcal{W}\, K^3_{\ 4} \mathcal{W}^{-1}&= \gamma\, [\overbrace{[K^3_{\ 4}, \tfrac{i}{2} s^{44}   ]}^{i s^{34}} , \tfrac{i}{2}s^{44}] \\
&= - \tfrac{1}{2} \gamma\, r^{44|34}\, ,\\
\mathcal{W}\, r^{4} \mathcal{W}^{-1} &= \epsilon\, r_{4}\, ,\\
 \mathcal{W}\, \tilde{r}^{4} \mathcal{W}^{-1} &= -  \epsilon\, \tilde{r}_{4}\, .
\end{split}\ee
If we apply (\ref{eqn:nocomu}) and (\ref{eqn:temporal}), we find the action of $\Omega'$ on these generators:
\be \label{eqn: detailsinv}\begin{split}
\gamma\,  \Omega' K^3_{\ 4} &= - \tfrac{1}{2}\, \Omega' \, (\mathcal{W} r^{44|34}   \mathcal{W}^{-1})= - \tfrac{1}{2}\mathcal{W} \underbrace{\Omega_i\,  r^{44|34}}_{- \epsilon_3 \epsilon_4 r_{44|34}}  \mathcal{W}^{-1} \\
&= \gamma (-) \epsilon_3 \epsilon_4 K^4_{\ 3}\, ,\\
\epsilon\,  \Omega' r^4 &= \Omega' ( \mathcal{W} r_4         \mathcal{W}^{-1})=   \mathcal{W} \underbrace{\Omega_i r_4 }_{- \epsilon_4 r^4}       \mathcal{W}^{-1}\, ,\\
&=  \epsilon  (-\epsilon_4) r_4\, ,\\
- \epsilon \Omega' \tilde{r}^4 &= \Omega' (\mathcal{W} \tilde{r}_4         \mathcal{W}^{-1})=   \mathcal{W} \underbrace{\Omega_i\,  \tilde{ r}_4 }_{ \epsilon_4 \tilde{r}^4}       \mathcal{W}^{-1}\, ,\\
&= - \epsilon  (\epsilon_4) \tilde{r}_4\, .
\end{split} \ee  
From (\ref{eqn: detailsinv}), one gets 
\be \label{eqn:omegaf} \begin{split} \begin{aligned}
\Omega' K^{3}_{\ 4}&=\  - \epsilon_{3}\,  \epsilon_{4} \, K ^{4}_{\ 3} & &= \Omega_i \,  K^{4}_{3}\, ,\\
\Omega' r^{4}&=\ - \epsilon_4\,  r_4& & = \Omega_i \, r^4\, ,\\
\Omega' \tilde{r}^{4}&= \ \epsilon_4\, \tilde{ r}_4 &&= \Omega_i\,  \tilde{r}^4\, .
\end{aligned} \end{split} \ee
We find in (\ref{eqn:omegaf}) that the involution $\Omega' $ acts exactly in the same way that the involution $\Omega_i$ defined by (\ref{eqn:temporal}). We can then conclude that the Weyl reflection $s_{\alpha_4+ \alpha_5}$ does not affect the signature $(1,3)$ of $\mathcal{N}=2$ supergravity theory in $D=4$.

 \setcounter{equation}{0}
\section{Embedding of $\asuppp$ in $E_{11}$}
\label{sec:Embedding}

In this final section, we find a regular embedding of $\asuppp$ in the split real form of $E_{11}$ \footnote{An embedding of $\mathfrak{su}(2,1)$ in $E_{8(8)}$ has been discussed in \cite{Gunaydin:2001bt}.} . This embedding will be derived using
elegant arguments from brane physics. We will relate between themselves different extremal brane configurations of eleven-dimensional supergravity and pure $\mathcal{N}=2$ supergravity in $D=4$. We first describe the brane setting we use.

\subsection{The brane setting}

We build an extremal brane configuration  allowed by the intersection rules \cite{Argurio:1997gt,Argurio:1998cp} leading upon dimensional reduction down to four to an extremal Reissner-Nordstr\"om  electrically charged black hole solution \cite{Maldacena:1996ky}  of $\mathcal{N}=2$ supergravity in $D=4$.

The configuration, that we denote by configuration {\bf\sf A}, built out of two extremal M5 branes and two extremal M2 branes  is the following (again we choose the direction 4 to be time-like):

\begin{table}[h]
\begin{center}
\begin{tabular}{|c|cccc|ccccccc|}
\hline
Branes &1& 2 & 3 & 4 & 5 & 6 & 7 & 8 & 9 & 10 & 11\\
\hline\hline
$A_1$=M5 &\, & \, & \, &$\bullet$ & \, & \, &$\bullet$ &$\bullet$  &$\bullet$ &$\bullet$& $\bullet$ \\
\hline
$A_2$=M5 &\, & \, & \, &$\bullet$ & $\bullet$ &$\bullet$ &\, &\,   &$\bullet$ &$\bullet$& $\bullet$ \\
\hline
$A_3$=M2 &\, & \, & \, &$\bullet$ & \, &$\bullet$ &\, &$\bullet$  &\, &\, & \, \\
\hline
$A_4$=M2 &\, & \, & \, &$\bullet$ &$\bullet$ & \, &$\bullet$ &\, &\, &\,& \, \\
\hline
\end{tabular}
\end{center}
\caption{\sl \small  Configuration \textbf{\sf{A}}: the extremal brane configuration leading to a four-dimensional extremal Reissner-Nordstr\"om  electrically charged black hole. The directions 1 to $4$ are non-compact (where 4 is time) and the directions 5 to 11 are compact.}
\label{tab:braneconf1}
\end{table}

This extremal configuration is generically characterised by four different harmonic functions in three dimensions, one for each brane. Here we choose the harmonic function to be the same for all the branes: $H= 1+\frac{q}{r}$ where $r$ is the radial coordinate in the four-dimensional non-compact space-time (we denote also $\phi \in [0, 2\pi] $ and $\theta \in [0, \pi]$ the usual angles, considering spherical coordinates). The metric of this intersecting branes configuration, depending only on the $q$ parameter  is:
\begin{equation}
\dd s^2_{11}=-H^{-2} \dd x_4^2 +H^2 (\dd x_1^2+\dd x_2^2+\dd x_3^2) + \sum_{i=5}^{11} \dd x_i^2.
\label{conf11}
\end{equation}
Upon dimensional reduction down to four dimensions the metric (\ref{conf11}) is the four-dimensional extremal  
Reissner-Nordstr\"om  electrically charged black hole solution of $\mathcal{N}=2$ supergravity in $D=4$
given by (\ref{Isom}) with $m=q$ and $n=h=0$ and with $t=x^4$.

The eleven-dimensional solution is characterised by four non-zero components $A^{(i)}, \  i=1 \dots 4$, of the three form potential, one for each brane. These are given by (see for instance \cite{Argurio:1998cp}):
\begin{equation}
\label{Acompo}
A^{(1)}= A_{\phi 5 6}, \qquad A^{(2)}= A_{\phi 78}, \qquad A^{(3)}=A_{468}, \qquad A^{(4)}=A_{457}.
\end{equation}
The corresponding non-vanishing components of the  field strengths are such that $\star F^{(1)}=\star F^{(2)}=F^{(3)}= F^{(4)}=\partial_r(H^{-1})$ where $\star$ denotes the Hodge dual in eleven dimensions.  As a consequence, if we want to interpret the configuration after dimensional reduction as an electric Reissner-Nordstr\"om black hole we have to identify the four-dimensional Maxwell field strength ${}^{(4)} F$ of (\ref{eqn:SugraAction4d}) as being the dimensional reduction of the diagonal eleven-dimensional field strength  ${}^{(11)}F^{diag} \equiv \star F^{(1)}+\star F^{(2)}+  F^{(3)}+ F^{(4)}$. This gives indeed back (\ref{Isom}) with $m=q$ and $n=h=0$.

Having the eleven-dimensional origin of the electrically charged extremal Reissner-Nordstr\"om black hole, we can now easily deduce the eleven-dimensional configuration corresponding to the magnetically charged extremal Reissner-Nordstr\"om by Hodge dualising {\it in four dimensions}  (i.e. the internal coordinates $x_i,  \ i=5 \dots  11$, playing now a passive role) and uplifting back to eleven dimensions. One immediately deduces that the non-zero components  ${\tilde A}^{(i)}$ of the dual configuration are:
\begin{equation}
\label{Acompo2}
{\tilde A}^{(1)}= A_{4 5 6}, \qquad {\tilde A}{(2)}= A_{478}, \qquad {\tilde A}^{(3)}=A_{\phi 68}, \qquad {\tilde A}^{(4)}=A_{\phi 57}.
\end{equation}
From (\ref{Acompo}), we deduce that the dual configuration, denoted with the letter {\bf\sf B}, is the one given in Table \ref{tab:braneconf2}.
\begin{table}[h]
\begin{center}
\begin{tabular}{|c|cccc|ccccccc|}
\hline
Branes &1& 2 & 3 & 4 & 5 & 6 & 7 & 8 & 9 & 10 & 11\\
\hline\hline
$B_1$=M2 &\, & \, & \, &$\bullet$ &$\bullet$ & $\bullet$ & \, &\, &\, &\,& \,\\
\hline
$B_2$=M2 &\, & \, & \, &$\bullet$ & \, &\, &$\bullet$ &$\bullet$  &\, &\,& \, \\
\hline
$B_3$=M5 &\, & \, & \, &$\bullet$ &$\bullet$ &\, &$\bullet$ &\, &$\bullet$ &$\bullet$ &$\bullet$ \\
\hline
$B_4$=M5 &\, & \, & \, &$\bullet$ &\, &$\bullet$ &\, &$\bullet$ &$\bullet$ &$\bullet$&$\bullet$ \\
\hline
\end{tabular}
\end{center}
\caption{\sl \small Configuration \textbf{\sf B}: the extremal brane configuration leading to a four-dimensional extremal Reissner-Nordstr\"om  magnetically charged black hole.}
\label{tab:braneconf2}

\end{table}

The knowledge of the two dual configurations in eleven dimensions will permit us to find an embedding of 
$\asuppp$ in $E_{11}$. In order to do that we first recall how branes are encoded in the algebraic structure of $E_{11}$.

\subsection{Description of the brane configuration in $E_{11}$}

We first briefly recall the algebraic structure of $E_{11}$. The Dynkin diagram is depicted in Figure
\ref{ffirst}. The Lorentzian Kac-Moody algebra $E_{11}$ contains a  subalgebra $\mathfrak{gl}(11, \mathbb R)$ such that $\mathfrak{sl}(11, \mathbb R) \cong A_{10}
\subset \mathfrak{gl}(11, \mathbb R) \subset
E_{11}$. We can again perform a level decomposition of $E_{11}$. The level $\ell$ here
is defined by the number of times the root $\alpha_{11}$ appears in the decomposition of the adjoint representation of $E_{11}$ into irreducible representation of $A_{10}$. The first levels up to $\ell=3$ are listed in Table \ref{tab:leve11} (for more details see Section \ref {subsec:leveldece11}). Here, the indices are vector indices of $\mathfrak{sl}(11,\mathbb{R})$ and hence take values $a=1,\ldots,11$.

\begin{table}[h]
\begin{center}
\begin{tabular}{|c|c|c|}
\hline
$\ell$ &$ \mathfrak{sl}(11,\mathbb{R})$ Dynkin labels& Generator of $E_{11}$\\
\hline \hline
$0$ &$[ 1,0,0,0,0,0,0,0,0,1] $ & $K^a_{\ b}$\\
$1$ &$[ 0,0,0,0,0,0,0,1,0,0]  $ & $R^{\,abc}$\\
$2$ &$[ 0,0,0,0,1,0,0,0,0,0]  $ & $R^{\, abcdef}$\\
$3$ &$[ 0,0,1,0,0,0,0,0,0,1]  $ & $\tilde{R}^{\, abcdefgh|i}$\\

\hline
\end{tabular}
\caption{\sl \small Level decomposition of $E_{11}$ under $\mathfrak{sl}(11,\mathbb{R})$ up to level $\ell=3$ and height $29$. }
\label{tab:leve11}
\end{center}
\end{table}
The positive Chevalley generators of
$E_{11}$ are
${\tilde e}_m=\delta_m^{a} K^a{}_{a+1},\ m=1,\ldots ,10$, and  ${\tilde e}_{11}= R^{\, 9\,  
10\,11}$
where
$R^{abc}$ is the level 1 generators in $E_{11}$. One gets for the Cartan generators
\be
\begin{split}
\begin{aligned}
\label{eqn:aa125} {\tilde h}_m&=\ \delta_m^{a}(K^a{}_a-K^{a+1}{}_{a+1}) \qquad \mathrm{for}  \
m=1,\dots,10\, ,\\
 {\tilde h}_{11}&=\ -\frac{1}{3}(K^1{}_1+\ldots +K^8{}_8) +\frac{2} 
{3}(K^9{}_9+
K^{10}{}_{10}+K^{11}{}_{11})\,.
\end{aligned}
\end{split}
\ee
We now recall how the extremal branes of eleven-dimensional supergravity are encoded in
the algebraic structure of $E_{11}$ (see \cite{Englert:2003py,Englert:2004it,Englert:2004ph,West:2004st}).

Each extremal brane $B_i$ corresponds  to one real root $\alpha_{B_i}$ (or one positive step operator) of $E_{11}$ and the description is always electric namely each M2 brane  is described by a definite component of the three form potential at level one and each M5 is described by a component of the six-form potential of level two. The non-zero component is the one with the indices corresponding to longitudinal directions of the extremal brane $B_i$.  The only other non-zero fields are the Cartan ones which encode the form of the metric \cite{Englert:2003py}. The intersection rules \cite{Argurio:1997gt} are neatly encoded through a pairwise orthogonality condition between the roots corresponding to each brane \cite{Englert:2004it}.

It is worthwhile to recall that such an algebraic description of extremal brane configurations extends to all space-time theories characterized by a $\mathfrak{g}^{+++}$ with simple $\mathfrak{g}$. Here, we will see that it also applies to  pure $\mathcal{N}=2$ supergravity in $D=4$ where $\mathfrak{g}^{+++}=\mathfrak{su}(2,1)^{+++}$, this will be crucial in the next subsection to uncover the embedding.

In Table \ref{tab:step}, we list the positive step operators corresponding to each brane entering in configuration {\bf\sf A} and {\bf\sf B}.
\begin{table}[h]

\begin{center}
\begin{tabular}{|c|c||c|c|}
\hline
Brane of conf. {\bf\sf A} &  step operator  &Brane of conf. {\bf\sf B} &   step operator  \\
\hline\hline
$A_1$ & $R^{4\, 7\, 8\, 9\, 10\, 11}$&  $B_1$ & $R^{4\,5\, 6}$  \\

$A_2$ & $R^{\, 4\, 5\, 6\, 9\, 10 \, 11}$ &  $B_2$ & $R^{4\, 7\, 8}$ \\

$A_3$ & $R^{4\, 6\, 8}$ &  $B_3$ &  $R^{4\, 5\, 8\, 9\, 10\, 11}$\\

$A_4$ & $R^{4\, 5\, 7}$ &  $B_4$ & $R^{4\, 6\,8\, 9\, 10\,11}$\\

\hline
\end{tabular}
\end{center}
\caption{\small The positive step operator corresponding to each brane of configurations {\bf\sf A} and {\bf \sf B}.}
\label{tab:step}
\end{table}
Since all the harmonic functions are the same in   configuration {\bf\sf A} and the dual one {\bf\sf B},  each one is characterized by an unique element of $E_{11}$.  We have
\begin{eqnarray}
\label{stepconA}
\rm{conf. \, \bf \sf{A}}& \Leftrightarrow & c\,  ( \epsilon_1 R^{4\, 7\, 8\, 9\, 10\, 11}+\epsilon_2 R^{4\, 5\, 6\, 9\, 10\, 11}+\epsilon_3 R^{4\, 6\, 8}+\epsilon_4 R^{4\, 5\, 7})\, , \\
\label{stepconB}
\rm{conf. \, \bf\sf{B}} &\Leftrightarrow& c\, (\epsilon_1 R^{4\, 5\, 6}+\epsilon_2 R^{4\, 7\, 8}+\epsilon_3 R^{4\, 5\, 7\, 9\, 10\, 11}+\epsilon_4 R^{4\, 6\, 8 \, 9 \, 10\, 11}) \, ,
\end{eqnarray}
where $c$ is a real constant and $\epsilon_i, i=1 \dots 4$ are signs. We will fix them in the next section.

\subsection{The regular embedding}

We are now in the position to find a regular embedding of $\asuppp$ in $E_{11}$.\footnote{An embedding of the split $\mathfrak{g}_2^{+++}$ in $E_{11}$ was found in \cite{Kleinschmidt:2008jj}. In this reference additional generators were added to take into account the higher rank forms that can be added consistently to the supersymmetry algebra in $D=5$ and to the tensor hierarchy~\cite{Gomis:2007gb,deWit:2008ta}.} We first discuss the non-compact Cartan generators of $\asuppp$: $\mathbf{h_i}$ with $i=1, \ldots, 4$.

\subsubsection{The non-compact Cartan generators of $\asuppp$  }

We first recall that we can describe the Cartan fields of $E_{11}$ in two bases, the $\mathfrak{gl}(11, \mathbb R)$ one and the Chevalley base given by the ${\tilde h}_m$ (see (\ref{eqn:aa125})). The relation between these two bases  (see \eqref{basecartsu21}) is:
\begin{equation}
\sum_{a=1}^{11} \, p_a K^a_{\ a}=\sum_{a=1}^{11} \, q_a\,  \tilde{h}_a\, .
\label{basecart}
\end{equation}
To find the non-compact Cartan generators of $\asuppp$ out of the eleven Cartan generators of $E_{11}$, we have simply to enforce 
\begin{equation}
p_a =0\, , \qquad a=5, \dots, 11.
\label{embc1}
\end{equation}
One can easily understand  this embedding condition in several different ways. In the brane context by noticing that the metric (\ref{conf11}) is characterized by $g_{aa}=1$ for all the longitudinal coordinates $(a=5 \dots 11)$. In a more general way this amounts to demanding that all the scalars coming from the dimensional reduction from eleven down to four should be zero. It is  a necessary condition to have a consistent truncation of eleven-dimensional supergravity to pure $\mathcal{N}=2$ supergravity in $D=4$.

Using (\ref{basecart}) we can translate the embedding condition (\ref{embc1}) in terms of the $q_a$'s using (\ref{eqn:aa125}), we find
\be \begin{split} \begin{aligned} \label{embc2}
q_a &= \tfrac{a-2}{3}\,  q_{11}, \qquad a=4, \dots, 8\, , 
\\
q_9&=\tfrac{4}{3}\,  q_{11},\\
q_{10}&=\tfrac{2}{3}\,  q_{11}.
\end{aligned} \end{split} \ee
Plugging back  (\ref{embc2}) into the Cartan fields of $E_{11}$ in the Chevalley basis, we find
\begin{equation}
\sum_{a=1}^{11} q_a\, {\tilde h}_a = q_1 \mathbf{h_1}+q_2 \mathbf{h_2}+q_3 \mathbf{h_3} +\frac{q_{11}}{3} \mathbf{h_4},
\label{embc4}
\end{equation}
where the $\mathbf{h_i}$ are the four non-compact Cartan generators of $\asuppp$ (see (\ref{eqn:su21+++Chev})).
This completes the discussion of the embedding for the non-compact Cartan generators.
\subsubsection{The other generators  of $\asuppp$  }

We now find  the embedding of the simple step operators and of the compact Cartan generator $\mathbf{h_5}$.
The simple step operators corresponding to the first three nodes of Figure \ref{fig2:su21+++} are of course trivially identified with the step operators of the first three nodes of Figure \ref{ffirst}.
We turn to the generators corresponding to the nodes 4 and 5 of Figure {\ref{fig2:su21+++}, respectively $r^4$ and $\tilde{r}^4$. An extremal Reissner-Nordstr\"om  electrically (resp. magnetically) charged black hole is a zero brane (the only longitudinal direction 4 being time-like). We recall that it is described in $\mathfrak{su}(2,1)^{+++}$ by the step operator $r^4$ (resp. $\tilde{r}^4$) \cite{Englert:2003py}. Consequently, using the brane picture expressions (\ref{stepconA}) and (\ref{stepconB}),  we have the identification 
\be \begin{split} \begin{aligned} \label{embr4}
r^4&=\tfrac{1}{\sqrt 2}\, ( \epsilon_1\,  R^{4\, 7\, 8\, 9\, 10\, 11}+\epsilon_2\,  R^{4\, 5\, 6\, 9\, 10\, 11}+\epsilon_3\,  R^{4\, 6\, 8}+\epsilon_4\, R^{4\, 5\, 7})\, ,  \\
\tilde{r}^4&=\tfrac{1}{\sqrt 2} \,  (\epsilon_1\,  R^{4\, 5\, 6}+\epsilon_2\,  R^{4\, 7\, 8}+\epsilon_3\,  R^{4\, 5\, 7\, 9\, 10\, 11}+\epsilon_4\,  R^{4\, 6\, 8 \, 9 \, 10\, 11}) \, ,
\end{aligned} \end{split} \ee
where the constant $c$ in (\ref{stepconA}) and (\ref{stepconB})  has been fixed to fulfill the normalization of    $r^4$ and $\tilde{r}^4$ in $\asuppp$ (see (\ref{eqn:bililevel1})). We still have to determine the signs $\epsilon_i$. We will fix them in the process of  determining  the compact Cartan $\mathbf{h_5}$ of  $\asuppp$.  The  commutation relations ({\ref{comutenc45}) imply that basically $\mathbf{h_5}$ interchanges the electric and magnetic configuration. The operator $\mathbf{h_5}$ embedded  in   $E_{11}$ should thus correspond, in the brane picture, to the operator interchanging configuration {\bf\sf A} and {\bf\sf B} (see Tables \ref{tab:braneconf1} and \ref{tab:braneconf2}). In order to map configuration {\bf\sf A} onto configuration {\bf\sf B}, brane by brane (i.e $B_i\rightarrow {\tilde B}_i,  \quad i=1 \dots 4$), we have to perform three operations: a double T-duality in the directions 9 and 10, an exchange of the direction 6 and 7 and an exchange of the direction 5 and 8. A double T-duality in the directions 9 and 10 (followed by the exchange of the directions 9 and 10) is described in $E_{11}$ by the Weyl reflection corresponding to the simple root $\alpha_{11}$ (see Figure \ref{ffirst}) \cite{Englert:2003zs,Elitzur:1997zn,Obers:1998rn}. The associated compact generator is: $R^{9\, 10\, 11}-R_{9\, 10\, 11}$.
The exchange of coordinates 6 and 7 (resp. 5 and 8) is generated by the compact generator 
$K^{6}_{\ 7}-K^{7}_{\ 6}$ (resp. $K^{5}_{\ 8}-K^{8}_{\ 5}$). We thus deduce that
\begin{equation}
\label{embedcartcom}
\mathbf{h_5}=K^{6}_{\ 7}-K^{7}_{\ 6}+K^{5}_{\ 8}-K^{8}_{\ 5}+R^{9\, 10\, 11}-R_{9\, 10\, 11},
\end{equation}
we have $(\mathbf{h_5} | \mathbf{h_5})= -6$ as it should (see (\ref{eqn:h5}), (\ref{eqn:bilinearzero})).

To fix the signs in (\ref{embr4}) we use the relation $\left[ r^4, \tilde{r}_4 \right]=\mathbf{h_5}$, we find
\be \begin{split} \begin{aligned}
r^4&=\ \tfrac{1}{\sqrt 2}\, (  R^{4\, 7\, 8\, 9\, 10\, 11}+ R^{4\, 5\, 6\, 9\, 10\, 11} - R^{4\, 6\, 8}+\ R^{4\, 5\, 7})
\label{embs1} \, ,\\
\tilde{r}^4&=\ \tfrac{1}{\sqrt 2} \,  ( R^{4\, 5\, 6}+ R^{4\, 7\, 8}- R^{4\, 5\, 7\, 9\, 10\, 11}+ R^{4\, 6\, 8 \, 9 \, 10\, 11})\, , \\
r_4&=\ \tfrac{1}{\sqrt 2}\, (  R_{4\, 7\, 8\, 9\, 10\, 11}+ R_{4\, 5\, 6\, 9\, 10\, 11} - R_{4\, 6\, 8}+\ R_{4\, 5\, 7})\, , \\
\tilde{r}_4&=\ \tfrac{-1}{\sqrt 2} \,  ( R_{4\, 5\, 6}+ R_{4\, 7\, 8}-R_{4\, 5\, 7\, 9\, 10\, 11}+ R_{4\, 6\, 8 \, 9 \, 10\, 11}) .
\end{aligned} \end{split} \ee
We can the check that the definitions (\ref{embedcartcom})-(\ref{embs1}) together with the $\mathbf{h_i}, \, i=1\dots 4$ (see (\ref{embc4})) satisfy all the relations of $\asuppp$. 

The expressions (\ref{embedcartcom})-(\ref{embs1}) and (\ref{embc2})-(\ref{embc4}) define thus  a regular embedding of  the non-split $\asuppp$ in the split form of $E_{11}$, proving the algebraic counterpart of the truncation of maximal supergravity to the $\mathcal{N}=2$ theory.


\setcounter{section}{0}

\renewcommand{\thesection}{\Alph{section}}

   \setcounter{equation}{0}
    \chapter*{Appendices of  Part II \markboth{Appendices of  Part II}{Appendices of  Part II}}
\addcontentsline{toc}{chapter}{Appendices of  Part II} 
 
 \hrule
\vspace{2cm}
 
 \setcounter{equation}{0}
\section{Signature changes and compensations} \label{appw}
Expressing a Weyl transformation $W$  as a conjugation by a group
element
$U_W$ of $E_{11}$ ($E_{10}$), one defines the involution
$\Omega^\prime$ operating on 
the conjugate elements by
\begin{equation}
\label{newinvolve}
\Omega^\prime (T^\prime)=U_W\,\Omega(\underbrace{U^{-1}_W T^\prime U_W}_{T}) \, U^{-1}_W\, ,
\end{equation}
where $T$ and $T^\prime$ are any conjugate pair of generators in
$E_{11}$ ($E_{10}$). The subgroup invariant under $\Omega$ is
conjugate to the subgroup invariant under $\Omega^\prime$. However, we have seen in Section \ref{sec:tempinvolutintro} that
Weyl reflexions in general do not commute with the temporal
involution~\cite{Keurentjes:2004bv,Keurentjes:2004xx}. 

The Weyl transformations on the gravity line of $E_{11}$ (or $E_{10}$)
simply changes the time coordinate but do not modify the global
signature $(1,10)$ (or $(1,9)$).  This need not be the case for Weyl
transformations from roots pertaining to higher levels. We shall
determine  the different signatures for the M2 and M5 sequences. 
In order to do that, we have to study the effect of the Weyl
reflexions $s_{\alpha_{11}}$ and $s_{-\alpha_{11}+\delta}$ on  the
involutions characterising  the M2 and the M5 we started with. 
We consider separately the two sequences. We will also see that the
nature of the compensation 
transformations ($\mathrm{SO}(2)$ or $\mathrm{SO}(1,1)$) is determined and follows from
this analysis. Finally we shall consider the signatures induced on the
gravity tower by the mapping (\ref{bragramap}), (\ref{bragra1})
and (\ref{bragra2}). 
\subsection{The brane towers}
\subsubsection{Signatures of the M2 sequence} \label{appw1}
We start with the conventional M2 described by the solution
(\ref{2M2}) of 11-dimensional supergravity 
with the signature (1,10, +). Here the first entry denotes the number
of timelike directions 
 (in our case the single direction 9), the second denotes the number
of spacelike directions and the third 
gives the sign of the kinetic energy term in the action ($+$ being the
usual one). We want to determine the space-time 
signature for all the solutions of the M2 sequence. The M2 of level 1
is characterised by 
the involution  $\Omega_{\ell_1}\equiv \Omega_9$ fixing 9 as a time coordinate.

To determine the signature at  level 5, we  perform the Weyl
transformation $s_{-\alpha_{11}+\delta}$ to find  
the corresponding  new involution $\Omega_{\ell_5}$.
The generators of the gravity line affected by this reflexion are $
K^2_{\ 3}$ and $K^8_{\ 9}$. From  
(\ref{newinvolve}) we have (from now on, we drop irrelevant signs
$\rho,\sigma,\tau$ appearing in (\ref{permute})) 
\begin{eqnarray}
\label{ww1a}
\Omega_{\ell_5}
 K^2{}_ 3 &=&U_{s_{-\alpha_{11}+\delta}}\, \Omega_{\ell_1} R^{2\, 4\,5 \,6 \,7 \,8 } \, U^{-1}_{s_{-\alpha_{11}+\delta}}\nn\\
 &=&- K^3{}_ 2 \\
\Omega_{\ell_5}
 K^8{}_ 9&=& U_{s_{-\alpha_{11}+\delta}}\, \Omega_{\ell_1} R_{3\, 4\,5 \,6 \,7 \,9 } \, U^{-1}_{s_{-\alpha_{11}+\delta}}
 \label{ww1ab}\nn\\
&=& +K^9{}_8\, .
\end{eqnarray}
The signature of the level 5 solution is thus unchanged, the only time coordinate is still 9.
The action of the involution $\Omega_{\ell_5}$ on $R^{9\, 10\, 11}$  follows from  $s_{-\alpha_{11}+\delta}(\alpha_{11})= 2\delta  - \alpha_{11}$. We get 
\begin{equation}
\label{ww1b}
\Omega_{\ell_5}
 R^{\,  9 \, 10\, 11}= U_{s_{-\alpha_{11}+\delta}}\, \Omega_{\ell_1} R^{[6]}_5\, U^{-1}_{s_{-\alpha_{11}+\delta}}= + R_{\,  9 \, 10\, 11}\, .
 \end{equation}
As 9 is a timelike coordinate, this yields the usual sign for this generator and one sticks to the signature $(1,10,+)$  as it should. We have 
$\Omega_{\ell_5}\, (R^{[6]}_2-R^{[6]}_{-2})=R^{[6]}_2-R^{[6]}_{-2}$ and the compensation  for the $\mathrm{SL}(2)$ of level 5 (see (\ref{matrix5})-(\ref{comp5})) lies in      its $\mathrm{SO}(2)$ subgroup.

We now  perform the Weyl reflexion $ s_{\alpha_{11}}$   to reach level 7. We have
\begin{eqnarray}
\Omega_{\ell_7} K^8{}_ 9 &=& U_{s_{\alpha_{11}}}\, \Omega_{\ell_5} R^{\,8 \,10 \,11 } \, U^{-1}_{s_{\alpha_{11}}} \nn \\ \label{ww2a}
&=& - K^9{}_8\\
 \label{ww2b}
\Omega_{\ell_7}
R^{\, 9\, 10\, 11} &=& U_{s_{\alpha_{11}}}\, \Omega_{\ell_5} R_{\, 9\, 10\, 11} \, U^{-1}_{s_{\alpha_{11}}}\nn\\
&=& +R_{\, 9\, 10\, 11}\, .
\end{eqnarray}
From (\ref{ww2a}), we deduce immediately that the time coordinates are now 10 and 11 and from
(\ref{ww2b}) we deduce that the sign of the kinetic terms is the `wrong' one, namely it corresponds to the (2,9,-) theory as it should \cite{Hull:1998vg, Hull:1998ym}. We have $\Omega_{\ell_7}\, (
R^{\, 9\, 10\, 11} +R_{\, 9\, 10\, 11})=R^{\, 9\, 10\, 11} +R_{\, 9\, 10\, 11}$ and  the compensation for the $\mathrm{SL}(2)$ of level 7 solution (\ref{fin7}) lies in its $\mathrm{SO}(1,1)$   subgroup.
\begin{table}[h]
\begin{center}
\begin{tabular}{|c|c|c|c|}
\hline
levels ($n>0$)&
times & $(t,s,\pm)$ & compensation \\
\hline\hline
1&9&$(1,10,+) $&$--$\\
\hline 
5&9&$(1,10,+) $&$\mathrm{SO}(2)$\\
\hline
 1+6n , n odd&10,11&$(2,9,-) $&$\mathrm{SO}(1,1)$\\
 -1+6(n+1), n odd &10,11&$(2,9,-) $&$\mathrm{SO}(2)$\\
\hline 
1+6n , n even&9&$(1,10,+) $&$\mathrm{SO}(1,1)$\\
 -1+6(n+1), n even &9&$(1,10,+) $&$\mathrm{SO}(2)$\\
\hline
\end{tabular}
 \caption{\sl \small Involution switches from $\Omega_{\ell_1}$ to
$\Omega_{\ell_{\pm 1+6n}}$ in the M2 sequence due to the application of
  the successive Weyl reflexions  $s_{-\alpha_{11}+\delta}$ and
  $s_{\alpha_{11}}$}  \label{tab:M2seq}
\end{center}
\end{table}

\noindent
We can now repeat the analysis to all levels of the M2 sequence. We use $ s_{-\alpha_{11}+\delta}$  to go from level $1+6n$ to level $-1+6(n+1)$. Replacing in (\ref{ww1a}), (\ref{ww1ab}) and  (\ref{ww1b})  
$\Omega_{\ell_1}$ by $\Omega_{\ell_{1+6n}}$ and $\Omega_{\ell_5}$ by $\Omega_{\ell_{-1+6(n+1)}}$, we see that the signature of the theory is unchanged. Furthermore,  analysing the action of   $\Omega_{\ell_{-1+6(n+1)}}$  on $R^{[6]}_2$, we conclude that the compensation for level  $-1+6(n+1)$ is always an $\mathrm{SO}(2)$ one.
We use $s_{\alpha_{11}}$  to go from level $-1+6(n+1)$ to level $1+6(n+1)$. Replacing in (\ref{ww2a}) and  (\ref{ww2b}) $\Omega_{\ell_5}$ by $\Omega_{\ell_{-1+6(n+1)}}$ and $\Omega_{\ell_7}$ by $\Omega_{\ell_{1+6(n+1)}}$, we see that theories $(1,10,+)$ and $(2,9,-)$ are interchanged. The action of $\Omega_{\ell_{1+6(n+1)}}$ on  $R^{[3]}_1$ shows that the compensation at level  $1+6(n+1)$ lies always in $\mathrm{SO}(1,1)$. The results are summarised in Table \ref{tab:M2seq}.
\subsubsection{Signatures of the M5 sequence} \label{appw2}
We start with the non-exotic M5 described by (\ref{2M5}) solution of eleven-dimensional  supergravity
with the signature (1,10, +) and time in 3. We want to determine the space-time
signature of all the solutions of the M5 sequence depicted in Figure \ref{ger1fig}. The M5 of level 2 is characterised by
the involution  $\Omega_{\ell_2} \equiv \Omega_3$ fixing 3 as a time coordinate.

To determine the signature of  level 4, we  perform the Weyl reflexion $ s_{\alpha_{11}}$  to find 
the  new involution $\Omega_{\ell_4}$.
The generator of the gravity line affected by this reflexion is $K^8{}_9$. From 
(\ref{newinvolve}) we have
\begin{align}
\label{ww52a}
\Omega_{\ell_4}
 K^8_{\ 9} & \ = \ U_{s_{\alpha_{11}}}\, \Omega_{\ell_2} R^{\,8 \,10 \,11 } \, U^{-1}_{s_{\alpha_{11}}}=\  - K^9_{\ 8} \\ \label{ww52b}
\Omega_{\ell_4}
R^{\, 9\, 10\, 11} &\ =\  U_{s_{\alpha_{11}}}\, \Omega_{\ell_2} R_{\, 9\, 10\, 11} \, U^{-1}_{s_{\alpha_{11}}}= \ -R_{\, 9\, 10\, 11}\, ,
\end{align}

From (\ref{ww52a}), we deduce immediately that there is no change of signature and thus 3 remains the only timelike direction. From
(\ref{ww52b}) the sign of the kinetic terms is unchanged and  the phase is still
(1,10,+) as it should. We have $\Omega_{\ell_4} (R^{\, 9\, 10\, 11} -R_{\, 9\, 10\, 11})=R^{\, 9\, 10\, 11} -R_{\, 9\, 10\, 11}$. Hence the coset characterising the level 4 solution is
$\mathrm{SL}(2)/\mathrm{SO}(2)$ and the compensation at level 4 (see (\ref{fin4})) lies in $\mathrm{SO}(2)$.
\begin{table}[h]
\begin{center}
\begin{tabular}{|c|ccccccccc|c|c|}
\hline
$\ell$ &$K^2_{\ 3}$&$K^3_{\ 4}$&$K^4_{\ 5}$&$K^5_{\ 6}$&$K^6_{\
7}$&$K^7_{\ 8}$&$K^8_{\ 9}$&$K^9_{\ 10}$&$K^{10}_{\ 11}$&
times & $(t,s,\pm)$ \\
\hline\hline
4&$+$&$+$&$-$&$-$&$-$&$-$&$-$&$-$&$-$&3&$(1,10,+) $\\
\hline 8& $-$&$+$&$-$&$-$&$-$&$-$&$+$&$-$&$-$&4,5,6,7,8&$(5,6,+)
$\\
\hline
\end{tabular}
\caption{\sl \small Involutions at level 4 and 8.}
\label{tab:signam5seql4}
\end{center}
\end{table}

\vskip -.5cm
To determine the signature of  level 8, we  perform the Weyl reflexion $s_{-\alpha_{11}+\delta}$ to find 
the  new involution $\Omega_{\ell_8}$.
We have
\begin{eqnarray}
\label{ww51a}
&&\Omega_{\ell_8}
 K^2{}_3=U_{s_{-\alpha_{11}+\delta}}\, \Omega_{\ell_4} R^{2\, 4\,5 \,6 \,7 \,8 } \, U^{-1}_{s_{-\alpha_{11}+\delta}}= - K^3{} _2 \nonumber \\
&&\Omega_{\ell_8}
 K^8{}_9 = U_{s_{-\alpha_{11}+\delta}}\, \Omega_{\ell_4} R_{3\, 4\,5 \,6 \,7 \,9 } \, U^{-1}_{s_{-\alpha_{11}+\delta}}= +K^9{}_8\, .
\end{eqnarray}
The flip of sign in  $K^2_{\  3}$ and $K^8_{\  9}$ 
illustrated in Table \ref{tab:signam5seql4} shows that the resulting theory comprises the 5 time coordinates 4, 5, 6, 7 and 8.
The involution $\Omega_{\ell_8}$ acts on $R^{9\, 10\, 11}$ according to
 \begin{equation}
\label{ww51b}
\Omega_{\ell_8}
 R^{\,  9 \, 10\, 11}=U_{s_{-\alpha_{11}+\delta}}\, \Omega_{\ell_4} R^{[6]}_5\, U^{-1}_{s_{-\alpha_{11}+\delta}}= - R_{\,  9 \, 10\, 11}\, .
 \end{equation}
The directions 9, 10 and 11 being spacelike, the action of the involution on $R^{\,  9 \, 10\, 11}$ yields the `right' kinetic energy terms. The new theory is $(5,6,+)$  as it should \cite{Hull:1998vg, Hull:1998ym}.
We have  $\Omega_{\ell_8}(R^{[6]}_2+R^{[6]}_{-2})=R^{[6]}_2+R^{[6]}_{-2}$. Hence the compensation
at level 8 (see (\ref{fin8})) lies in $\mathrm{SO}(1,1)$.

\begin{table}[h]
\begin{center}
\begin{tabular}{|c|c|c|c|}
\hline
levels ($n>0$)&
times & $(t,s,\pm)$ & compensation \\
\hline\hline
2&3&$(1,10,+) $&$--$\\
\hline 
4&3&$(1,10,+) $&$\mathrm{SO}(2)$\\
\hline
 2+6n , n odd&4,5,6,7,8&$(5,6,+) $&$\mathrm{SO}(1,1)$\\
 -2+6(n+1), n odd &4,5,6,7,8&$(5,6,+) $&$\mathrm{SO}(2)$\\
\hline 
2+6n , n even&3&$(1,10,+) $&$\mathrm{SO}(1,1)$\\
 -2+6(n+1), n even &3&$(1,10,+) $&$\mathrm{SO}(2)$\\
\hline
\end{tabular}
\caption{\sl \small Involution switches from $\Omega_{\ell_2}$ to
$\Omega_{\ell_{\pm 2+6n}}$ in the M5 sequence due to the application of the successive Weyl reflexions  $s_{\alpha_{11}}$ and $s_{-\alpha_{11}+\delta}$. }
\label{tab:signam5seq}
\end{center}
\end{table}

\vskip -.5cm
We can repeat the analysis to find the signature of all  the levels of the M5 sequence.
Again we find only the two signatures found at the lower levels alternating every two steps while, as for the M2 sequence, the nature of the compensation alternates at each step.
The results are summarised in Table \ref{tab:signam5seq}.

\subsection{The gravity towers} \label{appsg}
The gravity towers were obtained by performing Weyl transformations on the brane towers. In Section \ref{sec:gravitytowersec} we have showed that the M2 sequence is mapped to the wave sequence and the M5 sequence to the monopole sequence. We shall take advantage of this Weyl mapping to find the signatures of the gravity towers.

The brane towers comprise $4$ different signatures:
\begin{itemize}
\item $(1,10,+)$ with time in $9$ and $(2,9,-)$ with time in $10$ and $11$ for the M2 sequence,
\item  $(1,10,+)$ with time in $3$ and $(5,6,+)$ with time in $4,\, 5,\, 6,\, 7$ and $8$  for the M5 sequence.
\end{itemize}

We perform the Weyl mapping on the $4$ signatures in three steps.  The first Weyl transformation $W_{(1)}$ interchanges 9 and 3 on the gravity line. The second Weyl transformation is the Weyl reflexion $W_{(2)}\equiv s_{\alpha_{11}}$  and the last Weyl transformation $W_{(3)}$ interchanges 9 and 11 on the gravity line. 

The first and last Weyl transformations $W_{(1)}$ and $W_{(3)}$ permute tensor indices and do not alter the global signature. Only  $W_{(2)}$ can change the global signature $(t,s,\pm)$. There are $2$ simple roots affected by $W_{(2)}\equiv s_{\alpha_{11}}$: $\alpha_8$ and $\alpha_{11}$ defining respectively the generators $K^8_{\ 9}$ and $R^{9\, 10\, 11}$. 
From (\ref{newinvolve}) we get, dropping irrelevant signs, 
\begin{align}\label{k89}
\Omega' \, K^8_{\ 9}&\ =\  s_1\, K^9_{\ 8}=\  U_{W_{(2)}} \, \underbrace {\Omega \,  R^{8\, 10\, 11}}_{s_1\, R_{8\, 10\, 11}} \, U_{W_{(2)}}^{-1} \\
\label{r91011}
\Omega' \, R^{9\, 10\, 11}&\ = \ s_2\, R_{9\, 10\, 11}= \ U_{W_{(2)}} \, \underbrace {\Omega \,  R_{9\, 10\, 11}}_{s_2\, R_{9\, 10\, 11}} \, U_{W_{(2)}}^{-1}\, .
\end{align}
where $s_1$ and $s_2$ are signs. The possible change of signatures by the Weyl transformation $W_{(2)}$ will be deduced from the signs $s_1$ and $s_2$.
\subsubsection{Signatures of the wave sequence}
\noindent
{\bf $\bullet$ Mapping of the signature $(1,10,+)$ with time in $9$}

The global signature $(1,10,+)$ is not modified by $W_{(1)}$ but the time coordinate is no longer in $9$ but in $3$. From (\ref{k89}) and (\ref{r91011}), with $s_1=-1 $ and $s_2=-1$, we deduce that the signature is unchanged by the second Weyl reflexion $W_{(2)}$. The last transformation $W_{(3)}$ does not change the signature either. 

The signature $(1,10,+)$ with time in $9$ is thus mapped by the three successive Weyl transformations to the signature $(1,10,+)$ with time in $3$.

\medskip
\noindent
{\bf $\bullet$ Mapping of the signature $(2,9,-)$ with time in $10$ and $11$}

The first Weyl transformation $W_{(1)}$ does not modify the signature. We then perform the Weyl transformation $W_{(2)}$. From (\ref{k89}) with $s_1=+1$ we find that the time coordinate becomes $9$ and from (\ref{r91011}) with $s_2=+1$, we deduce the sign of the kinetic term. This sign is the `usual' one and the signature becomes $(1,10,+)$ with time coordinate $9$. The last Weyl reflexion $W_{(3)}$ does not change the global signature but puts the time coordinate  in $11$.

The signature $(2,9,-)$ with times in $10$ and $11$ is thus mapped by the three successive Weyl transformations to the signature $(1,10,+)$ with time in $11$.

\medskip
\noindent{\bf $\bullet$ Signatures of the wave sequence}

\begin{table}[h]
\begin{center}
\begin{tabular}{|c|c|c|}
\hline
levels ($n, n' >0$)&
times & $(t,s,\pm)$  \\
\hline\hline
0&3&$(1,10,+) $\\
\hline 
6n , n odd&11&$(1,10,+) $\\
6n, n even&3&$(1,10,+)  $\\
\hline 
6n' , n' odd&3&$(1,10,+) $\\
6n', n' even&11&$(1,10,+)  $\\
\hline
\end{tabular}
\caption{\sl \small Signatures of the wave sequence}
\label{tab:signawavseq}
\end{center}
\end{table}

The coset representatives of the M2 sequence $\mathcal{V}_{1+6n}$ (\ref{seqM21}) and $\mathcal{V}_{-1+6n}$ (\ref{seqM22}) are mapped respectively to the coset representatives of the wave sequence $\mathcal{V}_{6n}$ (\ref{seqG01}) and $\mathcal{V}_{6n'}$ (\ref{seqG02}). All the signatures of the wave sequence are summarised in Table \ref{tab:signawavseq}.

\subsubsection{Signatures of the monopole sequence}

\noindent
{\bf $\bullet$ Mapping of the signature $(1,10,+)$ with time in $3$}

The global signature $(1,10,+)$ is not modified by $W_{(1)}$ but the time coordinate is no longer in $3$ but in $9$. We then perform the Weyl transformation $W_{(2)}$. From (\ref{k89}) with $s_1=-1$, we find that the time coordinates become $10$ and $11$ and from (\ref{r91011}) with $s_2=+1$ we deduce the sign of the kinetic term. This sign is the `wrong' one and the signature becomes $(2,9,-)$. The last Weyl reflexion $W_{(3)}$ does not change the global signature $(2,9,-)$ but puts the time coordinates  in $9$ and $10$. 

The signature $(1,10,+)$ with time in $3$ is mapped by the three successive Weyl transformations to the signature $(2,9,-)$ with times in $9$ and  $10$.

\medskip
\noindent
{\bf $\bullet$ Mapping of the signature $(5,6,+)$ with times in $4, 5, 6, 7, 8$}

The first Weyl transformation $W_{(1)}$ does not modify the signature. From (\ref{k89}) and (\ref{r91011}) with $s_1=+1 $ and $s_2=-1$ we see that the signature is invariant  under the second Weyl reflexion $W_{(2)}$.  The last Weyl reflexion $W_{(3)}$ also leaves the signature unchanged. 

The signature $(5,6,+)$ with time in $4, 5, 6, 7, 8$ is left invariant by the Weyl mapping.

\begin{table}[h]
\begin{center}
\begin{tabular}{|c|c|c|}
\hline
levels ($n, n' >0$)&
times & $(t,s,\pm)$  \\
\hline\hline
3&9,10&$(2,9,-) $\\
\hline 
3+ 6n' , n' odd&4,5,6,7,8&$(5,6,+) $\\
3+6n', n' even&9,10&$(2,9,-)  $\\
\hline 
-3+6n , n odd&9,10&$(2,9,-) $\\
-3+6n, n even&4,5,6,7,8&$(5,6,+)  $\\
\hline
\end{tabular}
\caption{\sl \small Signatures of the monopole sequence}
\label{tab:signamonopolseq}
\end{center}
\end{table} 

\noindent
{\bf $\bullet$ Signatures of the monopole sequence}

The coset representatives of the M5 sequence $\mathcal{V}_{2+6n}$ (\ref{seqM51}) and $\mathcal{V}_{-2+6n}$ (\ref{seqM52}) are mapped respectively to the coset representatives of the monopole sequence $\mathcal{V}_{3+ 6n'}$ (\ref{seqG31}) and $\mathcal{V}_{-3+ 6n}$ (\ref{seqG32}). All the signatures of the monopole sequence are summarised in Table \ref{tab:signamonopolseq}.

 \setcounter{equation}{0}
 \section{Weyl transformations commute with compensations}\label{appc}
We will show  that the set of dualitites, compensations and Weyl transformations needed to express, in the M2 and M5 sequences, the Borel representative at a given level  in terms of the level~1 supergravity field  does not depend on the path chosen in Figure \ref{second} and Figure \ref{fig:sequenceb}. Equivalently we will prove that Weyl transformations and compensations do commute. The same proof can be done for the gravity tower.

We note that the  nature of the compensation matrix, i.e. $\mathrm{SO}(2)$ or
$\mathrm{SO}(1,1)$,  is unaltered by  the Weyl reflexions $s_{\alpha_{11}}$ or
$s_{-\alpha_{11}+\delta}$. In other words it is the same along a
column in Figure \ref{second}  and Figure \ref{fig:sequenceb}. 
Indeed (\ref{newinvolve}) shows that the Weyl reflexion mapping the level $k$ generator to the level $k+n$ acts on the involution $\Omega R_k = \epsilon R_{-k}$ where $\epsilon$ is a sign, to yield $\Omega^{\prime} R_{k+n} = U \Omega R_k U^{-1}= \epsilon R_{-k-n}$ with the same sign (irrelevant signs in the Weyl transformed have been dropped). Taking this fact into account, we will analyse simultaneously the $\mathrm{SO}(2)$ and $\mathrm{SO}(1,1)$ compensations.

We start at a given level from the Borel representative  given by (\ref{borel3n}) and  (\ref{borel6n}) for the M2 sequence and by (\ref{6Nborel}) and  (\ref{3Nborel}) for the M5 sequence. After a number of dualities and compensations we reach $\boldsymbol{R_i}\equiv R^{[6]}_{1-3n}$ or $\boldsymbol{R_i} \equiv R^{[3]}_{-1-3n}$  defining a Borel representative ${\cal V}_i$. Both cases are shown in Table \ref{t2case}. We will show that ${\cal V}_f$  defined by $\boldsymbol{R_f}$ in the Table is independent of the path joining $\boldsymbol{R_i}$ to $\boldsymbol{R_f}$. Hence compensations and Weyl transformations do commute.
  \begin{table}[top]
\begin{center}
\begin{tabular}{c|ccccc}
\emph{Level} & &  & &   &\\
-1+3p&$R^{[6]}_{-1+ 3p}$& $\dots$&$ \boldsymbol{R_i} \equiv R^{[6]}_{1-3n}$&$\stackrel{comp}{\longrightarrow} $&$R^{[6]}_{-1+3n}$  \\
{} & & & & & \\
 $  s_{\alpha_{11}}$& & & $\Downarrow$ & & $\downarrow$\\
{} & & & & &\\
1+3p&$R^{[3]}_{1+ 3p}$ &$\dots$ &$R^{[3]}_{-1-3n}$&$\stackrel{comp}{\Longrightarrow}$&$\boldsymbol{R_f} \equiv R^{[3]}_{1+3n}$\\
\hline
.& & &. & &. \\
.& & &. & &. \\
.& & &. & &. \\
\hline
1+3r&$R^{[3]}_{1+ 3r}$& $\dots$&$\boldsymbol{R_i}\equiv R^{[3]}_{-1-3n}$&$\stackrel{comp}{\longrightarrow} $&$R^{[3]}_{1+3n}$\\
{} & & & & & \\
$ s_{-\alpha_{11}+\delta}$  & & & $\Downarrow$ & & $\downarrow$\\
{} & & & & &\\
-1+3 (r+2) &$R^{[6]}_{-1+ 3(r+2)}$ & $\dots$&$R^{[6]}_{1-3(n+2)}$&$\stackrel{comp}{\Longrightarrow}$&$\boldsymbol{R_f}\equiv R^{[6]}_{-1+3(n+2)}$ \\
\end{tabular}
\caption{\sl \small Commutation of compensations and Weyl transformations: the paths depicted by
  simple arrows $``\rightarrow ''$  and by  double arrows
 $``\Rightarrow ''$ lead to the same result. The table applies to both  M2 and M5 sequences described in Figure \ref{second}  and Figure \ref{fig:sequenceb}. The Weyl transformations are    $ s_{\alpha_{11}}$ and $ s_{-\alpha_{11}+\delta}$, and  $1<p$,  $0<r$, $n< p,r$.}
 \label{t2case}
\end{center}
\end{table}

\subsection{case 1: $R_i \equiv R^{[6]}_{1-3n}$ and $R_f \equiv R^{[3]}_{1+3n}$}

Let ${\cal V}_a$ be a coset representative along a path joining $\boldsymbol{R_i}$ to $\boldsymbol{R_f}$.  We write the Cartan contribution ${\cal V}_a^{(0)}$ to ${\cal V}_a$ as
\begin{eqnarray}
 \label{reducedC}
\mathcal{V}_a^{(0)}= \exp \left [ \big (\frac{1}{2} \ln \mathcal{R}e \,  \mathcal{E}_a \big) \  \mathcal{C}_a + \Xi\, K^2{}_2 \right ]  \, .
\end{eqnarray}
Here we isolated in $\mathcal{V}_a^{(0)}$ a contribution $\Xi\, K^2{}_2$ left invariant along the paths by both compensations and Weyl transformations. The relevant contribution of the Cartan generators  is the transformed   $[(1/2)\ln\mathcal{R}e \,  \mathcal{E}_a  \  \mathcal{C}_a]$  of $[(1/2)\ln\mathcal{R}e \,  \mathcal{E}_i ) \  \mathcal{C}_i]$ at $\boldsymbol{R_i}$ where
$\mathcal{C}_i$ is the linear combination of the Cartan generators $h_{11}$ and $K^2{}_2$ pertaining to the $\mathrm{SL}(2)$ subgroup containing
$R^{[6]}_{1-3n}$. From (\ref{slb}), we have  $\mathcal{C}_i = \alpha( -h_{11}-n K^2{}_2)$, and $\alpha \,=\, +1\  [-1]$ if the compensation is in $\mathrm{SO}(2)$ [$\mathrm{SO}(1,1)$]. The Ernst potential $\mathcal{E}_a $ is invariant along a column of Table \ref{t2case}.

We now examine the transformations of $\mathcal{V}_a^{(0)}$ along the paths.

\begin{itemize}
\item{ path ``$\rightarrow$''}

After compensation, both for the $\mathrm{SO}(2)$ and $\mathrm{SO}(1,1)$ compensations, one gets\footnote{The action of the compensation on $ \mathcal{C}_i$ is obtained by  straightforward generalisation of the compensations performed in Sections \ref{subsec:meseqth} and \ref{subsec:m5seqth}.} $\mathcal{C}_c = - \, \mathcal{C}_i=  \alpha( h_{11}+n K^2{}_2)$. Performing  the Weyl transformation $s_{\alpha_{11}}$, $K^2{}_2$ is left invariant and the sign of $h_{11}$ changes. The final expression for  $\mathcal{C}_a$ is thus $\mathcal{C}_f= \alpha(- h_{11}+n K^2{}_2)$.

\item{ path ``$\Rightarrow$''}

Performing the Weyl reflexion $ s_{\alpha_{11}}$, we get $\mathcal{C}_{ s_{\alpha_{11}}} = \alpha(h_{11}-n K^2{}_2$). To perform the subsequent compensation, we must identify the Cartan generator of $\mathrm{SL}(2)$ subgroup containing $R^{[3]}_{-1-3n}$. From 
(\ref{sla}) this is precisely $\mathcal{C}_{ s_{\alpha_{11}}}$. After compensation we then obtain $\mathcal{C}_f=- \, \mathcal{C}_{ s_{\alpha_{11}}}=   \alpha (-h_{11}+n K^2{}_2)$.

\end{itemize}

The two different paths yield the same $\mathcal{C}_f$, {\em QED}. 

\noindent
The proof of path equivalence for the contribution of the step operators to ${\cal V}_a$
is immediate except for possible sign shifts in the Weyl transformations. Taking these into account, one easily verifies that this  affects in the same way both columns of Table \ref{t2case} connecting $\boldsymbol{R_i}$ to $\boldsymbol{R_f}$, as these columns list generators corresponding to opposite roots (see discussion after (\ref{permute})). This completes the proof.

\subsection{case 2: $R_i \equiv R^{[3]}_{-1-3n}$ and $R_f \equiv R^{[6]}_{-1+(3n+2)}$ }

Our starting point is (\ref{reducedC}).  In the relevant contribution of the Cartan generators, $\mathcal{C}_i$ is here the linear combination of  $h_{11}$ and $K^2{}_2$ pertaining to the $\mathrm{SL}(2)$ subgroup containing
$R^{[3]}_{-1-3n}$. From (\ref{sla}), we have  $\mathcal{C}_i = \alpha( h_{11}-n K^2{}_2)$.

\begin{itemize}

\item{ path ``$\rightarrow$''}

After compensation, both for the $\mathrm{SO}(2)$ and $\mathrm{SO}(1,1)$ compensations, one gets $\mathcal{C}_c = - \, \mathcal{C}_i=  \alpha(- h_{11}+n K^2{}_2)$. Performing  the Weyl transformation $s_{-\alpha_{11}+\delta}$, $K^2{}_2$ is left invariant and  $ h_{11} \rightarrow -h_{11}- 2\, K^2{}_2 $. The final expression for  $\mathcal{C}_a$ is thus $\mathcal{C}_f= \alpha (h_{11}+(n+2)K^2{}_2)$.

\item{ path ``$\Rightarrow$''}

Performing the Weyl reflexion $s_{-\alpha_{11}+\delta}$, we get $\mathcal{C}_{ s_{-\alpha_{11}+\delta}} = \alpha (-h_{11} -(n+2)K^2{}_2 )$. To perform the subsequent compensation, we must identify the Cartan generator of $\mathrm{SL}(2)$ subgroup containing $R^{[6]}_{1-3(n+2)}$. From 
(\ref{slb}) this is precisely $\mathcal{C}_{ s_{-\alpha_{11}+\delta}}$. After compensation we then obtain $\mathcal{C}_f=- \, \mathcal{C}_{ s_{-\alpha_{11}+\delta}}=   \alpha (h_{11}+(n+2)K^2{}_2)$.

\end{itemize}

The two different paths yield the same $\mathcal{C}_f$. It
can also be checked that not only the actions on the generators
commute  but that also the corresponding fields are transformed in the
same way. {\em QED}.

 \setcounter{equation}{0}
\section{Redundancy of the two  gravity towers}
\label{apprg}

We will show that there are redundancies at each level 3n of the wave sequence and of the monopole sequence. 

Let us consider first the wave sequence. The metric of the wave sequence solutions follow from the representatives (\ref{seqG01}) and (\ref{seqG02}). One has
\begin{align}
\label{wave6n}
\dd s^2_{[6n]}&= g_{ab} \, \dd x^a\, \dd x^b \nonumber \\
&=   {\cal F}_{2n-1}\bar{\cal F}_{2n-1} \Big [(\dd x^1)^2+(\dd x^2)^2 \Big ] + (-1)^{n+1} H^{-1}_{2n+1}(\dd x^3)^2  \\
& \quad +\Big [(\dd x^4)^2\dots+(\dd x^{10})^2 \Big ]\nn  \\
& \quad  + (-1)^{n} H_{2n+1} \Big [ \dd x^{11} -  \Big( (-1)^{n}   H^{-1}_{2n+1}+ (-1)^{n+1} \Big)\dd x^3 \Big]^2 \, , \nn \\
 \nonumber \\
 \label{wave6np}
\dd s^2_{[6 n']}&= \widetilde{g}_{ab} \, \dd x^a\, \dd x^b\nonumber \\
&= {\cal F}_{2n'-1}\bar{\cal F}_{2n'-1} \Big [(\dd x^1)^2+(\dd x^2)^2 \Big ]  + (-1)^{n'} H^{-1}_{2n'}(\dd x^3)^2 \\
& \quad+\Big [(\dd x^4)^2\dots+(\dd x^{10})^2 \Big ] \nn\\
& \quad + (-1)^{n' +1} H_{2n'} \Big [ \dd x^{11} -  \Big( (-1)^{n' +1}   H^{-1}_{2n'}+ (-1)^{n'} \Big)\dd x^3 \Big]^2  \, ,   \nonumber 
\end{align}
where $H_p = {\cal R}e\, {\cal E}_p$.  

We now show that these metric are identical at each level $6n=6n'$ up to interchange of the coordinates $3$ and $11$. Using the relation
\begin{eqnarray} \label{hdn}
{\cal E}(z)_{2n+1}= 2-{\cal E}(z)_{2n} \qquad \Longrightarrow \qquad H_{2n+1}= 2-H_{2n} \, ,
\end{eqnarray}
 in (\ref{wave6n})  and performing  the following coordinate transformation 
\begin{eqnarray} \label{tsfc}
\left \{ \begin{array}{ll} 
x_3^{\prime} = - x_{11} & \, , \\
x_{11}^{\prime} = x_3& \, ,  \\
x_a^{\prime} = x_a  & a \neq 3, 11 \, ,
 \end{array}  \right.
\end{eqnarray}
we get the transformed metric 
\begin{align}
ds^{\prime \, 2}_{[6n]}&=  {\cal F}_{2n-1}\bar{\cal F}_{2n-1} \Big [(\dd x^1)^2+(\dd x^2)^2 \Big ] + (-1)^{n+1} (2- H_{2n})^{-1} (\dd x^{11})^2 \nn \\
&\quad +\Big [(\dd x^4)^2\dots+(\dd x^{10})^2 \Big ] \\
 & \quad + (-1)^{n}(2- H_{2n}) \Big [ \dd x^{3} + \Big( (-1)^{n}  (2- H_{2n})^{-1}+ (-1)^{n+1} \Big)\dd x^{11} \Big]^2 \,. \nn
\end{align}
We thus conclude that $ds^{\prime \, 2}_{[6n]} =  \dd s^2_{[6n']} $ for $n = n'$ :
\begin{eqnarray}
\begin{array}{lll}
g^\prime _{11\, 11} &= (-1)^{n+1} H_{2n} &= \widetilde {g} _{11\, 11} \, ,   \\
 g^\prime _{3\, 3}& = (-1)^{n}(2- H_{2n})&= \widetilde {g} _{3\, 3}  \, , \\
g^\prime _{3\, 11} &= -1+ H_{2n} &= \widetilde {g} _{3\, 11}  \, , \\
g^\prime _{a\, a} & = \, \widetilde {g} _{a\, a} \quad a\neq 3, 11 \, .
 & \end{array}
\end{eqnarray}
Let us now  consider the monopole sequence. The metric of the monopole sequence solutions follow from the representatives (\ref{seqG31}) and (\ref{seqG32}). One has
\begin{align}
\label{m6n}
\dd s^2_{[3+ 6n']}&= g_{ab} \, \dd x^a\, \dd x^b \nonumber\\
&= {\cal F}_{2n'-1}\bar{\cal F}_{2n'-1}  H_{2n'+1}\Big [(\dd x^1)^2+(\dd x^2)^2 \Big ]\! + H_{2n'+1}(\dd x^3)^2 \\
& \quad +(-1)^{n'}  \Big [(\dd x^4)^2\dots +(\dd x^{8})^2 \Big ] +(-1)^{n'+1}  \Big [(\dd x^9)^2+(\dd x^{10})^2 \Big ]  \nn  \\
 &\quad  + H^{-1}_{2n'+1} \Big [ \dd x^{11} -  \Big( (-1)^{n'}   B_{2n'+1}\Big)\dd x^3 \Big]^2 \nn \\
 \nonumber \\
 \label{m63n}
\dd s^2_{- 3+ 6n}&= \widetilde{g}_{ab} \, \dd x^a\, \dd x^b\nonumber\\  
&= {\cal F}_{2n-1}\bar{\cal F}_{2n-1}  H_{2n}\Big [(\dd x^1)^2+(\dd x^2)^2 \Big ] + H_{2n}(\dd x^3)^2\\
&\quad +(-1)^{n+1}  \Big [(\dd x^4)^2\dots+(\dd x^{8})^2 \Big ] +(-1)^{n}  \Big [(\dd x^9)^2+(\dd x^{10})^2 \Big ]     \nonumber \\
 &\quad  + H^{-1}_{2n} \Big [ \dd x^{11} -  \Big( (-1)^{n+1}   B_{2n}\Big)\dd x^3 \Big]^2 \,. \nn
 \end{align}
 where $B_p = {\cal I}m\, {\cal E}_p$.
 
We now show that these metric are identical at each level $3+ 6n'= - 3+ 6n$, i.e. $n= n'+1$, up to interchange of the coordinates $3$ and $11$. Replacing in (\ref{m63n}) $n$ by $n'+1$, we get 
\be \label{m63n2} \begin{split}
\dd s^2_{[-3+ 6(n'+1)]} &=   {\cal F}_{2n'-1}\bar{\cal F}_{2n'-1} {\cal E}_{2n'+1}\bar{\cal E}_{2n'+1}   H_{2n'+2}\Big [(\dd x^1)^2+(\dd x^2)^2 \Big ] \\&\quad + H_{2n'+2}(\dd x^3)^2+(-1)^{n'}  \Big [(\dd x^4)^2\dots+(\dd x^{8})^2 \Big ]  \\
 &\quad + (-1)^{n'+1}  \Big [(\dd x^9)^2+(\dd x^{10})^2 \Big ]   + H^{-1}_{2n'+2} \Big [ \dd x^{11} -  \Big( (-1)^{n'}   B_{2n'+2}\Big)\dd x^3 \Big]^2 
\end{split} \ee
Using the relation
\be\label{hdn2} \begin{split}
{\cal E}(z)_{2n'+2}= ({\cal E}(z)_{2n'+1} )^{-1}\qquad &\Longrightarrow \qquad H_{2n'+2}= \frac{H_{2n'+1}}{ {\cal E}_{2n'+1}\bar{\cal E}_{2n'+1} } \, \cvv \\
&\Longrightarrow \qquad B_{2n'+2}= -\frac{ B_{2n'+1}}{ {\cal E}_{2n'+1}\bar{\cal E}_{2n'+1} } \, \cvv
\end{split} \ee
 in (\ref{m63n2})  and performing  the  transformation of coordinates (\ref{tsfc}), we get the transformed metric 
\begin{align}
ds^{2 \prime}_{[-3+ 6(n'+1)]} &=   {\cal F}_{2n'-1}\bar{\cal F}_{2n'-1}   H_{2n'+1}\Big [(\dd x^1)^2+(\dd x^2)^2 \Big ] + \frac{H_{2n'+1}}{ {\cal E}_{2n'+1}\bar{\cal E}_{2n'+1} }(\dd x^{11})^2\nonumber \\&\quad +(-1)^{n'}  \Big [(\dd x^4)^2\dots+(\dd x^{8})^2 \Big ]  
 + (-1)^{n'+1}  \Big [(\dd x^9)^2+(\dd x^{10})^2 \Big ] \nonumber \\  &\quad + (\frac{H_{2n'+1}}{ {\cal E}_{2n'+1}\bar{\cal E}_{2n'+1} })^{-1} \Big [ \dd x^{3} +  \Big( (-1)^{n'+1}   \frac{B_{2n'+1}}{ {\cal E}_{2n'+1}\bar{\cal E}_{2n'+1} }\Big)\dd x^{11} \Big]^2 \, . 
\end{align}
We thus conclude that $ds^{2 \prime}_{[-3+ 6(n'+1)]} =  \dd s^2_{[3+6n']}$:
\begin{eqnarray}
\begin{array}{lll}
g^\prime _{11\, 11} &=  H^{-1}_{2n'+1} &= g _{11\, 11} \, ,   \\
 g^\prime _{3\, 3}& = H^{-1}_{2n'+1}  (  {\cal E}_{2n'+1}\bar{\cal E}_{2n'+1} )&= g _{3\, 3}  \, ,\\
g^\prime _{3\, 11} &= (-1)^{n'+1}B_{2n'+1} H^{-1}_{2n'+1}  &= g _{3\, 11}  \, , \\
g^\prime _{a\, a} & =g _{a\, a} \quad a\neq 3, 11 .& 
 \end{array}
\end{eqnarray}

We see that there is only one gravity tower, the left and the right tower of Figure \ref{fig:mappingbranegrav}b being equivalent (except for the level 0 KK-wave) as each of them contains the full wave and monopole sequences.

\setcounter{equation}{0} 
\section{Level 4 by Buscher's duality } \label{appb}
We review the Buscher formulation \cite{Buscher:1987sk,
  Buscher:1987qj}  of T-duality in 10-dimensional superstring theories
for backgrounds admitting one Killing vector\footnote{Here we are
  interested  in NS-NS backgrounds  as it was discussed originally
  \cite{Buscher:1987sk, Buscher:1987qj}, thus the formula apply not
  only for type II but also to the bosonic string in 26 dimensions. We
  will not need here the generalisation to R-R backgrounds
  \cite{Bergshoeff:1995as}.} . 
In Buscher's  construction one starts with a manifold $\cal M$ with
  metric $g_{ij}$  in the string frame, dilaton background $\phi$ and
  NS-NS background potential  $b_{ij}$. If the background is invariant
  under $x^{10}$ translations, it becomes under T-duality in the
  direction 10 $(a,b=1, \dots ,9)$ 
\begin{eqnarray}
\label{busch}
&& \widetilde{g}_{10 \, 10} = 1/g_{10\, 10}\nn \\
&&  \widetilde{g}_{10 \, a}=b_{10\, a}/g_{10\, 10} \nonumber\\
&& \widetilde{g}_{ab} = g_{ab}-(g_{10\, a}g_{10 \, b}-b_{10\, a}b_{10\, b})/g_{10 \, 10} \nonumber\\
&& \widetilde{b}_{10 \, a}=g_{10\, a}/g_{10\, 10}  \\
&& \widetilde{b}_{ab} = b_{ab}-(g_{10\, a}b_{10 \, b}-b_{10\, a}g_{10\, b})/g_{10 \, 10} \nonumber\\
&&\widetilde{\phi}= \phi-{1\over 2} \ln g_{10\, 10} \, .\nonumber
\end{eqnarray}

We apply these transformations to a M5 brane with longitudinal spacelike directions $4 \dots 8$ and longitudinal timelike direction 3 to generate the level 4 BPS solution. We smear the M5 in the directions 9, 10, 11.  We perform a double T-duality in the directions 9 and 10. 
Upon dimensional reduction along $x^{11}$ to type IIA, this M5 yields a NS5 brane smeared in the directions 9,10 with non-compact transverse directions are
1 and 2. The  smeared NS5  is given in the string frame by 
\begin{align}
\label{sns5}
 \dd s^2_{NS5}&=  -(\dd x^3)^2 +(\dd x^4)^2+\dots  +(\dd x^8)^2 +
H \left[ (\dd x^1)^2  +(\dd x^2)^2 +(\dd x^9)^2  +(\dd x^{10})^2\right ] \nonumber \\
 H(r)&=\ln r \nonumber\\
 e^{\phi}&= H^{1/2} \nonumber\\
 \widetilde{F}_{r345678}&=  \partial_r (1/ H)\, ,
\end{align}
where $ r^2\equiv (x^1)^2+(x^2)^2$ and  $\widetilde{F}_{r345678}$ is the Hodge dual of the 3-form NS field strength $\dd b$. We use the dual because the smearing procedure is always performed in the `electric' description of a  brane.

We first perform the T-duality in the direction 10.  Performing a T-duality on a direction transverse to a NS5 yields a KK5 monopole with the Taub-NUT direction in this transverse direction. Thus the T-duality on the configuration (\ref{sns5}) generate a KK5 monopole with Taub-NUT direction 10 and smeared along 9. 
To find from Buscher's rule the transformed configuration, we need the value of the
non-zero $b$ field. Using the Hodge duality and (\ref{sns5}), we find that the non-zero component of $b$ is
\begin{equation}
\label{bfield}
b_{9\, 10} =  {\rm arctg}(x^2/x^1) \equiv B\, .
\end{equation}
Using  (\ref{busch}), we find  the metric of the smeared KK5 monopole, 
 \begin{equation}\begin{split}
 \label{skk5}
 \dd s^2_{s\mathrm{KK5}}=  &-(\dd x^3)^2 +(\dd x^4)^2+\dots  +(\dd x^8)^2 +
H \left [ (\dd x^1)^2  +(\dd x^2)^2 +(\dd x^9)^2\right]\\
&+H^{-1} \left[(\dd x^{10}) - B (\dd x^9)\right]^2
\end{split}\end{equation}
and $\phi=0$, as it should. We then perform the second T-duality in the direction 9, 
applying (\ref{busch}) on (\ref{skk5}). We get (in the string frame)
\begin{align}
\label{2tns5}
\dd s^2&=  -(\dd x^3)^2 +(\dd x^4)^2+\dots  +(\dd x^8)^2 +
H \left ( (\dd x^1)^2  +(\dd x^2)^2\right)\nn \\
&\qquad  + \widetilde{H} \left((\dd x^9)^2+(\dd x^{10})^2\right) \nonumber\\
e^\phi&= \widetilde{H}^{1/2} \\
 b_{9\, 10} &=\widetilde{B}, \nonumber
 \end{align}
 where $ \widetilde{H}= H/(H^2+B^2)$ and $ \widetilde{B}=-B/(H^2+B^2)$. 
Finally uplifting the configuration (\ref{2tns5}) back to eleven dimension, we find
 \begin{align}
 \label{Ml4}
\dd s^2&=  H\widetilde{H}^{-1/3} \left( (\dd x^1)^2  +(\dd x^2)^2\right) + \widetilde{H}^{-1/3}\left(-(\dd x^3)^2 +(\dd x^4)^2+\dots  
 +(\dd x^8)^2\right) \nonumber\\
& \qquad +\widetilde{H}^{2/3}\left((\dd x^9)^2+(\dd x^{10})^2+(\dd x^{11})^2\right) \nonumber\\
 A_{9\, 10\, 11}&=\widetilde{B}\, .
 \end{align}
 
 The solution (\ref{Ml4}) of eleven-dimensional  supergravity is exactly the level 4 solution  (\ref{metric4}) obtained starting with the level 2 solution (describing the double smeared M5) and performing a Weyl reflexion $s_{\alpha_{11}}$ to go to level 4. This result is in agreement with  the interpretation of the Weyl reflexion  $s_{\alpha_{11}}$ as a double T-duality in the directions 9,10 plus the interchange of the two directions \cite{Elitzur:1997zn, Obers:1998rn, Banks:1998vs, Englert:2003zs}. 

 \setcounter{equation}{0}
\section{Structure of the $A_1^+$ U-duality group}
\label{affapp}

Our U-duality group in two non-compact dimensions is the infinite
order Weyl group ${\cal{W}}$ of an affine group. The structure of such Weyl
groups is known in terms of translations and the finite Weyl group of
the underlying finite group of rank $r$ \cite{Kac:book}. The affine Weyl
group is the semi-direct product of translations $\mathbb{Z}$ and the
finite Weyl group. For the case of affine
$A_1^+$, which features prominently in this
paper, the affine Weyl group is simply
\begin{eqnarray}
{\cal{W}} = \mathbb{Z}_2 \ltimes \mathbb{Z}.
\end{eqnarray}

To derive this fact it is useful to denote the two simple roots of 
$A_1^+$ by $\alpha_1$ and $\alpha_2$. These can be identified to
$\alpha_{11}$ and $-\alpha_{11} +\delta$   for the M2 and the M5
sequences,  and to $\lambda$ and $-\lambda + \delta$ for the KK-wave
and the KK-monopole sequences.  The simple roots $\alpha_1$ and 
$\alpha_2$ are a basis of the root lattice (they span the ladder
diagrams of Figure \ref{ger1fig} and Figure \ref{fig:mappingbranegrav}). To describe the fundamental reflexions
$s_1,s_2$ in these two  roots, it is sufficient to give their action
on a basis: 
\be \begin{split}\begin{aligned}
s_1(\alpha_1) &= -\alpha_1,&\quad s_1(\alpha_2)&= \alpha_2 + 2
\alpha_1,\\ 
s_2(\alpha_1) &= \alpha_1+2\alpha_2,&\quad s_2(\alpha_2)&= -\alpha_2.
\end{aligned}\end{split} \ee
We take the $\mathbb{Z}_2$ to be generated by the horizontal
Matzner--Misner reflexion $s_2$. The Coxeter relations for this Weyl
group are $(s_1 s_2)^\infty = \text{id}$ and $(s_2 s_1)^\infty =
\text{id}$, in other words there are no mixed relations. Therefore all
Weyl group elements are of the form
\be
\label{cox}
(s_1 s_2)^n, \quad\quad (s_2 s_1)^n, \quad\quad s_1 (s_2 s_1)^n,\quad\quad s_2 (s_1 s_2)^n, \quad\quad\text{id}\, ,
\ee
for some $m,n\ge 0$. Defining $t^n = (s_1s_2)^n$
for $n\ge 0$ and $t^n=(s_2s_1)^{(-n)}$ for $n\le 0$ one deduces for the
$\mathbb{Z}_2$ generated by $s_2$ the structure
\begin{eqnarray}
s_2 t^n = t^{-n} s_2\, ,
\end{eqnarray}
illustrating the semi-directness of the product in this case. The set
of all elements of the infinite order Weyl group is thus 
\begin{eqnarray}
{\cal{W}}=\{ t^n : n\in\mathbb{Z}\} \cup \{ s_2 t^n : n\in\mathbb{Z}\}\, ,
\end{eqnarray}
with relations
\be \begin{split}
s_2s_2 &= 1,\\
s_2 t^n &= t^{-n} s_2,\\
t^n t^m &= t^{n+m}.
\end{split} \ee
The translations $t$ act vertically in the diagrams  of Figure \ref{ger1fig}a, Figure \ref{second}, Figure \ref{fig:sequenceb}, Figure \ref{fig:mappingbranegrav}a and Figure \ref{fig:mappingbranegrav}b. They connect the points lying both on the same tower and on the same sequence.

We can act with the Weyl group on any integrable representation $\rho:A_1^+\to End(V)$, by letting \cite{Kac:book}
\begin{eqnarray}\label{weyl}
U_i = \exp(\rho (f_i)) \exp(-\rho (e_i)) \exp(\rho (f_i))\in \mathrm{GL}(V).
\end{eqnarray}
Here $ e_i$ and $f_i$ are the simple Chevalley generators. It is not hard to see that this definition implies that $U_i$ is
actually an element of $\mathrm{SO}(V)$ in the sense that
$U_iU_i^T=\text{id}_V$ for $i=1,2$ where the {\em Chevalley}
transposed element is\footnote{The
  definition 
  (\ref{weyl}) does not necessarily imply $U_i U_i=\text{id}_V$. In
  order to arrive at the proper Weyl group one has to factor out the
  subgroup generated by $U_i U_i$ from the $\mathrm{GL}(V)$
  subgroup generated by the $U_i$, see \S3.8 in~\cite{Kac:book}.}
  \be
U_i^T=\exp(\rho(e_i))\exp(-\rho (f_i))\exp(\rho (e_i)).
\ee

From (\ref{weyl}) it is straightforward to show that Weyl reflexions
are always elements of the compact subgroup of the split real form of
the associated group, so in our case this means $\mathrm{K}_{10}^+$.

  \setcounter{equation}{0}
\section{Masses and U-duality} \label{appm}
We first review the KK6-monopole mass and then derive all the masses or actions ${\cal A}_{\ell}$ of the $E_9$ multiplet in the M-theory context from the Weyl reflexions interpreted as T-dualities.
Their spacelike or timelike nature  is discussed.

\subsection{The level $3$ solution}
\label{appm3}

The KK6-monopole mass can be derived from the M5 mass using the relation between 11-dimensional supergravity and type IIA theory and T-duality. We recall that the relations between the 11-dimensional parameters $R_{11}$ and the Planck  length  $l_p$ and the string parameters $g$ and $l_s$ are (we neglect all numerical factors):
\begin{eqnarray}
\label{1110}
l_p= g^{1/3}\, l_s \qquad ,\qquad R_{11}= g \, l_s \,.
\end{eqnarray}
On the other hand if one compactifies a direction of type IIA theory  on a circle of radius R, the T-duality along this direction acts as:
\begin{eqnarray}
\label{tdualp}
R \rightarrow \frac{l_s^2}{R}\qquad , \qquad g \rightarrow \frac{g l_s}{R} \cvp
\end{eqnarray}
We start with a M5 along the directions $4 \dots 8$ and 3 is the longitudinal timelike direction. The mass of this elementary M5 is given by
\begin{equation}
\label{mm5}
M_{\ell=2}={R_4 \dots R_8 \over l_p^6}\cvp
\end{equation}
We smear this M5 along the directions 10 and 11. Reducing along the eleventh direction, using (\ref{1110}), one obtains a NS5 in type IIA then performing a T-duality along the direction 10, using (\ref{tdualp}) one get the KK5-monopole of type IIA with  Taub-NUT direction 10 
\begin{equation}
\label{mkk5}
M_{\mathrm{KK5}}={R_4 \dots R_8  R^2_{10} \over g_s^2 \, l_s^8}\, \cvp
\end{equation}
Uplifting back to eleven dimension, using (\ref{1110}), one obtains a KK6 monopole with longitudinal directions $4 \dots 8, 11$ and with Taub-NUT direction 10 .
Using (\ref{1110}) and  (\ref{tdualp}), one find the mass of this KK6-monopole\footnote{The KK6-monopole discussed here has timelike direction is 3 and the Taub-NUT direction is 10. The level 3 KK6-monopole of the gravity tower in the mapping from the 6-tower in Chapters \ref{chap:infiniteudualgroup} and 4 have timelike directions 9 and 10 and Taub-NUT direction 11. Note that in Section \ref{subsec:kkwkkm} the non-exotic KK6-monopole has its single timelike direction in 4 and Taub-NUT direction in 11. }
\begin{equation}
\label{mkk}
M_{\ell=3}={R_4 \dots R_8 R_{11} R^2_{10} \over l_p^9}\cvp
\end{equation}

\subsection{The level $4$ solution}
\label{appm4}
To go to level 4 from the M5 considered in (\ref{mm5}) we perform as in Section \ref{subsec:m5seqth} the Weyl reflexion $s_{\alpha_{11}}$ interpreted here as a double T-duality plus exchange in the direction 9 and 10. We thus smear the KK5-monopole with mass given in (\ref{mkk5}) in the direction 9 and perform a further T-duality in this direction. From 
(\ref{1110}) and (\ref{tdualp}), we find the mass of the level 4 solution \begin{equation}
\label{ml4}
M_{\ell=4}={R_4 \dots R_8 R^2_9 R^2_{10}R^2_{11}  \over l_p^{12}}\cvp
\end{equation}

\subsection{The level $5$ solution}
\label{appm5}

As in Section \ref{subsec:meseqth}, to   reach the level $5$ solution we perform,
the Weyl reflexion $s_{-\alpha_{11}+\delta}$ sending the level $1$ generator $R_{1}^{[3]}$ to the level $5$ generator $R_{5}^{[6]}$. 
  
We  first decompose the Weyl reflexion $s_{-\alpha_{11}+\delta}$  in terms of simple Weyl reflexions
which have an interpretation in terms of permutations of coordinates and double T-duality plus exchange of the directions
9 and 10. This decomposition will permit us to compute the mass of the level 5 solution and also to check that the timelike direction 9 is unaffected by $s_{-\alpha_{11}+\delta}$.  We write
\begin{eqnarray} 
s_{-\alpha_{11}+\delta}= s_i \, s_j \ldots s_k\, ,
\end{eqnarray}
where $s_i \equiv s_{\alpha _i}$ is the simple Weyl reflexion corresponding to the simple root $\alpha_i$. To perform this decomposition, we can use the following lemma (whose proof is straightforward computing both sides of the equality) :

\noindent
{\bf Lemma}
\it 
{If a real root $\gamma$ can be written as the sum of two real roots $\gamma = \gamma_1 + \gamma_2$ such that $ (\gamma_1|\gamma_2)= -1$, then the Weyl reflexion $s_{\gamma }$ can be decomposed as $s_{\gamma } = s_{\gamma_1 }\, s_{\gamma_2 }\, s_{\gamma_1 }$  .} \rm

\noindent

One may write the root $- \alpha _{11}+ \delta$ as the sum $\gamma_1 + \gamma_2 $ where $\gamma_1$ is the root associated to $R^{345}$ and $\gamma_2$ is the root associated to $R^{6\, 7\, 8}$. Using the lemma, we get
\begin{eqnarray}
s_{-\alpha_{11}+\delta} = s_{\gamma_1+ \gamma_2 } = s_{\gamma_1 }\, s_{\gamma_2 }\, s_{\gamma_1 }\, ,
\end{eqnarray}
with
\be \begin{split}
s_{\gamma_1}&= \underbrace{T_{39}\, T_{4\, 10}\, T_{5\, 11}}_{\mathcal{R}_1}\,  s_{\alpha_{11}} T_{5\, 11}\, T_{4\, 10}\, T_{39} \\
s_{\gamma_2}&=\underbrace{ T_{69}\, T_{7\, 10}\, T_{8\, 11}}_{\mathcal{R}_2}\,  s_{\alpha_{11}} T_{8\, 11}\, T_{7\, 10}\, T_{69}\, ,
\end{split}\ee
where $T_{ij}$ is the  Weyl reflexion of the gravity line permuting
the $i$ and $j$ indices namely it permutes the compactified radii $R_i
\leftrightarrow R_j$ and $ s_{\alpha_{11}}$ is the simple Weyl
reflexion with respect to $\alpha_{11}$ interpreted in type IIA as a
double T-duality in the directions 9 and 10 followed by an exchange of
the two radii.  We can directly check that there are no timelike
T-dualities when we perform the Weyl reflexion $s_{\alpha_{11}}$. The
reason is that the $T_{69}$ and $T_{39}$ replace always the time
coordinate 9 by 3 or 6 before any double T-duality.

We can now compute the mass of the level $\ell=5$ solutions by applying the Weyl reflexion $s_{-\alpha_{11}+\delta}$ to the expression (\ref{massl1}) for the  mass of the M2 with longitudinal spacelike directions 10, 11 and smeared in all spacelike directions but two 1 and 2 . The M2 mass is given by
\begin{eqnarray}
\label{massl1}
M_{M2}=  \frac{R_{10}\, R_{11}}{l_p^{3}}\, \cvp
\end{eqnarray}
 To perform the sequence of simple  Weyl reflexions in the decomposition of $s_{-\alpha_{11}+\delta}$, we insist on two points :
\begin{itemize}
\item the permutations of the radii $R_i \leftrightarrow R_j$ are always performed in $11$ dimensions ;
\item the T-dualities are performed in $10$ dimensions.
\end{itemize}
Then, when reaching any $s_{\alpha_{11}}$ in a sequences of Weyl reflexions, we reduce the 11-dimensional theory to type IIA to perform the double T-duality plus exchange of radii,  and then do an uplifting to $11$ dimensions. 
Performing successive permutations  of radii, reduction on type IIA, T-duality, exchange of radii, uplifting to $11$ dimensions, and so on, we find the mass of the level $5$ solution 
\begin{eqnarray}
\label{massl5}
M_{M2}=  \frac{R_{10}\, R_{11}}{l_p^{3}}\quad \stackrel {s_{-\alpha_{11}+\delta}} {\longrightarrow } \quad M_{\ell=5}=  \frac{(R_3\, R_4 \ldots R_8)^{2}R_{10}\, R_{11}}{l_p^{15}} \, \cvp
\end{eqnarray}

\subsection{The level 6 solution}
\label{appm6}

To obtain the mass of the level 6 solution we start from the expression for the level 5 mass  (\ref{massl5}) and we use the brane to gravity tower map (\ref{bragramap})-(\ref{bragra2}).  In the string language this amounts to
perform  3 steps: first, exchange the radii $R_9$ and $R_3$ in 11 dimensions, second,  reduce to type IIA and, using (\ref{1110})-(\ref{tdualp}), perform a double T-duality in the directions 9 and 10 plus exchange of the radii $R_9$ and $R_{11}$, finally uplift back to 11 dimensions and exchange the radii $R_9$ and $R_{11}$. The result is\footnote{Note that as the previous cases no timelike T-duality has been performed.}
 \begin{eqnarray}
\label{massl6}
M_{\ell=6}=   \frac{(R_4 \ldots R_{10})^2 \, R_{11}^3}{l_p^{18}} \, \cvp
\end{eqnarray}

\subsection{Beyond level 6}
\label{appm7}

To cope with time-like T-dualities and exotic states involving more than one timelike direction, we first compactify time to a radius $R_t$ for the non-exotic states.  
We define the action ${\cal A}= M R_t$, where $M$ is the mass. Applying the relations (\ref{1110}) and (\ref{tdualp}) to ${\cal A}_{\ell}= M_{\ell} R_t$ for the U-dualities performed in Appendix \ref{appm} leaves all computations of masses unchanged. As $\cal A$ treats space and time symmetrically it is natural to assume that U-duality can be extended to timelike T-dualities by applying the relations (\ref{1110}) and (\ref{tdualp}) directly to $\cal A$ for both space and time radii.

We can now compute  $A_{\ell}$ for any level $\ell$ from  the $\ell \leq 6$ non-exotic solutions: we use the double T-duality plus inversion of radii encoded in the Weyl transformations $s_{\alpha_{11}}$ and $s_{-\alpha_{11}+\delta}$ used in Chapter \ref{chap:infiniteudualgroup} to reach {\em all} $E_9$ BPS states. It is easy to show by recurrence that one recovers in this way for all levels the results (\ref{act3}), (\ref{act6}) and (\ref{actKK}). 

Indeed, consider first the brane towers depicted in Figure \ref{ger1fig}a.  Perform the Weyl transformation $s_{\alpha_{11}}$ and assume that  the formula (\ref{act6}) is true for level $2+3m$. Using (\ref{1110}) and (\ref{tdualp}) one obtains the equation (\ref{act3}) for $n=m+1$. Perform the Weyl transformation $s_{-\alpha_{11}+\delta}$ and assume that the formula (\ref{act3}) is valid for level  $1+3m$. Using the decomposition of $s_{-\alpha_{11}+\delta}$ described in Appendix \ref{appm5} one finds  (\ref{act6}) for $n=m+1$. The validity of assumption at the levels $\ell< 6$ yields (\ref{act3}) and (\ref{act6}).

For the [8,1]-gravity tower, assume that the formula  (\ref{act6}) is valid for all level $2+3m$ and use the gravity tower map (\ref{bragramap})-(\ref{bragra2})   translated in string language as in Appendix \ref{appm6}. One finds (\ref{actKK}) for $n=m$, QED.

\bibliographystyle{jhep}
\bibliography{Refthesis} \markboth{Bibliography}{Bibilography}
\addcontentsline{toc}{chapter}{Bibliography}
\end{document}